\documentclass[11pt,twoside]{article}
\pdfoutput=1
\usepackage{RR}

\usepackage{geometry}
\usepackage[english]{babel}
\usepackage[utf8]{inputenc}
\usepackage{amsmath,amssymb}
\usepackage{mathtools}
\usepackage{graphicx}
\usepackage[colorinlistoftodos]{todonotes}
\usetikzlibrary{automata}
\usepackage{hyperref}
\usepackage{url}
\usepackage{prettyref}
\usepackage{comment}
\usepackage{cancel}
\usepackage{pifont}
\usepackage{algorithm}
\usepackage[noend]{algpseudocode}
\usepackage{manfnt}
\usepackage{microtype}
\hypersetup{
    colorlinks = true,
    urlbordercolor = {magenta},
    citecolor = {blue},
}

\usepackage{todonotes}

\definecolor{modelica-model}{HTML}{1E7A19}
\definecolor{modelica-bool}{HTML}{0F7877}
\definecolor{modelica-eq-ite}{HTML}{7F007E}
\definecolor{modelica-der}{HTML}{0000FC}

\newcommand{\shadow}[1]{\mathsf{Sh\!}\left(#1\right)}
\newcommand{\imporder}{\mathfrak{o}}
\newcommand{\magorder}[1]{[\![#1]\!]}

\usepackage{color}
\definecolor{darkred}{rgb}{0.6,0,0}
\definecolor{darkorange}{rgb}{0.7,0.2,0}
\definecolor{darkgreen}{rgb}{0,0.3,0}
\definecolor{lightblue}{rgb}{0.3,0.3,1}

\usepackage{xcolor}
\definecolor{mygreen}{rgb}{0.0, 0.42, 0.24}
\definecolor{myblue}{rgb}{0.16, 0.32, 0.75}
\definecolor{mychoc}{rgb}{0.48, 0.25, 0.0}

\usepackage{tikz}
\usetikzlibrary{arrows}
\usetikzlibrary{decorations.markings}
\tikzstyle{box}=[rectangle, draw, minimum width=10mm, minimum height=7mm]
\newcommand{\tikzbox}{\node[box]}
\tikzstyle{vertex}=[circle, draw, inner sep=0pt, minimum size=22pt]
\newcommand{\vertex}{\node[vertex]}
\tikzset{arrowedge/.style args={#1}{postaction=decorate,
decoration={markings,mark=at position #1 with {\arrow[very thick,scale=1.5]{>}}}}}

\newcommand{\bexp}{b}
\newcommand{\flow}{\varphi}

\newcommand{\tick}{\operatorname{Tick}}

\newcommand{\sigmamethod}{$\Sigma$-method}

\newcommand{\proj}[2]{\mathbf{proj}_{#1}\!\left(#2\right)}
\newcommand{\rref}[2][]{\prettyref{#2}}
\newrefformat{sec}{Section\,\ref{#1}}
\newrefformat{def}{Definition\,\ref{#1}}
\newrefformat{defn}{Definition\,\ref{#1}}
\newrefformat{thm}{Theorem\,\ref{#1}}
\newrefformat{prop}{Prop.\,\ref{#1}}
\newrefformat{lem}{Lem.\,\ref{#1}}
\newrefformat{cor}{Corollary\,\ref{#1}}
\newrefformat{ex}{Ex.\,\ref{#1}}
\newrefformat{ax}{Appendix\,\ref{#1}}
\newrefformat{tab}{Table\,\ref{#1}}
\newrefformat{fig}{Figure\,\ref{#1}}
\newrefformat{alg}{Algorithm\,\ref{#1}}
\newrefformat{exec}{Exec.\,Sch.\,\ref{#1}}
\newrefformat{op}{Line\,\ref{#1}}
\newrefformat{rk}{Remark\,\ref{#1}}
\newrefformat{pr}{Principle\,\ref{#1}}
\newrefformat{eq}{Eq.\,\eqref{#1}}
\newrefformat{sys}{System\,\eqref{#1}}

\newtheorem{theorem}{Theorem}
\newtheorem{lemma}[theorem]{Lemma}

\newtheorem{notation}[theorem]{Notations}
\newtheorem{example}[theorem]{Example}

\newtheorem{problem}[theorem]{Problem}
\newtheorem{question}[theorem]{Question}

\newtheorem{definition}[theorem]{Definition}

\newtheorem{corollary}[theorem]{Corollary}

\newtheorem{strategy}[theorem]{Strategy}
\newtheorem{ccomment}[theorem]{Comment}

\newtheorem{assumption}{Assumption}
\newtheorem{principle}{Principle}

\newenvironment{execscheme}[1][htb]{%
    \floatname{algorithm}{Execution Scheme}
        \begin{algorithm}[#1]%
}{\end{algorithm}}
\newenvironment{proof}{\paragraph{\it Proof}}{\eproof \\

}

\def\Jacobian{\mathbf{J}}
\def\ssuccessful{complete}

\def\bfc{\mathbf{c}}
\def\bfd{\mathbf{d}}
\def\bfD{\mathbf{D}}

\def\bff{\mathbf{f}}

\def\bfB{\mathbf{B}}
\def\bfH{\mathbf{H}}
\def\bfK{\mathbf{K}}

\def\bbV{\mathbb{V}}

\def\path{\pi}

\def\DM{{\rm DM}}

\def\straight{k}

\usepackage{graphicx}
\graphicspath{{./graphics/}}

\newcommand{\prog}[1]{\mathtt{#1}}

\newcommand{\compl}[1]{#1^{\sf i}}
\newcommand{\impuls}[1]{#1^{\sf i}}
\newcommand{\nonimpuls}[1]{#1}
\newcommand{\wemph}[1]{\mbox{{\color{white}$#1$}}}
\newcommand{\remph}[1]{\mbox{{\color{red}$#1$}}}

\newcommand{\bemph}[1]{\mbox{{\color{blue}$#1$}}}
\newcommand{\blemph}[1]{\mbox{{\color{black}$#1$}}}
\newcommand{\llong}{long}
\newcommand{\transient}{transient}

\newcommand{\up}[1]{#1^\uparrow}
\newcommand{\consistency}[1]{\overline{#1}}
\newcommand{\parconsistency}[1]{{\!\!}\wemph{(}\overline{#1}\wemph{)}{\!\!}}

\renewcommand{\dot}[1]{#1'}

\renewcommand{\ddot}[1]{#1''}

\usepackage[english]{babel}

\newcommand{\tension}{{\lambda}}

\newcommand{\myparagraph}[1]{\smallskip\noindent\textit{{#1.}}}
\newcommand{\bea}{\begin{array}}
\newcommand{\eea}{\end{array}}
\newcommand{\beq}{\begin{eqnarray}}
\newcommand{\eeq}{\end{eqnarray}}
\newcommand{\beqq}{\begin{eqnarray*}}
\newcommand{\eeqq}{\end{eqnarray*}}
\newcommand{\status}{\sigma}
\newcommand{\ttimes}{}

\newcommand{\ra}{\rightarrow}
\newcommand{\la}{\leftarrow}

\newcommand{\degree}[2]{d^o_{#1}(#2)}

\newcommand{\cJ}{\mathcal{J}}
\newcommand{\val}{\nu}
\newcommand{\vval}[1]{\val_{#1}}
\newcommand{\dom}{D}

\newcommand{\ScottVar}{S-variable}
\newcommand{\ScottVars}{\ScottVar s}

\newcommand{\cC}{\mathcal{C}}

\newcommand{\cG}{\mathcal{G}}

\newcommand{\success}{b}

\newcommand{\eqq}{e}
\newcommand{\Eqqs}{E}

\newcommand{\block}{\beta}

\newcommand{\DAE}{{{\rm DAE}}}

\newcommand{\previous}[1]{#1^-}
\newcommand{\preset}[1]{{^{\bullet\!}{#1}}}
\newcommand{\wpostset}[1]{#1^\circ}

\newcommand{\postset}[1]{#1^\bullet}

\newcommand{\ppreset}[2]{^{\bullet#1}#2}
\newcommand{\ppostset}[2]{#2^{\bullet#1}}
\newcommand{\pprime}[2]{#2^{\prime#1}}

\newcommand{\nst}[1]{{^{\star{\!\!}}#1}}

\newcommand{\eproof}{\hfill$\Box$}
\newcommand{\vsmall}{\partial}

\newcommand{\st}[1]{{\mathsf{st}\!\left(#1\right)}}

\newcommand{\regularU}[2]{U_{\!#1\!}\left(#2\right)}
\newcommand{\regularV}[2]{V_{\!#1\!}\left(#2\right)}

\newcommand{\eqdef}{=_{{\rm def}}}

\newcommand{\mode}{\mu}
\newcommand{\Modes}{M}
\newcommand{\guard}{\gamma}
\newcommand{\Guards}{\Gamma}

\newcommand{\nstime}{\mathfrak{t}}

\newcommand{\nstr}{\nst{\,\bR}}
\newcommand{\nstR}{\nst{\,\bR}}
\newcommand{\nstn}{\nst{\,\bN}}
\newcommand{\nstN}{\nst{\,\bN}}
\newcommand{\continuation}[1]{\mathsf{Continuations}\left(#1\right)}

\newcommand{\squared}{\mathfrak{E}}
\newcommand{\underapprox}{\mathfrak{U}}
\newcommand{\overapprox}{\mathfrak{O}}
\newcommand{\undef}{\mbox{\sc u}}
\newcommand{\irrelevant}{\mbox{\sc i}}
\newcommand{\dt}{dt}
\newcommand{\du}{du}

\newcommand{\eend}{{\sf end}}
\newcommand{\cF}{{\cal F}}
\newcommand{\cA}{{\cal A}}
\newcommand{\cM}{{\cal M}}
\newcommand{\cN}{{\cal N}}

\newcommand{\cE}{\mathbb{E}}

\newcommand{\cT}{{\cal T}}

\newcommand{\cS}{{\cal S}}
\newcommand{\CS}{\mathbf{ER}}

\newcommand{\when}{\prog{if}}
\newcommand{\doo}{\prog{then}}
\newcommand{\mdAE}{mdAE}
\newcommand{\deltaAE}{$\vsmall\!${A\!E}}

\newcommand{\mDAE}{mDAE}

\newcommand{\bE}{\mathbb{E}}

\newcommand{\bR}{\mathbb{R}}
\newcommand{\bT}{\mathbb{T}}

\newcommand{\bfx}{\mathbf{x}}

\newcommand{\bfa}{\mathbf{a}}

\newcommand{\bN}{\mathbb{N}}

\newcommand{\bX}{\mathbb{X}}

\newcommand{\bZ}{\mathbb{Z}}

\newcommand{\bQ}{\mathbb{Q}}

\newcommand{\ttt}{\mbox{\sc t}}
\newcommand{\fff}{\mbox{\sc f}}

\newcommand{\aand}{\mathtt{and}}

\newcommand{\system}{S}

\newcommand{\ifequation}{{\tt if}-equation}
\newcommand{\whenequation}{{\tt when}-equation}

\newcommand{\deronde}{\partial}

\newcommand{\head}{{\sf h}}
\newcommand{\tail}{{\sf t}}

\newcommand{\Vars}{\mathcal{X}}

\newcommand{\atomicact}[1]{\mathsf{#1}}
\newcommand{\asgnt}[1]{{#1}}
\newcommand{\asgnf}[1]{\overline{#1}}
\newcommand{\asgnd}[1]{\asgnf{#1}}
\newcommand{\asgnw}[1]{\sharp{#1}}
\newcommand{\nxt}[1]{\postset{#1}}
\newcommand{\rnd}[1]{#1^{\deronde}}

\newcommand{\nan}{\mbox{NaN}}
\newcommand{\mlast}[1]{#1^-}

\newcommand{\REMOVE}[1]{}
\newcommand{\redundent}[1]{\Delta\gets{#1}}

\newcommand{\enbld}[2]{\mathit{Enab}({#2})}
\newcommand{\disbld}[2]{\mathit{Disab}({#2})}
\newcommand{\ndef}[1]{\mathit{Undef}({#1})}

\newcommand{\cnfa}{$\asgnt{\omega_1},\asgnt{\omega_2}$}
\newcommand{\cnfb}{$\begin{array}{c}\asgnf{\gamma},\asgnt{\omega_1},\asgnt{\omega_2}, \\
    \asgnd{e_3},\asgnd{e_4}\end{array}$}
\newcommand{\cnfc}{$\begin{array}{c}\asgnf{\gamma},\asgnt{\omega_1},\asgnt{\omega_2}, \\
    \asgnt{\tau_1},\asgnt{\tau_2},\asgnt{\nxt{\omega_1}},\asgnt{\nxt{\omega_2}}, \\
    \asgnt{\rnd{e_1}},\asgnt{\rnd{e_2}},\asgnd{e_3}, \\
    \asgnd{e_4},\asgnt{e_5},\asgnt{e_6}\end{array}$}
\newcommand{\cnfd}{$\begin{array}{c}\asgnt{\gamma},\asgnt{\omega_1},\asgnt{\omega_2}, \\
    \asgnd{e_5},\asgnd{e_6} \end{array}$}
\newcommand{\cnfe}{$\begin{array}{c}\asgnt{\gamma},\asgnt{\omega_1},\asgnt{\omega_2}, \\
    \asgnt{\tau_1},\asgnt{\tau_2},\asgnt{\nxt{\omega_1}},\asgnt{\nxt{\omega_2}}, \\
    \asgnt{\rnd{e_1}},\asgnt{\rnd{e_2}},\asgnt{\nxt{e_3}}, \\
    \asgnt{e_4},\asgnd{e_5},\asgnd{e_6} \end{array}$}
\newcommand{\cnff}{$\asgnt{\omega_1},\asgnt{\omega_2},\asgnw{e_3}$}
\newcommand{\cnfg}{$\begin{array}{c}\asgnf{\gamma},\asgnt{\omega_1},\asgnt{\omega_2}, \\
    \asgnd{e_3},\asgnd{e_4}\end{array}$}
\newcommand{\cnfh}{$\begin{array}{c}\asgnf{\gamma},\asgnt{\omega_1},\asgnt{\omega_2}, \\
    \asgnt{\tau_1},\asgnt{\tau_2},\asgnt{\nxt{\omega_1}},\asgnt{\nxt{\omega_2}}, \\
    \asgnt{\rnd{e_1}},\asgnt{\rnd{e_2}},\asgnd{e_3}, \\
    \asgnd{e_4},\asgnt{e_5},\asgnt{e_6}\end{array}$}
\newcommand{\cnfi}{$\begin{array}{c}\asgnt{\gamma},\asgnt{\omega_1},\asgnt{\omega_2}, \\
    \asgnt{e_3},\asgnd{e_5},\asgnd{e_6} \end{array}$}
\newcommand{\cnfj}{$\begin{array}{c}\asgnt{\gamma},\asgnt{\omega_1},\asgnt{\omega_2}, \\
    \asgnt{\tau_1},\asgnt{\tau_2},\asgnt{\nxt{\omega_1}},\asgnt{\nxt{\omega_2}}, \\
    \asgnt{\rnd{e_1}},\asgnt{\rnd{e_2}},\asgnt{e_3},\asgnt{\nxt{e_3}}, \\
    \asgnt{e_4},\asgnd{e_5},\asgnd{e_6} \end{array}$}

\newcommand{\lbla}{$\asgnf{\gamma};\asgnd{e_3};\asgnd{e_4}$}
\newcommand{\lblb}{$\asgnt{\gamma};\asgnd{e_5};\asgnd{e_6}$}
\newcommand{\lblc}{$\begin{array}{c}\asgnt{e_5};\asgnt{e_6}; \\ \asgnt{\rnd{e_1}};\asgnt{\rnd{e_2}}\end{array}$}
\newcommand{\lbld}{$\asgnt{\rnd{e_1}}+\asgnt{\rnd{e_2}}+\asgnt{\nxt{e_3}}+\asgnt{e_4}$}
\newcommand{\lble}{$\asgnf{\gamma};\asgnd{e_3};\asgnd{e_4}$}
\newcommand{\lblf}{$\asgnt{\gamma};\asgnd{e_5};\asgnd{e_6};\redundent{e_3}$}
\newcommand{\lblg}{$\asgnt{e_5};\asgnt{e_6};\asgnt{\rnd{e_1}};\asgnt{\rnd{e_2}}$}
\newcommand{\lblh}{$\begin{array}{c}\asgnt{\rnd{e_1}}+\asgnt{\rnd{e_2}}+ \\ \asgnt{\nxt{e_3}}+\asgnt{e_4}\end{array}$}

\newcommand{\gencodea}{\begin{array}{l}\tau_1 = 0; \; \tau_2 = 0; \\ \omega_1' = a_1(\omega_1) + b_1(\omega_1) \tau_1; \\ \omega_2' = a_2(\omega_2) + b_2(\omega_2) \tau_2\end{array}}

\newcommand{\gencodeb}{\begin{array}{l}\tau_1 = \nan; \; \tau_2 = \nan; \\ \omega_1^+ = \frac{b_2(\omega_2^-) \mlast{\omega_1} + b_1(\omega_1^-) \mlast{\omega_2}}{b_1(\omega_1^-)+b_2(\omega_2^-)}; \\ \omega_2^+ = \omega_1^+\end{array}}

\newcommand{\gencodec}{\begin{array}{l}\tau_1 = (a_2(\omega_2) - a_1 (\omega_1)) / (b_1(\omega_1) + b_2(\omega_2)); \tau_2 = - \tau_1; \\ \omega_1' = a_1(\omega_1) + b_1(\omega_1) \tau_1; \; \omega_2' = a_2(\omega_2) + b_2(\omega_2) \tau_2;  \\ \mbox{\textbf{constraint}} \; \omega_1 - \omega_2 = 0\end{array}}

\newcommand{\gencoded}{\begin{array}{l}\tau_1 = 0; \; \tau_2 = 0; \\ \omega_1^+ = \mlast{\omega_1}; \\ \omega_2^+ = \mlast{\omega_2}\end{array}}

\newcommand{\ttrue}{\mathbf{T}}
\newcommand{\ffalse}{\mathbf{F}}

\RRdate{August 2020}

\RRetitle{
\bf The Mathematical Foundations of Physical Systems Modeling Languages
} 
\RRtitle{
\bf Fondements math\'ematiques des langages de mod\'elisation de syst\`emes physiques 
} 

\RRauthor{\rm Albert Benveniste\thanks{\rm Inria-Rennes, Campus de Beaulieu, 35042 Rennes cedex, France; email: surname.name@inria.fr}, Benoit Caillaud, 
Mathias Malandain
}
\authorhead{A. Benveniste, B. Caillaud, M. Malandain}

\vfill
\RRnote{This work was supported by the FUI ModeliScale DOS0066450/00 French national grant (2018-2021) and the Inria IPL ModeliScale large scale initiative (2017-2021, \url{https://team.inria.fr/modeliscale/}).}

\RRabstract{Modern modeling languages for general physical systems, such as Modelica, Amesim, or Simscape, rely on Differential Algebraic Equations (DAEs), i.e., constraints of the form $f(\dot{x},x,u)=0$. This drastically facilitates modeling from first principles of the physics, as well as the reuse of models. In this paper, we develop the mathematical theory needed to establish the development of compilers and tools for DAE-based physical modeling languages on solid mathematical bases.

Unlike Ordinary Differential Equations (ODEs, of the form $\dot{x}=g(x,u)$), DAEs exhibit subtle issues because of the notion of \emph{differentiation index} and related \emph{latent equations}---ODEs are DAEs of index zero, for which no latent equation needs to be considered. Prior to generating execution code and calling solvers, the compilation of such languages requires a nontrivial \emph{structural analysis} step that reduces the differentiation index to a level acceptable by DAE solvers.
  
The models supported by tools of the Modelica class involve multiple modes, with mode-dependent DAE-based dynamics and state-dependent mode switching. However, multimode DAEs are much more difficult to handle than DAEs, especially because of the events of mode change. Unfortunately, the large literature devoted to the mathematical analysis of DAEs does not cover the multimode case, typically saying nothing about mode changes.  This lack of foundations causes numerous difficulties to the existing modeling tools. Some models are well handled, others are not, with no clear boundary between the two classes.  

In this paper, we develop a comprehensive mathematical approach supporting compilation and code generation for this class of languages. Its core is the \emph{structural analysis of multimode DAE systems}. As a byproduct of this structural analysis, we propose sound criteria for accepting or rejecting multimode models. 
Our mathematical development relies on \emph{nonstandard analysis}, which allows us to cast hybrid system dynamics to discrete-time dynamics with infinitesimal step size, thus providing a uniform framework for handling both continuous dynamics and mode change events.}

\RRresume{Les langages modernes de mod\'elisation de syst\`emes physiques, tels que Modelica, Amesim, ou Simscape, s'appuient sur des \'Equations Alg\'ebro-Diff\'erentielles (DAE), qui sont des contraintes de la forme $f(\dot{x},x,u)=0$. Cette approche rend naturelle la mod\'elisation directe \`a partir des principes de la physique, ainsi que la réutilisation de mod\`eles. Dans cet article, nous d\'eveloppons les bases math\'ematiques qui fondent ces langages.

Par rapport aux \'Equations Diff\'erentielles Ordinaires (ODE, de la forme $\dot{x}=g(x,u)$), les DAE sont plus compliqu\'ees, en raison des notions d'\emph{index} et d'\emph{\'equations latentes}---les ODE sont des DAE d'index z\'ero, sans \'equations latentes. Avant toute g\'en\'eration de code, une \emph{analyse structurelle} des mod\`eles est requise, afin de ramener l'index \`a des valeurs acceptables pour les solveurs de DAE (deux ou trois maximum).

Les mod\`eles accept\'es par les langages de la classe Modelica sont des DAE \emph{multi-modes}, avec une DAE diff\'erente dans chaque mode, tandis que les changements de mode sont provoqu\'es par des conditions portant sur l'\'etat et/ou le temps. Les DAE multi-mode sont beaucoup plus difficiles \`a traiter que les DAE (mono-mode). La difficult\'e principale est le traitement des changements de mode et des conditions de red\'emarrage associ\'ees. Ce point n'est pas abord\'e dans la litt\'erature math\'ematique sur les DAE multi-mode: rien n'est dit sur comment les changements de mode doivent \^etre trait\'es. Cette situation est source de probl\`emes importants pour tous les outils existants. Certains mod\`eles sont correctement trait\'es, d'autres non, sans qu'on sache bien pourquoi---et les cas probl\'ematiques ne sont pas pathologiques, mais se rencontrent naturellement. 

Dans ce travail, nous d\'eveloppons une approche coh\'erente et compl\`ete pour la compilation des DAE multi-mode. Au c{\oe}ur de celle-ci se trouve une \emph{analyse structurelle multi-mode,} qui couvre \`a la fois les dynamiques dans les divers modes, et les changements de mode. Notre approche traite des mod\`eles pr\'esentant des valeurs impulsives pour certaines variables, lors des changements de mode. Notre approche permet d'accepter ou de rejeter des mod\`eles sur des crit\`eres bien fond\'es (mod\`ele sur- ou sous-d\'etermin\'e). Ces travaux eussent \'et\'e impossibles sans le recours \`a l'\emph{analyse non-standard.} Gr\^ace \`a elle, nous pouvons r\'einterpr\'eter les dynamiques en temps continu comme du temps discret \`a pas infinit\'esimal, ce qui fournit une vision uniforme de toute la trajectoire---en mode continu ou en changement de mode.
}

\RRmotcle{analyse structurelle, \'equations algébro-diff\'erentielles
  (DAE), syst\`emes multi-mode, mod\`eles \`a structure variable, analyse non-standard}
\RRkeyword{structural analysis, differential-algebraic equations
  (DAE), multi-mode systems, variable-structure models, nonstandard analysis}
	\RRprojet{Hycomes}  
\RCRennes
\RRNo{9334}

\begin{document}
\makeRR
\clearpage
\tableofcontents
\clearpage
\section{Introduction}\label{sec:intro}
\subsection{Overall motivations}
\emph{Multimode DAE systems} constitute the mathematical framework supporting the  physical modeling of systems possessing different \emph{modes}. Each mode exhibits a  different dynamics, captured by \emph{Differential Algebraic Equation}s (DAEs). {Multimode DAE systems} are the underlying framework of `object-oriented' modeling languages such as  \href{https://www.modelica.org/documents/ModelicaSpec33.pdf}{Modelica}, \href{https://www.plm.automation.siemens.com/global/en/products/simcenter/simcenter-amesim.html}{Amesim}, or \href{https://mathworks.com/products/simscape/}{Simscape}. Corresponding models can be represented as systems of guarded equations of the form
\beq
\mbox{if }  \guard_j(x_i \mbox{'s and derivatives}) 
\mbox{ then }  f_j(x_i \mbox{'s and derivatives}) = 0
\label{ow8ugthoerui}
\eeq
where $x_i,i=1,\dots,n$ denote the system variables, and, for $j=1,\dots,m$, $\guard_j(\dots)$ is a Boolean combination of predicates guarding the differential or algebraic equation $f_j(\dots)=0$. The meaning is that, if $\guard_j$ has the value $\ttt$ (the constant true), then equation $f_j(\dots)=0$ has to hold; otherwise, it is discarded. In particular, when all the predicates have the value $\ttt$, one obtains a \emph{single-mode} DAE, that is, a classical DAE defined by the set of equations 
\beq
f_j(x_i \mbox{'s and derivatives}) = 0
\label{eopiruthywhy}
\eeq
where $i{=}1,\dots,n$ and $j{=}1,\dots,m$.
When all $f_j$'s have the special form $x_j'{-}g_j(x_1,\dotsc,x_n)$, one recovers the Ordinary Differential Equations (ODE)  $x_j'{=}g_j(x_1,\dotsc,x_n).$   DAEs are a strict generalization of ODEs, where the so-called \emph{state variables} $x_1,\dotsc,x_n$ are implicitly related to their time derivatives $x'_1,\dotsc,x'_n$. This modeling framework is fully compositional, since systems of systems of equations of the form (\ref{ow8ugthoerui}) are just systems of equations of the form (\ref{ow8ugthoerui}), with no restriction.

DAE systems (having a single mode) are well understood. With comparison to ODE systems, a new difficulty arises with the notion of \emph{differentiation index}, introduced in the late 1980's \cite{Petzold82,CampbellGear1995}. For simplicity, in this introduction, we discuss it for the particular case in which System~(\ref{eopiruthywhy}) involves all derivatives $\dot{x_i}$ but no higher-order derivative. The problem addressed by this notion of differentiation index relates to the Jacobian matrix $\Jacobian$ associated to System (\ref{eopiruthywhy}), defined by  $\Jacobian_{ij}=\partial{f_j}/\partial{\dot{x_i}}$. If this Jacobian matrix is invertible (requiring, in particular, $m{=}n$), then all the derivatives $\dot{x_i}$ are uniquely determined as functions of the ${x_i}$, so that System (\ref{eopiruthywhy}) is equivalent to an ODE system.
If, however, matrix $\Jacobian$ is singular, then System~(\ref{eopiruthywhy}) as such is no longer equivalent to an ODE system. This situation can occur, even for a square DAE system possessing a unique solution for any given initial conditions. This situation makes it difficult to design DAE solvers that work for any kind of DAE.

 Now, $f_j{=}0$ in (\ref{eopiruthywhy}) implies $\frac{d}{dt}f_j{=}0$, revealing that such \emph{latent equations} come implicitly with System (\ref{eopiruthywhy}). Adding latent equations to the DAE system does not change the system solutions, but it can bring additional constraints on highest-order derivatives of the original system states---on the other hand, this can introduce ``spurious'' derivatives (of higher order than in the original system), which are to be eliminated. Adding more latent equations while eliminating spurious derivatives eventually leads to highest-order derivatives of the original system states being uniquely determined as functions of the lower order state derivatives. The \emph{differentiation index} \cite{Petzold82,CampbellGear1995} is the smallest integer $k$ such that it is enough to differentiate with respect to time every equation $f_j{=}0$ up to order at most $k$ for the above situation to occur.
 
  DAE solvers exist for DAE of index $1$, $2$, or $3$~\cite{MattssonSoderlin1993}. An efficient and popular approach, called the \emph{dummy derivatives} method~\cite{MattssonSoderlin1993}, consists in reducing the index of the DAE system to $1$, producing in addition a form in which the algebraic constraints are all preserved by the solver---this no longer holds if the index is reduced to zero, leading to an ODE. \emph{Structural analysis} methods have been proposed~\cite{pantelides,Pryce01} which perform index reduction by only exploiting the bipartite graph associated to the different equations and variables of the system. Structural analyses scale up much better and provide valid results outside exceptional values for the coefficients arising in the system equations~\cite{CampbellGear1995}; structural analysis with dummy derivatives is implemented in most Modelica tools. Efforts are still ongoing to improve the efficiency and range of applicability of these methods, but one can say that DAE systems are now reasonably well understood.

In contrast, the handling of mode changes is much less mature. No structural analysis exists for mode changes. Due to discontinuities in system trajectories, the notion of differentiation index does not apply to multimode DAE systems. The notion of solutions of such systems is not even well understood, except for some subclasses of models, such as \emph{semi-linear systems}~\cite{DBLP:series/lncs/BenvenisteCEGOP19}. Still, modeling languages exist that support multimode DAE systems, e.g., \href{https://www.modelica.org/documents/ModelicaSpec33.pdf}{Modelica}, \href{https://standards.ieee.org/findstds/standard/1076.1-2007.html}{VHDL-AMS}, and the proprietary languages \href{https://www.plm.automation.siemens.com/global/en/products/simcenter/simcenter-amesim.html}{Amesim} and \href{https://mathworks.com/products/simscape/}{Simscape}. However, the lack of mathematical understanding of mode changes can result  in a spurious handling of some physically meaningful models. Indeed, with the exception of the recent work~\cite{DBLP:series/lncs/BenvenisteCEGOP19}, the class of ``safe'' models (well supported by the considered tool) is never characterized, thus leading to a ``try-and-see'' methodology. This situation motivates our work.

Establishing the mathematical foundations for compilers\footnote{By \emph{compilation}, we mean here the suite of symbolic analyses and transformations that must be performed prior to generating simulation code.} of multimode DAE systems raises non-classical difficulties. To substantiate this claim, let us compare the following three subject matters:
\begin{enumerate}
	\item \label{leriughoi} Developing 
	formal verification for a given class of hybrid systems \cite{alur-et-al93,Platzer2018};
	\item \label{lsiuhpuih} Developing existence/uniqueness results and discretization schemes, for a given class of (single-mode or multimode) DAE systems \cite{Ascher1998};
	\item \label{leguihpiug} Developing the mathematical foundations for multiphysics/multimode DAE systems modeling languages (our focus).
\end{enumerate}
For subject~\ref{leriughoi}, restrictions on the class are stated, under which the proposed algorithms are proved correct, and their complexity is analyzed. Similarly, for subject~\ref{lsiuhpuih}, assumptions are stated on which discretization schemes are proved correct, and convergence rates can be given. In both cases, assumptions are formulated and it is the responsibility of the user to check their validity when performing verification or applying discretization schemes.

In contrast, for subject~\ref{leguihpiug}, one cannot expect a compiler to check the existence/uniqueness of solutions of a given multimode DAE system. Worse, we cannot expect the user to check this as part of his/her model design activity. In fact, the compiler must handle \emph{any} submitted model, and it has the responsibility for accepting or rejecting a model on the sole basis of syntactic or symbolic (but never numerical) analyses. This is a demanding task that we claim is not well addressed today.

\subsection{Some illustration examples}
\label{reuyerkyluiy}
To illustrate the above discussion, we now review two specific examples of multimode DAE systems.

\subsubsection{An ideal clutch}
\label{lireutyliuel}
This clutch is depicted in \rref{fig:clutch}. It is a simple, idealized clutch involving two rotating shafts where no motor or brake are connected. More precisely, we assume that this system is closed, with no interaction other than explicitly specified.

\begin{figure}[h]
  \centering
    \includegraphics[height=2.5cm]{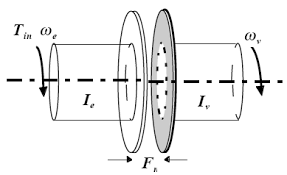}
  \caption{
  An ideal clutch with two shafts.}
 \label{fig:clutch}
\end{figure}

We provide hereafter in~(\ref{sys:coupledshafts}) a model for it, complying with the general form~(\ref{ow8ugthoerui}). The dynamics of each shaft $i$ is modeled by \mbox{$\omega_i'=f_i(\omega_i,\tau_i)$} for some, yet unspecified, function $f_i$, where $\omega_i$ is the angular velocity, $\tau_i$ is the torque applied to shaft $i$, and $\omega_i'$ denotes the time derivative of $\omega_i$. Depending on the value of the input Boolean variable $\gamma$, the clutch is either engaged ($\gamma = \ttt$, the constant ``true'') or released ($\gamma = \fff$, the constant ``false''). When the clutch is released, the two shafts rotate freely: no torque is applied to them ($\tau_i=0$).  When the clutch is engaged, it ensures a perfect join between the two shafts, forcing them to have the same angular velocity ($\omega_1-\omega_2=0$) and opposite torques ($\tau_1+\tau_2=0$). Here is the model:
\begin{equation}
\left\{
\bea{rllcc}
&&\omega'_1=f_1(\omega_1,\tau_1) &&(\eqq_1) \\
&&\omega'_2=f_2(\omega_2,\tau_2) &&(\eqq_2) \\
\when\;\guard&\doo&\omega_1-\omega_2=0 &&(\eqq_{3}) \\
&\aand&\tau_1+\tau_2=0 &&(\eqq_{4}) \\
\when\;\prog{not}\;\guard&\doo&\tau_1=0 &&(\eqq_{5}) \\
&\aand&\tau_2=0 &&(\eqq_{6}) \\
\eea
\right.
\label{sys:coupledshafts}
\end{equation}
When $\guard=\ttt$, equations $(\eqq_3,\eqq_4)$ are active and equations $(\eqq_5,\eqq_6)$ are disabled, and vice-versa when $\guard=\fff$. If the clutch is initially released, then, at the instant of contact, the relative speed of the two rotating shafts jumps to zero; as a consequence, an impulse is expected on the torques. In addition, the model yields an ODE system when the clutch is released, and a DAE system of index $1$ when the clutch is engaged, due to equation $(\eqq_3)$ being active in this mode. 

Note that the condition $\tau_1+\tau_2=0$ is an invariant of the system in both modes, which reflects the preservation of the angular momentum due to the assumption that the system is closed. Model (\ref{sys:coupledshafts}) is one possible specification. One could have instead put $\tau_1+\tau_2=0$ outside the scope of any guard, thus making explicit that this is an invariant---this is actually the form adopted for the Modelica model of \rref{fig:clutch_modelica}; then, mode $\guard=\fff$ only states $\tau_1=0$. Both models are equivalent. It turns out that our approach yields the same generated code for both source codes. 

\paragraph{The clutch in Modelica}
\begin{figure}[ht]
\footnotesize
\begin{tabular}{cc}
\begin{minipage}{6cm}\tt
{\color{modelica-model}{model}} ClutchBasic \\
\wemph{to}  {\color{modelica-model}{parameter}} Real w01=1;\\
\wemph{to}   {\color{modelica-model}{parameter}} Real w02=1.5;\\
\wemph{to}   {\color{modelica-model}{parameter}} Real j1=1;\\
\wemph{to}   {\color{modelica-model}{parameter}} Real j2=2;\\
\wemph{to}   {\color{modelica-model}{parameter}} Real k1=0.01;\\
\wemph{to}   {\color{modelica-model}{parameter}} Real k2=0.0125;\\
\wemph{to}   {\color{modelica-model}{parameter}} Real t1=5;\\
\wemph{to}   {\color{modelica-model}{parameter}} Real t2=7;\\
\wemph{to}   Real t(start=0, fixed={\color{modelica-bool}{true}});\\
\wemph{to}   Boolean g(start={\color{modelica-bool}{false}});\\
\wemph{to}   Real w1(start = w01, fixed={\color{modelica-bool}{true}});\\
\wemph{to}   Real w2(start= w02, fixed={\color{modelica-bool}{true}});\\
\wemph{to}   Real f1;  Real f2;
\end{minipage} & \begin{minipage}{6cm}\tt
{\color{modelica-eq-ite}equation} \\
\wemph{to}   {\color{modelica-der}der}(t) = 1;\\
\wemph{to}   g = (t >= t1) {\color{modelica-bool}and} (t <= t2);\\
\wemph{to}   j1*{\color{modelica-der}der}(w1) = -k1*w1 + f1;\\
\wemph{to}   j2*{\color{modelica-der}der}(w2) = -k2*w2 + f2;\\
\wemph{to}   0 = {\color{modelica-eq-ite}if} g {\color{modelica-eq-ite}then} w1-w2 {\color{modelica-eq-ite}else} f1;\\
\wemph{to}   f1 + f2 = 0;\\
{\color{modelica-model}{end}} ClutchBasic; 
\end{minipage}
\end{tabular}
\caption{{Modelica code for the idealized clutch.}}
\label{fig:clutch_modelica}
\end{figure}


\begin{figure}[ht]
  \centering\includegraphics[width=13.5cm]{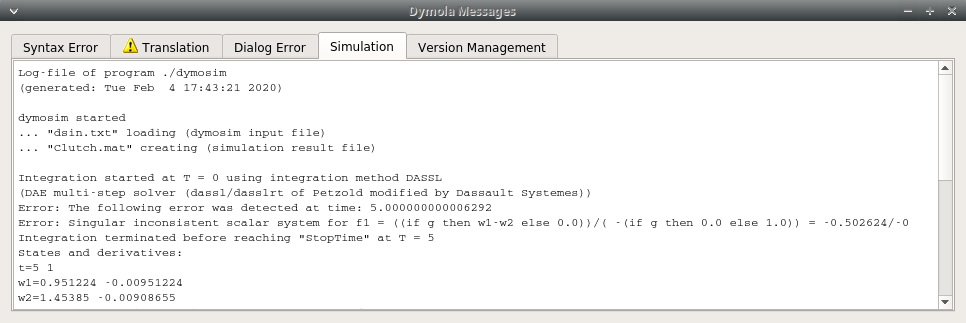}
  \caption{Division by zero exception occuring when simulating the Ideal Clutch Modelica model.} \label{fig:clutch:exception}
\end{figure}

\begin{figure}[ht]
  \centering\includegraphics[width=8cm]{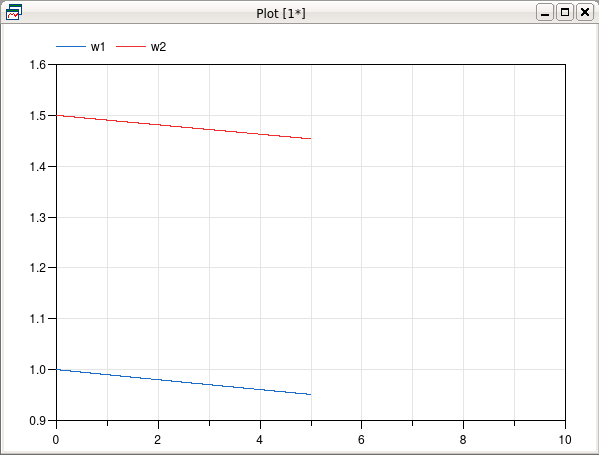}
  \caption{Trajectory of the Ideal Clutch Modelica model: it stops at $t=5s$, when the exception occurs.} \label{fig:clutch:trajexc}
\end{figure}

Figure~\ref{fig:clutch_modelica} details the Modelica model of the
Ideal Clutch system, with two shafts interconnected by an idealized
clutch. It is a faithful translation in the Modelica language of the
mDAE~(\ref{sys:coupledshafts}), with the exception that the two
differential equations have been linearized. Also, the trajectory of the input guard $\gamma$ (here called \verb!g!) has been fully defined: it takes the value $\ttt$ between instants $t_1$ and $t_2$, and $\fff$ elsewhere. The model is deemed to be
structurally nonsingular by the two Modelica tools we had the
opportunity to test: OpenModelica and Dymola. However, none of these
tools generates correct simulation code from this model.
Indeed, simulations fail precisely at the instant when the clutch
switches from the uncoupled mode (\verb|g=false|) to the coupled one
(\verb|g=true|). This is evidenced by a division by zero exception,
as shown in Figures~\ref{fig:clutch:exception} and~\ref{fig:clutch:trajexc}.

The root cause of this exception is that none of these tools
performs a multimode structural analysis. Instead, the
structure of the model is assumed to be invariant, and the structural
analysis implemented in these tools is a Dummy Derivatives method,
which is proved to be correct only on single-mode DAE systems; it is
possibly correct on multimode systems, but under the stringent assumption that the
model structure (including, but not limited to, its differentiation index) is independent of the mode. The consequence is that the
structural analysis methods implemented in these tools do not detect
that the differentiation index jumps from $0$ to $1$ when the shafts
are coupled, and that the structure is not invariant. 
The division by
zero results from the pivoting of a linear system of equations that
becomes singular when \verb|g| becomes equal to \verb|true|.

\paragraph{The clutch in Mathematica}

\begin{figure}[ht]
  \centering\includegraphics[width=10cm]{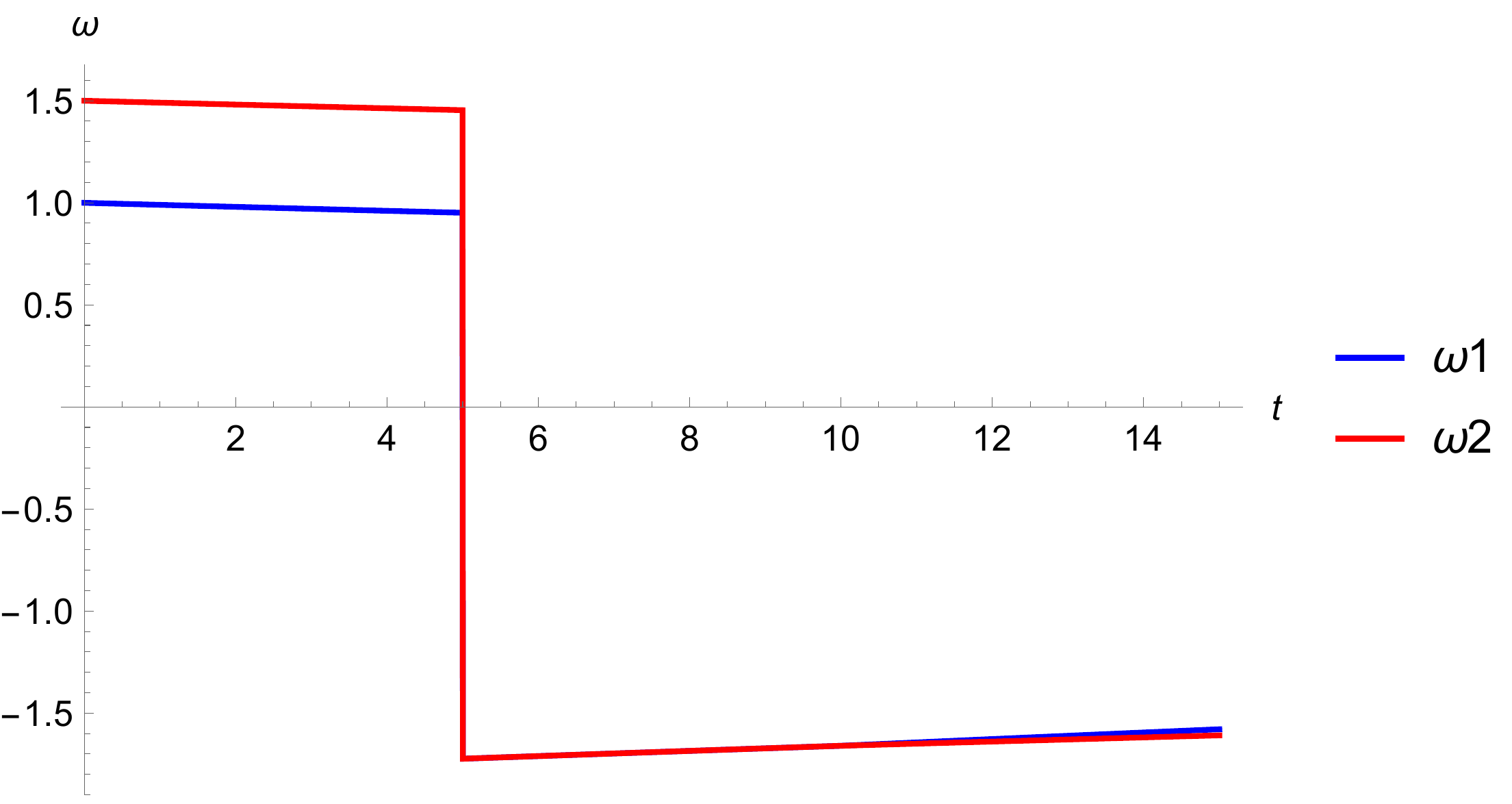}
  \caption{Solution of the Ideal Clutch model, computed by Mathematica.} \label{fig:clutch:mathematica}
\end{figure}

DAE systems can be solved with the Mathematica ``Advanced Hybrid \& DAE''
toolbox.\footnote{\url{https://www.wolfram.com/mathematica/new-in-9/advanced-hybrid-and-differential-algebraic-equations/}} This
library is also advertised as supporting multimode DAE
systems. Figure~\ref{fig:clutch:mathematica} shows the solution of the
Ideal Clutch system, computed by Mathematica. It can be seen that the
angular momentum of the mechanical system is not preserved, in contradiction with the physics.
Also, a small change in the velocities before the change results in totally different values for restart, although the system is not chaotic by itself.
A likely explanation is that Mathematica regards the mode switching as a consistent initialization problem, with one degree of freedom. Indeed, any choice of angular velocities satisfying the constraint $\omega_1 = \omega_2$ is a consistent initial state for the DAE in the engaged mode.

The conclusion of these two experiments is that, clearly, some fundamental study is needed to ensure a correct handling of events of mode change by the Modelica tools. Smoothing the `\texttt{if then else}' equation could help solving the above problem, but requires a delicate and definitely non-modular tuning, as it depends on the different time scales arising in the system. We believe that, as the tools reputedly support multimode DAE models, they should handle them correctly.

\paragraph{Our objective}
For this example, \emph{our objective is to be able to compile \rref{sys:coupledshafts} and produce correct simulations for it, without any need for smoothing the mode change}. This objective is addressed in \rref{sec:simpleclutch}.

\subsubsection{A Cup-and-Ball game}
\label{leriugfehliu}
\begin{figure}[h]
\centerline{\includegraphics[width=4cm]{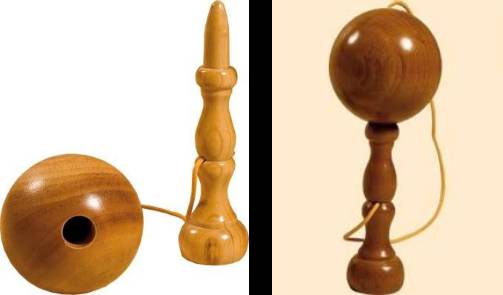}}
\caption[yres]{ The Cup-and-Ball game.}
\label{fig:cupandball}
\end{figure}
We sketch here a multimode extension of the popular example of the pendulum in Cartesian coordinates~\cite{pantelides}, namely the planar Cup-and-Ball game, whose 3D version is illustrated by \rref{fig:cupandball}. A ball, modeled by a point mass, is attached to one end of a rope, while the other end of the rope is fixed, to the origin of the plane in the model. The ball is subject to the unilateral constraint set by the rope, but moves freely while the distance between the ball and the origin is less than its actual length. Again, we assume that the system is closed and subject to no interaction other than specified. Using Cartesian coordinates $x$ and $y$, the equations of the planar Cup-and-Ball game are the following:
\beq
\left\{\bea{lll}
 0= \ddot{x}+{\tension}x && (\eqq_1) \\
 0= \ddot{y}+{\tension}y+g  && (\eqq_2) \\
0\leq{L^2}{-}(x^2{+}y^2)    && (\kappa_1) \\
0\leq\tension   && (\kappa_2) \\
0=\left[{L^2}{-}(x^2{+}y^2)\right]\times\tension   && (\kappa_3) 
\eea\right.
\label{reiutyekrtyk}
\eeq
where the unknowns (also called \emph{dependent variables}) are the position $(x,y)$ of the ball in Cartesian coordinates and the rope tension $\tension$.

The subsystem $(\kappa_1,\kappa_2,\kappa_3)$ expresses that the distance of the ball from the origin is less than or equal to $L$, the tension is nonnegative, and one cannot have a distance less than $L$ and a nonzero tension at the same time.
This is known as a \emph{complementarity condition}, written as
~\mbox{$
0\leq{L^2}{-}(x^2{+}y^2) \perp \tension \geq{0}
$}~
in the \emph{nonsmooth systems} literature~\cite{Acary08}, and is an adequate modeling of ideal valves, diodes, and contact in mechanics. 
Constraints $\kappa_1$ and $\kappa_2$ are unilateral, which does not belong to our framework of guarded equations. Therefore, using a technique communicated to us by H. Elmqvist and M. Otter, we redefine the graph
of this complementarity condition as a parametric curve,
represented by the following three equations:
\[\bea{rcl}
s &=& \prog{if}\; \guard \;\prog{then} -\tension \;\prog{else} \;{L^2}{-}(x^2{+}y^2)
\\
0 &=& \prog{if}\;\guard\;\prog{then} \;{L^2}{-}(x^2{+}y^2) \; \prog{else} \; \tension
\\
\guard &=& [s\leq 0]
\eea\]
which allows us to rewrite the complementarity condition $(\kappa_1,\kappa_2,\kappa_3)$ as follows:
\beq
\left\{\bea{rll}
& 0= \ddot{x}+{\tension}x & (\eqq_1) \\
& 0= \ddot{y}+{\tension}y+g  & (\eqq_2) \\
& \guard= [s\leq{0}]  & (\straight_0) \\
\when \; \guard\; \doo& 0={L^2}{-}(x^2{+}y^2)   & (\straight_1) \\
\prog{and}& 0=\tension+s   & (\straight_2) \\
\when \;\prog{not}\; \guard\; \doo& 0=\tension   & (\straight_3) \\
\prog{and}& 0=({L^2}{-}(x^2{+}y^2))-s   & (\straight_4) \\ 
\eea\right.
\label{riuytuikdastf}
\eeq
Unsurprinsingly (considering the reasons for the wrong handling of the clutch example), the Modelica tools also fail to handle this model correctly.

Some new issues emerge from Example~(\ref{riuytuikdastf}). First,  subsystem $(\kappa_1,\kappa_2,\kappa_3)$ of~(\ref{reiutyekrtyk}) leaves the impact law at mode change insufficiently specified: it could be fully elastic, fully inelastic, or inbetween. Second, (\ref{riuytuikdastf}) exhibits a logico-numerical fixpoint equation in $s$, which we regard as problematic. While the Jacobian matrix yields a regularity criterion for systems of smooth algebraic equations, logico-numerical systems of equations are unfriendly: there is no simple criterion for the existence and/or uniqueness of solutions. Note that a sensible modification of model (\ref{riuytuikdastf}) would consist in replacing, in $(k_0)$, $s$ by its left-limit $s^-(t)\eqdef\limsup_{u\nearrow{t}}s(u)$.\footnote{The left-limit is available in Simulink under the `state port' construct, whereas the $\prog{pre}$ keyword in Modelica serves the same purpose.}

\paragraph{Our objectives}
For this example, our objectives are \emph{to address the above difficulties at compile time, and to generate correct simulation code once structural problems are fixed.}
These objectives are addressed in Section\,\ref{loeruighlrigtuh}.
%
\subsubsection{A Westinghouse air brake} 
\label{elriuftrheliu}
\begin{figure}[h]
\vspace*{-2mm}
\centerline{\includegraphics[height=7cm,width=9cm]{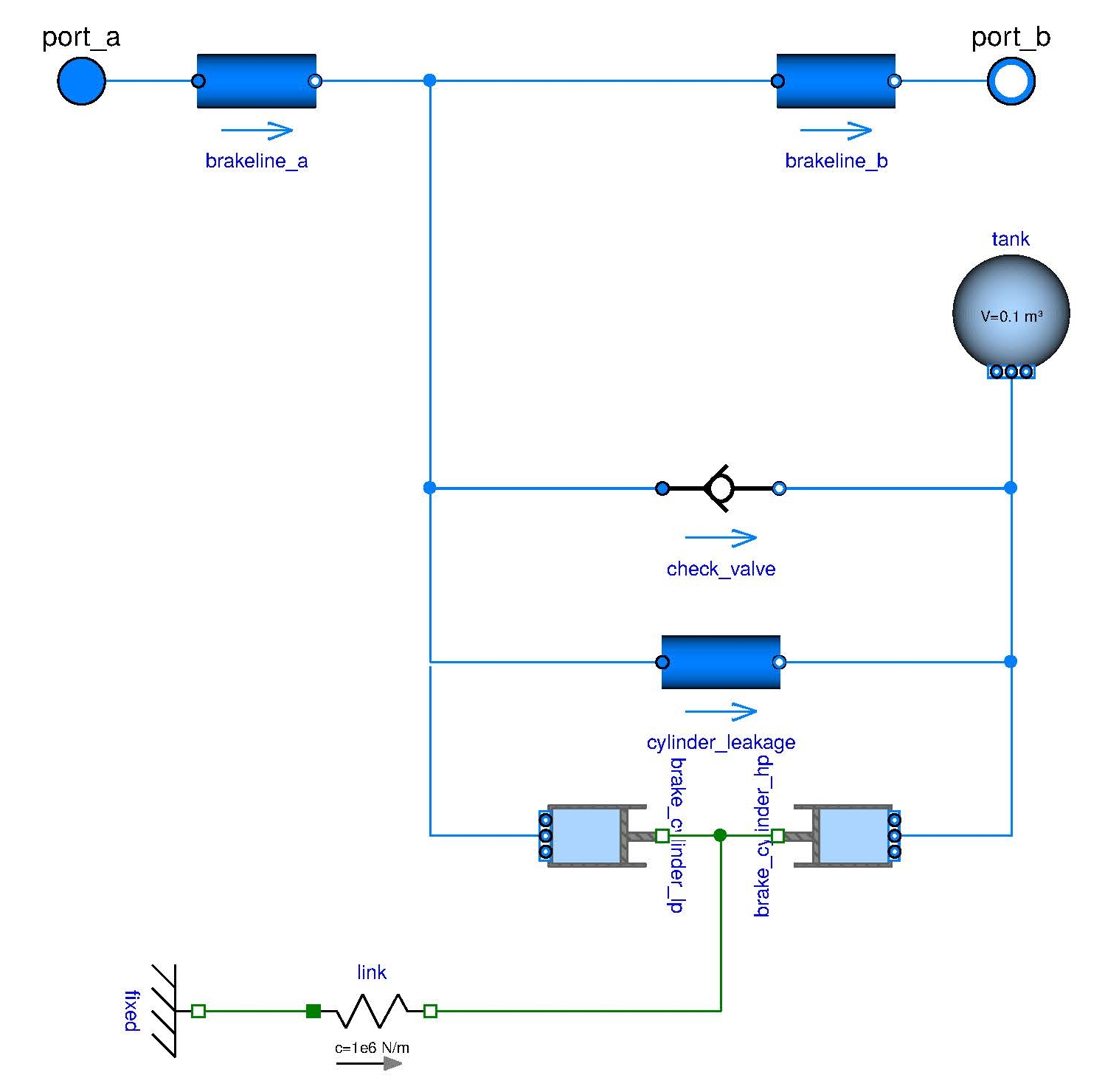}}
\vspace*{-2mm}
\caption{ Westinghouse brake for one car: schematics.}
\label{fig:railcardiagram}
\medskip
\small
%
{$\bea{rl}
    p_b  :& \mbox{connector pressure (a control)} \\
    p_{bn} :& \mbox{next connector pressure} \\
    p_r  :& \mbox{railcar pressure} \\
    p_t  :& \mbox{reservoir pressure} \\
    f_b  :& \mbox{mass flow brake connector} \\
    f_{bn} :& \mbox{mass flow next brake connector} f_{bn}=0\\
    f_v  :& \mbox{mass flow one way valve} \\
    f_l  :& \mbox{leakage mass flow} \\
    f_{cl} :& \mbox{mass flow low press side of cylinder} \\
    f_{ch} :& \mbox{mass flow high press side of cylinder} \\
    f_t  :& \mbox{mass flow reservoir} \\
    b   :& \mbox{brake force} \\
    x   :& \mbox{position of the piston} \\
~ &~\\
(m_1,m_2):& \mbox{Mass balance equations} \\
(f_1\!\!-\!\!f_3):& \mbox{Flow equations} \\
(r):& \mbox{Dynamics of the reservoir} \\
(p_1,p_2):& \mbox{Dynamics of the piston} \\
(l_1,l_2):& \mbox{Brake force and mech link movement} \\
(\nu_1\!\!-\!\!\nu_3):& \mbox{One way ideal valve}
\eea$\hfill$\left\{\bea{lc}
0=f_b - f_v - f_{cl} + f_l &(m_1)\\
0=f_v - f_{ch} - f_l - f_t &(m_2)\\
0=f_b - F_1.\flow(p_b-p_r) &(f_1)\\
0=F_1.\flow(p_{bn}-p_r) &(f_2)\\
0=f_l - F_2.\flow(p_t-p_r) &(f_3)\\
0=f_t - \frac{\rho.V}{P_0}.\dot{p_t} &(r)\\
0= \rho.S.(\dot{x}.p_r + (x-L).\dot{p_r}) + P_0.f_{cl} &(p_1)\\
0=\rho.S.(\dot{x}.S.p_t + x.\dot{p_t}) - P_0.f_{ch} &(p_2)\\
0=S.(p_t-p_r) - b &(l_1)\\
0= K.x-b &(l_2)\\
0\leq p_r-p_t & (\nu_1) \\
0\leq f_v & (\nu_2) \\
0=f_v\times(p_r-p_t) & (\nu_3) 
\eea\right.$}
\caption{ Model corresponding to the schematics of \rref{fig:railcardiagram}; $p_b$, the control, and $f_{bn}$ determined by the boundary condition $f_{bn}{=}0$, are known.}
	\label{fig:erpgfiou}
\end{figure}
The second example is a simplified model of the railway air brake patented by George Westinghouse in 1868. A railway car is equipped with the pneumatic system whose schematics is shown in \rref{fig:railcardiagram}. It has two pneumatic ports that are interconnected when several cars are coupled to form a train; in this simple example, only one car is considered, with the boundary conditions shown (\verb!control! is an input pressure). 
As for the clutch example, this program is not correctly executed by the Modelica 3.3 engine.
Using the smoothed version of a valve from the Modelica library of predefined components solves the problem. This, however, requires a careful tuning in accordance with the time scales of the whole system.

A simplified model corresponding to the schematics of \rref{fig:railcardiagram} is given in \rref{fig:erpgfiou}, where we have taken the boundary conditions into account\,---\,the equations are simplified with reference to hydraulics, but the model structure (states and index) are correct. The model of the ideal valve is given by the subsystem $(\nu_1,\nu_2,\nu_3)$. Inequalities $(\nu_1,\nu_2)$ express nonnegativity conditions satisfied by the pressure difference and the flow through the valve. Equality $(\nu_3)$ expresses that, either the valve is open (no pressure drop), or it is closed (no air pass).  This is known as a \emph{complementarity condition}, written 
\[
0\leq p_r-p_t\;\perp\;f_v\geq{0}
\]
in the \emph{nonsmooth systems} literature\,\cite{Acary08}, and is the adequate modeling of systems such as ideal valve, diode, and contact in mechanics. 
The two constraints $(\nu_1,\nu_2)$ are unilateral, which does not belong to our framework of guarded equations. Therefore, using a technique communicated to us by H. Elmqvist and M. Otter, we rewrite the complementarity condition $(\nu_1,\nu_2,\nu_3)$ by using an auxiliary variable $s$ as follows:
\beq
\left\{\bea{rlc}
&\guard=[s\leq 0] &(v_0)\\
\when \; \guard\; \doo& 0=p_r-p_t &(v_1)\\
\prog{and}& 0=f_v+s &(v_2)\\
\when \; \prog{not}\; \guard\; \doo& 0=f_v &(v_3)\\
\prog{and}& 0=p_r-p_t-s&(v_4)
\eea\right. 
\label{eroiguhopiuh}
\eeq
The system of \rref{fig:erpgfiou}, where $(\nu_1,\nu_2,\nu_3)$ is replaced by $(v_0,\dots,v_4)$ in (\ref{eroiguhopiuh}), has the form (\ref{ow8ugthoerui}).
In this model, the subsystem collecting equations $(v_0,\dots,v_4)$ constitutes a logico-numerical fixpoint equation, in that $\guard$ depends on $s$ and the equations determining $s$ are guarded by $\guard$. Given that no solver exists that reliably supports this kind of fixpoint equation, we think it should be rejected at compile time, with a proper diagnosis. Still, our model is physically meaningful. 

\paragraph{Our objectives:}
For this example, \emph{our objectives are to be able to reject, with a proper diagnosis, the model of Figure\,$\ref{fig:erpgfiou}$ (using $(\ref{eroiguhopiuh})$ for the valve), and to correct it so that the corrected model simulates as expected}. These objectives are addressed in \rref{sec:railcar}.

\subsection{Related work}

There exists a huge litterature devoted to the multi-physics DAEs (having a single mode); see, e.g., \cite{Brenan1996} as a basic reference. Many references related to the Modelica language can also be found on the website of the Modelica organization.\footnote{\url{https://www.modelica.org/publications}} We focus this review on works addressing multimode DAE systems with an emphasis on the compilation and code generation, for related models.

A large body of work was developed in the domain of multi-body systems with contacts where impulses can occur; see for instance~\cite{PfeifferGlocker2008,Pfeiffer2012} for an overview of this domain. The current research focuses on the handling of contacts between many bodies using time-stepping methods. Contact equations are usually expressed on the velocity or acceleration, and since the constraints on position level are not explicitly taken into account, a drift typically occurs. An exception is \cite{Schoeder2013}, where constraints on position level are enforced. It is not obvious how the specialized multimode methods for multi-body systems can be generalized to any type of multimode DAE.
For electrical circuits with idealized switches, like ideal diodes, impulses can also occur. Again, a large literature is available, concentrating mostly on specialized electrical circuits, such as piecewise-linear networks with idealized switches~\cite{Heemels2002}.
However, it is not obvious how to generalize methods in this restricted area to general multimode DAEs.

The article by Mehrmann \emph{et al.}~\cite{Mehrmann2009} contains interesting results regarding numerical techniques to detect chattering between modes. It however assumes that consistent reset values are explicitly given for each mode. Such an assumption does not hold in general, particularly for models derived from first principles of the physics; the clutch and cup-and-ball examples are good illustrations of this.

In the PhD thesis of Zimmer~\cite{Zimmer2010}, variable-structure multi-domain DAEs are defined thanks to an \emph{ad hoc} modeling language, and a runtime interpreter is used that processes the equations at runtime, when the structure and/or the index changes. Limitations of this work are that impulsive behavior is not supported and that the user has to explicity define the transfer of variable values from one mode to another, which is not practical for large models. 

Describing variable-structure systems with \emph{causal} state machines is discussed by Pepper \emph{et al.}~\cite{Pepper2011}. 
Dynamically changing the structural analysis at runtime is also performed by Hoeger~\cite{Hoger14,DBLP:conf/eoolt/Hoger17}. In~\cite{Hoger14}, the author proposes a dynamic execution of John Pryce's \sigmamethod~\cite{Pryce01}.

Elmqvist \emph{et al.}~\cite{MultiMode2014,VaryingIndex2015} propose a high-level description of multimode models as an extension to the synchronous Modelica 3.3 state machines, by using continuous-time state machines having continuous-time models as `states'. Besides ODEs as used in hybrid automata,  \emph{acausal DAE models with physical connectors} can also be a `state' of a state machine.
Such a state machine is mapped into a form such that the resulting equations can be processed by standard symbolic algorithms supported by Modelica tools. The major restrictions of this approach are that mode changes with impulsive behavior are not supported and that not all types of multimode systems can be handled due to the static code generation. 

Correct restart at mode changes is clearly a central issue in multi-mode DAE systems. Techniques from \emph{nonsmooth dynamical systems}~\cite{Acary08} address this issue by performing a warm restart at mode changes using time-stepping discretization schemes: not only the numerical scheme does not interrupt at mode changes, but these mode changes are not even detected. This is very elegant and effective. However, this relies on the handling of \emph{complementarity conditions} (see Section~\ref{leriugfehliu}), thus, it does not apply to general multimode DAE systems but only to a certain subclass. 


We observe that, in general, the issue of correct restart at mode changes is not addressed in the above mentioned literature. It is, however, considered in the few references to follow.

Benveniste \emph{et al.}~\cite{BenvenisteCEGOP17,benveniste:hal-01343967} considered this issue, as well as the problem of varying structure and index, from a fundamental point of view, by relying on nonstandard analysis to capture continuous-time dynamics and mode change events in a unified framework. A first structural analysis algorithm was presented in~\cite{BenvenisteCEGOP17}, by significantly modifying the original Pantelides algorithm~\cite{pantelides}. A proof-of-concept mockup named SunDAE was developed implementing this approach. This first attempt suffers from some issues: the proposed structural analysis does not boil down trivially to the Pantelides algorithm in the case of single-mode systems; it involves nondeterministic decisions, an unwanted feature for the mathematical foundation of compilers; and its mathematical study is incomplete. 

The work of Stephan Trenn is an important contribution to the subject. In his PhD thesis~\cite{TrennPhD2009} and his article~\cite{Trenn09}, he points out the difficulty in defining piecewise smooth distributions: there is no intrinsic definition for their `Dirac' part, as several definitions comply with the requirements for it.
This indicates that distributions are not the ultimate answer to deal with impulsive variables in multimode DAE systems. Still, Trenn was able in~\cite{LiberzonT12} to define complete solutions for a class of switched DAE systems in which each mode is in \emph{quasi-linear form}: switching conditions are time-based, not state-based. A main difference with our approach is that Trenn's solution is complete for all variables and relies on piecewise smooth distributions, whereas our works on restart conditions consist in eliminating impulsive variables while computing restart values for state variables. (Beyond the existence of solutions, the work of Trenn \emph{et al.} contains additional interesting results that are not relevant to our study.)

An important step forward was done in~\cite{DBLP:series/lncs/BenvenisteCEGOP19}. The interesting subclass of multimode DAE systems that is now called \emph{semi-linear} (see~\rref{sec:mainthm}) was first identified. Semi-linear multimode DAE systems can exhibit impulsive variables at mode changes. They generalize the `semi-linear systems' proposed by Trenn in the sense that switching conditions are no longer restricted to time-based ones, instead including state-based switching conditions.  The analysis and discretization schemes proposed in~\cite{DBLP:series/lncs/BenvenisteCEGOP19} are mathematically sound. Building on this work, Martin Otter has developed the  \href{https://modiasim.github.io/ModiaMath.jl/stable/man/Overview.html}{ModiaMath} tool for semi-linear multimode DAE systems. Hence, it makes sense to compare the schemes proposed in~\cite{DBLP:series/lncs/BenvenisteCEGOP19} to the ones we develop in this paper, for general \mDAE\ systems. It turns out that our general approach coincides with the schemes proposed in~\cite{DBLP:series/lncs/BenvenisteCEGOP19} when applied to the subclass of semi-linear systems, see \rref{sec:mainthm}. Our present work thus extends and significantly improves~\cite{DBLP:series/lncs/BenvenisteCEGOP19}.

\subsection{{Contributions}}
In this work, we develop a systematic approach that addresses, as particular cases, the objectives stated for the examples of Sections~\ref{lireutyliuel}
and~\ref{leriugfehliu}. Its main contributions can be detailed as follows.

\paragraph{Structural analysis of mode changes using nonstandard analysis}
As our main contribution, we extend the notion of structural analysis to mode change events. More precisely, we develop a structural analysis for multimode DAE systems that is valid at any time, that is, for both continuous dynamics and mode changes.

Doing so raises a number of difficulties. First, mode changes occur under various kinds of dynamics; there can be isolated events or finite cascades thereof, but Zeno situations can also occur, i.e., events of mode changes can accumulate and form an infinite sequence bounded by some finite instant. As a result, sliding modes can emerge out of mode changes, and the modeling and simulation tools can only identify such situations at runtime. 
%
Second, impulses may occur at mode changes for certain variables, which again is most of the time discovered by tools at runtime.
Compilation, however, is a task performed prior to generating simulation code; hence, it encounters serious difficulties in both the nature of time at mode changes and the domain of variables, due to impulsive behaviors. 

The main objective of structural analysis, akin to compilation in general, is precisely to keep away from numerical considerations, by focusing only on the structure (or the syntax) of the system of equations. This is captured by the bipartite graph associated to the considered DAE system (also referred to as `incidence graph' or `incidence matrix' or `$\Sigma$-matrix' depending on the context and authors). The basic principle supporting structural analysis consists in: (1) abstracting an equation $f(x,y,z)=0$ as the bipartite graph having $f,x,y,z$ as vertices and $(f,x),(f,y),(f,z)$ as edges, and a system of equations as the union of the above bipartite graphs; and (2) exploiting the resulting graph to generate efficient simulation code. This means focusing on syntax and symbols, while abstracting away numerical aspects.
We thus wish to \emph{symbolically} manipulate both normal and impulsive values for variables, as well as both continuous time and events (in finite or infinite cascades), and possibly stranger structures for time.

We see no vehicle for supporting all of this, other than \emph{nonstandard analysis.} Nonstandard analysis was proposed by Abraham Robinson in the 1960s to allow the explicit manipulation of `infinitesimals' and `infinities' in analysis~\cite{Robinson,Cutland}. Impulsive values become normal citizens in nonstandard analysis. Continuous time and dynamics can be discretized using infinitesimal time steps, e.g., by setting $\dot{x}(t)\approx\frac{x(t+\vsmall)-x(t)}{\vsmall}$ with $\vsmall$ infinitesimal, thus providing perfect fidelity up to infinitesimal errors. This discretized time can then support the various kinds of time needed when modeling mode changes. Since techniques from structural analysis
translate almost verbatim to discrete-time dynamics, having dynamical systems indexed by discrete (nonstandard) time in a uniform setting paves the way toward structural analysis for multimode DAE systems in their generality.
Using this approach, we are able to extend the structural analysis, from single-mode to multimode DAE systems, by encompassing both modes and mode changes.

\paragraph{Generating code for the restart at mode changes} 
The structural analysis developed hereafter works in the nonstandard analysis domain. As such, it cannot be executed on any real-life computer. A further stage is needed to produce executable code, by mapping the nonstandard execution schemes to ordinary world (standard) schemes; this is called \emph{standardization}. This allows us to produce mathematically justified schemes for the restart at mode changes.

\paragraph{Rejecting and accepting programs on a clear basis}
Our structural analysis is precise enough to properly identify models that are structurally over- or under-specified at mode change events, or models that contain a fixpoint equation involving a subset of variables and some guard. We reject such models, with a proper diagnosis explaining the rejection. In turn, mode-dependent index/state/dynamics are not reasons for rejection \emph{per se}. Our compilation technique is able to handle such cases.

\paragraph{The paper is organized as follows}
This report is an extended version of a paper under submission. With reference to this submission, it contains additional material, which the reader may skip for a first reading. In the following list of sections, we indicate in blue these {\color{blue} additional sections}.

The clutch example is informally developed in~\rref{sec:simpleclutch}; we show how to map this model to the nonstandard domain, how the structural analysis works (in particular, the handling of mode changes is emphazised), and how standardization is performed to generate actual simulation code. The Cup-and-Ball example is developed in Section~\ref{loeruighlrigtuh}, following the same guidelines; based on the case of elastic impact, it brings forth the need for considering transient modes, i.e., modes that last for zero time. {\color{blue} \rref{sec:railcar} develops the Westinghouse air brake example; we address the objectives stated in the introduction and we advocate for the interest of adding {assertions}, restricting the allowed sequences of modes in the system.} The insight gained from studying these examples leads to the approach described in Section~\ref{sec:liguholui}.

The extension of structural analysis to multimode DAE systems is then developed. In \rref{sec:mDAE}, the basics of structural analysis for algebraic systems of equations are recalled. Then, we review J.\,Pryce's \sigmamethod\ for the structural analysis of DAE systems. We extend the \sigmamethod\ to the multimode case, including the handling of mode change events, in \rref{sec:loguhiuliugh}, for the subclass of systems involving only long modes---which excludes the Cup-and-Ball example with elastic impact. Motivated by the latter case, we identify the need for handling \emph{transient} modes, in which the system spends zero time; this is addressed in \rref{sec:epwtgouihpiouh}. {\color{blue} In \rref{sec:mdDAE} we widen our scope by adding to our framework, in addition to derivatives, the left- and right-limits of a signal, which is a useful primitive operator provided by modeling languages such as Modelica or Simulink.}

\rref{sec:NSA} recalls the needed background on nonstandard analysis, mostly related to differential calculus. In \rref{sec:standardization}, we complement this background with additional results related to structural analysis and impulse analysis, and we apply this machinery to the task of standardization. 

Three important results are developed in \rref{sec:MainResults}. 

\begin{itemize}
\item In \rref{sec:otherschemes}, we prove that modifying the nonstandard expansion of the derivative (there are infinitely many possibilities) does not change the final generated code, thus showing that our approach is intrinsic. 
\item We prove in \rref{sec:mainthm} that our approach does compute correct solutions for \emph{semi-linear} multimode DAE systems, a nontrivial subclass possibly involving impulsive behaviors and containing, as a subclass, multi-body mechanics (see~\cite{DBLP:series/lncs/BenvenisteCEGOP19}). We would be happy to prove that the simulation code we generate does compute the solution of any multimode DAE system; this, however, is beyond what is doable since no notion of solution is known for multimode DAE systems in general.
\item We propose in \rref{sec:restartschemes} a numerical scheme for approximating restart values for non-impulsive variables at mode changes. This alleviates the need for performing standardization when generating code in practice; the proof of this result, however, uses standardization as a mathematical tool.
\end{itemize}

Finally, we draw in \rref{sec:RLDC2} how our approach can be implemented in a tool, and we illustrate this on the RLDC2 circuit example. This example, provided to us by Sven-Erik Mattsson, is problematic for the existing Modelica tools, as it has mode-dependent index and structure. The structural analysis of the corresponding DAE system in all its modes is successfully handled by our IsamDAE tool~\cite{Caillaud2020a}. This tool, currently under development, will support multimode DAE structural analysis following our theory. Its current version does not support impulsive variables at mode changes---the RLDC2 circuit example, however, does not exhibit this difficulty.

{\color{blue} In \rref{sec:loguhiuliugh}, we assumed that the effect of guards is  delayed by one nonstandard time shift. This is a necessary condition for being able to support both long and transient modes. In Apprendix~\ref{eriughepuih}, we propose an alternative execution scheme that is \emph{constructive} in that it can handle models in which the effect of guards is not delayed by one nonstandard time shift. In turn, only models with long modes are supported.}
\section{The ideal clutch}
\label{sec:simpleclutch}
In this section, we introduce our approach by discussing in detail the
ideal clutch example presented in Section\,\ref{lireutyliuel}. Its
model was given in~(\ref{sys:coupledshafts}).
We first analyze separately the model for each mode of the clutch
(\rref{sec:monomode}).  Then, we discuss the difficulties arising when
handling mode changes (\rref{sec:events}).  Finally, we propose a
global comprehensive analysis in Sections
\ref{sec:nsa}\,--\,\ref{sec:standardizeclutch}. For convenience, we
present an intuitive and informal introduction to nonstandard analysis
in \rref{sec:nsa}.

\subsection{Separate Analysis of Each Mode}
\label{sec:monomode}
In the released mode, when $\gamma$ is false in
\rref{sys:coupledshafts}, the two shafts are independent and one
obtains the following two independent ODEs for $\omega_1$ and
$\omega_2$:
\begin{equation}
\begin{array}{lr}
\omega'_1=f_1(\omega_1,\tau_1) &(\eqq_1) \\
\omega'_2=f_2(\omega_2,\tau_2) &(\eqq_2)
\end{array}
\quad
\begin{array}{lr}
\tau_1=0 &(\eqq_{5}) \\
\tau_2=0 &(\eqq_{6})
\end{array}
\label{sys:daemode1}
\end{equation}
In the engaged mode, however, $\gamma$ holds true, and both the two
velocities and the two torques are algebraically related:
\begin{equation}
\begin{array}{lr}
\omega'_1=f_1(\omega_1,\tau_1) &(\eqq_1) \\
\omega'_2=f_2(\omega_2,\tau_2) &(\eqq_2)
\end{array}
\quad
\begin{array}{rr}
\omega_1-\omega_2=0 &(\eqq_{3}) \\
\tau_1+\tau_2=0 &(\eqq_{4})
\end{array}
\label{sys:daemode2}
\end{equation}
Equation $(\eqq_4)$, that relates the two torques, is not an issue by
itself, as it can be rewritten as $\tau_2=-\tau_1$ and used to eliminate
$\tau_2$. Equation $(\eqq_3)$, however, relates the two velocities
$\omega_1$ and $\omega_2$ that are otherwise subject to the ODE system
$(\eqq_1,\eqq_2)$, which makes \rref{sys:daemode2} a DAE instead of an ODE.

In \rref{sys:daemode2}, the torques $\tau_1$ and $\tau_2$ are not
differentiated: we call them \emph{algebraic variables}. The
derivatives of the velocities $\omega_1$ and $\omega_2$ appear: they are \emph{state variables}. The
$\dot{\omega_1},\dot{\omega_2},\tau_1,\tau_2$ are called the
\emph{leading variables} of \rref{sys:daemode2}.
%

Now, suppose one is able to uniquely determine the {leading variables}
$\dot{\omega_1},\dot{\omega_2},\tau_1,\tau_2$ given \emph{consistent}
values for the {state variables} $\omega_1,\omega_2$, i.e., values
satisfying $(\eqq_3)$. Then, by using an ODE solver, one could perform
an integration step and update the current velocities
${\omega_1},{\omega_2}$ using the computed values for their
derivatives $\dot{\omega_1},\dot{\omega_2}$. In this case, we say that
the considered DAE is an ``extended ODE''~\cite{Pryce01}.

It turns out that this condition does not hold for \rref{sys:daemode2}
as is.  To intuitively explain the issue, we move to
discrete time by applying an explicit first-order Euler scheme with
constant step size $\delta>0$:
\begin{equation}
\begin{array}{lr}
\postset{\omega_1}=\omega_1+\delta\cdot f_1(\omega_1,\tau_1) &(\eqq_1^\delta) \\
\postset{\omega_ 2}=\omega_ 2+\delta\cdot f_ 2(\omega_ 2,\tau_ 2) &(\eqq_2^\delta)
\end{array}
\quad
\begin{array}{rr}
\omega_1-\omega_2=0 &(\eqq_{3}) \\
\tau_1+\tau_2=0 &(\eqq_{4})
\end{array}
\label{sys:expliciteulermode2}
\end{equation}
where 
$
\postset{\omega}(t)\eqdef\omega(t{+}\delta) 
$
denotes the forward time shift operator by an amount of $\delta$.
Suppose we are given consistent values $\omega_1=\omega_2$ and we wish
to use \rref{sys:expliciteulermode2}, seen as a system of algebraic
equations, to determine the value of the dependent variables (i.e., the
unknowns) $\tau_1,\tau_2,\postset{\omega_1},\postset{\omega_2}$.  This
attempt fails since we have only three equations
$\eqq_1^\delta$, $\eqq_2^\delta$  and $\eqq_4$ to determine four
unknowns $\tau_1$, $\tau_2$, $\postset{\omega_1}$ and
$\postset{\omega_2}$; indeed, the omitted equation $(\eqq_3)$ does
not involve any dependent variable.

However, since \rref{sys:expliciteulermode2} is time invariant, and
assuming that the system remains in the engaged mode for at least
$\delta$ seconds, there exists an additional \emph{latent equation} on
the set of variables, namely
\beq
\postset{\omega_1}-\postset{\omega_2}=0 && (\postset{\eqq_3})
\label{eq:eq3boulet}
\eeq
obtained by shifting $(\eqq_3)$ forward. Now,
replacing $(\eqq_3)$ with $(\postset{\eqq_3})$ in \rref{sys:expliciteulermode2} yields a system with four equations and four
dependent variables; moreover, this system is, in fact, nonsingular in a generic sense.
One can now use \rref{sys:expliciteulermode2} along with equation $(\postset{\eqq_3})$ to get an execution
scheme for the engaged mode of the clutch. This is shown in
\rref{exec:daemode2} below.
\begin{execscheme}
  \caption{\rref{sys:expliciteulermode2}$+$\rref{eq:eq3boulet}.}\label{exec:daemode2}
  \begin{algorithmic}[1]
    \Require consistent $\omega_1$ and $\omega_2$, i.e., satisfying  $(\eqq_3)$. 
    \State \texttt{Solve} $\{\eqq_1^\delta,\eqq_2^\delta,\postset{\eqq_3},\eqq_4\}$ for $(\postset{\omega_1},\postset{\omega_2},\tau_1,\tau_2)$
		\label{op:solve} \Comment{$4$ equations, $4$ unknowns} 
    \State $(\omega_1,\omega_2) \gets (\postset{\omega_1},\postset{\omega_2})$ \Comment{update the states $({\omega_1},{\omega_2})$}
    \State \texttt{Tick} \Comment{move to next discrete step}
  \end{algorithmic}
\end{execscheme}

Since the new values of the state variables satisfy
(\ref{eq:eq3boulet}) by construction, the consistency condition is met
at the next iteration step (should the system remain in the same
mode).

\begin{ccomment}[structural nonsingularity]\rm 
	\label{erlihtroiu} 
        The implicit assumption behind \rref{op:solve} in
        \rref{exec:daemode2} is that solving
        $\{\eqq_1^\delta,\eqq_2^\delta,\postset{\eqq_3},\eqq_4\}$
        always returns a unique set of values.
        In our example, this is true in a `generic' or
        `structural' sense, meaning that it holds for all but
        exceptional values for the parameters in the considered
        equations (through the unspecified functions $f_1$ and
        $f_2$). We will repeatedly use \emph{structural
          nonsingularity} (also called \emph{structural regularity})
        in the sequel. A formalization of structural reasoning and
        analysis is developed in \rref{sec:mDAE}.\eproof
\end{ccomment}
Observe that the same analysis could be applied to the original
continuous-time dynamics from \rref{sys:daemode2} by augmenting it with the following \emph{latent equation}: \beq
\dot{\omega_1}-\dot{\omega_2}=0 && (\eqq_3')
\label{eq:eq3prime}
\eeq obtained by differentiating $(\eqq_3)$; as a matter of fact, since $(\eqq_3)$ holds
at any instant, $(\eqq_3')$ follows as long as the solution is smooth
enough for the derivatives $\dot{\omega_1}$ and $\dot{\omega_2}$ to be
defined.  The resulting execution scheme is given
in~\rref{exec:daemode2c}, that parallels \rref{exec:daemode2}.
\begin{execscheme}
  \caption{\rref{sys:daemode2}$+$\rref{eq:eq3prime}.}\label{exec:daemode2c}
  \begin{algorithmic}[1]
    \Require consistent $\omega_1$ and $\omega_2$, i.e., satisfying  $(\eqq_3)$. 
    \State \texttt{Solve} $\{\eqq_1,\eqq_2,\eqq_3',\eqq_4\}$  for $(\dot{\omega_1},\dot{\omega_2},\tau_1,\tau_2)$ \label{op:solve2} \Comment{$4$ equations, $4$ unknowns} 
    \State \texttt{ODESolve} $({\omega_1},{\omega_2})$ \Comment{update the states $({\omega_1},{\omega_2})$}
    \State \texttt{Tick} \Comment{move to next discretization step}
  \end{algorithmic}
\end{execscheme}

\rref{op:solve2} is identical for the two schemes and is assumed to
give a unique solution, generically; it fails if one omits the latent
equation $(\eqq_3')$.  Then, although getting the next
values for the $\omega_1$ and $\omega_2$ is easy in \rref{exec:daemode2}, the situation changes in \rref{exec:daemode2c}: the derivatives
$(\omega_1',\omega_2')$ are first evaluated, and then an ODE solver (here denoted by \texttt{ODESolve})
is used to update the state $({\omega_1},{\omega_2})$ for the next
discretization step. Note that, when considering an exact mathematical
solution, if ${\omega_1}-{\omega_2}=0$ initially holds and
$\dot{\omega_1}-\dot{\omega_2}=0$, then the linear constraint
$(\eqq_3)$ will be satisfied at any positive time.

The replacement of $(\eqq_3)$ by $(\dot{\eqq_3})$, which resulted in
\rref{exec:daemode2c}, is known in the literature as
\emph{index reduction}~\cite{MattssonSoderlin1993}.  It requires
adding the (smallest set of) latent equations needed for the system in \rref{op:solve2} of the execution scheme to become solvable and
deterministic.  The number of successive time differentiations
needed for obtaining all the latent equations is called the
\emph{differentiation index}~\cite{CampbellGear1995}, which we simply
refer to as the \emph{index} in the sequel.\footnote{There are indeed a
  great number of possible definitions for ``the index'' of a DAE, serving various purposes,
  see~\cite{CampbellGear1995} for instance. This, however, is not relevant to the
  present work.} 
Finally, observe that both execution schemes
rely on an algebraic equation system solver to
evaluate the leading variables as a function of state variables.

We conclude this part by briefly discussing the initialization
problem.  Unlike for ODE systems, initialization is far from trivial
for DAE systems. For the clutch example, if one considers
\rref{sys:daemode2} as a standalone DAE, the initialization is
performed as indicated in \rref{exec:init}.
\begin{execscheme}
  \caption{Initialization of \rref{sys:daemode2}$+$\rref{eq:eq3boulet}.}\label{exec:init}
  \begin{algorithmic}[1]
    \State $(\omega_1,\omega_2)$
    $\gets$ \texttt{Solve}$\{\eqq_3\}$  for $({\omega_1},{\omega_2})$
     \label{op:init} 
     \hfill 
     \Comment{$2$ unknowns for only $1$ equation} 
  \end{algorithmic}
\end{execscheme}

The system for solving possesses two unknowns for only one
equation, which leaves us with one degree of freedom; in mathematical terms, the set of all initial values for the $2$-tuple of variables
is a manifold of dimension 1, generically or structurally. For
example, one may freely fix the initial common rotation speed so that
($\eqq_3$) is satisfied.

\subsection{Mode Transitions}
\label{sec:events}
In an attempt to reduce the full clutch model to the analysis of the
DAE of each mode, one hopes that the handling of a mode change boils down
to applying the initialization given in \rref{exec:init}.
If one was to handle restarts at mode changes like initializations, it
would make the clutch system nondeterministic, due to the
\emph{extra} degree of freedom of \rref{exec:init}. In contrast, the
physics tells us that the state of the system is entirely determined
when the clutch gets engaged. This comforts the intuition
that resets at mode changes strongly differ from mere initializations.

However, inferring by hand the reset values for rotation
velocities when the clutch gets engaged is definitely non-trivial. Furthermore, these values depend on the whole system model, so that the task of determining them becomes complex if external components are added. \emph{It is therefore highly desirable, for this example, to
  let the compiler infer these reset values from model
  \emph{(\ref{sys:coupledshafts})}.}

This problem involves a mix of continuous-time and discrete-time
dynamics, the latter consisting in one or several computation steps
to restart the dynamics of the next mode. As a result, it will be convenient to cast everything
in discrete time thanks to nonstandard analysis. We develop this
next.
 
\subsection{Nonstandard Semantics}
\label{sec:nsa} {Nonstandard
  analysis}~\cite{Robinson,Lindstrom,JCSS11} extends the set $\bR$ of
real numbers into a superset $\nstR$ of \emph{hyperreals} (also
called \emph{nonstandard reals}) that includes infinite sets of infinitely
large numbers and infinitely small numbers. The key properties of
hyperreals, needed for the informal discussion of the clutch
example, are the following \cite{Lindstrom}:\footnote{The skeptical reader is gently referred to \rref{sec:NSA} for mathematical details.}

\begin{itemize}
\item There exist \emph{infinitesimals}, defined as hyperreals that are
  smaller in absolute value than any real number: an infinitesimal
  $\deronde \in \nstR$ is such that $|\deronde|<a$ for any
  positive $a\in\bR$. For $x,y$ two hyperreals, write $x\approx{y}$ if
  $x-y$ is an infinitesimal.
\item All relations, operators, and propositional formulas that are valid over
  $\bR$ are also valid over $\nstR$; for example, $\nstR$ is a
  totally ordered set. The arithmetic operations $+$, $\times$,
  etc. can be lifted to $\nstR$. We say that a hyperreal
  $x$ is \emph{finite} if there exists some standard finite positive
  real number $a$ such that $|x|<a$.
\item For every finite hyperreal $x\in\nstR$, there exists a unique
  standard real number $\st{x}\in\bR$ such that
  $\st{x} \approx x$, and $\st{x}$ is called the \emph{standard part} (or \emph{standardization}) of
  $x$. As we shall see, standardizing more complex objects, such as functions or systems of equations, requires some care.
\item Every real function lifts in a systematic way to a hyperreal
  function (regardless of its continuity properties).\footnote{The result of this lifting may not be intuitive in general, see \rref{sec:NSA} for mathematical details.} This allows us to write $f(x)$ where $f$ is a real
  function and $x$ is a nonstandard number.
\item Let $t\mapsto{x}(t), t\in\bR$ be an $\bR$-valued (standard)
  signal. Then: \beq \mbox{
\begin{minipage}{11cm}
  $x$ is continuous at instant $t\in\bR$ if and only if, for any
  infinitesimal $\deronde\in\nstR$, one has
  $x(t+\deronde)\approx{x(t)}$;
\end{minipage}
} 
\label{erlfuihliu} \\
[1mm]
  \mbox{
\begin{minipage}{11cm}
  $x$ is differentiable at instant $t \in \bR$ if and only if there
  exists $a \in \bR$ such that, for any infinitesimal
  $\deronde \in \nstR$,
  $\frac{x(t+\deronde)-x(t)}{\deronde} \approx {a}$. In this case,
  $a=\dot{x}(t)$.
\end{minipage}
} 
\label{erpgtuihpui}
\eeq
\end{itemize}
We can then consider, for a given positive infinitesimal $\deronde$, the following time index set
$\bT\subseteq\nstR$:
\beq
\bT = 0,\deronde,2\deronde,3\deronde,\dots =
     \left\{n\deronde \mid n \in \nstN \right\}
\label{non-standard-time-line}
\eeq
where $\nstN$ denotes the set of \emph{hyperintegers},
consisting of all natural integers augmented with additional infinite numbers
called \emph{nonstandard}. For convenience, the elements
of $\bT$ will be generically denoted by the symbol $\nstime$, the
explicit expression $n\vsmall$, or even simply by $t$ when the
possible confusion with (usual) real time does not matter.  The
important features of $\bT$ are:\footnote{In French, one uses to
  qualify this as: \emph{avoir le beurre et l'argent du beurre.}}
\begin{itemize}
\item Any finite positive real time $t\,{\in}\,\bR$, $t \geq 0$ is infinitesimally close to some element of $\bT$ (informally, $\bT$ covers the set $\bR_+$ of nonnegative real numbers
  and can be used to index continuous-time dynamics); and
\item $\bT$ is `discrete': every instant $n\deronde$ has a
  predecessor $(n{-}1)\deronde$ and a successor $(n{+}1)\deronde$, except 0 that has no predecessor.
\end{itemize}
%
\begin{definition}[Forward and Backward Shifts]
    \label{defn:tshifts}
    Let $x$ be a nonstandard signal indexed by $\bT$. We define the
    forward shifted signal $\postset{x}$ through
    ~$ \postset{x}(n\vsmall) \eqdef x((n{+}1)\vsmall) $ and, for
    $f(X)$ a function of the tuple $X$ of signals, we set
    $\postset{(f(X))}\eqdef f(\postset{X})$ where the forward shift
    $X\mapsto\postset{X}$ applies pointwise to all the components of
    the tuple.
    It will also be convenient to consider the \emph{backward shift}
    $\preset{x}$ defined by
    ~$ \preset{x}((n{+}1)\vsmall) \eqdef x(n\vsmall) $, implying that
    an initial value for $\preset{x}(0)$ must be provided.
\end{definition}
For example,
$
\postset{f}(x,y)(t)\eqdef{f}(\postset{x},\postset{y})(t)=f(x(t{+}\vsmall),y(t{+}\vsmall))
$.  Observe that the definition of the forward shift depends on the
infinitesimal $\vsmall$. Thanks to (\ref{erpgtuihpui}), using forward
shifts allows us to represent, up to an infinitesimal, the derivative
$\dot{x}$ of a signal by its first-order explicit Euler approximation
$\frac{1}{\vsmall}(\postset{x}-x)$.

Solutions of multimode DAE systems may, however, be non-differentiable and even non-continuous at events of mode change. To
give a meaning to $x'$ at any instant, \emph{we decide to
  {\textbf{define}} it everywhere as the nonstandard first-order Euler
  increment}
\begin{equation}
\label{eq:differenceeq}
 \dot{x} \eqdef \frac{1}{\vsmall}(\postset{x}-x) \enspace .
 \end{equation}
\begin{definition}\label{erliutwyeoiu} Substituting, in a multi-mode
  DAE system, every occurrence of $x'(t)$ by its expansion
  $(\ref{eq:differenceeq})$ yields a difference algebraic equation
  \emph{(dAE)}
  system\footnote{Throughout this paper, we consistently use, in acronyms, the letters
    ``D'' and ``d'' to refer to ``Differential'' and ``difference'',
    respectively.} that we call its \emph{nonstandard semantics}.
\end{definition}
For instance, the nonstandard semantics of \rref{sys:coupledshafts} is
the following:
\beq
\left\{
\bea{rllcc}
&&\frac{\postset{\omega_1}-\omega_1}{\vsmall}=f_1(\omega_1,\tau_1) &&(\eqq_1^\vsmall) \\
&&\frac{\postset{\omega_2}-\omega_2}{\vsmall}=f_2(\omega_2,\tau_2) &&(\eqq_2^\vsmall) \\
\when\;\guard&\doo&{\omega_1}-{\omega_2}=0 &&({\eqq_3}) \\
&\aand&\tau_1+\tau_2=0 &&(\eqq_{4}) \\
\when\;\prog{not}\;\guard&\doo&\tau_1=0 &&(\eqq_{5}) \\
&\aand&\tau_2=0 &&(\eqq_{6}) \\
\eea
\right.
\label{sys:nscoupledshaftsbroken}
\eeq
The state variables are $\omega_1$, $\omega_2$ whereas the leading
variables are now
$\tau_1$, $\tau_2$, $\postset{\omega_1}$, $\postset{\omega_2}$, in both modes
$\guard=\fff$ and $\guard=\ttt$.
%
Reproducing the reasoning of \rref{sec:monomode}, one obtains, for
each mode, a discrete system very much like the explicit Euler scheme
of \rref{sec:monomode}, except that the step size is now infinitesimal
and that the variables are all nonstandard.  The added value of
\rref{sys:nscoupledshaftsbroken} with respect to the explicit Euler
scheme of \rref{sec:monomode} is twofold: first, it is exact up to
infinitesimals within each mode; second, it will allow us to
carefully analyze what happens at events of modes change.

\subsection{Nonstandard structural analysis}
\label{sec:orguhuoi}
We now develop a structural analysis for the nonstandard multimode dAE
system (\ref{sys:nscoupledshaftsbroken}).
As a preliminary step, we perform the structural analysis in each mode
and combine the two dynamics with the due guards (the added latent
equation is shown in {\color{red}red}): \beq \left\{ \bea{rllcc}
  &&\frac{\postset{\omega_1}-\omega_1}{\vsmall}=f_1(\omega_1,\tau_1) &&(\eqq_1^\vsmall) \\
  &&\frac{\postset{\omega_2}-\omega_2}{\vsmall}=f_2(\omega_2,\tau_2) &&(\eqq_2^\vsmall) \\
  \when\;\guard&\doo&{\omega_1}-{\omega_2}=0 &&({\eqq_3}) \\
  &\remph{\aand}&\remph{\postset{\omega_1}-\postset{\omega_2}=0} &&\remph{(\postset{\eqq_3})} \\
  &\aand&\tau_1+\tau_2=0 &&(\eqq_{4}) \\
  \when\;\prog{not}\;\guard&\doo&\tau_1=0 &&(\eqq_{5}) \\
  &\aand&\tau_2=0 &&(\eqq_{6}) \\
  \eea \right.
\label{sys:loiguwhuip}
\eeq \rref{sys:loiguwhuip} coincides with the known structural
analysis in the interior of each of the two modes $\guard=\ttt$ and
$\guard=\fff$.  What about the instants of mode change?

\paragraph{Mode change $\guard:\ttt\ra\fff$}
At the considered instant, we have $\preset{\guard}=\ttt$ and
$\guard=\fff$, where $\preset{\guard}$ denotes the backward shift of
$\guard$ following \rref{defn:tshifts}. Unfolding
\rref{sys:loiguwhuip} at the two successive previous (shown in
{\color{blue}blue}) and current (shown in black) instants yields, by
taking the actual values for the guard at those instants into account:
\beq\bea{rl} \bemph{\mbox{previous}}&\left\{ \bea{lcl}
  \bemph{\frac{{\omega_1}-\preset{\omega_1}}{\vsmall}=f_1(\preset{\omega_1},\preset{\tau_1})} && \\ [1mm]
  \bemph{\frac{{\omega_2}-\preset{\omega_2}}{\vsmall}=f_2(\preset{\omega_2},\preset{\tau_2})} && \\
  \bemph{\preset{\omega_1}-\preset{\omega_2}=0} && \bemph{(\preset{\eqq_3})} \\
  \bemph{{\omega_1}-{\omega_2}=0} && \\
  \bemph{\preset{\tau_1}+\preset{\tau_2}=0} && \\ [1mm]
\eea
\right.
\\
\blemph{\mbox{current}} &\left\{
\bea{lcl}
\frac{\postset{\omega_1}-\omega_1}{\vsmall}=f_1(\omega_1,\tau_1) &~~\;& \\ [1mm]
\frac{\postset{\omega_2}-\omega_2}{\vsmall}=f_2(\omega_2,\tau_2) && \\ [1mm]
\tau_1=0 && \\
\tau_2=0 && \\
\eea
\right.
\eea
\label{sys:erguihuiop}
\eeq We assume that values are given for the state variables at the
previous instant, namely $\preset{\omega_1}$ and $\preset{\omega_2}$, and we
regard \rref{sys:erguihuiop} as an algebraic system of equations with
dependent variables
$\bemph{\preset{\tau_i},{\omega_i}};{\tau_i},\postset{\omega_i}$ for
$i=1,2$, i.e., the leading variables of \rref{sys:loiguwhuip} at the
previous and current instants. Note that equation
$\bemph{(\preset{\eqq_3})}$ does not involve any dependent variable:
it is actually a \emph{fact} that is inherited from the instant
$t-2\vsmall$, where $t$ is the current instant; we can thus ignore
this equation in \rref{sys:erguihuiop}.  Seen this way, \rref{sys:erguihuiop} is structurally nonsingular
and we can solve it as two successive blocks, corresponding to the
previous instant followed by the current instant: the latter is our
desired code for the mode change. This yields the following code for
the mode change $\guard:\ttt\ra\fff$: \beq \left\{ \bea{lcc}
  \bemph{{\omega_1},{\omega_2} \mbox{ set by previous instant}} && \\ [1mm]
  \frac{\postset{\omega_1}-\omega_1}{\vsmall}=f_1(\omega_1,\tau_1) && \\ [1mm]
  \frac{\postset{\omega_2}-\omega_2}{\vsmall}=f_2(\omega_2,\tau_2) && \\ [1mm]
\tau_1=0 && \\
\tau_2=0 && \\
\eea
\right.
\label{sys:eprgfiuaerhfop}
\eeq The above reasoning may seem an overshoot for this mode change
$\guard:\ttt\ra\fff$, since one could simply say: let the previous
part of \rref{sys:erguihuiop} fix $\omega_1,\omega_2$ and, then, let the
current part of \rref{sys:erguihuiop} determine
${\tau_i},\postset{\omega_i}$ for $i=1,2$.
However, the same line of reasoning will allow us to address the other, more
difficult, mode change as well.

\paragraph{Mode change $\guard:\fff\ra\ttt$} At the
considered instant, we have $\preset{\guard}=\fff$ and
$\guard=\ttt$. We proceed as for the other case. Unfolding
\rref{sys:loiguwhuip} at the two successive previous (shown in
{\color{blue}blue}) and current (shown in black) instants yields, by
taking the actual values for the guard at those instants into account:
\beq\bea{rl} \bemph{\mbox{previous}}&\left\{ \bea{lcl}
  \bemph{\frac{{\omega_1}-\preset{\omega_1}}{\vsmall}=f_1(\preset{\omega_1},\preset{\tau_1})} && \bemph{(\preset{\eqq_1^\vsmall})} \\ [1mm]
  \bemph{\frac{{\omega_2}-\preset{\omega_2}}{\vsmall}=f_2(\preset{\omega_2},\preset{\tau_2})} && \bemph{(\preset{\eqq_2^\vsmall})} \\
  \bemph{\preset{\tau_1}=0} && \\
  \bemph{\preset{\tau_2}=0} && \\ [1mm]
\eea
\right.
\\
\blemph{\mbox{current}}&\left\{
\bea{lcl}
\frac{\postset{\omega_1}-\omega_1}{\vsmall}=f_1(\omega_1,\tau_1) &~~\;&  \\ [1mm]
\frac{\postset{\omega_2}-\omega_2}{\vsmall}=f_2(\omega_2,\tau_2) && \\ [1mm]
\omega_1-\omega_2=0 && (\eqq_3) \\ 
\postset{\omega_1}-\postset{\omega_2}=0 && \\ 
\tau_1+\tau_2=0 && 
\eea
\right.
\eea
\label{sys:welguihwtuilh}
\eeq
Once more, we regard \rref{sys:welguihwtuilh} as an algebraic system
of equations with dependent variables
$\bemph{\preset{\tau_i},{\omega_i}};{\tau_i},\postset{\omega_i}$ for
$i=1,2$, i.e., the leading variables of \rref{sys:loiguwhuip} at the
previous and current instants.
In contrast to the previous mode change, \rref{sys:welguihwtuilh} is
no longer structurally nonsingular since the following subsystem
possesses five equations and only four dependent variables
$\omega_1,\omega_2,\preset{\tau_1},\preset{\tau_2}$:
\beq \left\{
  \bea{lcl}
  \bemph{\frac{{\omega_1}-\preset{\omega_1}}{\vsmall}=f_1(\preset{\omega_1},\preset{\tau_1})} && \bemph{(\preset{\eqq_1^\vsmall})} \\ [1mm]
  \bemph{\frac{{\omega_2}-\preset{\omega_2}}{\vsmall}=f_2(\preset{\omega_2},\preset{\tau_2})} && \bemph{(\preset{\eqq_2^\vsmall})} \\
  \bemph{\preset{\tau_1}=0} && \\
  \bemph{\preset{\tau_2}=0} && \\
\omega_1-\omega_2=0 && (\eqq_3) \\ 
\eea
\right.
\label{oreiuyoweriu}
\eeq
This subsystem exhibits a conflict. We propose to resolve it by
applying the following causality principle:
\begin{principle}[causality]
	\label{oeuryfgy}  What was done at the previous instant cannot be undone at the current instant.
\end{principle}
Applying Principle~\ref{oeuryfgy} leads
to removing, from subsystem (\ref{oreiuyoweriu}), the conflicting equation $(\eqq_3)$.
This yields the following code for the mode change $\guard:\fff\ra\ttt$:
\beq
\left\{
\bea{lcl}\bemph{{\omega_1},{\omega_2},\preset{\tau_1},\preset{\tau_2} \mbox{ set by previous instant}} && \\ [1mm]
\frac{\postset{\omega_1}-\omega_1}{\vsmall}=f_1(\omega_1,\tau_1) &&  \\ [1mm]
\frac{\postset{\omega_2}-\omega_2}{\vsmall}=f_2(\omega_2,\tau_2) && \\ [1mm]
\postset{\omega_1}-\postset{\omega_2}=0 && \\ 
\tau_1+\tau_2=0 && 
\eea
\right.
\label{sys:segfuihpeiu}
\eeq
Note that the consistency equation $(\eqq_3):\omega_1-\omega_2=0$
has been removed from \rref{sys:segfuihpeiu}, thus modifying the original model. However, this removal occurs only at mode
change events \mbox{$\guard:\fff\ra\ttt$}. What we have done amounts
to \emph{delaying by one nonstandard instant the satisfaction of some
  of the constraints in force in the new mode} $\guard=\ttt$. Since
our time step $\vsmall$ is infinitesimal, this takes zero standard
time.

The corresponding nonstandard execution scheme is summarized in
\rref{exec:nsclutch}.  We use variable $\Delta$ to encode the
\emph{context}, that is, the set of equations known to be satisfied by the
state variables as a result of having executed the previous
instant. For instance, $\eqq_3\not\in\Delta$ corresponds to the mode
change $\guard:\fff\ra\ttt$, expressing that $(\eqq_3)$ is not
provably true if the guard was false at the previous
instant. 
At each tick, that is, every time the scheme moves to a new instant, the context gets updated to account for the equations satisfied by the new state.
The procedure `\texttt{Solve}' solves the (algebraic) system to
determine the values of the leading variables.
\begin{execscheme}
  \caption{for Nonstandard \rref{sys:nscoupledshaftsbroken}. Comments are written in {\color{blue}blue}.}\label{exec:nsclutch}
  \begin{algorithmic}[1]
    \Require $\omega_1$ and $\omega_2$. 
    \If {$\gamma$ {\color{blue}(engaged mode)}} 
        \If {$\eqq_3 \notin \Delta$ {\color{blue}(mode change)}}
        \State $(\postset{\omega_1},\postset{\omega_2})$ $\gets$ \texttt{Solve} $\left\{\eqq_1^\vsmall,\eqq_2^\vsmall,\postset{\eqq_3},\eqq_4\right\}$ \label{op:reset} 
        \State \texttt{Tick}: $\Delta \gets \Delta \cup \{\eqq_3\}$
            \Else {} 
            \State $(\tau_1,\tau_2,\postset{\omega_1},\postset{\omega_2})$ $\gets$ \texttt{Solve} $\left\{\eqq_1^\vsmall,\eqq_2^\vsmall,\postset{\eqq_3},\eqq_4\right\}$ 
            \State \texttt{Tick}: $\Delta$ unchanged
        \EndIf
     \Else { {\color{blue}(released mode)}}
     \State $(\tau_1,\tau_2,\postset{\omega_1},\postset{\omega_2})$ $\gets$ \texttt{Solve} $\left\{\eqq_1^\vsmall,\eqq_2^\vsmall,\eqq_5,\eqq_6\right\}$ \label{op:blocks}
     \State \texttt{Tick}: $\Delta \gets \Delta \setminus \{\eqq_3\}$ 
    \EndIf
  \end{algorithmic}
\end{execscheme}

Observe that \rref{exec:nsclutch} would work without changes if the
guard
$\guard$ was a predicate on the state variables $\omega_1$ and $\omega_2$.

\subsection{Standardization}
\label{sec:standardizeclutch}
\rref{exec:nsclutch} cannot be run in its present form, since it
involves nonstandard reals.  To recover executable code over the real
numbers, a supplementary \emph{standardization} step is needed.
Recall that any finite nonstandard real $x$ has a unique standard part
$\st{x}\in\bR$ such that $x\approx\st{x}$.  The standardization
procedure aims at recovering the standard parts of the leading
variables from their nonstandard version.  We distinguish two cases:
continuous evolutions within each mode, assuming that the sojourn time
in each mode has (standard) positive duration, and discrete evolutions
at events of mode change.

\subsubsection{Within continuous modes}
Let $x:t\mapsto x(t)$, $t \in [s,p)$, denote the real continuous and differentiable solution in a given mode (assuming it exists).  
	 In the sequel, we use the notation $o(\xi)$ to denote any (unspecified) nonstandard real number $\zeta$ such that $\zeta/\xi$ is infinitesimal.
Using (\ref{erpgtuihpui}) and this notation, equations $(\eqq_1^\vsmall)$ and
$(\eqq_2^\vsmall)$ rewrite as $\dot{\omega_i}=f_i(\omega_i,\tau_i)+o(1)$,
which can be shown to standardize as the differential equations
$(\eqq_1)$ and $(\eqq_2)$, respectively. This suffices to recover the
dynamics (\ref{sys:daemode1}) in the released mode $\guard=\fff$.

In the engaged mode $\guard=\ttt$, however, we need to handle
$(\postset{\eqq_3}):\postset{\omega_1}-\postset{\omega_2}=0$. The
nonstandard characterization of derivatives writes
\mbox{$\postset{\omega_i}=\omega_i+\vsmall.\dot{\omega_i}+o(\vsmall)$},
where $\omega_i$ and $\dot{\omega_i}$ are both standard. Since
${\omega_1}-{\omega_2}=0$ and
$\postset{\omega_1}-\postset{\omega_2}=0$ both hold, we inherit
$\dot{\omega_1}-\dot{\omega_2}=o(1)$, which implies
$\dot{\omega_1}-\dot{\omega_2}=0$ since the $\dot{\omega_i}$ are
standard. We thus recover the dynamics (\ref{sys:daemode2}) augmented
with the latent equation (\ref{eq:eq3prime}) in the engaged mode
$\guard=\ttt$.

To summarize, within continuous modes, we perform structural analysis
as usual for DAE systems in continuous time: nothing new happens.

\subsubsection{At the instants of mode change}
\label{lwergtuiherlu}
Suppose we have an event of mode change at time $t$, meaning that
$\guard(t)\neq\guard({t}-\vsmall)$. Our aim is to use the next values
$\postset{\omega_i}(t)=\omega_i(t+\vsmall)$ to initialize the state
variables for the next mode with the value
${\omega_i^+}(t)=\postset{\omega_i}(t)$. Thus, unlike for continuous
time dynamics, the occurrence of the infinitesimal time step $\vsmall$
in the expression $\omega_i(t+\vsmall)$ is not an issue: it just
points to the supplementary instant at which restart is performed.

However, the equations defining $\postset{\omega_i}$ as functions of
the ${\omega_i}$'s involve the hyperreal $\vsmall$ \emph{in space}: standardization must be performed to get rid of it.

\paragraph{Transition $\guard:\ttt\ra\fff$} It is easy, as the new mode
has an ODE dynamics whose state variables are initialized with
corresponding exit values when leaving the previous mode.


\paragraph{Transition $\guard:\fff\ra\ttt$} This one is more involved. As
established in \rref{exec:nsclutch}, in order to compute the reset
values, we use the system of $4$ equations
$\{\eqq_1^\vsmall,\eqq_2^\vsmall,\postset{\eqq_3},\eqq_4\}$ to
determine the leading variables
$(\tau_1,\tau_2,\postset{\omega_1},\postset{\omega_2})$.  In
particular, from $\eqq_i^\vsmall$, we get
\begin{equation}
\label{eq:omegaboulet}
\frac{\postset{\omega_i}-\omega_i}{\vsmall} = f_i(\omega_i,\tau_i),\quad i=1,2 . 
\end{equation}
At this point, we must identify possible impulsive variables, which is the objective of the following \emph{impulse analysis.}

\myparagraph{Impulse analysis} Before engaging the clutch, we must
assume $\omega_1-\omega_2 \neq 0$. Since
$\postset{\omega_1}-\postset{\omega_2}=0$ holds,
$
\frac{(\postset{\omega_1}-\postset{\omega_2})-(\omega_1-\omega_2)}{\vsmall}
= f_1(\omega_1,\tau_1) - f_2(\omega_2,\tau_2) $ cannot be a finite
nonstandard real because, if it was, then the function
$t \mapsto \omega_1(t)-\omega_2(t)$ would be right-continuous, in contradiction with the above
assumption.  Hence, the nonstandard real
$f_1(\omega_1,\tau_1) - f_2(\omega_2,\tau_2)$ is necessarily infinite.
However, we assumed continuous functions $f_i$ and finite state
$(\omega_1,\omega_2)$.  Thus, one of the torques $\tau_i$ must be
infinite at $t$, and because of equation $(\eqq_4): \tau_1+\tau_2=0$,
both torques are in fact infinite, i.e., are \emph{impulsive}.

\myparagraph{Eliminating impulsive variables} It remains to compute
the (standard) restart values for the state variables.  To this end,
we assume that the $f_i$'s are linear in the torques, i.e., each $f_i$ has the
following form:
\beq
f_i(\omega_i,\tau_i) &=& a_i(\omega_i)+b_i(\omega_i)\tau_i \enspace ,
\label{eq:peowigtu9p}
\eeq
where $b_1$ and $b_2$ are the inverse moments of
inertia of the rotating masses and $a_1$ and $a_2$ are damping factors
divided by the corresponding moments of inertia.
This yields the following system of equations, to be solved for
$\postset{\omega_1},\postset{\omega_2},\tau_1,\tau_2$ at the instant
when $\guard$ switches from $\fff$ to $\ttt$: \beq \left\{ \bea{lcc}
  \postset{\omega_1}=\omega_1+{\vsmall}.(a_1(\omega_1)+b_1(\omega_1)\tau_1) &&(\eqq_1^\vsmall) \\
  \postset{\omega_2}=\omega_2+{\vsmall}.(a_2(\omega_2)+b_2(\omega_2)\tau_2) &&(\eqq_2^\vsmall) \\
  \postset{\omega_1}-\postset{\omega_2}=0 &&(\postset{\eqq_3}) \\
  \tau_1+\tau_2=0 &&(\eqq_{4}) \eea \right.
\label{sys:wp49guhsopiu}
\eeq It is tempting to standardize \rref{sys:wp49guhsopiu} by simply
setting $\vsmall=0$ in it. Unfortunately, doing this leaves us with a
structurally singular system, since the two torques are then involved
in only one equation.

\begin{ccomment}\rm\textbf{(standardizing functions vs. standardizing equations)}
	\label{keruif} To get a better insight into this problem, let us replace for a while the infinitesimal $\vsmall$ by a `small' (standard) time step $\delta$ and rewrite \rref{sys:wp49guhsopiu} as $F(\postset{\omega_1},\postset{\omega_2},\tau_1,\tau_2,{\omega_1},{\omega_1},\delta)=0\,,$ where $F{=}0$ stacks the four equations of (\ref{sys:wp49guhsopiu}) into one vector equation. Then, the function $\delta\mapsto{F(\postset{\omega_1},\postset{\omega_2},\tau_1,\tau_2,{\omega_1},{\omega_1},\delta)}$ is continuous in $\delta$ at $\delta{=}0$. Consequently, by (\ref{erlfuihliu}), the standardization of  $F(\postset{\omega_1},\postset{\omega_2},\tau_1,\tau_2,{\omega_1},{\omega_1},\vsmall)$, seen as a function, is obtained by setting $\vsmall{=}0$ in it. 
	
	This, however, becomes incorrect if, instead of the function $F$, we are interested in the equation $F{=}0$. More precisely, in our example, let $Y$ be the vector that collects the states $\omega_1,\omega_2$, and $X$ be the vector that collects the dependent variables $\postset{\omega_1},\postset{\omega_2},\tau_1,\tau_2$; then, the function $F$ defined by \rref{sys:wp49guhsopiu} is linear in $\delta$, i.e., it has the form $Y\mapsto{X}$, where $X$ is solution of the algebraic equation 
	$$[A-\delta{\times}B(Y)]X=C(\delta,Y),$$ 
	in which
	\beqq
	A=\left[\bea{cccc}
	1 & 0 & 0 & 0 \\
	0 & 1 & 0 & 0 \\
	1 & -1 & 0 & 0 \\
	0 & 0 & 1 & 1 
	\eea\right] 
	&;&
	B=B(Y)=\left[\bea{cccc}
	0 & 0 & b_1(\omega_1) & 0 \\
	0 & 0 & 0 & b_2(\omega_2) \\
	0 & 0 & 0 & 0 \\
	0 & 0 & 0 & 0
	\eea\right] \\
	\eeqq
	and vector $C(\delta,Y)$ is appropriately fixed; from now on, we omit $Y$ in the writing.
For every $\delta > 0$, the matrix $A{-}\delta{\times}B$ is invertible and
we have 
$
X{=}(A{-}\delta{\times}B)^{-1}C.
$

If $A$ was nonsingular, we could simply write
$X = (I{-}\delta{\times}D)^{-1}A^{-1}C$, where
$D = A^{-1}B$. Expanding the inverse around $\delta = 0$ would yield
$(I{-}\delta{\times}D)^{-1}=I{+}\delta{\times}D{+}o(\delta)$,
implying $X = (I{+}\delta{\times}D)A^{-1}C{+}o(\delta)$. Therefore, the
solution $X(\vsmall)$ of equation
$(A{-}\vsmall{\times}B(Y))X{=}C(\vsmall,Y)$ would standardize as the
solution $X$ of the standard equation $AX=C(0,Y)$.
All of this, however, requires that $A$ be nonsingular, which is
wrong in our case.\footnote{This informal reasoning will be formalized
  in Section~\ref{ptw9ugfdll}, in the context of nonstandard
  analysis.}\eproof
\end{ccomment}
When $A$ is singular (as in our case), the solution consists in eliminating the impulsive variables from \rref{sys:wp49guhsopiu}, namely, the two torques. This is performed by symbolic pivoting:
\begin{enumerate}
	\item Using $(\eqq_{4})$ yields $-\tau_2=\tau_1\eqdef\tau$;
	\item Premultiplying the system of equations
	\[
	\left\{\bea{lcc}
\postset{\omega_1}=\omega_1+{\vsmall}.(a_1(\omega_1)+b_1(\omega_1)\tau) &&(\eqq_1^\vsmall) \\
\postset{\omega_2}=\omega_2+{\vsmall}.(a_2(\omega_2)-b_2(\omega_2)\tau) &&(\eqq_2^\vsmall) 
\eea\right.
	\] 
	by the row matrix $
	\left[
\bea{cc} b_2(\omega_2) & b_1(\omega_1)\eea
	\right]
	$
	yields the equation
	\[	b_2(\omega_2)\postset{\omega_1}+b_1(\omega_1)\postset{\omega_2}=b_2(\omega_2)({\omega_1}+\vsmall.a_1(\omega_1))
+ b_1(\omega_1)({\omega_2}+\vsmall.a_2(\omega_2)) \enspace .
\]
	\item Using in addition $(\postset{\eqq_3})$ and setting $\postset{\omega}\eqdef\postset{\omega_1}=\postset{\omega_2}$ finally yields
\end{enumerate}
\beq
\postset{\omega} = \displaystyle\frac{b_2(\omega_2) \omega_1 + b_1(\omega_1) \omega_2}{b_1(\omega_1) + b_2(\omega_2)} 
 + \vsmall\,\displaystyle\frac{a_1(\omega_1) b_2(\omega_2) + a_2(\omega_2) b_1(\omega_1)}{b_1(\omega_1) + b_2(\omega_2)}  \enspace .
\label{wrthoijoipouih}
\eeq
The symbolic-only rewriting used here actually alleviates any difficulty in standardizing equation (\ref{wrthoijoipouih}): this can now be done by just setting $\vsmall=0$ in its right-hand side. This yields, by identifying $\st{\omega_i}=\omega_i^-$ and $\st{{\postset{\omega_i}}}=\omega_i^+$: 
\beq
\omega_1^+=\omega_2^+=
\frac{b_2(\omega_2^-) \omega_1^- + b_1(\omega_1^-) \omega_2^-}{b_1(\omega_1^-) + b_2(\omega_2^-)}  ~ ,
\label{eq:owgtuiho}
\eeq
that is, the reset value for the two velocities is the weighted arithmetic mean of $\omega_1^-$ and $\omega_2^-$.
 \rref{eq:owgtuiho} provides us with the reset values for the positions in the engaged mode, which is enough to restart the simulation in this mode. 
The actual impulsive values for the torques are useless and can be discarded.
Note, however, that their sum explicitly remains equal to 0 during the instant of mode change.
As a result, the momentum is preserved by the mode change, so that the obtained solution is satisfactory from a physical point of view.

As a final observation, instead of computing the exact standard part of $\postset{\omega_i}$, one could instead 
attempt to approximate it numerically by substituting $\vsmall$ with a small (non-infinitesimal) step size 
$\delta$ in \rref{sys:wp49guhsopiu}. This alternative approach is developed in Section~\ref{sec:restartschemes}. It has the advantage of not requiring the elimination of the impulsive variables.

We thus achieved the objectives stated at the end of
Section~\ref{lireutyliuel}.
\begin{figure}[bt]
    \centering
    \includegraphics[width=0.6\textwidth]{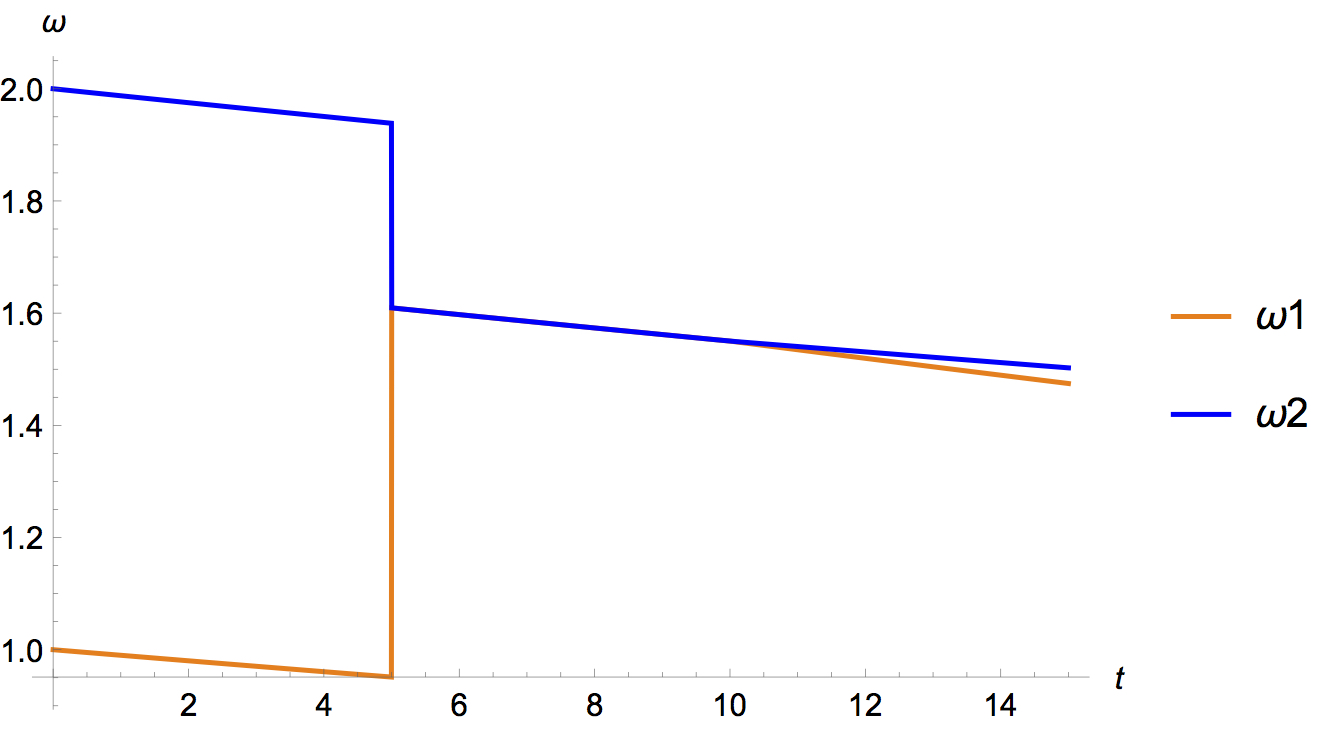}
    \caption{
      \small\sf Simulation of the clutch model with resets, using our method. Mode
      change $\fff\ra\ttt$ occurs at $t=5s$ and mode change
      $\ttt\ra\fff$ occurs at $t=10s$.}
    \label{fig:clutch0}
\end{figure}
\rref{fig:clutch0} shows a simulation of the clutch model where the
resets are computed following our approach developed hereabove.  As one may
expect, the reset value sits between the two values of
$\omega_1^-$ and $\omega_2^-$ when $\guard: \fff \to \ttt$ (at $t=5$s),
and the transition is continuous at the second reset (at $t=10$s).
\section{{The Cup-and-Ball example}}
\label{loeruighlrigtuh}
In this section, we study the example of Section~\ref{leriugfehliu} and address the objectives stated at the end of that section.

\subsection{Rejecting the submitted model}
We first explain why we reject the model proposed in
Section~\ref{leriugfehliu}, and how we do so. Let us consider
again System~(\ref{riuytuikdastf}).

\paragraph{Nonstandard semantics} Following \rref{sec:nsa} for the clutch, the derivatives $\dot{x},\dot{y},\ddot{x},\ddot{y}$ are replaced using the explicit first-order Euler scheme with an infinitesimal time step $\vsmall$: 
\beq
\dot{x} \gets \frac{\postset{x}-x}{\vsmall} ~~;~~
\ddot{x} \gets \frac{\ppostset{2}{x}-2\postset{x}+x}{\vsmall^2} \enspace .
\label{ergoiuwhio}
\eeq
In the sequel, when referring to model (\ref{riuytuikdastf}), substitution (\ref{ergoiuwhio}) will be implicit.

\paragraph{Structural analysis} $x,y,\postset{x},\postset{y}$ are the state variables, the value of which is known from previous nonstandard instants. The system has guard $\guard$. In contrast to the clutch example, $\guard$ is no longer an input, but is rather evaluated based on system variables using equation $(\straight_0)$. For both modes, the leading variables are $\ppostset{2}{x},\ppostset{2}{y},\tension$; the first two of these arise from the second derivatives $\ddot{x},\ddot{y}$ by using substitution (\ref{ergoiuwhio}). 
\begin{ccomment} \rm\textbf{(logico-numerical fixpoints)}
	\label{elritheliru} We could consider the direct solving of System (\ref{riuytuikdastf}), seen as a mixed logico-numerical system of equations with dependent variables $\ppostset{2}{x},\ppostset{2}{y},\tension,\guard$. However, this system mixes guarded numerical equations with the evaluation of a predicate. A possible solution would consist in performing a relaxation, by iteratively updating the numerical variables based on the previous value for the guards, and then re-evaluating the guard based on the updated values of the numerical variables, hoping for a fixpoint to occur. Such fixpoint equation, however, can have zero, one, several, or infinitely many solutions. No characterization exists that could serve as a basis for a (graph-based) structural analysis. We thus decided not to try to solve such mixed logico-numerical systems.\eproof
\end{ccomment}
As a consequence, we are unable to evaluate guard $\guard$, so that the mode the system is in cannot be determined. Should we give up? Not yet. 

It could be that, based on information obtained by solving the subsystem consisting of the equations that are always enabled, guard $\guard$ can be evaluated. In this model, this subsystem is $G=\{(\eqq_1),(\eqq_2)\}$, with dependent variables $\ppostset{2}{x},\ppostset{2}{y},\tension$. Unfortunately, subsystem $G$ is underdetermined: it cannot be solved, and the evaluation of guard $\guard$ is still impossible. 
We thus reject model~(\ref{riuytuikdastf}), because we are unable to evaluate guard $\guard$ on the basis of the subsystem that is enabled prior to evaluating this guard.

\begin{ccomment}\rm\textbf{(forbidding logico-numerical fixpoints based on syntax)} 
	\label{leriuhlu}	
	The reader may wonder why we did not simply reject the model right from the beginning, on the sole basis of the existence of the logico-numerical fixpoint equation. The rationale is that the method discussed above allows, in principle, more models to get accepted than just the ones without such equations. This design issue will be addressed more thoroughly in Comment~\ref{guiohpiouh} of \rref{sec:oguihwpgui}.\eproof
\end{ccomment}

\subsection{Correcting the model}

To break the fixpoint equation defining $\guard$, the idea is to replace the predicate $s\leq 0$ defining the guard $\guard$ by its left-limit: ${\guard}= [s^-\leq{0}]$, where $s^-(t)\eqdef\lim_{u\nearrow{t}}s(u)$:  
\beqq
\left\{\bea{rll}
& 0= \ddot{x}+{\tension}x & (\eqq_1) \\
& 0= \ddot{y}+{\tension}y+g  & (\eqq_2) \\
& \remph{\guard= [s^-\leq{0}]}  & \remph{(\straight_0)} \\
\when \; \guard\; \doo& 0={L^2}{-}(x^2{+}y^2)   & (\straight_1) \\
\prog{and}& 0=\tension+s   & (\straight_2) \\
\when \;\prog{not}\; \guard\; \doo& 0=\tension   & (\straight_3) \\
\prog{and}& 0=({L^2}{-}(x^2{+}y^2))-s   & (\straight_4) \\ 
\eea\right.
\eeqq
We map this model to nonstandard semantics by stating that predicate $s\leq 0$ defines the value of the guard \emph{at the next nonstandard instant}---therefore, an initial condition for guard $\guard$ is required; we initialize the system in free motion mode $\guard(0)=\fff$. This yields the corrected model (\ref{loeifuhpwoui}), where the modification is highlighted in {\color{red}red}:
\beq\left\{\bea{rll}
& 0= \ddot{x}+{\tension}x & (\eqq_1) \\
& 0= \ddot{y}+{\tension}y+g  & (\eqq_2) \\
& \remph{\postset{\guard}= [s\leq{0}]; \guard(0)=\fff}  & \remph{(\straight_0)} \\
\when \; \guard\; \doo& 0={L^2}{-}(x^2{+}y^2)   & (\straight_1) \\
\prog{and}& 0=\tension+s   & (\straight_2) \\
\when \;\prog{not}\; \guard\; \doo& 0=\tension   & (\straight_3) \\
\prog{and}& 0=({L^2}{-}(x^2{+}y^2))-s   & (\straight_4) \\
\eea\right.
\label{loeifuhpwoui}
\eeq
This model is understood in the nonstandard semantics,
meaning that the derivatives are expanded using (\ref{ergoiuwhio}). Therefore, the leading variables in all modes are  $\tension,s,\ppostset{2}{x},\ppostset{2}{y}$. 


\subsection{Nonstandard structural analysis} 
\label{sec:lerighugholuih}
This section focuses on the core difficulty, namely the mode change $\guard:\fff\ra\ttt$, i.e., the event when the rope gets straight. Due to equation $(\straight_1)$, the mode $\guard=\ttt$ (where the rope is straight) requires index reduction. We thus augment model~(\ref{loeifuhpwoui}) with the two latent equations shown in {\color{red}red}:
\beq\left\{\bea{rll}
& 0= \ddot{x}+{\tension}x & (\eqq_1) \\
& 0= \ddot{y}+{\tension}y+g  & (\eqq_2) \\
& {\postset{\guard}= [s\leq{0}]; \guard(0)=\fff}  & {(\straight_0)} \\
\when \; \guard\; \doo& 0={L^2}{-}(x^2{+}y^2)   & (\straight_1) \\
\prog{and}& \remph{0={L^2}{-}\postset{(x^2{+}y^2)}}   & \remph{(\postset{\straight_1})} \\
\prog{and}& \remph{0={L^2}{-}\ppostset{2}{(x^2{+}y^2)}}   & \remph{(\ppostset{2}{\straight_1})} \\
\prog{and}& 0=\tension+s   & (\straight_2) \\
\when \;\prog{not}\; \guard\; \doo& 0=\tension   & (\straight_3) \\
\prog{and}& 0=({L^2}{-}(x^2{+}y^2))-s   & (\straight_4) \\
\eea\right.
\label{uioghrpui}
\eeq
In this model, the two equations $(\straight_1)$ and $\remph{(\postset{\straight_1})}$ are in conflict with equations from the previous two instants, shown in {\color{blue}blue} in the following subsystem:
\[
\left\{\bea{ll}
 \bemph{0= \frac{x-2\,{\preset{\!x}}+\,\ppreset{2}{\!x}}{\vsmall^2}+{\ppreset{2}{\!\tension}}\,{\ppreset{2}{\!x}}} & \bemph{(\ppreset{2}{\eqq_1})} \\
[1mm] \bemph{0= \frac{y-2\,\preset{\!y}+\,\ppreset{2}{\!y}}{\vsmall^2}+{\ppreset{2}{\!\tension}}\,{\ppreset{2}{\!y}}+g} & \bemph{(\ppreset{2}{\eqq_2})} \\
[1mm]
 \bemph{0= \frac{\postset{x}-2x+\,{\preset{\!x}}}{\vsmall^2}+{\preset{\!\tension}}\,{\preset{\!x}}} & \bemph{(\preset{\eqq_1})} \\
[1mm] \bemph{0= \frac{\postset{y}-2y+\,{\preset{\!y}}}{\vsmall^2}+{\preset{\!\tension}}\,{\preset{\!y}}+g} & \bemph{(\preset{\eqq_2})} \\
[1mm]
0={L^2}{-}(x^2{+}y^2)   & (\straight_1) \\ [1mm]
\remph{0={L^2}{-}\postset{(x^2{+}y^2)}}   & \remph{(\postset{\straight_1})} \\
\eea\right.
\]
Applying again Principle~(\ref{oeuryfgy}) of causality, we erase, in model~(\ref{uioghrpui}), equations $(\straight_1)$ and $(\postset{\straight_1})$ at the instant of mode change $\preset{\guard}{=}\fff,\guard{=}\ttt$. This yields the following system:
\beq
\mbox{at }\left[\!\!\bea{r}\preset{\guard}{=}\fff
\\
{\guard}{=}\ttt\eea\!\!\right]
:
\left\{\bea{ll}
 0= \ddot{x}+{\tension}x & (\eqq_1) \\
 0= \ddot{y}+{\tension}y+g  & (\eqq_2) \\
{0={L^2}{-}\ppostset{2}{(x^2{+}y^2)}}   & {(\ppostset{2}{\straight_1})} \\
0=\tension+s   & (\straight_2) \\
\eea\right.
\label{sys:louihpoui}
\eeq
It uniquely determines all the leading variables from the state variables $x,y$ and $\postset{x},\postset{y}$. In turn, equations $(\straight_1)$ and $(\postset{\straight_1})$, which were erased from this model, are not satisfied.

Also, we erase, in model~(\ref{uioghrpui}), the only equation $({\straight_1})$ at the next instant, i.e., when $\ppreset{2}{\guard}{=}\fff,\preset{\guard}{=}\ttt,\guard{=}\ttt$. This yields the following system:
\beq
\mbox{at }\left[\!\!\bea{r}\ppreset{2}{\guard}{=}\fff
\\
\preset{\guard}{=}\ttt
\\
\guard{=}\ttt\eea\!\!\right]
:
\left\{\bea{ll}
 0= \ddot{x}+{\tension}x & (\eqq_1) \\
 0= \ddot{y}+{\tension}y+g  & (\eqq_2) \\
{0={L^2}{-}\postset{(x^2{+}y^2)}}   & {(\postset{\straight_1})} \\
{0={L^2}{-}\ppostset{2}{(x^2{+}y^2)}}   & {(\ppostset{2}{\straight_1})} \\
0=\tension+s   & (\straight_2) \\
\eea\right.
\label{sys:rgiouhpruio}
\eeq
This uniquely determines all the leading variables, given values for the state variables $x,y$ and $\postset{x},\postset{y}$. Note that $(\postset{\straight_1})$ is a consistency equation that is satisfied by the state variables $\postset{x},\postset{y}$. In turn, equation $(\straight_1)$, which was erased from this model, is not satisfied.
At subsequent instants, equation erasure is no longer needed.
This completes the nonstandard structural analysis of the mode change $\guard:\fff\ra\ttt$, i.e., when the rope gets straight.

\subsection{Standardization}
\label{sec:oguhkiu}

Code generation for restarts consists in standardizing nonstandard systems (\ref{sys:louihpoui}) and (\ref{sys:rgiouhpruio}). Remember, from our development of the clutch example, that standardizing systems of equations requires more care than standardizing numbers, due to impulsive behaviors and singularity issues that result.
The simplest method from a conceptual point of view is to eliminate impulsive variables from the restart system, as they are of no use for restarting the new mode. 

We focus on the standardization of the mode change $\guard:\fff\ra\ttt$, i.e., when the rope gets straight. Our task is to standardize systems (\ref{sys:louihpoui}) and (\ref{sys:rgiouhpruio}), by targeting discrete-time dynamics, for the two successive instants composing the restart phase. This will provide us with restart values for positions and velocities.

Due to the expansion of derivatives in equations $(\eqq_1,\eqq_2,\postset{\eqq_1},\postset{\eqq_2})$, tensions $\tension$ and $\postset{\tension}$ are both impulsive, hence so are $s$ and $\postset{s}$ by $({\straight_2},\postset{\straight_2})$. We eliminate the impulsive variables by ignoring $({\straight_2},\postset{\straight_2})$, combining $(\eqq_1)$ and $(\eqq_2)$ to eliminate $\tension$, and $(\postset{\eqq_1})$ and $(\postset{\eqq_2})$ to eliminate $\postset{\tension}$. This
yields:
\beq
\mbox{at }\left[\!\!\bea{r}\preset{\guard}{=}\fff
\\
{\guard}{=}\ttt\eea\!\!\right]
&:&
\left\{\bea{l}
0=\ddot{y}x+gx-\ddot{x}y \\ 
{0={L^2}{-}\ppostset{2}{(x^2{+}y^2)}} \\
\eea\right.
\label{sys:erkuyfgeouryiyu}
\\
\mbox{at }\left[\!\!\bea{r}\ppreset{2}{\guard}{=}\fff
\\
\preset{\guard}{=}\ttt
\\
\guard{=}\ttt\eea\!\!\right]
&:&
\left\{\bea{l}
{0={\ddot{y}}{x}+g {x}-{\ddot{x}}{y}} \\ 
{0={L^2}{-}\postset{(x^2{+}y^2)}}\\ 
{0={L^2}{-}\ppostset{2}{(x^2{+}y^2)}}
\eea\right.
\label{sys:rkufygekyu}
\eeq
In \rref{sys:erkuyfgeouryiyu}, we expand second derivatives using (\ref{ergoiuwhio}). 
Note that substitution (\ref{ergoiuwhio}) implies that we can as well expand the second derivatives as functions of first derivatives:
\beq
\ddot{x}&\gets&\frac{\postset{\dot{x}}-\dot{x}}{\vsmall} \enspace .
\label{leriugeliuw}
\eeq
Since we will use \rref{sys:rkufygekyu} to evaluate restart values for the velocities,
we will use expansion (\ref{leriugeliuw}) in it. Consequently, 
\rref{sys:erkuyfgeouryiyu} has dependent variables $\ppostset{2}{x},\ppostset{2}{y}$, whereas \rref{sys:rkufygekyu} has dependent variables  $\postset{\dot{x}},\postset{\dot{y}}$. We are now ready to standardize the two systems.

\paragraph{\rref{sys:erkuyfgeouryiyu} to define restart positions} We expand second derivatives using (\ref{ergoiuwhio}):
\beq
\left\{\bea{l}
0=(\ppostset{2}{y}-2\postset{y}+y)x-(\ppostset{2}{x}-2\postset{x}+x)y+\vsmall^2 gx \\ 
{0={L^2}{-}\ppostset{2}{(x^2{+}y^2)}} \\
\eea\right.
\label{kejfygeoru}
\eeq
Regarding~(\ref{kejfygeoru}), one can safely set $\vsmall={0}$ to standardize it, since the resulting system is still structurally regular---see Comment\,\ref{keruif} regarding the clutch example. The standardization of (\ref{kejfygeoru}) is thus:
\beq
\left\{\bea{l}
0=(\ppostset{2}{y}-2\postset{y}+y)x-(\ppostset{2}{x}-2\postset{x}+x)y \\ 
{0={L^2}{-}\ppostset{2}{(x^2{+}y^2)}} \\
\eea\right.
\label{rieufgoeiu}
\eeq
In contrast, had we directly set $\vsmall=0$ in \rref{sys:louihpoui} (without eliminating impulsive variable $\tension$), we would get a structurally singular system, an incorrect standardization.

In System (\ref{rieufgoeiu}), we can interpret $x$ and $\postset{x}$ as the left-limit $x^-$ of state variable $x$ in previous mode, and $\ppostset{2}{x}$ as the restart value $x^+$ for the new mode. This yields
\beq
\left\{\bea{l}
0=({y^+}-y^-)x^--({x^+}-x^-)y^- \\ 
0={L^2}{-}{(x^2{+}y^2)^+} \\
\eea\right.
\label{leriugheku}
\eeq
which determines the restart values for positions. Note that the constraint that the rope is straight is satisfied. Furthermore, if we remember that the rope is straight at mode change (this is the mode change condition), then $0={L^2}{-}{(x^2{+}y^2)^-}$ also holds, which ensures that $x^+=x^-,y^+=y^-$ is the unique solution of (\ref{leriugheku}): positions are continuous.

\paragraph{\rref{sys:rkufygekyu} to define restart velocities} We expand second derivatives using (\ref{leriugeliuw}):
\beq
\left\{\bea{l}
0=(\postset{\dot{y}}-{\dot{y}}){x}-(\postset{\dot{x}}-{\dot{x}}){y} +\vsmall.g{x} \\ 
0={L^2}{-}\postset{(x^2{+}y^2)}\\ 
0={L^2}{-}\ppostset{2}{(x^2{+}y^2)}
\eea\right.
\label{gtliuhrliu}
\eeq
By expanding $\ppostset{2}{x}=\postset{x}+\vsmall\postset{\dot{x}}$, the right-hand side of the last equation rewrites 
\beq\bea{rcl}
{L^2}{-}\ppostset{2}{(x^2{+}y^2)} =&&  {L^2}{-}\postset{(x^2{+}y^2)}
\\ [1mm]
&\!\!\!\!+\!\!\!\!&  2\vsmall (\postset{x}\postset{\dot{x}}+\postset{y}\postset{\dot{y}})
\\ [1mm]
&\!\!\!\!+\!\!\!\!& \vsmall^2\bigl((\postset{\dot{x}})^2+(\postset{\dot{y}})^2\bigr)
\\ [1mm]
=&& {0} \mbox{ ~(by second equation of (\ref{gtliuhrliu}))}
\\ [1mm]
&\!\!\!\!+\!\!\!\!& 
2\vsmall (\postset{x}\postset{\dot{x}}+\postset{y}\postset{\dot{y}})\\ [1mm]
&\!\!\!\!+\!\!\!\!& O(\vsmall^2)
\eea
\label{ekljgkeuygku}
\eeq
Using this expansion of ${L^2}{-}\ppostset{2}{(x^2{+}y^2)}$, setting $\vsmall=0$ in (\ref{gtliuhrliu}) yields
\beq
\left\{\bea{l}
0=(\postset{\dot{y}}-{\dot{y}}){x}-(\postset{\dot{x}}-{\dot{x}}){y}  \\ 
{0=\postset{x}\postset{\dot{x}}+\postset{y}\postset{\dot{y}}}
\eea\right.
\label{sys:lwreiupowteiu}
\eeq
where the dependent variables are now $\postset{\dot{x}},\postset{\dot{y}}$, whereas other variables are state variables whose values are determined by previous time steps. Note that \rref{sys:lwreiupowteiu} is structurally regular, hence it is the correct standardization of System (\ref{gtliuhrliu}).
We are now ready to get restart values for velocities. In \rref{sys:lwreiupowteiu}, we perform the following substitutions, recalling that superscripts $^-$ and $^+$ denote left- and right-limits:
\beq
{x}=x^- \;;\; \postset{x}=x^+ \;;\;  {\dot{x}}= \dot{x}^- 
\label{lehtrdriuhpiu}
\eeq
and the restart velocity is evaluated by
\beq
\dot{x}^+=\postset{\dot{x}} ~~ . 
\label{somedumblabel}
\eeq
After similar substitutions for $y$, we get
\beq\left\{\bea{l}
0=({\dot{y}}^+-{\dot{y}}^-){x^-}-({\dot{x}}^+-{\dot{x}}^-){y^-}  \\ 
{0={x^+}{\dot{x}}^++{y^+}{\dot{y}}^+}
\eea\right.
\label{sys:kerfygouy}
\eeq
\rref{sys:kerfygouy} determines ${\dot{x}}^+$ and ${\dot{y}}^+$, which are the velocities for restart. The second equation guarantees that the velocity will be tangent to the constraint. With (\ref{leriugheku}) and (\ref{sys:kerfygouy}), we determine the restart conditions for positions and velocities. Invariants from the physics are satisfied.

Our reasoning so far produces a behavior in which the two modes (free motion and straight rope) gently alternate; the system always stays in one mode for some positive period of time before switching to the other mode. \emph{This indeed amounts to assuming that the impact is totally inelastic at mode change, an assumption that was
not explicit at all  in $(\ref{loeifuhpwoui})$.} So, what happened?

\subsection{Handling transient modes}
\label{sec:leiruioepui}

In fact, \emph{the straight rope mode was implicitly assumed to last for at least three nonstandard successive instants, since we allowed ourselves to shift $(\straight_1)$ twice}.  

Now, let us instead assume elastic impact. Hence, the cascade of mode changes $\guard:\fff\ra\ttt\ra\fff$ is possible in three successive nonstandard instants, reflecting that the straight rope mode is \emph{transient} (it is left immediately after being reached). We wish to adapt our reasoning to handle this transient mode.

Consider again model~(\ref{loeifuhpwoui}). We regard the instant of the cascade when $\guard=\ttt$ occurs as the current instant.
%
We cannot add latent equations by simply shifting $(\straight_1)$, since these shifted versions are not active in the mode $\guard=\fff$. Instead, our line of reasoning goes as follows. Set
\[\bea{rcl}
\system(\ttt)&=& \{(\eqq_1),(\eqq_2),(\straight_1),(\straight_2)\} \enspace , \\
\system(\fff)&=& \{(\eqq_1),(\eqq_2),(\straight_3),(\straight_4)\} \enspace .
\eea
\]
Systems $\postset{\system}(\ttt)$ and $\postset{\system}(\fff)$ are obtained by shifting the equations constituting ${\system}(\ttt)$ and ${\system}(\fff)$; systems $\ppostset{2}{\system}(\ttt)$ and $\ppostset{2}{\system}(\fff)$ are defined similarly.

Now, the idea is to consider the \emph{differentiation array} originally proposed by Campbell and Gear~\cite{CampbellGear1995}, except that we take into account the trajectory $\ttt,\fff,\fff,\dots$ for guard $\guard$. Using shifting instead of differentiation yields the following \emph{array of shifted systems}, where size $n$ is a parameter: 
\beq
\cA_n(\system)&\eqdef&
\left[\bea{r}
\system(\ttt) \\ \postset{\system}(\fff) \\ \ppostset{2}{\system}(\fff) \\ \vdots~~~~ \\ \ppostset{n}{\system}(\fff)
\eea\right]
\label{weoriguwehroigu}
\eeq
We search for an $n$ such that $\cA_n(\system)$ provides us with the due latent equations. Thus, the question is: can we find latent equations for $\system(\ttt)$ by shifting appropriate equations from $\system(\fff)$? Unfortunately, the answer is no: although shifting twice $(\straight_4)$ in System\,(\ref{loeifuhpwoui}) produces one more equation involving the leading variables $\ppostset{2}{x},\ppostset{2}{y}$, this equation also involves the new variable $\ppostset{2}{s}$, which keeps the augmented system underdetermined; shifting other equations fails as well.

Therefore, the structural analysis rejects this model as being underdetermined at transient mode $\guard=\ttt$. The user is then asked to provide one more equation. For example, he/she could specify an impact law for the velocity $\dot{y}$ by providing the equation 
\[
(\dot{y})^+=-(1 - \alpha) {(\dot{y})^-} \enspace ,
\]
where $0\leq\alpha<1$ is a fixed damping coefficient. This is reinterpreted in the nonstandard domain as
\beq
\postset{\dot{y}}=-(1 - \alpha) {\dot{y}} \enspace ,
\label{lruifhepriuhiu}
\eeq
yielding the following refined system for use at mode $\guard{=}\ttt$ within the cascade $\guard{:}\fff{\ra}\ttt{\ra}\fff$:
\beq
\left\{\bea{ll}
0= \ddot{x}+{\tension}x & (\eqq_1) \\
0= \ddot{y}+{\tension}y+g  & (\eqq_2) \\ 
0= \postset{\dot{y}}+(1 - \alpha) {\dot{y}}  & (\tau_1) \\ 
0={L^2}{-}(x^2{+}y^2)   & (\straight_1) \\
0=\tension+s   & (\straight_2) 
\eea\right.
\label{reotpowilwefuh}
\eeq
We proceed again with the structural analysis. Variables $x,y$ are the states, so that their values are set by the previous instants. Current equation $(\straight_1)$ creates a conflict with the past. Hence, we discard it from System\,(\ref{reotpowilwefuh}), which leaves us with the following system:
\beq
\left\{\bea{ll}
0= {\ddot{x}}+{\tension}x & (\eqq_1) \\
0= \ddot{y}+{{\tension}}y+g  & (\eqq_2) \\ 
0= {\postset{\dot{y}}}+(1 - \alpha) {\dot{y}}  & (\tau_1) \\ 
0=\tension+{s}   & (\straight_2) 
\eea\right.
\label{leriutywuily}
\eeq
Model (\ref{leriutywuily}) is structurally nonsingular, recalling that $\ddot{y}$ and $\postset{\dot{y}}$ can be interchanged for the structural analysis. This refined model is therefore accepted. Standardization proceeds accordingly, with no new difficulty.

\subsection{{Consequences for the modeling language}}
\label{sec:language}

Through the Cup-and-Ball example, we demonstrated the need for the following user-given information: is the current mode long or transient? We argue that this information cannot be derived from inspecting the system model.
 
In the clutch example, guard $\guard$ was a system input, and we implicitly assumed that each of the two modes was long. Had we composed the clutch model with a `feedback' determining $\guard$ as a predicate over system variables, this implicit assumption would have become questionable. This situation indeed occurs with the Cup-and-Ball example, where the elastic/inelastic impact assumption is a side information that cannot be deduced from the equation syntax.
%
%
 This leads to the following important observation: 
\begin{ccomment} \rm\textbf{(Long\,/\,Transient)}
	\label{elrigufehilu}
	Long\,/\,Transient is an information regarding modes, that cannot be found by inspecting the equations of the system. It must be inferred from understanding the system physics. Therefore, it has to be provided as an extra information by the model designer.\eproof
\end{ccomment}
The natural way of addressing Comment~\ref{elrigufehilu} is to provide a different syntax for specifying long modes on the one hand, and events corresponding to transient modes on the other hand. In contrast, mode changes separating two successive long modes need not be specified. A candidate syntax is:
\beq\mbox{
\begin{tabular}{|r|r|} \hline
	 long mode & \texttt{if predicate then $S$}  \\ \hline
	 transient mode & \texttt{when predicate then $S$} \\ \hline
\end{tabular}}
\label{erliuheoui}
\eeq
The `\texttt{when}' statement of Modelica, which we reuse here, points to the event when `\texttt{predicate}' switches from false to true. For this `\texttt{when}' statement, we could furthermore restrict `\texttt{predicate}' to be a zero-crossing condition, by which a real-valued expression crosses zero from below~\cite{lucy:hscc13}.
We devote the `\texttt{if}' statement to long-lasting modes specified by `\texttt{predicate}'. 
%

To illustrate this, our syntax offering \ifequation{s} and \whenequation{s} allows specifying the Cup-and-Ball example with elastic impact as shown in \rref{fig:elguioehpguio}.
\begin{figure}[ht]
\[
\left\{\bea{rll}
& 0= \ddot{x}+{\tension}x & (\eqq_1) \\
& 0= \ddot{y}+{\tension}y+g  & (\eqq_2) \\
& {{\guard}= [s^-\leq{0}]; \guard(0)=\fff}  & {(\straight_0)} \\
\remph{\prog{when}\;\guard\; \doo}& \remph{\dot{y}=-\alpha{\dot{y}}^-} & \remph{(\tau_1)}
\\
\when \;\prog{not}\; \guard\; \doo& 0=\tension   & (\straight_3) \\
\prog{and}& 0=({L^2}{-}(x^2{+}y^2))-s   & (\straight_4) \\
\eea\right.
\]
\caption{ Model of the Cup-and-Ball example with elastic impact specified by equation $(\tau_1)$, highlighted in {\color{red}red}. This equation replaces the two equations $(\straight_1),(\straight_2)$ specifying the straight rope mode in case of inelastic impact.
}
\label{fig:elguioehpguio}
\end{figure}
\section{The Westinghouse air brake example}
\label{sec:railcar}
In this section we develop the example of Section~\ref{elriuftrheliu} and we address the objectives stated at the end of that section. 

We reject the model proposed in Section~\ref{elriuftrheliu} for the very same reasons as for the Cup-and-Ball example, namely the existence of a logico-numerical fixpoint equation defining the guard. The corrected system is shown in \rref{fig:railcarcorrec}.
\begin{figure}[ht]
	 \centerline{
$\bea{rlc}
&0=f_b - f_v - f_{cl} + f_l &(m_1)\\
&0=f_v - f_{ch} - f_l - f_t &(m_2)\\
&0=f_b - F_1.\flow(p_b-p_r) &(f_1)\\
&0= F_1.\flow(p_{bn}-p_r) &(f_2)\\
&0=f_l - F_2.\flow(p_t-p_r) &(f_3)\\
&0=f_t - \frac{\rho.V}{P_0}.\dot{p_t} &(r)\\
&0= \rho.S.(\dot{x}.p_r + (x-L).\dot{p_r}) + P_0.f_{cl} &(p_1)\\
&0=\rho.S.(\dot{x}.S.p_t + x.\dot{p_t}) - P_0.f_{ch} &(p_2)\\
&0=S.(p_t-p_r) - b &(l_1)\\
&0= K.x-b &(l_2)\\
&\remph{\postset{\guard}=[s\leq 0]; \guard(0)=\fff} &\remph{(v_0)}\\
\when \; \guard\; \doo& 0=p_r-p_t &(v_1)\\
\prog{and}& 0=f_v+s &(v_2)\\
\when \; \prog{not}\; \guard\; \doo& 0=f_v &(v_3)\\
\prog{and}& 0=p_r-p_t-s&(v_4)
\eea 
$}
\caption{The system of of \rref{fig:erpgfiou} corrected. The correction is highlighted in {\color{red}red}.}
	\label{fig:railcarcorrec}
\end{figure}
The two systems differ in equation $\remph{(v_0)}$, where the predicate $s<0$ sets the value of guard $\guard$ for the next nonstandard instant. This models the fact that $\guard$ is evaluated based on $\preset{s}$, which is the nonstandard semantics of the left-limit of $s$. Note that an initial condition for $\guard$ is needed in this corrected version. This model is interpreted in the nonstandard semantics, meaning that substitution (\ref{ergoiuwhio}) applies.

\subsection{Nonstandard structural analysis}
The value of guard $\guard$ is known from the previous instant, together with the value for the two state variables ${x},{p_r},{p_t}$. For both modes $\guard=\ttt/\fff$, the leading variables for evaluation are the $12$ variables $\postset{x},\postset{p_r},\postset{p_t},f_b,f_v,f_{cl},f_l,f_{ch},f_t,p_{bn},b,s$, hence the active system is square in both modes. We successively explore the two long modes $\guard=\ttt/\fff$ for the guard. Then we investigate the mode changes.

\myparagraph{Long mode $\guard=\fff$} The active system is $G_{\fff}=G\cup\{(v_3),(v_4)\}$, see \rref{fig:soerigu}. 
\begin{figure}[ht]
\[
\bea{lc}
0=f_b - f_v - {f_{cl}} + f_l &(m_1)\\
0={f_v} - f_{ch} - f_l - {f_t} &(m_2)\\
0={f_b} - F_1.\flow(p_b-p_r) &(f_1)\\
0= F_1.\flow({p_{bn}}-p_r) &(f_2)\\
0={f_l} - F_2.\flow(p_t-p_r) &(f_3)\\
0={f_t} - \frac{\rho.V}{P_0}.{\dot{p_t}} &(r)\\
0= \rho.S.(\dot{x}.p_r + (x-L).{\dot{p_r}}) + P_0.f_{cl} &(p_1)\\
0=\rho.S.({\dot{x}}.S.p_t + x.\dot{p_t}) - P_0.f_{ch} &(p_2)\\
\remph{0=S.({p_t}-p_r) - b} &\remph{(l_1)}\\
\remph{0= K.x-{b}} &\remph{(l_2)}\\
0={f_v} &(v_3)\\
0=p_r-p_t-{s}&(v_4)
\eea
\]
\caption{The active system when $\guard=\fff$. We show  in {\color{red}red} a structurally singular subsystem.
}
\label{fig:soerigu}
\[
\bea{lc}
0=f_b - f_v - \remph{f_{cl}} + f_l &(m_1)\\
0={f_v} - \remph{f_{ch}} - f_l - {f_t} &(m_2)\\
0=\remph{f_b} - F_1.\flow(p_b-p_r) &(f_1)\\
0= F_1.\flow(\remph{p_{bn}}-p_r) &(f_2)\\
0=\remph{f_l} - F_2.\flow(p_t-p_r) &(f_3)\\
0=\remph{f_t} - \frac{\rho.V}{P_0}.{\dot{p_t}} &(r)\\
0= \rho.S.(\dot{x}.p_r + (x-L).\remph{\dot{p_r}}) + P_0.f_{cl} &(p_1)\\
0=\rho.S.(\remph{\dot{x}}.S.p_t + x.\dot{p_t}) - P_0.f_{ch} &(p_2)\\
0=S.(\remph{\postset{p_t}}-\postset{p_r}) - \postset{b} &(\postset{l_1})\\
0= K.\postset{x}-\remph{\postset{b}} &(\postset{l_2})\\
0=\remph{f_v} &(v_3)\\
0=p_r-p_t-\remph{s}&(v_4)\\
\mbox{consistency equations:} \\
0=S.({p_t}-p_r) - b &(l_1)\\
0= K.x-{b} &(l_2)
\eea
\]
\caption{The system of \rref{fig:soerigu} in which latent equations are found by shifting $(l_1),(l_2)$. Recall that $\dot{x}$ and $\postset{x}$ are related by (\ref{ergoiuwhio}). The assignment of a variable to each equation, highlighted in {\color{red}red}, defines a complete matching.}
\label{fig:esrsrsrhu}
\end{figure}

The subsystem $\{({l_1}),({l_2})\}$ involves two equations but only one leading variable: $b$. It is structurally singular. 
 Following the technique of index reduction used in the clutch example, we  augment $G_{\fff}$ with the latent equations $(\postset{l_1}),(\postset{l_2})$. The result is shown in \rref{fig:esrsrsrhu}.
Note that $b$ is now a state variable. 
Equations $(l_1),(l_2)$ are  consistency conditions. They are satisfied as a result of executing the previous instant. We are thus left with the system $G^\Sigma_{\fff}=G\cup\{(\postset{l_1}),(\postset{l_2})\}$. Can we check its regularity by using structural (graph based) methods? 

A \emph{complete matching} is shown in {\color{red}red} in \rref{fig:esrsrsrhu}, i.e., a one-to-one assignment of leading variables to equations. The intent is that each variable is evaluated using its assigned equation, by pivoting. For example, $f_{cl}$ is evaluated using equation $(m_1)$. The heart of structural analysis for algebraic systems of equations is that, if a complete matching exists, then the system is regular in a generic sense, i.e., outside exceptional values for the non-zero coefficients of the equations. See \rref{sec:mDAE} for a formal development. 

The system $G^\Sigma_{\fff}$ is thus structurally nonsingular. Hence, the system of \rref{fig:soerigu} was of index $1$ and adding the latent equations reduced the index to $0$. 

\myparagraph{Long mode $\guard{=}\ttt$} The active system is $G_{\ttt}=G\cup\{(v_1),(v_2)\}$, shown in \rref{fig:eslriguhu}.
\begin{figure}[ht]
\[
\bea{lc}
0=f_b - f_v - {f_{cl}} + f_l &(m_1)\\
0={f_v} - f_{ch} - f_l - f_t &(m_2)\\
0={f_b} - F_1.\flow(p_b-p_r) &(f_1)\\
0= F_1.\flow({p_{bn}}-p_r) &(f_2)\\
0={f_l} - F_2.\flow(p_t-p_r) &(f_3)\\
0={f_t} - \frac{\rho.V}{P_0}.\dot{p_t} &(r)\\
0= \rho.S.(\dot{x}.p_r + (x-L).{\dot{p_r}}) + P_0.f_{cl} &(p_1)\\
0=\rho.S.({\dot{x}}.S.p_t + x.\dot{p_t}) - P_0.f_{ch} &(p_2)\\
\remph{0=S.({p_t}-p_r) - b} &\remph{(l_1)}\\
\remph{0= K.x-{b}} &\remph{(l_2)}\\
\remph{0=p_r-p_t} &\remph{(v_1)}\\
0={f_v}+{s}&(v_2)
\eea
\]
\caption{The active system when $\guard=\ttt$. We show  in {\color{red}red} a structurally singular subsystem. Compare with \rref{fig:soerigu}.}
\label{fig:eslriguhu}
\[
\bea{lc}
0=f_b - f_v - \remph{f_{cl}} + f_l &(m_1)\\
0=\remph{f_v} - {f_{ch}} - f_l - f_t &(m_2)\\
0=\remph{f_b} - F_1.\flow(p_b-p_r) &(f_1)\\
0= F_1.\flow(\remph{p_{bn}}-p_r) &(f_2)\\
0=\remph{f_l} - F_2.\flow(p_t-p_r) &(f_3)\\
0=\remph{f_t} - \frac{\rho.V}{P_0}.\dot{p_t} &(r)\\
0= \rho.S.(\remph{\dot{x}}.p_r + (x-L).{\dot{p_r}}) + P_0.f_{cl} &(p_1)\\
0=\rho.S.({\dot{x}}.S.p_t + x.\dot{p_t}) - P_0.\remph{f_{ch}} &(p_2)\\
{0=S.(\remph{\postset{p_t}}-\postset{p_r}) - \postset{b}} &{(\postset{l_1})}\\
{0= K.{\postset{x}}-\remph{\postset{b}}} &{(\postset{l_2})}\\
{0=\remph{\postset{p_r}}-\postset{p_t}} &{(\postset{v_1})}\\
0={f_v}+\remph{s}&(v_2) \\
\mbox{consistency equations:}\\
{0=S.({p_t}-p_r) - b} &{(l_1)}\\
{0= K.x-{b}} &{(l_2)} \\
{0=p_r-p_t} &{(v_1)}
\eea
\]
\caption{Adding latent equations to the system of \rref{fig:eslriguhu}.}
\label{fig:peguiohpwiuh}
\end{figure}
The subsystem $\{(l_1),(l_2),(v_1)\}$ is structurally singular since it has $3$ equations and only one dependent variable $b$. We thus shift it forward. The result is shown in \rref{fig:peguiohpwiuh}. 

Equations $(l_1),(l_2),(v_1)$ are consistency conditions. They are satisfied as a result of executing the previous instant. We are thus left with the index reduced system $G^\Sigma_{\ttt}=G\cup\{(\postset{l_1}),(\postset{l_2}),(\postset{v_1})\}$, for which a complete matching is shown in {\color{red}red} in \rref{fig:peguiohpwiuh}. This system is structurally nonsingular. The system of \rref{fig:eslriguhu} was of index $1$. The set of latent equations differs from that of the other mode, however.

So far we have completed the study of the long modes. It remains to investigate the mode changes.

\myparagraph{Mode change $\guard:\ttt\ra\fff$} We reproduce the reasoning following \rref{sys:welguihwtuilh} for the clutch example. Joining together the systems at previous and current instant may result in a conflict between the consistency equations of the current instant and selected equations from the previous instant. This does not occur here, however: consistency equations $({l_1}),({l_2})$ in \rref{fig:esrsrsrhu} coincide with the latent equations $(\postset{l_1}),(\postset{l_2})$ associated with the previous instant in \rref{fig:peguiohpwiuh} (the set of latent equations of mode $\guard=\ttt$ contains the set of latent equations of mode $\guard=\fff$).

\myparagraph{Mode change $\guard:\fff\ra\ttt$} 
Unlike for the previous case, joining together the systems at previous and current instant actually results in a conflict between the consistency equations of the current instant and the dynamics at the previous instant. More precisely, the consistency equation $(v_1):0{=}p_r-p_t$, which holds in current mode $\guard=\ttt$, is not a consequence of executing the previous instant (in mode $\guard=\fff$). As for the clutch example (\rref{sec:orguhuoi} and particularly (\ref{oeuryfgy})), 
\beq
\mbox{
\begin{minipage}{9cm}
	 we remove, at the instant of mode change $\guard:\fff\ra\ttt$, the consistency equation $(v_1):0{=}p_r-p_t$.
\end{minipage}
}
\label{tgoierpguip}
\eeq
This completes the structural analysis.

\subsection{Standardization} 
\label{sec:eoguihptuihoi}
To derive the actual executable code we must return to the standard domain by performing standardization. The task is very much like for the clutch example and we only provide a brief summary:

\paragraph{Within long modes $\guard=\fff$ or $\ttt$:} Here, time and dynamics are continuous: We standardize $\frac{\postset{x}-x}{\vsmall}$ as the derivative $\dot{x}$. Also, we standardize equation $(v_0)$ as $\guard{=}[s^-{\leq}0]$, where $s^-$ denotes the left limit of $s$, i.e., the value of $s$ before the mode change in case a mode change occurs. Doing this removes all the occurrences of $\vsmall$ in both modes. The resulting system is standard.

\paragraph{For mode changes:} Here we target discrete time: $\postset{x}$ is interpreted as the next value for $x$ and raises no difficulty. In contrast, the occurrence of $\vsmall$ in the Euler schemes $\frac{\postset{x}-x}{\vsmall}$, where it acts in space, is now the difficulty. 

\myparagraph{Mode change $\guard:\ttt\ra\fff$}
The restart conditions are solutions of the system of \rref{fig:esrsrsrhu} in which the consistency equations are ignored. This system coincides with the nonstandard dynamics of the long mode $\guard=\fff$. However, we do not seek for a strandardization in continuous time, but rather in discrete time, which is different:
derivatives are expanded using (\ref{ergoiuwhio}) and we solve the system for the discrete time leading variables $\postset{p_t},\postset{p_r},\postset{x}$, etc. This nonstandard system is shown in \rref{fig:weopgiuhoiuh}, where the infinitesimal $\vsmall$ is highlighted in {\color{blue}blue}.
\begin{figure}[ht]
\footnotesize
\[
\bea{lc}
0=f_b - f_v - {f_{cl}} + f_l &(m_1)\\
0={f_v} - {f_{ch}} - f_l - {f_t} &(m_2)\\
0={f_b} - F_1.\flow(p_b-p_r) &(f_1)\\
0= F_1.\flow({p_{bn}}-p_r) &(f_2)\\
0={f_l} - F_2.\flow(p_t-p_r) &(f_3)\\
0=\bemph{\vsmall}.{f_t} - \frac{\rho.V}{P_0}.({\postset{p_t}-p_t}) &(r)\\
0= \rho.S.((\postset{x}-x).p_r + (x-L).({\postset{p_r}-p_r})) + \bemph{\vsmall}.P_0.f_{cl} &(p_1)\\
0=\rho.S.((\postset{x}-x).S.p_t + x.({\postset{p_t}-p_t})) - \bemph{\vsmall}.P_0.f_{ch} &(p_2)\\
0=S.({\postset{p_t}}-\postset{p_r}) - \postset{b} &(\postset{l_1})\\
0= K.\postset{x}-{\postset{b}} &(\postset{l_2})\\
0={f_v} &(v_3)\\
0=p_r-p_t-{s}&(v_4)
\eea
\]
\caption{Restart at mode change $\guard:\ttt\ra\fff$.}
\label{fig:weopgiuhoiuh}
\end{figure}

We need to standardize this system. We leave to the reader as an exercise that setting $\vsmall=0$ in this system makes it structurally nonsingular (no complete matching can be found). By Comment\,\ref{keruif} of Section\,\ref{lwergtuiherlu}, this implies that setting $\vsmall=0$ in this system does not lead to a correct standardization. 

The problem originates from equations $(r),(p_1),(p_2)$. Focus on $(r)$: the state variable $x$ is finite before and after the mode change and $\postset{x}-x\neq{0}$ possibly holds; then, $\vsmall$ infinitesimal requires that $f_t$ is possibly infinite, i.e., impulsive. The same holds for $f_{cl}$ and $f_{ch}$.

Following the approach of Section\,\ref{lwergtuiherlu}, we must eliminate, by algebraic manipulations, the possibly impulsive variables $f_t,f_{cl},f_{ch}$. This is easy since the latter enter linearly in equations $(r),(p_1),(p_2)$. 

Setting $\vsmall=0$ in the resulting system leaves us with a structurally singular system that sets the restart values for the system state variables $p_r,p_t,b,x$. This yields the desired standardization, which is the code for correct restart at this mode change. Details are omitted.

\myparagraph{Mode change $\guard:\fff\ra\ttt$}
We proceed similarly. The restart equations at mode change $\guard:\fff\ra\ttt$ are given in \rref{fig:eprogihepou}, where the infinitesimal $\vsmall$ is highlighted in {\color{blue}blue}.
\begin{figure}[ht]
\footnotesize
\[
\bea{lc}
0=f_b - f_v - {f_{cl}} + f_l &(m_1)\\
0={f_v} - {f_{ch}} - f_l - f_t &(m_2)\\
0={f_b} - F_1.\flow(p_b-p_r) &(f_1)\\
0= F_1.\flow({p_{bn}}-p_r) &(f_2)\\
0={f_l} - F_2.\flow(p_t-p_r) &(f_3)\\
0=\bemph{\vsmall}.{f_t} - \frac{\rho.V}{P_0}.(\postset{p_t}-p_t) &(r)\\
0= \rho.S.({(\postset{x}-x)}.p_r + (x-L).{(\postset{p_r}-p_r)}) + \bemph{\vsmall}.P_0.f_{cl} &(p_1)\\
0=\rho.S.((\postset{x}-x).S.p_t + x.(\postset{p_t}-p_t)) - \bemph{\vsmall}.P_0.{f_{ch}} &(p_2)\\
{0=S.({\postset{p_t}}-\postset{p_r}) - \postset{b}} &{(\postset{l_1})}\\
{0= K.{\postset{x}}-{\postset{b}}} &{(\postset{l_2})}\\
{0={\postset{p_r}}-\postset{p_t}} &{(\postset{v_1})}\\
0={f_v}+{s}&(v_2) 
\eea
\]
\caption{Restart at mode change $\guard:\fff\ra\ttt$.}
\label{fig:eprogihepou}
\end{figure}
We conclude as for the other mode change.

\begin{ccomment}\rm\textbf{(fake impulsive variables)}
	\label{kwerywiur} The reader acquainted with the physics of this system may be quite surprised by our reasoning above. We found that, for the two mode changes, algebraic variables $f_t,f_{cl}$, and $f_{ch}$ are possibly impulsive. However, the physics of the Westinghouse air brake does not tell so. There seems to be something wrong here.\eproof	
\end{ccomment}
In the next section we investigate how the right solution can be found at compile time.
	
\subsection{Assertions}
Focus on the system for restart at mode change $\guard:\fff\ra\ttt$, see \rref{fig:eprogihepou} and assume for a while that the trajectories of system variables remain continuous, which applies in particular to $s$. Now, the predicate defining the guard $\guard$ tells us that the value of $s$ at both mode changes is zero. 
As a consequence, at the instant of both mode changes $\guard:\fff\ra\ttt$ and $\ttt\ra\fff$, the two subsystems $\{(v_1),(v_2)\}$ and $\{(v_3),(v_4)\}$ of the system of \rref{fig:railcarcorrec} become identical due to the zero value for $s$. As a consequence, whatever next mode is, the consistency equations for it ($\{(v_1),(v_2)\}$ or $\{(v_3),(v_4)\}$) will be satisfied. Therefore, no  action is required at mode changes other than commuting from code of \rref{fig:esrsrsrhu} to code of \rref{fig:peguiohpwiuh} or vice-versa. 
	
	Since our compiler is only equipped with our graph based structural analysis, it will not be able to perform the above reasoning. We can, however, overcome this difficulty by giving to the user the possibility to manually refine her model by adding \emph{guarded assertions} of the form
	\beq	
	\prog{when}\;\guard\;\prog{assert}\;{\beta}(X) \label{7yuflieuh}
	\eeq
	where $\guard$ is some guard and $\beta(X)$ is a predicate in the numerical variables $X$ of the system. For example, in the Westinghouse example, the user could state the following assertion, which follows from the fact that $s=0$ holds at mode changes:
	\beq
\prog{when}\;{\guard^-}\;\prog{and}\;\prog{not}\;\guard~\prog{assert}~s=0
		\label{elrituywliruey}
	\eeq
	Assertion (\ref{elrituywliruey}) can be used by the compiler in deducing that no specific restart action is required other than commuting between the code of the two modes, which, in turn, implies that no state variable is impulsive.

\begin{ccomment}\rm\textbf{(are assertions needed for correctness?)}
Not quite. In fact, the approach followed in \rref{sec:eoguihptuihoi} for deriving the restart conditions at mode changes is still correct. It is just incomplete in that it does not allow to compute the value for the algebraic variables that were erroneously considered impulsive. This is indeed harmless. 
Still, our approach of \rref{sec:eoguihptuihoi} was an overshoot and adding a proper assertion is preferred. Also, assertions can be very useful in narrowing the domain of reachable modes, which may improve efficiency of the compilation.\eproof
	\label{seoriuho} 
\end{ccomment}
Assertions are provided in our IsamDAE tool under development~\cite{Caillaud2020a,caillaud:hal-02476541}, see also \rref{sec:RLDC2}. We will not discuss them any further in this paper, however.
\section{Issues and systematic approach}
\label{sec:liguholui}
From our discussion of the examples, we see that the key issue is the handling of mode changes, and particularly the reconciliation of the dynamics just before the change, and 
 the consistency conditions set by the new mode. 
How to properly address this in the original model of continuous real-time was an open question;
mapping the system dynamics to the nonstandard domain clarified this task.

In the rest of this paper, we extend our method, from the toy examples to a systematic approach for \mDAE\ system compilation. This approach, illustrated in \rref{fig:approach}, decomposes into several steps.
\begin{itemize}
\item First, the mDAE system is transformed into a system of \emph{multi-mode difference Algebraic Equations} (mdAE) using the nonstandard Euler expansion of derivatives $\dot{x}\gets\frac{\postset{x}-x}{\vsmall}$. The resulting mdAE system is the \emph{nonstandard semantics} of the original model. 
\item Second, the structural analysis of this nonstandard semantics is performed, resulting in a new mdAE system where latent equations are made explicit and conflicts at mode changes are solved. 
\item Finally, the standardization step recovers both the smooth dynamics in each mode, and the possibly
discontinuous/impulsive restart conditions at mode changes. The resulting code can then be executed with the help of a solver.
\end{itemize}
\begin{figure}[h]
\centerline{\scalebox{1}{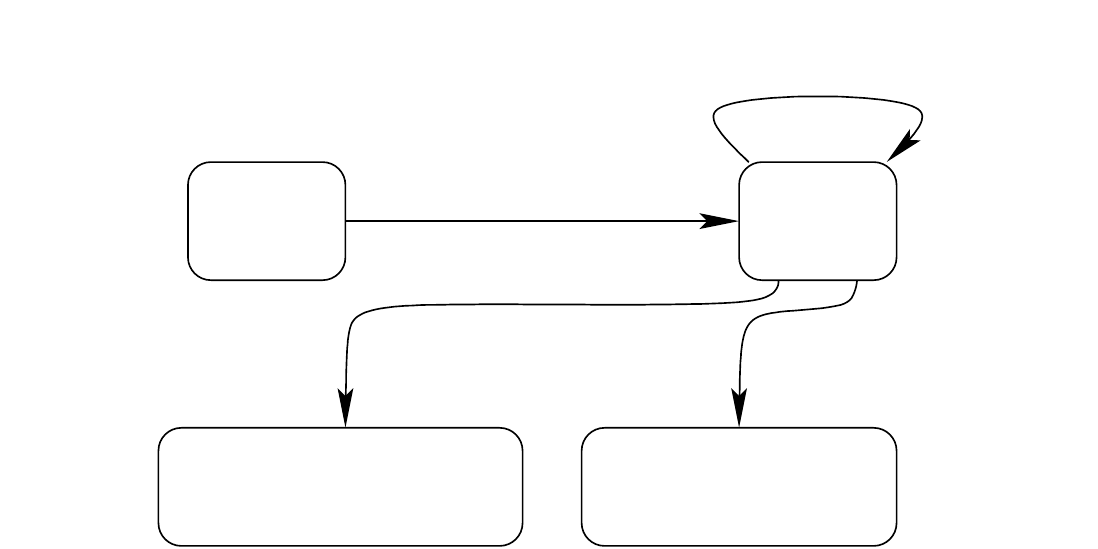}}
\caption[tyf]{\small\sf Our approach to the compilation of mDAE systems.}
\label{fig:approach} 
\end{figure}

The intuitions developed on the examples are formalized in the sequel of the paper. We first recall the background on structural analysis and we extend it to our context. We then focus on standardization. 

\section{Background on Structural Analysis}
\label{sec:mDAE}
We recall in this section the basics of structural analysis for algebraic systems of equations and provide in passing a few new lemmas for subsequent use.  
We then present the \sigmamethod, which is the structural analysis method we adopt for DAE systems.
\begin{notation}\rm 
	\label{lguitp} In this paper, $\bN=\{0,1,2,\dots\}$ is the set of nonnegative natural integers and $\bR_+=[0,\infty)$ is the set of nonnegative real numbers.\eproof
\end{notation}
We will sometimes need to handle variables, their valuations, and vectors thereof. To this end, the following conventions will be used (unless no confusion occurs from using more straightforward notations):
\begin{notation}\rm 
	\label{shltrguioh}  Lowercase letters ($x,y$, etc.) denote scalar real variables; capitals ($X,Y$, etc.) denote vectors of real variables; whenever convenient, we regard $X$ as a set of variables and write $x\in{X}$ to refer to an entry of $X$. 
	 We adopt similar conventions for functions ($f,g$, etc.) and vectors of functions ($F,G$, etc.).
	 A value for a variable $x$ is generically denoted by $\vval{x}$, and similarly for $X$ and $\vval{X}$.\eproof
\end{notation}

\subsection{Structural analysis of algebraic equations}
\label{sec:erligtuhouih}
In this section we focus on systems of algebraic equations, that is, equations involving no time and no dynamics. 

\subsubsection{Structural nonsingularity of algebraic equations}
\label{epfuiowerhpfoui}
Consider a system of smooth algebraic equations:
\begin{equation}
f_i(y_1,\dotsc,y_k,x_1,\dotsc,x_n)=0,\quad i=1,\dotsc,m
\label{eq:F}
\end{equation}
rewritten as $F(Y,X)=0$, where $Y$ and $X$ denote the vectors
$(y_1,\dotsc,y_k)$ and $(x_1,\dotsc,x_n)$, respectively, and $F$ is
the vector $(f_1,\dotsc,f_m)$. The system has $m$ equations, $k$ input variables collected in vector $Y$, and $n$ \emph{dependent variables} (or unknowns) collected in vector $X$. Throughout this section , we assume that the $f_i$'s are all $\cC^1$ at least.

If the system is square, i.e., if $m=n$, the \emph{Implicit Function Theorem} (see,
e.g., Theorem $10.2.2$ in \cite{DieudonneEA1}) states that, if
$(\vval{Y},\vval{X})\in\bR^{k+n}$ is a value for the pair $(Y,X)$ such that
$F(\vval{Y},\vval{X})=0$ and the Jacobian matrix of $F$ with respect to $X$  evaluated at $(\vval{Y},\vval{X})$ is nonsingular, then there
exists, in an open neighborhood $U$ of $\vval{Y}$, a unique vector of functions
$G$ such that $F(v,G(v))=0$ for all $v\in{U}$.  In words,
\rref{eq:F} uniquely determines $X$ as a function of $Y$, locally
around $\vval{Y}$. Denote by $\Jacobian_{\!X}{F}$ the above mentioned Jacobian matrix. Solving $F=0$ for $X$, given a value $\vval{Y}$ for $Y$,
requires forming $\Jacobian_{\!X}{F}(\vval{Y})$ as well as inverting it.

One could avoid considering $\Jacobian_{\!X}{F}$ by focusing on  \emph{structural} nonsingularity, a weaker but much cheaper criterion previously introduced for handling sparse matrices in high performance computing. 

\subsubsection{{Structural analysis}}
\label{oeriufhyeoiu}
We begin with some background on graphs, for which a basic reference is~\cite{Berge1962}. Given a graph $\cG=(V,E)$, where $V$ and $E$ are the sets of vertices and edges, a \emph{matching} $\cM$ of $\cG$ is a set of edges of $\cG$ such that no two edges in $\cM$ are incident on a common vertex. Matching $\cM$ is called \emph{complete} if it covers all the vertices of $\cG$. 
An \emph{$\cM$-alternating path} is a path of $\cG$ whose edges are alternatively in $\cM$ and not in $\cM$---we say simply ``alternating path'' when $\cM$ is understood from the context. A vertex is said to be \emph{matched} if there is an edge in $\cM$ that is incident on the vertex, and \emph{unmatched} otherwise. 

Let $F(Y,X){=}0$ be a system of algebraic equations of the form (\ref{eq:F});  its \emph{bipartite graph} $\cG_F$ is the graph having $F{\cup}{X}$ as set of vertices and an edge $(f_i,x_j)$ if and only if  variable $x_j$ occurs in function $f_i$. 

\paragraph{Square systems} We first consider square systems, in which equations and dependent variables are in equal numbers.
\begin{definition}[structural nonsingularity]
	\label{erfuilehui} System $F(Y,X){=}0$ is called \emph{structurally nonsingular} if its bipartite graph $\cG_F$ possesses a complete matching.
\end{definition}
A structurally nonsingular system is necessarily square, i.e., with equations and variables in equal numbers.
The link to numerical regularity is formalized in the following lemma (implicitly used in~\cite{pantelides,MattssonSoderlin1993,PothenF90,Pryce01}):
\begin{lemma}
	\label{oerigtuerio} Assume that $F$ is $\cC^1$. The following properties are equivalent:
\begin{enumerate}
	\item System $F(Y,X){=}0$ is {structurally nonsingular}.
	\item \label{oweruigyweo} For every $(\vval{Y},\vval{X})$ satisfying $F(\vval{Y},\vval{X}){=}0$, the Jacobian matrix $\Jacobian_{\!X}{F}(\vval{Y},\vval{X})$ remains {generically}\,\footnote{Generically means: outside a set of values of Lebesgue measure zero.\label{erpgfuieh}} nonsingular when its nonzero coefficients vary over some neighborhood.
\end{enumerate}
\end{lemma}  
Let us detail Property\,{\ref{oweruigyweo}}. We can represent  graph $\cG_F$ by the $n{\times}n$-incidence matrix $M_F=(m_{ij})$, having a $1$ at $(i,j)$ if $(f_i,x_j)$ is an edge of $\cG_F$, and a $0$ otherwise. We order the entries $\{a_{ij}\mid m_{ij}=1\}$ of Jacobian $\Jacobian_{\!X}{F}(\vval{Y},\vval{X})$, for example by lexicographic order over the pair $(i,j)$. In this way, every $n{\times}n$-matrix whose pattern of nonzero coefficients is $M_F$ identifies with a unique vector of $\bR^K$, where $K$ is the number of $1$'s in $M_F$.
We denote by $\cJ_{\!X}{F}\in\bR^K$ the image of the Jacobian $\Jacobian_{\!X}{F}(\vval{Y},\vval{X})$ obtained via this correspondence. Property\,{\ref{oweruigyweo}} says: 
\beq
\mbox{
\begin{minipage}{11cm}
	 There exist an open neighborhood $U$ of $\cJ_{\!X}{F}$ in $\bR^K$, and a subset $V{\subseteq}{U}$ such that: $(i)$ the set $U\setminus{V}$ has zero Lebesgue measure, and $(ii)$ every $\cJ\in{V}$ yields a regular matrix.
\end{minipage}
} 
\label{leiuthpieu}
\eeq
In the sequel, the so defined sets $U$ and $V$ will be denoted by 
\beq
 \regularU{F}{\vval{Y},\vval{X}} &\mbox{and}& \regularV{F}{\vval{Y},\vval{X}}
 \label{lerigfui}
\eeq
or simply $U_{\!F}$ and $V_{\!F}$ when no confusion can result.

\myparagraph{Practical use of structural nonsingularity} If a square matrix is structurally singular, then it is singular. The converse is false: structural nonsingularity does not guarantee nonsingularity. Therefore, the practical use of Definition\,\ref{erfuilehui} is as follows: We first check if $F{=}0$ is structurally nonsingular. If not, then we abort searching for a solution. Otherwise, we then proceed to computing a solution, which may or may not succeed depending on the actual numerical regularity of the Jacobian matrix $\Jacobian_{\!X}F$.\eproof

Let $F(Y,X){=}0$ be structurally nonsingular and let $\cM$ be a complete matching for it. Using $\cM$, graph $\cG_F$ can be directed as follows. Edge $(f_i,x_j)$ is directed from $f_i$ to $x_j$ if $(f_i,x_j)\in\cM$, from $x_j$ to $f_i$ otherwise. Denote by ${\vec{\cG}}^\cM_F$ the resulting directed graph. The following result holds~(\cite{Duff86}, Chapter 6.10):
\begin{lemma}
	\label{jytdfeoguip} The strongly connected components of ${\vec{\cG}}^\cM_F$ are independent of $\cM$. 
\end{lemma}
Each strongly connected component defines a \emph{block} of $F$, consisting of the set of all equations that are vertices of this component. Blocks are partially ordered and we denote by $\preceq_F$ this order. Extending $\preceq_F$ to a total order (via topological sorting) yields an ordering of equations that puts the Jacobian matrix $\Jacobian_{\!X}F$ in \emph{Block Triangular Form} (BTF).

\begin{figure}[ht]
  \centering
		\centerline{
		\resizebox{\textwidth}{!}{
\begin{tabular}{cc|cc}
        $\begin{tikzpicture}[scale=.7]
        \tikzstyle{every node}=[font=\Large]
        \vertex[] (f1) at (0,7.5) []{\color{mygreen}{$f_1$}};
        \vertex[] (f2) at (0,6) []{\color{mygreen}{$f_2$}};
        \vertex[] (f3) at (0,4.5) []{\color{mygreen}{$f_3$}};
        \vertex[] (f4) at (0,3) []{\color{mygreen}{$f_4$}};
        \vertex[] (f5) at (0,1.5) []{\color{mygreen}{$f_5$}};
        \vertex[] (f6) at (0,0) []{\color{mygreen}{$f_6$}};
        \vertex[] (x1a) at (4,7.5) []{\color{red}{$x_1$}};
        \vertex[] (x2a) at (4,6) []{\color{red}{$x_2$}};
        \vertex[] (x3a) at (4,4.5) []{\color{red}{$x_3$}};
        \vertex[] (x4a) at (4,3) []{\color{red}{$x_6$}};
        \vertex[] (x5a) at (4,1.5) []{\color{red}{$x_4$}};
        \vertex[] (x6a) at (4,0) []{\color{red}{$x_5$}};
        	\path 
        	    (f1) edge[] (x2a)
        	    (f1) edge[] (x4a)
        	    (f2) edge[] (x3a)
        	    (f2) edge[] (x6a)
        	    (f3) edge[] (x4a)
        	    (f4) edge[] (x3a)
        	    (f5) edge[] (x4a)
        	    (f6) edge[] (x1a)
        	    (f6) edge[] (x5a)
        	    (f1) edge[line width = 3,color=myblue] (x1a)
        	    (f2) edge[line width = 3,color=myblue] (x2a)
        	    (f3) edge[line width = 3,color=myblue] (x3a)
        	    (f4) edge[line width = 3,color=myblue] (x4a)
        	    (f5) edge[line width = 3,color=myblue] (x5a)
        	    (f6) edge[line width = 3,color=myblue] (x6a)
        	    ;
        \end{tikzpicture}$
&
        $\begin{tikzpicture}[scale=.7]
        \tikzstyle{every node}=[font=\Large]
        \tikzbox[] (f1) at (3,3) []{\color{mygreen}{$f_1$}};
        \tikzbox[] (f2) at (5,6) []{\color{mygreen}{$f_2$}};
        \tikzbox[] (f3) at (0,6) []{\color{mygreen}{$f_3$}};
        \tikzbox[] (f4) at (0,3) []{\color{mygreen}{$f_4$}};
        \tikzbox[] (f5) at (0,0) []{\color{mygreen}{$f_5$}};
        \tikzbox[] (f6) at (5,0) []{\color{mygreen}{$f_6$}};
        	\path 
        	    (f1) edge[arrowedge=1,line width = 3] (f6)
        	    (f6) edge[arrowedge=1,line width = 3] (f2)
        	    (f2) edge[arrowedge=1,line width = 3] (f1)
        	    (f3) edge[arrowedge=1] (f2)
        	    (f3) edge[arrowedge=1,line width = 3] (f4)
        	    (f4) edge[arrowedge=1,line width = 3] (f3)
        	    (f4) edge[arrowedge=1] (f1)
        	    (f4) edge[arrowedge=1] (f5)
        	    (f5) edge[arrowedge=1] (f6)
        	    ;
        \draw[densely dotted,line width=2,color=mychoc] (4.65,3) ellipse (2.6 and 4);
        \draw[densely dotted,line width=2,color=mychoc] (0,4.5) ellipse (1.6 and 2.6);
        \draw[densely dotted,line width=2,color=mychoc] (0,0) ellipse (1.3 and 0.8);
        \end{tikzpicture}$
&
        $\begin{tikzpicture}[scale=.7]
        \tikzstyle{every node}=[font=\Large]
        \tikzbox[] (f1) at (3,3) []{\color{mygreen}{$f_1$}};
        \tikzbox[] (f2) at (5,6) []{\color{mygreen}{$f_2$}};
        \tikzbox[] (f3) at (0,6) []{\color{mygreen}{$f_3$}};
        \tikzbox[] (f4) at (0,3) []{\color{mygreen}{$f_4$}};
        \tikzbox[] (f5) at (0,0) []{\color{mygreen}{$f_5$}};
        \tikzbox[] (f6) at (5,0) []{\color{mygreen}{$f_6$}};
        	\path 
        	    (f6) edge[arrowedge=1,line width = 3] (f1)
        	    (f2) edge[arrowedge=1,line width = 3] (f6)
        	    (f1) edge[arrowedge=1,line width = 3] (f2)
        	    (f4) edge[arrowedge=1] (f2)
        	    (f4) edge[arrowedge=1,line width = 3] (f3)
        	    (f3) edge[arrowedge=1,line width = 3] (f4)
        	    (f3) edge[arrowedge=1] (f1)
        	    (f3) edge[bend right,arrowedge=0.999] (f5)
        	    (f5) edge[arrowedge=1] (f6)
        	    ;
        \draw[densely dotted,line width=2,color=mychoc] (4.65,3) ellipse (2.6 and 4);
        \draw[densely dotted,line width=2,color=mychoc] (0,4.5) ellipse (1.6 and 2.6);
        \draw[densely dotted,line width=2,color=mychoc] (0,0) ellipse (1.3 and 0.8);
        \end{tikzpicture}$
&
        $\begin{tikzpicture}[scale=.7]
        \tikzstyle{every node}=[font=\Large]
        \vertex[] (f1) at (0,7.5) []{\color{mygreen}{$f_1$}};
        \vertex[] (f2) at (0,6) []{\color{mygreen}{$f_2$}};
        \vertex[] (f3) at (0,4.5) []{\color{mygreen}{$f_3$}};
        \vertex[] (f4) at (0,3) []{\color{mygreen}{$f_4$}};
        \vertex[] (f5) at (0,1.5) []{\color{mygreen}{$f_5$}};
        \vertex[] (f6) at (0,0) []{\color{mygreen}{$f_6$}};
        \vertex[] (x1) at (4,7.5) []{\color{red}{$x_1$}};
        \vertex[] (x2) at (4,6) []{\color{red}{$x_2$}};
        \vertex[] (x3) at (4,4.5) []{\color{red}{$x_3$}};
        \vertex[] (x4) at (4,3) []{\color{red}{$x_6$}};
        \vertex[] (x5) at (4,1.5) []{\color{red}{$x_4$}};
        \vertex[] (x6) at (4,0) []{\color{red}{$x_5$}};
        	\path 
        	    (f1) edge[] (x1)
        	    (f1) edge[] (x4)
        	    (f2) edge[] (x2)
        	    (f2) edge[] (x3)
        	    (f3) edge[] (x3)
        	    (f4) edge[] (x4)
        	    (f5) edge[] (x4)
        	    (f6) edge[] (x6)
        	    (f6) edge[] (x5)
        	    (f1) edge[line width = 3,color=myblue] (x2)
        	    (f2) edge[line width = 3,color=myblue] (x6)
        	    (f3) edge[line width = 3,color=myblue] (x4)
        	    (f4) edge[line width = 3,color=myblue] (x3)
        	    (f5) edge[line width = 3,color=myblue] (x5)
        	    (f6) edge[line width = 3,color=myblue] (x1)
       	    ;
        \end{tikzpicture}$
\end{tabular}
}
		}
  \caption{\small\sf Illustrating Lemma~\ref{jytdfeoguip} (we only show the projection of strongly connected components on function nodes).}
 \label{fig:BTF}
\end{figure}
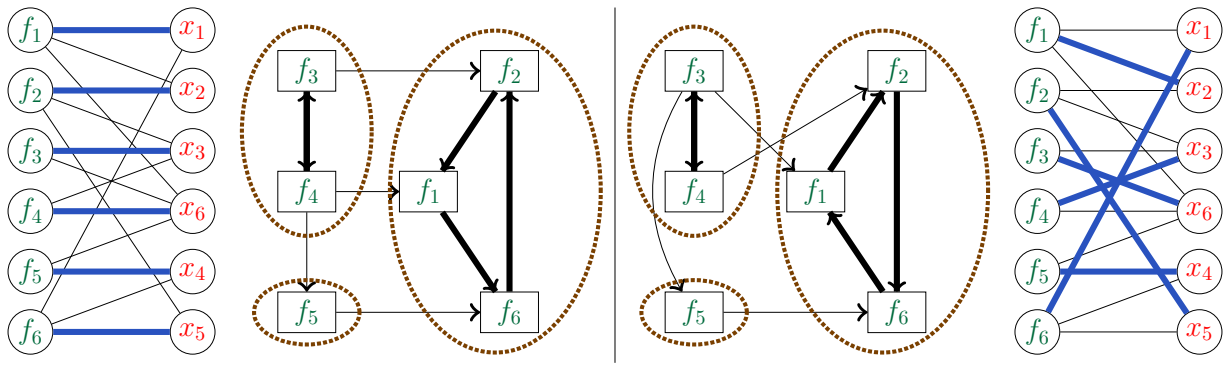

This lemma is illustrated in \rref{fig:BTF}, inspired by 
\href{https://graal.ens-lyon.fr/~bucar/CR07/lecture-matching.pdf}{lecture notes by J-Y. L'Excellent and Bora U\c car, 2010}.
In this figure, a same bipartite graph $\cG$ is shown twice (left and right), with the equations sitting on the left-hand side in green, and the variables on the right-hand side in red. Two complete matchings $\cM_1$ (left) and $\cM_2$ (right) are shown in thick blue, with other edges of the bipartite graph being black. The restriction, to the equation vertices, of the two directed graphs ${\vec{\cG}}^{\cM_1}_F$ (left) and ${\vec{\cG}}^{\cM_2}_F$ (right) are shown on both sides.
Although the two directed graphs ${\vec{\cG}}^{\cM_1}_F$ and ${\vec{\cG}}^{\cM_2}_F$ differ, the resulting block structures (encircled in dashed brown) are identical.

\paragraph{Nonsquare systems}
In our development, we will also encounter nonsquare systems of algebraic equations. For general systems (with a number of variables not equal to the
 number of equations, $n{\neq}{m}$ in (\ref{eq:F})), call \emph{block} a pair $\block=(\bff,\bfx)$ collecting a subset $\bff$ of the set of $f_i$'s and a subset $\bfx$ of variables $x_j$ involved in it. 
\begin{definition}[Dulmage-Mendelsohn]
	\label{elrftuierhfperui}  For $F{=}0$ a general system of algebraic equations, the \emph{Dulmage-Mendelsohn decomposition~\cite{PothenF90}} of $\cG_F$ yields the partition of system $F{=}0$ into the following three blocks, given some matching of maximal cardinality for $\cG_F$:
\begin{itemize}

\vspace*{-2mm}

 \item 
block $\block_\overapprox$ collects the variables and equations reachable via some alternating path from some unmatched equation;

\vspace*{-2mm}

	\item 
  block $\block_\underapprox$ collects the variables and equations reachable via some alternating path from some unmatched variable;
  
\vspace*{-2mm}

 \item 
block $\block_\squared$ collects the remaining variables and equations. 
  
\vspace*{-1mm}

\end{itemize}
Blocks $\block_\overapprox,\block_\underapprox,\block_\squared$ are the $\overapprox$verdetermined, $\underapprox$nderdetermined, and $\squared$nabled parts of $F$.
\end{definition}
Statement~\ref{erpoigfehpguio} of the following lemma ensures that Definition~\ref{elrftuierhfperui} is meaningful:
\begin{lemma}
	\label{erlihpfiu} \
\begin{enumerate}
	\item \label{erpoigfehpguio} The triple $(\block_\overapprox,\block_\squared,\block_\underapprox)$ defined by the Dulmage-Mendelsohn decomposition does not depend on the particular maximal matching used for its definition. 
We thus write $(\block_\underapprox,\block_\squared,\block_\overapprox)=\DM(F)$.
\item System $F$ is structurally nonsingular if and only if the  overdetermined and underdetermined blocks $\block_\overapprox$ and $\block_\underapprox$ are both empty.
\end{enumerate}
\end{lemma}
A refined DM decomposition exists, which in addition refines block $\block_\squared$ into its indecomposable block triangular form. This is a direct consequence of applying Lemma~\ref{jytdfeoguip} to $\block_\squared$ in order to get its BTF.

The following obvious lemma will be useful in the sequel.
\begin{lemma}
\label{togiethgouiho} Let $\DM({F}){=}(\block_\underapprox,\block_\squared,\block_\overapprox)$ be the Dulmage-Mendelsohn decomposition of $F{=}0$. Then, no overdetermined block exists in the Dulmage-Mendelsohn decomposition of ${F}\setminus\block_\overapprox$.
\end{lemma}
\begin{ccomment}\rm 
	\label{esrltuiophu} 
Suppose that, instead of removing from $F$ all equations belonging to $\block_\overapprox$, we 
only remove all unmatched equations with reference to the matching of maximal cardinality used for generating $\DM({F})$,  and call $F'$ the remaining subsystem of $F$. Then, we still get that no overdetermined block exists in the Dulmage-Mendelsohn decomposition of $F'$ and we have $F'\supseteq({F}\setminus\block_\overapprox)$, the inclusion being strict in general. However, this policy depends on the particular matching. In contrast, our policy is canonical since the DM decomposition is independent from the matching. 
Since we will rely on Lemma~\ref{togiethgouiho} for our compilation scheme in \rref{sec:loguhiuliugh}, adopting the finer policy would introduce a risk of nondeterminism due to the arbitrary choice of the particular complete matching selected. This is the reason for sticking with the policy stated in Lemma~\ref{togiethgouiho}.\eproof
\end{ccomment}

\subsubsection{Local versus non-local use of structural analysis}
\label{wlerifuhwlpi}
Structural analysis relies on the Implicit Function Theorem for its justification, as we have seen. Therefore, the classical use of structural analysis is local: let $F(Y,X){=}0$ be a structurally nonsingular system of equations with dependent variables $X$, and let $\vval{X}$ satisfy $F(\vval{Y},\vval{X}){=}0$ for given values $\vval{Y}$ for $Y$; then, the system remains (generically) nonsingular if $\vval{}$ varies in a neighborhood of $\vval{Y}$. This is the situation encountered in the running of ODE or DAE solvers, since a close guess of the solution is known from the previous time step.

We are, however, also interested in invoking the structural nonsingularity of system $F(Y,X)=0$ when no close guess is known. This is the situation encountered when handling mode changes, see our examples of Sections\,\ref{sec:simpleclutch} and \ref{loeruighlrigtuh}. We thus need another argument to justify the use of structural analysis in this case.

As a prerequisite we recall basic definitions regarding smooth manifolds (see, e.g.,~\cite{Cartan}, Section 4.7). In the next lemma, we consider the system $F(Y,X){=}0$ for a fixed value of $Y$, and thus we omit $Y$. 
\begin{lemma}
	\label{rleiufheliru} Let $k,m,n\in\bN$ be such that $m<{n}$ and set $p=n-m$. For $\cS\subset\bR^n$ and $x^*=(x^*_1,\dots,x^*_n)\in\cS$, the following properties are equivalent:
\begin{enumerate}
	\item \label{lreiuhleuirh} There exists an open neighborhood $V$ of $x^*$ and a $\cC^k$-diffeomorphism $F:V\ra{W}{\subset}\bR^n$ such that $F(x^*)=0$ and $F(\cS\cap{V})$ is the intersection of $W$ with the subspace of $\bR^n$ defined by the $m$ equations $w_{p+1}{=}0,\dots,w_n{=}0$, where $w_i$ denote the coordinates of points belonging to $W$.
	
	\item  \label{weuygweouertpoi} There exists an open neighborhood $V$ of $x^*$, an open neighborhood $U$ of $0$ in $\bR^p$ and a homeomorphism $\psi:U\rightarrow\cS\cap{V}$, such that $\psi$, seen as a map with values in $\bR^n$, is of class $\cC^k$ and has rank $p$ at $0$---i.e., $\psi'(0)$ has rank $p$.
\end{enumerate}
\end{lemma}
Statement\,\ref{lreiuhleuirh} means that $\cS$ is, in a neighborhood of $x^*$, the solution set of the system of equations $f_1(X){=}0,\dots,f_m(X){=}0$ in the $n$-tuple of dependent variables $X$, where $F=(f_1,\dots,f_m)$. Statement\,\ref{weuygweouertpoi} expresses that $\psi$ is a parameterization of $\cS$ in a neighborhood of $x^*$, of class $\cC^k$ and rank $p$---meaning that $\cS$ has dimension $p$ in a neighborhood of $x^*$.

{In order to apply Lemma\,$\ref{rleiufheliru}$ to the non-local use of structural analysis, we have to study the case of square systems of equations (i.e., let $m = n$).}
Consider a system $F(X){=}0$, where $X{=}(x_1,\dots,x_n)$ is the $n$-tuple of dependent variables and $F{=}(f_1,\dots,f_n)$ is an $n$-tuple of functions $\bR^n\rightarrow\bR$ of class $\cC^k$. Let $\cS\subseteq\bR^n$ be the solution set of this system, {that we assume is non-empty}. To be able to apply Lemma\,\ref{rleiufheliru}, we augment $X$ with one extra variable $z$, i.e., we set $Z\eqdef{X}\cup\{z\}$ and we now regard our formerly square system $F(X){=}0$ as an augmented system $G(Z){=}0$ where $G(X,z)\eqdef{F}(X)$. Extended system $G$ possesses $n$ equations and $n{+}1$ dependent variables, hence, $p{=}1$,  and its solution set is equal to $\cS\times\bR$. 
Let $x^* \in \cS$ be a solution of $F=0$. Then, we have $G(x^*,z^*){=}0$ for any $z^*\in\bR$. Assume further that the Jacobian matrix $\Jacobian_{\!X}F(x)$ is nonsingular at $x^*$. Then, it remains nonsingular in a neighborhood of $x^*$. Hence, the Jacobian matrix $\Jacobian_{\!X}G(x,z)$ has rank $n$ in a neighborhood $V$ of $(x^*,z^*)$ in $\bR^{n+1}$. We can thus extend $G$ by adding one more function in such a way that the resulting $(n{+}1)$-tuple $\bar{G}$ is a $\cC^k$-diffeomorphism, from $V$ to an open set $W$ of $\bR^{n+1}$. The extended system $\bar{G}{=}0$ satisfies Property\,\ref{lreiuhleuirh} of Lemma\,\ref{rleiufheliru}. By Property\,\ref{weuygweouertpoi} of Lemma\,\ref{rleiufheliru}, $(\cS\times\bR)\cap{V}$ has dimension $1$, implying that $\cS\cap\proj{\bR^n}{V}$ has dimension $<1$. This analysis is summarized in the following lemma:
\begin{lemma}
	\label{erkltuiherliu}  Consider the square system $F(X){=}0$, where $X{=}(x_1,\dots,x_n)$ is the $n$-tuple of dependent variables and $F{=}(f_1,\dots,f_n)$ is an $n$-tuple of functions $\bR^n\rightarrow\bR$ of class $\cC^k, k{\geq}1$. Let $\cS\subseteq\bR^n$ be the solution set of this system. Assume that system $F{=}0$ has solutions and let $x^*$ be such a solution. Assume further that the Jacobian matrix $\Jacobian_{\!X}F(x)$ is nonsingular at $x^*$. Then, there exists an open neighborhood $V$ of $x^*$ in $\bR^n$, such that $\cS\cap{V}$ has dimension $<1$.
\end{lemma}
{The implications for the non-local use of structural analysis are as follows.} Consider the square system $F(X){=}0$, where $X{=}(x_1,\dots,x_n)$ is the $n$-tuple of dependent variables and $F{=}(f_1,\dots,f_n)$ is an $n$-tuple of functions $\bR^n\rightarrow\bR$ of class $\cC^k$. 
We first check the structural nonsingularity of $F{=}0$. If structural nonsingularity does not hold, then we abort solving the system. Otherwise, we proceed to solving the system and the following cases can occur:
\begin{enumerate}
	\item \label{epriuthpeo} the system possesses no solution;
	
	\item \label{kerufygeouy} its solution set is nonempty, of dimension $<1$; or
	
	\item \label{leriughpoiuh} its solution set is nonempty, of dimension $\geq 1$.
\end{enumerate}
Based on Lemma\,\ref{erkltuiherliu}, one among cases\,\ref{epriuthpeo} or\,\ref{kerufygeouy} will hold, generically, whereas case\,\ref{leriughpoiuh} will hold exceptionally. Thus, structurally, we know that the system is not underdetermined. 
This is our justification of the non-local use of structural analysis. Note that the existence of a solution is not guaranteed. Furthermore, unlike for the local use, uniqueness is not guaranteed either.\footnote{{Lemma\,\ref{erkltuiherliu} only states that there cannot be two arbitrarily close solutions of $F=0$, as its conclusion $\cS \cap V = \{ x^* \}$ involves an open neighborhood $V$ of $x^*$.}} Of course, subclasses of systems for which existence and uniqueness are guaranteed are of interest---the illustration examples we develop in this paper have this property.

\subsubsection{{Equations with existential quantifiers}}
To our knowledge, the results of this section are new.
\label{sec:operiuhlo}
For the handling of mode change events, we will also need to develop a structural analysis for the following class of (possibly nonsquare) algebraic systems of equations with existential quantifier:
	  \beq
	 \exists{W}:F(X,W,Y){=}0 ~ ,
	 \label{elrifuhuio}
	 \eeq
	  where $(X,W)$ collects the dependent variables of system $F(X,W,Y){=}0$. 
Our aim is to find structural conditions ensuring that (\ref{elrifuhuio}) defines a partial function $Y\ra{X}$, meaning that the value of tuple $X$ is uniquely defined by the satisfaction of (\ref{elrifuhuio}), given a value for tuple $Y$.
	   
	   At this point, a fundamental difficulty arises. Whereas (\ref{elrifuhuio}) is well defined as an abstract relation, it cannot in general be represented by a projected system of smooth algebraic equations of the form $G(X,Y)=0$. Such a $G$ can be associated to (\ref{elrifuhuio}) only in subclasses of systems.\footnote{Examples are linear systems for which elimination is easy, and polynomial systems for which it is doable---but expensive---using Gr\"obner bases.} In general, no extension of the Implicit Function Theorem exists for systems of the form (\ref{elrifuhuio}); thus, one cannot apply as such the arguments developed in Section~\ref{wlerifuhwlpi}. Therefore, we will reformulate our requirement differently.
	   
	   Say that $\vval{Y}$ is \emph{consistent} if the system of equations $F(X,W,\vval{Y}){=}0$ possesses a solution for $(X,W)$. We thus consider the following property for the systen $F(X,W,Y){=}0$: for every consistent $\vval{Y}$,
	   \beq
	  \left.\bea{l}
	  F(\vval{X}^1,\vval{W}^1,\vval{Y})=0 \\ [1mm]
	  F(\vval{X}^2,\vval{W}^2,\vval{Y})=0
	  \eea\right\}\implies& \vval{X}^1=\vval{X}^2 ~~ , 
	   \label{lriuyokiyugti}
	   \eeq
	   expressing that $X$ is independent of $W$ given a consistent tuple of values for $Y$. 
	   To find structural criteria guaranteeing (\ref{lriuyokiyugti}), we will consider the algebraic system $F(X,W,Y){=}0$ with $X,W,Y$ as dependent variables (i.e., values for the entries of $Y$ are no longer seen as being given).
\begin{definition}
	\label{lerigfueriopg} System $(\ref{elrifuhuio})$ is \emph{structurally nonsingular} if the following two conditions hold, almost everywhere when the nonzero coefficients of the Jacobian matrix $\Jacobian_{X,W,Y}F$ vary over some neighborhood:\footnote{The precise meaning of this statement is the same as in Lemma~\ref{oerigtuerio}, see the discussion therafter.} consistent values for $Y$ exist and condition $(\ref{lriuyokiyugti})$ holds.
\end{definition}
The structural nonsingularity of (\ref{elrifuhuio}) can be checked by using the DM decomposition of $F$ as follows. Let $(\block_\underapprox,\block_\squared,\block_\overapprox)=\DM(F)$ be the DM decomposition of $F(X,W,Y){=}0$, with $(X,W,Y)$ as dependent variables. We further assume that the regular block $\block_\squared$ is expressed in its block triangular form, see Lemma~\ref{jytdfeoguip} and comments thereafter. Let $\bfB$ be the set of all indecomposable blocks of $\block_\squared$ and $\preceq$ be the partial order on $\bfB$ following Lemma~\ref{jytdfeoguip}. The following holds:
\begin{lemma}
	\label{lergfuioerhpuip} System $(\ref{elrifuhuio})$ is \emph{structurally nonsingular} if and only if:
\begin{enumerate}
	\item \label{sguiholghouiy} Block $\block_\overapprox$ is empty;
	\item \label{odigpiufr} Block $\block_\underapprox$ involves no variable belonging to $X$;
	\item \label{lpiguhlosgioj} For every $\block\in\bfB$ containing some variable from $X$, then, $\block$ contains no variable from $W$, and for every $\block'{\prec}\block$ and every directed edge $(z',{z}){\in}\vec{\cG}^\cM_F$ such that $z'{\in}\block'$, $z{\in}\block$, and $\cM$ is an arbitrary complete matching for $\cG_F$, then 	$z'\not\in{W}$.
\end{enumerate}
\end{lemma}
Conditions~\ref{sguiholghouiy} and~\ref{odigpiufr} speak by themselves. Regarding Condition~\ref{lpiguhlosgioj}, the intuition is the following. Every directed path of $\vec{\cG}^\cM_F$, originating from $W$ and terminating in $X$, of minimal length, must traverse $Y$. Consequently, $W$ influences $X$ ``through'' $Y$ only.
\begin{proof}
\myparagraph{``If'' part} By Condition~\ref{sguiholghouiy}, there exist consistent values for $Y$. By condition~\ref{odigpiufr}, no variable of $X$ belongs to the underdetermined part of $\DM(F)$. It remains to show that condition $(\ref{lriuyokiyugti})$ holds. By condition (\ref{lpiguhlosgioj}), each indecomposable block $\block\in\bfB$ involving a variable of $X$ has the form $G(Z,U){=}0$,
where 
\begin{enumerate}
	\item $Z$ is the $n$-tuple of dependent variables of $G=(g_1,\dots,g_n)$, 
	\item no variable of $W$ belongs to $Z$, and
	\item $U\subseteq{X}{\cup}Y$ is a subset of the dependent variables of the blocks $\block'$ immediately preceding $\block$.
\end{enumerate}
Since $G{=}0$ is structurally regular, fixing a value for $U$ entirely determines all the dependent variables of block $\block$. This proves the ``if'' part.

\myparagraph{``Only if'' part} We prove it by contradiction. If condition~\ref{sguiholghouiy} does not hold, then the existence of consistent values for $Y$ is not structurally guaranteed. If condition~\ref{odigpiufr} does not hold, then the variables of $X$ that belong to $\block_\underapprox$ are not determined, even for given values of $W,Y$. If condition~\ref{lpiguhlosgioj} does not hold, then two cases can occur:
\begin{itemize}
	\item Some $x\in{X}$ and $w\in{W}$ are involved in a same indecomposable block $\block\in\bfB$. Then, condition (\ref{lriuyokiyugti}) will not hold for $x$ due to the structural coupling with $w$ through block $\block$.
	\item There exists an indecomposable block $\block\in\bfB$  of the form $G(Z,U){=}0$,
where 
\begin{enumerate} 
	\item $Z$ is the $n$-tuple of dependent variables of $G=(g_1,\dots,g_n)$, and $Z$ contains a variable $x\in{X}$,
	\item no variable of $W$ belongs to $Z$, and
	\item some variable $w\in{W}$ belongs to $U$, the set of the dependent variables of the blocks $\block'$ immediately preceding $\block$.
\end{enumerate}
\end{itemize}
Hence, $w$ influences $x$, structurally, which again violates (\ref{lriuyokiyugti}). This finishes the proof.
\end{proof}
We complement Lemma~\ref{lergfuioerhpuip} with the algorithm $\atomicact{ExistQuantifEqn}$ (\rref{alg:erofiueopiu}), which
requires a system $F$ of the form (\ref{elrifuhuio}). If condition~\ref{sguiholghouiy} of Lemma~\ref{lergfuioerhpuip} fails to be satisfied, then $b_\overapprox\gets\fff$ is returned, indicating overdetermination. If conditions~\ref{odigpiufr} or \ref{lpiguhlosgioj} of Lemma~\ref{lergfuioerhpuip} fail to be satisfied, then $b_\underapprox\gets\fff$ is returned, indicating underdetermination. Otherwise, $\atomicact{ExistQuantifEqn}$ succeeds and returns the value $\ttt$ for both Booleans, together with
the decomposition $F_\Sigma\cup\consistency{F}_\Sigma$ of $\block_\squared$. In this decomposition:
\begin{itemize}
\item Subsystem $F_\Sigma$ collects the indecomposable blocks involving variables belonging to $X$, so that ${F}_\Sigma$ determines $X$ as a function of $\vval{Y}$ when $\vval{Y}$ is consistent;
\item Subsystem $\consistency{F}_\Sigma$ collects the consistency conditions, whose dependent variables belong to $W\cup{Y}$.
\end{itemize}
\begin{algorithm}[ht]
	\caption{$\atomicact{ExistQuantifEqn}$} 
	\label{alg:erofiueopiu}
 \begin{algorithmic}[1]
\Require $F$;
\Return $(b_\overapprox,b_\underapprox,F_\Sigma,\consistency{F}_\Sigma)$
\State $(\block_\underapprox,\block_\squared,\block_\overapprox)=\DM(F)$
\If{condition~\ref{sguiholghouiy} of Lemma~\ref{lergfuioerhpuip} holds} 
\State $b_\overapprox\gets\ttt$;
\If{\ref{odigpiufr} and \ref{lpiguhlosgioj} of Lemma~\ref{lergfuioerhpuip} hold} 
\State $b_\underapprox\gets\ttt$;
  \State partition $\block_\squared=F_\Sigma\cup\consistency{F}_\Sigma$
\label{op:efuiehoi}
\Else
\State $b_\underapprox\gets\fff$
\EndIf
\Else
\State $b_\overapprox\gets\fff$
\EndIf
	\end{algorithmic}
\end{algorithm}

Our background on the structural analysis of algebraic equations is now complete. 
In the next section, we recall the background on the structural analysis of (single-mode) DAE systems.

\subsection{The \sigmamethod\ for DAE systems}
\label{sec:structural}
\label{reoiughfuil}
The DAE systems we consider are ``square'', i.e., have the form 
\beq
f_j(\mbox{the } x_i \mbox{'s and their derivatives}) = 0
\label{sepguiohp}
\eeq
where $x_1,\dotsc,x_n$ denote the dependent variables and $f_1=0,\dotsc,f_n=0$ denote the equations. Call \emph{leading variables} of System~(\ref{sepguiohp}) the $d_i$-th derivatives\footnote{The notation $\pprime{k}{x}$ is adopted throughout this paper, instead of the more classical $x^{(k)}$, for the $k$-th derivative of $x$. See Notations~\ref{lerigftuheroiuof}.} $\pprime{d_i}{x_i}$ for $i=1,\dots,n$, where $d_i$ is the maximal differentiation degree of variable $x_i$ throughout $f_1=0,\dotsc,f_n=0$. The problem addressed by the structural analysis of DAE systems of the form (\ref{sepguiohp}) is the following. Regard (\ref{sepguiohp}) as a system of algebraic equations with the leading variables as unknowns. If this system is structurally nonsingular, then, given a value for all the $\pprime{k}{x_i}$ for $i=1,\dots,n$ and $k=0,\dots,d_i{-}1$, a unique value for the leading variables can be computed, structurally; hence, System~(\ref{sepguiohp}) is like an ODE. If this is not the case, finding additional \emph{latent equations} by differentiating suitably selected equations from (\ref{sepguiohp}) will bring the system to an ODE-like form, while not changing its set of solutions. Performing this is known as \emph{index reduction.} 

In our simple examples, finding latent equations was easy. In general, this is difficult and algorithms were proposed in the literature for doing it efficiently. Among them, the Pantelides algorithm~\cite{pantelides} is the historical solution. We decided, however, to base our subsequent developments on the beautiful method proposed in~\cite{Pryce01} by J. Pryce, called the \sigmamethod. The {\sigmamethod} also covers the construction of the Block Triangular Form and addresses numerical issues, which we do not discuss here. 

\paragraph{Weighted bipartite graphs}
We consider System~(\ref{sepguiohp}), which is entirely characterized by its set of dependent variables $X$ (whose generic element is denoted by $x$) and its set of equations $F{=}0$ (whose generic element is written $f{=}0$). 
We attach to (\ref{sepguiohp}) the bipartite graph $\cG$ having $F\cup{X}$ as set of vertices and having an edge $(f,x)$ if and only if $x$ occurs in function $f$, regardless of its differentiation degree. Recall that a matching is \emph{complete} iff it involves all equations of $F$.

So far, $\cG$ is agnostic with respect to differentiations. To account for this, we further equip $\cG$ with \emph{weights}: to each edge $(f,x){\in}\cG$ is associated a nonnegative integer $d_{\!f\!x}$, equal to the maximal differentiation degree of variable $x$ in function $f$. This yields a \emph{weight} for any matching $\cM$ of $\cG$ by the formula $w(\cM)=\sum_{(f,x)\in\cM}d_{\!f\!x}$.
Suppose we have a solution to the following problem:
\begin{problem}
	\label{liftuerhpituhepu8} 
Find a complete matching $\cM$ for $\cG$ and integer \emph{offsets $\{c_f\mid{f\in{F}}\}$ and $\{d_x\mid{x\in{X}}\}$}, satisfying the following conditions, where the $d_{\!f\!x}$ are the weights as before:
\beq
\bea{rcl}
d_x-c_f &\!\!\!\geq\!\!\!& d_{\!f\!x} \; \mbox{ with equality if } (f,x){\in}\cM
 \\
c_f &\!\!\!\geq\!\!\!& 0
\eea
\label{ltuhltrli}
\eeq
\end{problem}
Then, differentiating $c_f$ times each function $f$ yields a DAE system $F_\Sigma{=}0$ having the following properties. Inequality $d_x\geq{c}_f+d_{\!f\!x}$ holds for each $x{\in}X$, and $d_x={c}_f+d_{\!f\!x}$ holds for the unique $f$ such that  $(f,x){\in}\cM$. Hence, the leading variables of DAE system $F_\Sigma{=}0$ are the 
$d_x$-th derivatives $\pprime{d_x}{x}$. Consequently, system $F_\Sigma{=}0$, now seen as a system of algebraic equations having $\pprime{d_x}{x}$ as dependent variables, is structurally nonsingular by Definition~\ref{erfuilehui}. Hence, $F_\Sigma{=}0$ is like an ODE. The integer $k=\max_{f\in{F}}\,c_f$ is called the \emph{index} of the system.\footnote{We should rather say the \emph{differentiation index} as, once again, other notions of index exist for DAEs~\cite{CampbellGear1995} that are not relevant to our work.}
\begin{definition}
	\label{def:erlgtuioo} For $F$ a DAE system, the solution to Problem\,$\ref{liftuerhpituhepu8}$ yields the DAE system $F_\Sigma{=}\{f^{\prime{c_f}}{\mid} f{\in}{F}\}$ together with the system of \emph{consistency constraints} $\overline{F}_\Sigma{=}\{f^{\prime{k}}{\mid} f{\in}{F},0{\leq}{k}{<}c_f\}$.
\end{definition}
Knowing the offsets also allows transforming $F$ into a system of index~$1$, by not performing the final round of differentiations. 
There are infinitely many solutions to Problem~\ref{liftuerhpituhepu8} with unknowns $c_f$ and $d_x$, since, for example, adding the same $\ell\in\bN_{\geq{0}}$ to all $c_f$'s and $d_x$'s yields another solution. We thus seek for a \emph{smallest} solution, elementwise.
Hence, Problem~\ref{liftuerhpituhepu8} is the fundamental problem we must solve, and we seek a smallest solution for it.

The beautiful idea of J. Pryce is to propose a linear program encoding Problem~\ref{liftuerhpituhepu8}. As a preliminary step, we claim that a bruteforce encoding of Problem~\ref{liftuerhpituhepu8} is the following: Find integers $\xi_{f\!x},d_x,c_f$ such that
\begin{equation}
\bea{rl}
\left.\bea{r}
\sum_{f:(f,x)\in{\cG}}\;\xi_{f\!x}=1 \\
\sum_{x:(f,x)\in{\cG}}\;\xi_{f\!x}={1} \\
\xi_{f\!x}\geq{0} \eea\right\}\hspace*{-3mm} &\mbox{complete matching}
\\ [6mm]
\left.\bea{r}
d_x{-}c_f-d_{\!f\!x}\geq{0} \\ 
c_f\geq{0} \eea\right\}\hspace*{-3mm}  &\mbox{encodes ``$\geq$'' in (\ref{ltuhltrli})}
\\ [3mm]
\displaystyle\sum_{(f,x)\in{\cG}}\;\xi_{f\!x}(d_x{-}c_f-d_{\!f\!x})=0\hspace*{-1mm}  &\mbox{encodes ``$=$'' in (\ref{ltuhltrli})}
\eea
\label{kerfyugwerluh}
\end{equation}
We now justify our claim. Focus on the first block. Since the $\xi$'s are integers, they can only take values in $\{0,1\}$ and one can define a bipartite graph by stating that edge $(f,x)$ exists iff $\xi_{f\!x}=1$. The first two equations formalize that this graph is a matching, which  is complete since all equations are involved. The second block is a direct encoding of (\ref{ltuhltrli}) if we ignore the additional statement ``with equality iff''. The latter is encoded by the last constraint (since having the sum equal to zero requires that all the terms be equal to zero).
Constraint problem (\ref{kerfyugwerluh}) does not account for our wish for a ``smallest'' solution: this will be handled separately.

Following the characterization of solutions of linear programs via  \emph{complementary slackness conditions}, every solution of problem (\ref{kerfyugwerluh}) is a solution of the following dual linear programs (LP), where the $\xi_{f\!x}$ are real:
\beq\mbox{\emph{primal}} \hspace*{-3mm}&:&\bea{rl}
\mbox{maximize}&\sum_{(f,x)\in{\cG}}\;d_{\!f\!x}\;\xi_{f\!x} \\
[1mm] \mbox{subject to}&\sum_{f:(f,x)\in{\cG}}\;\xi_{f\!x}=1 \\
\mbox{and}&\sum_{x:(f,x)\in{\cG}}\;\xi_{f\!x}\geq{1} \\
\mbox{and}&\xi_{f\!x}\geq{0}
\eea
\label{erofuierhopiu}
\\ [2mm]
\mbox{\emph{dual}} \hspace*{-3mm}&:&\bea{rl}
\mbox{minimize}&\sum_xd_x{-}\sum_f\;c_f \\
[1mm] \mbox{subject to}&d_x{-}c_f\geq{d}_{\!f\!x} \\
\mbox{and}&c_f\geq{0}
\eea
\label{guiohpui}
\eeq
In these two problems, $f$ ranges over $F$, $x$ ranges over $X$, and $(f,x)$ ranges over $\cG$. Also, we have relaxed the integer LP to a real LP, as the two possess identical solutions in this case. 
Note that LP (\ref{erofuierhopiu}) encodes the search for a complete matching of maximum weight for $\cG$. By the principle of complementary slackness in linear programming, 
\beq
\mbox{
\begin{minipage}{7cm}
	 for respective optima of problems (\ref{erofuierhopiu}) and (\ref{guiohpui}), $\xi_{f\!x}{>}0$ if and only if $d_x{-}c_f{=}d_{\!f\!x}$,
\end{minipage}
}
\label{eogiheoiwy}
\eeq
which is exactly the last constraint of (\ref{kerfyugwerluh}). Using this translation into linear programs (\ref{erofuierhopiu}) and (\ref{guiohpui}), it is proved in~\cite{Pryce01}, Thm 3.6, that, if a solution exists to Problem~\ref{liftuerhpituhepu8}, then a unique elementwise smallest solution exists. 
Based on the above analysis, the following algorithm was proposed in~\cite{Pryce01} for solving (\ref{erofuierhopiu},\ref{guiohpui}) and was proved to provide the smallest solution for (\ref{guiohpui}). 
Let $\cG$ be a bipartite graph with weights $\{d_{\!f\!x}|(f,x){\in}\cG\}$:
\begin{enumerate}
	\item 
\label{elriuu}	Solve LP (\ref{erofuierhopiu}), which gives a complete matching $\cM$ of maximum weight for $\cG$; any method for solving LP can be used;
	\item \label{elrgfuehliu} Apply the following iteration until a fixpoint is reached (in finitely many steps), from the initial values $c_f=0$:
\begin{enumerate}
	\item \label{erlfiuhpui} $\forall x:d_x\la\max\{d_{\!f\!x}+c_f\mid(f,x)\in\cG\}$~;
	\item \label{louighgpuig} $\forall f:\, c_f\,\la d_x-d_{\!f\!x}$ where  $(f,x)\in\cM$~.
\end{enumerate}
\end{enumerate}
The reason for using the special iterative algorithm for solving the dual LP (\ref{guiohpui}) is that a standard LP-solving algorithm will return an arbitrary solution, not necessarily the smallest one. The following lemma, that will be crucially put to use in Section~\ref{sec:otherschemes}, is an obvious consequence of the linearity of problems (\ref{erofuierhopiu}) and (\ref{guiohpui}):
\begin{lemma}
	\label{rilugtlo}  Let $\cG$ be a given bipartite graph and let two families of weights 
	$(d^1_{\!f\!x})_{(f,x)\in\cG}$ and 	$(d^2_{\!f\!x})_{(f,x)\in\cG}$ be related by 
	$d^2_{\!f\!x}=M\times{d^1_{\!f\!x}}$ for every ${(f,x)\in\cG}$, where $M$ is a fixed positive integer. Then, the offsets of the corresponding \sigmamethod\ are also related in the same way: $d^2_x=M\times{d^1_x}$ for every variable $x$, and $c^2_f=M\times{c^1_f}$ for every function $f$.
\end{lemma}
A necessary and sufficient condition for Problem~\ref{liftuerhpituhepu8} to have a solution is that the set of complete matchings for $\cG$ is non-empty. We thus consider the  function
\beq
\success\gets\atomicact{ExistsMatching}(\cG)\,,
\label{ergoijpojii}
\eeq
which returns a Boolean $\success$ such that $\success=\ttt$ iff there exists a complete matching for the bipartite graph $\cG$ (regardless of its associated weights). In this case, step~\ref{elriuu} of the above algorithm is guaranteed successful (since the set of complete matchings for $\cG$ is nonempty and finite). In the sequel, we encode the above algorithm  as the function
\beq
(\bfc,\bfd)\gets\atomicact{FindOffsets}(\cG,\bfD)\,, &\mbox{where}& 
\left\{\bea{rr}
\bfc:&\hspace*{-2mm} F\ra\bN \\
\bfd:&\hspace*{-2mm} X\ra\bN \\
\bfD:&\hspace*{-2mm} F{\times}X\ra\bN  
\eea\right. \label{erilughroghwgouyg} 
\eeq
Having $\success=\ttt$ returned by $\atomicact{ExistsMatching}()$ is a prerequisite to  calling  $\atomicact{FindOffsets}()$, which is then guaranteed successful. 

\begin{algorithm}[ht]
	\caption{$\atomicact{SigmaMethod}$} 
	\label{alg:indexreduction}
 \begin{algorithmic}[1]
\Require $F$;
\Return $(\success,\consistency{F}_\Sigma,F_\Sigma)$
\State $\success\gets\atomicact{ExistsMatching}(\cG)$.
\If{$\success$}
\State $(\bfc,\bfd)\gets\atomicact{FindOffsets}(\cG,\bfD)$ \label{op:lghuoiygfori}
\State $(F_\Sigma,\consistency{F}_\Sigma)\gets\atomicact{LatentEqns}(\bfc,F)$
\label{op:oeriugheriopu}
\EndIf
	\end{algorithmic}
\end{algorithm}
We encapsulate these two functions as~\rref{alg:indexreduction} ($\atomicact{SigmaMethod}$). It requires a DAE system $F{=}0$. $\cG$ is the weighted bipartite graph of $F$ with weights $\bfD$. $\atomicact{SigmaMethod}$ returns a Boolean $\success$ and, if $\success=\ttt$, a pair $(\consistency{F}_\Sigma,F_\Sigma)$ defining an index-0 system $F_\Sigma{=}0$ and a (possibly underdetermined) system $\consistency{F}_\Sigma{=}0$ of consistency constraints following \rref{def:erlgtuioo}.
In \rref{op:oeriugheriopu}, $F_\Sigma$ collects the $\pprime{c_f}{f}$, for $f\in{F}$, and $\consistency{F}_\Sigma$ collects the $\pprime{k}{f}$, for $f\in{F}$ and $0\leq{k}<c_f$. This completes our background material on structural analysis.
\section{Structural analysis of multimode DAE systems}
\label{sec:loguhiuliugh}
In this section, we extend structural analysis, from (single-mode) DAE systems, to multimode DAE systems. As a prerequisite, we precisely define the class of multimode DAE systems considered. Since the structural analysis will actually be applied to the nonstandard semantics, we also need to define multimode \emph{difference Algebraic Systems} (\emph{dAE}), operating in discrete time.

\subsection{Defining multimode DAE/dAE systems}
\label{sec:mdAE}
In this section, the class of systems of multimode Differential/difference Algebraic Equations considered in this paper is formally defined. By analogy with DAE's, we use the acronym dAE to mean \emph{difference} algebraic equations. Notation $\postset{x}$ denotes the shift of $x$, as per \rref{defn:tshifts}.
\begin{notation} \rm
	\label{lerigftuheroiuof}
Let $\Vars$ be an underlying set of variables. For
$x{\in}\Vars$ and $m{\in}\bN$, the $m$-differentiation and $m$-shift of $x$ are denoted by
$\pprime{m}{x}$ and $\ppostset{m}{x}$, respectively.  
Let 
$\pprime{m}{\Vars}$ and $\ppostset{m}{\Vars}$ denote the set of all
$\pprime{m}{x}$ and $\ppostset{m}{x}$, for $x$ ranging over the set
$\Vars$ of variables. Let
\beq
\Vars^{(\prime)}\eqdef\bigcup_{m\in\bN}\pprime{m}{\Vars} \mbox{ and }
\Vars^{(\bullet)}\eqdef\bigcup_{m\in\bN}\ppostset{m}{\Vars} \enspace .
\label{erpfuhqf9}
\eeq
For $X{\subset}\Vars$, sets $\pprime{m}{X},\ppostset{m}{X},X^{(\prime)},X^{(\bullet)}$ are defined in a similar way.\eproof
\end{notation}
The following assumption will be in force in the remainder of this paper:
\begin{assumption}
	\label{liruyeliutyl} 
All the functions and predicates over $\Vars^{(\prime)}$ are assumed to involve only finitely many of the $\pprime{m}{x}$ for $x{\in}\Vars$ and $m{\in}\bN$. The same holds regarding $\Vars^{(\bullet)}$.
\end{assumption}
\subsubsection{Syntax and meaning}
\label{sec:loeriuehoriu}

		From the discussion in \rref{sec:language} regarding the Cup-and-Ball example, we know we need to type the modes as either \emph{long} or \emph{transient}, since the two situations require different index reduction techniques. This justifies the consideration of two different kinds of guarded equations.
\begin{definition}[\mDAE/\mdAE\ syntax]
\label{defn:mDAEattempt}
\emph{Multimode DAE systems (\mDAE{s})},
  respectively \emph{multimode dAE systems (\mdAE{s})}, are defined by the following syntax:
	\beq
	S&::=& s \mid S;s \nonumber \\
	s &::=& g \mid e \nonumber \\
g &::=& \guard=b \label{leirufghlouei} \\
	e&::=&~~~\;\wemph{\mid}\prog{if}\;\guard\;\doo\;f{=}0 \label{oui9hgpo} \\
&&\mid\prog{when}\;\guard\;\doo\;f{=}0 \label{werlituhwliu} 
\eeq
\end{definition}
\mDAE\ $S$ is a finite set of $s$; $s$ is either a \emph{guard evaluation} $g$ (\ref{leirufghlouei}), or a \emph{guarded equation} $e$, of two different kinds (\ref{oui9hgpo}) or (\ref{werlituhwliu}); $f$ is a smooth scalar
function over $\Vars^{(\prime)}$, and $b$ is a Boolean expression of predicates over $\Vars^{(\prime)}$. The same holds for the syntax of \mdAE, with $\Vars^{(\bullet)}$ replacing $\Vars^{(\prime)}$.
\begin{notation}\rm 
	\label{ilgftuerhogftuie} 
 Let $\eqq$ be a guarded equation. Its \emph{body} $f{=}0$ and  \emph{guard} $\guard$ are denoted by $f(\eqq)$ and $\guard(\eqq)$. For $S$ an \mDAE\ or \mdAE\ system, we denote by $X(S)$, $\Eqqs(S)$, and $\Guards(S)$ its sets of numerical variables, guarded equations, and guards, respectively. The mention of $S$ is omitted when it is clear from the context. Finally, for $e:\prog{if/when}\;\guard\;\doo\;f{=}0$ a guarded equation following \rref{defn:mDAEattempt}, we define
 \beq
 \postset{e}/\dot{e} &:& \prog{if/when}\;\guard\;\doo\;\postset{f}/\dot{f}{=}0
 \label{eprguiheguil}
 \eeq
 where $\postset{f}$ follows \rref{defn:tshifts}, and, for $E$ a set of equations, Notations~\ref{lerigftuheroiuof} are adapted so that $\pprime{m}{E},\ppostset{m}{E},E^{(\prime)},E^{(\bullet)}$ can be considered. We insist that, in (\ref{eprguiheguil}), the guard $\guard$ is left unchanged.
\eproof
\end{notation}
\paragraph{\bf Meaning} The meanings of \ifequation\ (\ref{oui9hgpo}) and \whenequation\ (\ref{werlituhwliu}) differ:
\begin{itemize}
	\item In the \ifequation\ (\ref{oui9hgpo}),  equation $f{=}0$
is enabled if and only if the guard $\guard$ holds.  Otherwise, this equation is
disabled.
	\item In the \whenequation\ (\ref{werlituhwliu}),  equation $f{=}0$
is enabled exactly at the events when guard $\guard$ switches from $\fff$ to $\ttt$. Otherwise, this equation is disabled.\eproof
\end{itemize}
For \mdAE\ systems, the formalization of the second statement is straightforward: equation $f{=}0$ is enabled when
\beq
\up{\guard} \mbox{ holds, where } \up{\guard}~\eqdef~(\prog{not}\;{\preset{\guard}})\;\prog{and}\;\guard \enspace .
\label{lwtguiohpui}
\eeq
For \mDAE\ systems: equation $f{=}0$ is enabled when
\beq
\up{\guard} \mbox{ holds, where } \up{\guard}~\eqdef~(\prog{not}\;{{\guard}^-})\;\prog{and}\;\guard
\label{epgiuhwligu}
\eeq
where $\guard^-$ denotes the left-limsup of $\guard$ at the current instant, defined by:
\[\bea{rll}
\guard^-(t) &\eqdef& \limsup_{s\nearrow{t}}\guard(s) \\ [1mm]
&=& \prog{if}\; \exists s_n\nearrow t: \guard(s_n){=}\ttt ~~ \prog{then}\; \ttt \; \prog{else}\; \fff
\eea\]
The left-limsup is always defined and coincides with the left limit when the latter exists.
\begin{ccomment}\rm 
	\label{guihui} Why bother with the two types of equations, since \whenequation{s} are mapped to \ifequation{s} by mapping $\guard$ to $\up{\guard}$? The motivation was given in \rref{sec:language}. We need to type the modes as either \emph{long} or \emph{transient}, since the two situations require different index reduction techniques. This is precisely what our \whenequation{s} guarantee, since, by construction, switching from $\fff$ to $\ttt$ occurs in a single nonstandard instant. \emph{Hence, by construction, \whenequation{s} are active at transient modes only.} Actually, this distinction between \ifequation{s} and \whenequation{s} already exists in object-oriented DAE-based languages such as Modelica. \eproof
\end{ccomment}
Regarding \ifequation{s}, we state the following assumption:
\begin{assumption} In \ifequation{s}, guard $\guard$ holds true for a (standard) positive duration.
	\label{erjtdfguiohu}  
\end{assumption}
It is the responsibility of the programmer to make sure that \ifequation{s} are properly used so that Assumption\,\ref{erjtdfguiohu} holds.

\subsubsection{Long versus transient modes}
\label{sec:oguihwpgui}
For our theoretical development and the presentation of algorithms, it will be convenient to simplify our syntax by replacing \whenequation{s} by \ifequation{s} using (\ref{lwtguiohpui}) and Assumption\,\ref{erjtdfguiohu}. Our syntax of \rref{defn:mDAEattempt} with arbitrary guards is thus replaced with the following (equivalent) syntax:
\begin{definition}[\mDAE/\mdAE~revisited]
\label{defn:mDAE}
\emph{Multimode DAE systems (\mDAE{s})},
  respectively \emph{multimode dAE systems (\mdAE{s})}, are defined by the following syntax:
	\beq
	S&::=& s \mid S;s \nonumber \\
	s &::=& g \mid e \nonumber \\
g &::=& \guard=b \nonumber \\
	e&::=&\prog{if}\;\guard\;\doo\;f{=}0 \label{eprioguhoiu} 
\eeq
where $b$ is a $\prog{\llong}/\prog{\transient}$-\emph{typed} Boolean expression of predicates over $\Vars^{(\prime)}$ (respectively $\Vars^{(\bullet)}$).
\end{definition}
From now on, and unless otherwise specified, we will use the modified syntax of \rref{defn:mDAE}. The notion of mode is then easily introduced.

\begin{definition}[mode]
	\label{defn:egfuiher} For $S$ a \mDAE\ or \mdAE, a \emph{mode} is a valuation, in $\{\fff,\ttt\}$, of all of its guards $\guard\in\Guards(S)$. The mode is called \emph{\transient} if at least one guard typed {\tt transient} takes the value $\ttt$ in it; otherwise, it is called \emph{\llong}. Modes will be generically denoted by the symbol $\mode$.
\end{definition}
 A mode enables a subset of the equations $f{=}0$ and disables the other ones.

\myparagraph{Nonstandard semantics}
An {\mDAE} $S$ is transformed to a (nonstandard) {\mdAE} $\nst{S}$ through the following syntactic transformation:
\beq\bea{rcl}
 \nst{S} & \eqdef & S\bigl[x' \; \mbox{is replaced with} \; \frac{\postset{x}-x}{\vsmall}\bigr]
\eea
\label{whgp45g8hseluio}
\eeq 
Since we perform the structural analysis on the nonstandard semantics, we must develop a structural analysis for \mdAE\ systems, operating in discrete time. Therefore, in the rest of this section, we work in nonstandard semantics.\eproof

\begin{ccomment}{\bf (Progressive evaluation of the modes)}\rm 
	\label{guiohpiouh} 
It may be that some guards have their value unknown when the execution of the current instant starts; this was, for example, the case with the Cup-and-Ball example in its first form given in Section~\ref{leriugfehliu}. We can use the known enabled equations to evaluate part of the variables, hoping that some of the yet unknown guards get their value defined. This yields more equations getting enabled/disabled, allowing for more guards being evaluated. If this process leads to the evaluation of all guards, then our execution scheme succeeds; otherwise, we report a failure due to a logico-numerical fixpoint equation, as was the case in Section~\ref{leriugfehliu}.

This progressive evaluation of guards becomes problematic when both long and transient modes are considered. Indeed, different structural analyses are used for each case (see \rref{sec:lerighugholuih} for the long mode case of the Cup-and-Ball example, versus~\rref{sec:leiruioepui} for the transient mode case). Unfortunately, we do not know which one should be used if the mode is only partially known, since the  long/transient qualification is only known when all guards have been evaluated. 
Consequently, we have to choose between the following two incompatible features for our compilation method:
\begin{enumerate}
	\item either we support the progressive evaluation of modes, but are unable to handle transient modes; 
	\item or we support both long and transient modes, but we do not support the progressive evaluation of modes.
\end{enumerate}
Since many physically meaningful models involve transient modes, our opinion is that the second option is much more useful in practice, hence it is adopted in our work.\eproof
\end{ccomment}
Comment~\ref{leriuhlu}, regarding the Cup-and-Ball example, discussed what we miss by favoring the second option. 
Indeed, delaying the effect of guards was the appropriate correction to remove logico-numerical fixpoint problems in the Cup-and-Ball example. Remember that the so introduced delay is infinitesimal, so this does not really matter.\footnote{This is unlike Synchonous Languages for discrete time systems, in which the restriction resulting from adopting \rref{defn:mDAEmodified} would be significant.}

Therefore, throughout the rest of this work, \rref{defn:mDAE} is modified to restrict ourselves \emph{by the syntax} to systems in which all guards can be evaluated at the initialization of each nonstandard instant:\footnote{With this new syntax, the original forms for the Cup-and-Ball and Westinghouse examples (given in the introduction) could not be expressed.}
\begin{definition}[\mdAE\ modified]
\label{defn:mDAEmodified}
\emph{Multimode dAE systems (\mdAE{s})}, are defined by the following syntax:
	\beq
	S&::=& s \mid S;s \nonumber \\
	s &::=& g \mid e \nonumber \\
g &::=& \postset{\guard}=b \label{leriughou} \\
	e&::=&\prog{if}\;\guard\;\doo\;f{=}0 \nonumber 
\eeq
where $b$ is a $\prog{\llong}/\prog{\transient}$-\emph{typed} Boolean expression of predicates over $\Vars^{(\bullet)}$.
\end{definition}
\begin{notation}\rm 
	\label{lguihuio} 
Referring to (\ref{leriughou}), we write $\bexp(\guard)$ to refer to the expression defining $\postset{\guard}$.\eproof
\end{notation}
From now on, by \mdAE, we will mean a system of the form given in \rref{defn:mDAEmodified}.

\subsection{Structural analysis: intuition}
\label{seoiguhweroui}
In this section, we give an intuition about how to extend the structural analysis from dAE systems to mdAE systems. To simplify our explanation, 
\emph{we consider the restricted case in which both modes (previous and current) are long.}
Also, it will be convenient to assume, for the two modes before and after the change, that we have two index-$1$ models, of the form
\beqq
S_i:
\left\{\bea{rcl}
\dot{X_i}&=&F_i(X_i,Y_i) \\ 0&=&H_i(X_i)
\eea
\right., \mbox{ for }i=0,1\,,
\eeqq
where indices $0$ and $1$ refer to the system before and after the change. Observe that system state $X_i$ may collect different variables before and after change, and so does the algebraic tuple $Y_i$. Of course, the dynamics $F_i,H_i$ are subject to changes as well. 
Let $\Jacobian_i\eqdef\Jacobian_{\!X_i}{H_i}$ denote the Jacobian of $H_i$ w.r.t. $X_i$. Hence, the system augmented with its latent equations 
\beqq
S_{i\Sigma}:
\left\{\bea{rcl}
\dot{X_i}&=&F_i(X_i,Y_i) \\ 0&=&H_i(X_i) \\ 0&=&\Jacobian_i(X_i)\dot{X_i}
\eea
\right., \mbox{ for }i=0,1 \enspace ,
\eeqq
seen as a system of algebraic equations with $\dot{X_i},Y_i$ as dependent variables,
is structurally nonsingular. Equivalently, we may replace, in this system, the differentiation of the algebraic constraint, $\frac{d}{dt}H_i(X_i)$, by its shifting $H_i(\postset{X_i})$ (their bipartite graphs are equal):
\beq
S_{i\Sigma}:
\left\{\bea{rcl}
\dot{X_i}&=&F_i(X_i,Y_i) \\ 0&=&H_i(X_i) \\ 0&=&H_i(\postset{X_i})
\eea
\right., \mbox{ for }i=0,1 \enspace .
\label{rtlgihsfoih}
\eeq
Form~(\ref{rtlgihsfoih}) can be seen as the result of applying Pantelides' or Pryce's structural analysis in each mode, before and after change. 
\begin{ccomment}[key idea]
	\label{erkuiyui}\rm 
{\emph{Our key idea is to handle the mode change as a task of reconciling the possible conflict between the predictions generated by the previous mode and the consistency conditions associated to the new mode.}}
Indeed, at the instant of mode change, there is a possible conflict between:
\begin{itemize}
	\item on the one hand, the value for $X_0$ predicted by model $S_{0\Sigma}$ which sets the value of the derivative $\dot{X_0}$ just before the change;
	\item on the other hand, the consistency constraints on the pair $X_1,Y_1$ resulting from model $S_{1\Sigma}$ after the change.\eproof
\end{itemize}
\end{ccomment}
To perform this reconciliation, we move to the nonstandard semantics by \emph{syntactically replacing}, in (\ref{rtlgihsfoih}), the derivatives $\dot{x}$ by their first-order Euler approximations with infinitesimal time step $\vsmall$:
\beqq
\dot{x}(t) \underbrace{\gets}_{\rm is\;replaced\;by} \frac{{x}(t+\vsmall)-x(t)}{\vsmall} = \frac{\postset{x}(t)-x(t)}{\vsmall} \enspace .
\eeqq
In addition, we replace, in~(\ref{rtlgihsfoih}), the derivative $\frac{d}{dt}H_i(X_i)=\Jacobian_i(X_i)\dot{X_i}$ by the time shift $H_i(\postset{X_i})$. This substitution is legitimate, since the bipartite graph of the system is left unchanged by performing this replacement.

Performing the above substitutions in (\ref{rtlgihsfoih}) and sampling it on the nonstandard discrete time basis $\bT{=}\{0,\vsmall,2\vsmall,\dots\}$ yields a discrete-time two-modes system. At the instant of mode change, we inherit the following system of algebraic equations:
\beq
S &=& \underbrace{~\preset{S}_{0\Sigma}~}_{\rm previous\;instant} \bigcup \underbrace{~S_{1\Sigma}~}_{\rm current\;instant}
\label{ekerughlruigth}
\eeq
where 
\begin{equation}
\preset{S}_{0\Sigma}:
\left\{\bea{rcl}
\frac{X_0-\preset{X_0}}{\vsmall}&=&F_0(\preset{X_0},\preset{Y_0}) \\ 0&=&H_0(\preset{X_0}) \\ 0&=&H_0(X_0)
\eea
\right. ~,~ 
{S}_{1\Sigma}:
\left\{\bea{rclc}
\frac{\postset{X_1}-{X_1}}{\vsmall}&=&F_1({X_1},{Y_1}) & (\eqq_0) \\ 0&=&H_1({X_1})  & (\eqq_1) \\ 0&=&H_1(\postset{X_1})  & (\postset{\eqq_1}) 
\eea
\right.
\label{lerigtuerhpiuwh}
\end{equation}
and $\cup$ refers to the union of sets  of equations. For each state variable $x$ that is shared between the two systems, i.e., $x\in{X_0}\cap{X_1}$, the value for $x$ must be identical in $\preset{S}_{0\Sigma}$ and ${S}_{1\Sigma}$.
To avoid duplicates of a same equation while taking the union of previous and current systems in $S$, 
	 \emph{we identify previous and current equations that are identical in their syntax, if any.} 
The dependent variables of $S$ are the $X_0$ inherited from $S_{0\Sigma}$ and the $X_1,Y_1,\postset{X_1}$ inherited from $S_{1\Sigma}$. When seen as sets of variables, $X_0$ on the one hand, and $X_1,Y_1,\postset{X_1}$ on the other hand, possess in general a nonempty intersection.
Being the result of a structural analysis, system ${S}_{1\Sigma}$ is:
\begin{itemize}
	\item structurally nonsingular, when seen as a system of algebraic equations in its leading variables $Y_1,\postset{X_1}$, for given values for $X_1$; and
	\item possibly structurally underdetermined, when 
seen as a system of algebraic equations in all its variables $X_1,Y_1,\postset{X_1}$ (ensuring consistent initialization).
\end{itemize}
Focus on $S_{1\Sigma}$ in decomposition (\ref{ekerughlruigth}). Denote by $\Phi$ the subsystem of $S_{1\Sigma}$ collecting all its equations that also belong to $\preset{S}_{0\Sigma}$, and are thus already solved. 

The remaining system $S_{1\Sigma}\setminus\Phi$ must then be solved for its dependent variables, which consist of the subset of $X_1\cup Y_1\cup \postset{X_1}$ collecting all variables not yet determined by $\preset{S}_{0\Sigma}$. 

To prepare for this, we submit $S_{1\Sigma}\setminus\Phi$ to the Dulmage-Mendelsohn (DM) decomposition, see Definition~\ref{elrftuierhfperui}: $(\block_\underapprox,\block_\squared,\block_\overapprox)\gets\mathrm{DM}(S_{1\Sigma}\setminus\Phi)$. If ${\block}_\underapprox={\block}_\overapprox=\emptyset$, then $S_{1\Sigma}\setminus\Phi$ is structurally nonsingular by Lemma\,\ref{erlihpfiu} and we solve it. Otherwise, the following  can occur:

\myparagraph{Case ${\block}_\underapprox\neq\emptyset$} This points to an underspecified subsystem at the instant of mode change. This information is returned to the designer for a correcting action.

\myparagraph{Case ${\block}_\underapprox=\emptyset$ and ${\block}_\overapprox\neq\emptyset$} This points to a conflict between $\preset{S}_{0\Sigma}$ and $S_{1\Sigma}$ and is handled differently. We solve this conflict by applying Principle~\ref{oeuryfgy} of causality. Accordingly, the equations belonging to ${\block}_\overapprox$ are removed from the consistency equations $(\eqq_1)$ of the new system $S_{1\Sigma}$. We thus keep the system $S^\downarrow_{1\Sigma}\eqdef\{(\eqq_0),(\eqq_1)\setminus{\block}_\overapprox,(\postset{\eqq_1})\}$. 
By Lemma~\ref{togiethgouiho}, since the system $\{(\eqq_0),(\postset{\eqq_1})\}$ is structurally nonsingular, performing the DM decomposition on system $S^\downarrow_{1\Sigma}$ returns a triple with an empty overdetermined block.
Hence, the so obtained reduced system $S^\downarrow_{1\Sigma}$ is structurally regular and  ready to be solved. 
\begin{ccomment}\rm
	\label{trlguihp} 
If the two systems ${S}_{i\Sigma},i=0,1$ are identical (there is no mode change), then $X_1=X_0,Y_1=Y_0,F_1=F_0$, and $H_1=H_0$, so that 
$\Phi$ consists of the consistency subsystem $0=H_1(X_1)$, see the definition of $S_{1\Sigma}$ in (\ref{lerigtuerhpiuwh}). Hence, $S_{1\Sigma}\setminus\Phi$ coincides in this case with the structurally nonsingular system obtained by performing index reduction (see \rref{def:erlgtuioo}), and we are left with the usual structural analysis for DAE.

The reconciliation may take several nonstandard instants following the mode change, where less and less equations from ${S}_{1\Sigma}$ are erased when applying DM, until we reach the steady regime. The interpretation of the erasure of some equations from ${S}_{1\Sigma}$ is that their satisfaction gets postponed for a few instants. Recall that we are in nonstandard semantics, so that all of this will occur in a total zero duration, in standard real time.
\end{ccomment}
In the rest of this section, we formalize this intuition and generalize it.

\subsection{Structural Analysis of long modes}
\label{sec:csemantics}
In this section we develop the structural analysis of \mdAE\ systems for the subclass of systems possessing only long modes. Within each long mode, we simply reuse the existing structural analysis based on the \sigmamethod. The only difficulty sits in the mode changes, for which we use conflict reconciliation as explained above. Throughout this section, we work in the nonstandard semantics as given by Equation~(\ref{whgp45g8hseluio}).

\subsubsection{Constructive Semantics}
\label{eproguihpeghp}
We reuse the techniques of \emph{constructive semantics,} first introduced in the context of reactive synchronous programming languages~\cite{Berry96,BenvenisteCG00,BenvenisteCEHGS03,EsterelBook}, for the purpose of grounding compilation on solid mathematical foundations. 
For synchronous languages, the execution of a program proceeds through successive \emph{reactions}, by which discrete time progresses from the current instant to the next one. The compilation of a program consists in generating the code for executing a reaction. This task is formalized through the notion of \emph{constructive semantics}, which relies on the following two pillars:
\begin{enumerate}
\item The specification of the set of \emph{{atomic actions}}, which are effective, non-interruptible, state transformations. For synchronous languages, atomic actions consist of: $({a})$ the evaluation of a single expression, and $({b})$ control flow actions. Executing an instance of an atomic action is referred to as performing a \emph{micro-step}. 
\item The \emph{causality analysis} of the reaction, which is a partial order on the set of micro-steps, abstracting the dependencies between them. A scheduling of the micro-steps is correct if and only if it complies with the causality analysis.
\end{enumerate}
The \emph{constructive semantics} consists in decomposing a reaction as the symbolic execution of a sequence of micro-steps subject to causality constraints. 
Intermediate stages of this execution are represented by the `status' of each variable, belonging to the abstract alphabet \{{not\,evaluated}/{evaluated}\}. The assignment of a status to all variables is called a \emph{configuration}. Such executions are monotonic, in that the so produced sequence of configurations is increasing with respect to the product order derived from the order \emph{not\,evaluated}$\,<\,$\emph{evaluated} for each variable. The constructive semantics is complete if it terminates by having all variables evaluated. Failure occurs if an incomplete sequence cannot be extended.
This approach is suited to \mdAE\ systems, since it will provide the formal support for correctly chaining the different actions composing the structural analysis, and describing the resulting increase in knowledge on all involved equations, variables and guards.
For {\mdAE} systems, the \emph{atomic actions}, however,  consist of: 
\begin{itemize}
\item[$({a})$]  the evaluation of a guard; 

\vspace*{-2mm}

\item[$({b})$]  manipulations of systems of equations: adding latent equations or deleting conflicting equations at mode changes; and 

\vspace*{-2mm}

\item[$({c})$] solving a block of algebraic equations.
\end{itemize}
Actions $(a)$ and $(b)$ participate in the structural analysis, whereas $(c)$ is the runtime solving of systems of equations to evaluate the leading variables.
To keep this evaluation symbolic, we abstract away the values of numerical variables and regard the values of guards as oracles, so that all the possible valuations must be explored.

\paragraph{Statuses and Contexts}
\label{wpgtuihhj}
The different statuses of the guards, variables and equations of an \mdAE~are the following: an evaluated guard may be {\it true} or {\it false};
 a variable may be {\it not evaluated}, or {\it evaluated};
 an equation may be {\it not evaluated}, {\it disabled} (its guard is false), 
or {\it solved}. This is summed up in Table\,\ref{ouilfhl}. 
There is no need to consider the status {\it not evaluated} for a guard, since guards can be evaluated at the initialization of each nonstandard instant, by \rref{defn:mDAEmodified}.
\begin{table}[ht]
\begin{center}
\begin{tabular}{|r|c|c|c|c|}\hline
	 & \irrelevant &\undef &\fff &\ttt \\ \hline
	 guard & irrelevant &  & evaluated to $\fff$ & evaluated to  $\ttt$ \\ \hline
	 variable & irrelevant & not evaluated &  & evaluated \\ \hline
	 equation & irrelevant & not evaluated & disabled & solved \\ \hline
\end{tabular}
\end{center}
\caption{
The partially ordered domain $\dom = \{\irrelevant,\undef,\fff,\ttt\}$ and its interpretation for guards, variables, and equations.}
\label{ouilfhl}
\end{table}

Unlike for single-mode DAEs, the set of equations defining the dynamics of an \mdAE\ is mode-dependent. This information should therefore be traced as part of the constructive semantics of \mdAE\ systems. To capture this, we tag as \emph{irrelevant} in the considered mode all the variables, guards, and equations that are not used to define the dynamics of the \mdAE\ system in this mode.   
Formally, the domain of statuses is the partially ordered set $(\dom,<)$ defined by
\beq
\dom = \{\irrelevant,\undef,\fff,\ttt\} &\text{and}& \irrelevant \; < \; \undef \;<\; \fff,\ttt \enspace .
\label{eq:D}
\eeq
The interpretation of this domain for variables, equations and guards, is summarized in Table~\ref{ouilfhl} and detailed next:
\begin{itemize}
\item The minimal element $\irrelevant$ is used to represent the fact
  that a variable, a guard, or an equation is \emph{irrelevant}.
\item Value $\undef$ means that a variable or equation has \emph{not
  been evaluated yet}. At the beginning of a time step, only state
  variables are known, and all other variables are set to $\undef$,
  reflecting that their numerical values are not known yet.
\item Maximal element $\ttt$ has different meanings, depending on
  whether it applies to a variable, a guard or an equation. In the
  case of a variable, it means that the numerical value of the
  variable has been computed, whatever it could be. For a guard, it
  means that the guard has been evaluated to true. For an equation, it
  means that the equation has been solved.
\item Maximal element $\fff$ also has different meanings, depending on
  whether it applies to a guard or an equation. In the context of a
  guard, it means that the guard has been evaluated to false. When it
  applies to an equation, it means that the equation is
  disabled, i.e., its guard is false. This value does not apply to variables.
\end{itemize}
We call \emph{status} the assignment, to each guard, variable, and equation, of a value over $\dom$. It is formally defined as follows:

\begin{definition}[status]\label{def:status} Let $S$ be an \mdAE\ system. Its set $\bbV$ of \emph{\ScottVars}\footnote{The prefix ``\textbf{S}'' is reminiscent of constructive \textbf{S}emantics.}  is
\beq
\bbV &\eqdef&  \Guards \cup X^{(\bullet)} \cup E^{(\bullet)}
  \label{kldfykl76}
  \eeq
(using Notations~\ref{lerigftuheroiuof} and~\ref{ilgftuerhogftuie}). A \emph{status} $\status$ of $S$ is a valuation in $\dom$ of the
\ScottVars, that is, a mapping $\status:\bbV\rightarrow \dom$.
A status is called \emph{finite} if it is
equal to $\irrelevant$, except for a finite number of \ScottVars.
The set of statuses is partially ordered by the product order:
$\sigma_1 \leq \sigma_2$ if and only if  for all $v\in\bbV$, $\sigma_1(v) \leq
\sigma_2(v)$.
\end{definition}
We denote by $\irrelevant\nearrow\undef$, $\undef\nearrow\fff$, and $\undef\nearrow\ttt$ the increasing changes of status for a given \ScottVar.

\myparagraph{Enabled Equations and Leading Variables}
Let $\status$ be a status. 
Equation $\ppostset{m}{\eqq}$ is \emph{enabled in} $\status$ (respectively \emph{disabled in} $\status$) if and only if $\status(\guard)=\ttt$ (respectively $\status(\guard)=\fff$), where $\guard$ is the guard of $\ppostset{m}{\eqq}$, see (\ref{eprguiheguil}).
Denote by  
$$\enbld{\eqq}{\status} \mbox{ (respectively  $\disbld{\eqq}{\status}$)}$$
 the set of equations that
are enabled (respectively disabled) in $\status$. Recall that, for any finite status $\status$,
these sets are finite. 
The set $\enbld{\eqq}{\status}$ defines the DAE system that is enabled at status $\status$. The \emph{leading variables} in status $\status$ are defined, with reference to this enabled DAE system, following the introduction of \rref{sec:structural}. 
We define the set
$$\bea{rcl}
\ndef{\status} &\eqdef& \{ v\in \bbV \mid \sigma(v) \leq \undef \}
\eea
$$ 
collecting the \ScottVars\ that are 
either irrelevant or not evaluated in status $\status$. We are now ready to introduce micro-steps, and runs as finite sequences of micro-steps.
\begin{definition}[micro-step]
	\label {ilugfehoriur} 
A \emph{micro-step} is a transformation, via some atomic action, of a status $\status$ to a status $\status'$ by updating the values of a finite subset of
\ScottVars, from $\irrelevant$ to $\undef$, or from $\undef$ to $\ttt$ or $\fff$. 
\end{definition}
A run is a finite sequence of micro-steps. It is complete when every equation is either irrelevant, disabled, or solved, and all the enabled leading variables are evaluated. 
Formally (the reader is referred to \rref{def:status} for the notions used):
\begin{definition}[runs]
\label{defn:success} 
Given a finite initial status $\status_0$,
a \emph{run} of \mdAE\ system $S$ is a finite increasing sequence of statuses
\begin{equation}
\status_0<\status_1<\dots<\status_k<\status_{k+1}<\dots<\status_\eend
\end{equation}
such that, for every $k<\eend$, $\status_{k+1}$ is obtained from $\status_k$ by performing a micro-step.
The final status $\status_\eend$ is called \emph{{\ssuccessful}} if, in $\status_\eend$, all equations $\eqq_i$ have either the value $\ttt$ or $\fff$ and no leading variable has the value $\undef$.
A run $\status_0 <  \ldots < \status_\eend$ of \mdAE\ system $S$ is  \emph{{\ssuccessful}}
if and only if its final status $\status_\eend$ is \ssuccessful.
\end{definition}
\myparagraph{Handling conflicts between past and present}
To handle the conflicts between the previous and current instants when a mode change occurs, we will need to keep track of the possible sources of conflict between previous and current instants. If equation $\ppostset{(m+1)}{\eqq},m{\geq}0$, was solved at the previous instant (meaning that $\status(\ppostset{(m+1)}{\eqq})=\ttt$ for the final status of that instant), then $\ppostset{m}{\eqq}$ is a possible source of conflict at the current instant. The set of all such equations is called the \emph{context} at the current instant. All the variables involved in the context are already evaluated at the initial status of the current run (value $\ttt$). 

If, however, equation $\ppostset{m}{\eqq}\in\Delta$ is also enabled at status $\status$ of the current instant, then  $\ppostset{m}{\eqq}$ is no longer a source of conflict, but rather a satisfied consistency equation. We call it a \emph{fact} at $\status$. Formally:
\begin{definition}[contexts and facts]
	\label{eporgfuieropiufl} 
	Let $\system$ be an \mdAE\ system. Its \emph{context} $\Delta$ at the current instant, and \emph{set of facts} $\Phi(\status)$ at the current status, are defined as follows:
\[\bea{l}
\Delta =
\left\{
\ppostset{m}{\eqq}
\mid
\emph{the final status of $\ppostset{(m+1)}{\eqq}$ at previous instant is $\ttt$}
\right\}
\\ [1mm]
\Phi(\status)= \Delta\cap\enbld{\eqq}{\status}
\eea\]
By abuse of notation, we will also denote by $\Delta$ and $\Phi$ the sets of functions $f(e)$, for $e$ ranging over $\Delta$ and $\Phi$.
\end{definition}
To prepare for the algorithm computing a run for given initial status $\status_0$ and context $\Delta$, we now collect a few useful auxiliary algorithms. 

\subsubsection{Auxiliary algorithms}

\paragraph{$\atomicact{Tick}$ algorithm \emph{(\rref{alg:rltoghiuiop})}} When a run is {\ssuccessful}, the system can proceed to the next time step by executing the algorithm $\atomicact{Tick}$. $\atomicact{Tick}$ requires a complete (final) status $\status$ and returns a pair $\left(\wpostset{\status},\wpostset{\Delta}\right)$ consisting of an initial status $\wpostset{\status}$ and a context $\wpostset{\Delta}$, for use as initial conditions for the next run.
\begin{algorithm}[ht]
\caption{$\atomicact{Tick}$}
\label{alg:rltoghiuiop}
	\begin{algorithmic}[1]
	\Require $\status$, a complete status; 
	\Return pair $\left(\wpostset{\status},\wpostset{\Delta}\right)$.
	\State $
\atomicact{Tick}(\status) \eqdef \left(\wpostset{\status},\wpostset{\Delta}\right)$,
where:
\begin{eqnarray*}
 \wpostset{\status}({\guard}) & = & \mbox{\bf if} \; {\guard} \; \mbox{is involved in} \, S \;\mbox{\bf then} \; \status({\bexp(\guard)}) 
                                          \; \mbox{\bf else} \; \irrelevant \\ 
 \wpostset{\status}\bigl(\ppostset{m}{x}\bigr) & = & \mbox{\bf if} \; \status\bigl(\ppostset{m+1}{x}\bigr) = \ttt \; \mbox{\bf then} \; \ttt \;\mbox{\bf else} 
\\&& \mbox{\bf if}\; \ppostset{m}{x} \; \mbox{is a variable of} \, S \; \mbox{\bf then} \; \undef\;  \mbox{\bf else} \; \irrelevant \\ 
\wpostset{\status}\bigl(\ppostset{m}{\eqq}\bigr) & = & \mbox{\bf if} \; \ppostset{m}{\eqq} \; \mbox{is a variable of} \, S \;\mbox{\bf then} \; \undef \; \mbox{\bf else} \; \irrelevant \\ 
\wpostset{\Delta} & = & \left\{ \ppostset{m}{\eqq} \left|\,\eqq{{\in}}\Eqqs, m{\in}\bN: \status\bigl(\ppostset{m+1}{\eqq}\bigr){=}\ttt \right. \right\}
\end{eqnarray*}
\end{algorithmic}
\end{algorithm}

Performing $\atomicact{Tick}$ has the effect of
shifting backward defined variables, and setting all other \ScottVars\ $v\in
\bbV$ to either $\undef$, if $v$ is relevant to the \mdAE\ $S$, or $\irrelevant$ (irrelevant) otherwise. Guards are set either to $\irrelevant$ or to the value $\status({\bexp(\guard)})\in\{\fff,\ttt\}$, using Notations~\ref{lguihuio} and the definition (\ref{leriughou}) for $\postset{\guard}$. Thus, the value of all guards at the next instant is set.
The new context is defined to be the set of equations that are
known to be satisfied in the next instant. $\atomicact{Tick}$ is not increasing and it does not have to be, since it applies when moving to the next time step.

\paragraph{$\atomicact{IndexReduc}$ algorithm \emph{(\rref{alg:lkciyughsliu})}} 
Given a system $G$ with its set of leading variables determined by the status $\status$, the $\atomicact{IndexReduc}$ algorithm is a renaming of the  \sigmamethod\ (\rref{alg:indexreduction}), together with a refined interpretation of its result. \emph{It is understood that we use the discrete-time adaptation of it, in which forward shift replaces differentiation but everything else remains unchanged.} 
\begin{algorithm}[ht]
	\caption{$\atomicact{IndexReduc}$}
	\label{alg:lkciyughsliu} 
\begin{algorithmic}[1]
	\Require $(G,\status)$; 
	\Return $(\status,b,G_\Sigma,\consistency{G}_\Sigma)$
	  \State $(b,G_\Sigma,\consistency{G}_\Sigma)\gets\atomicact{SigmaMethod}(G)$ with dependent variables set by $\status$; increase $\status$
      \label{op:trguiohpuih}
      \If{$b$} \label{op:lerigfueoiu}
       \Return $(G_\Sigma,\consistency{G}_\Sigma)$
      \EndIf
\end{algorithmic}
\end{algorithm}

Recall that this algorithm requires a square system $G{=}0$ and returns a success/failure flag $\success$, as well as, in case of success, a pair consisting of $G_\Sigma{=}0$, a structurally nonsingular system determining some leading variables, and the system $\consistency{G}_\Sigma{=}0$ collecting the consistency equations. As a consequence, the status of every latent equation added by the \sigmamethod\ is updated as $\irrelevant\nearrow\undef$.

\paragraph{$\atomicact{SolveConflict}$ algorithm \emph{(\rref{alg:elsgouihuip})}} 
This algorithm requires a system ${K}{=}0$ and a status $\status$ for all variables involved in $K$. System $K$ is submitted to the DM decomposition (\rref{op:roghuphpu}); in doing so, the dependent variables are those with status $\undef$. Possible conflicts are collected in $\block_\overapprox$. If the latter is nonempty,
corresponding equations are deleted (\rref{op:welriuwou}). The status is updated as $\undef\nearrow\fff$ for every removed equation.
The so reduced system is again submitted to the DM decomposition (\rref{op:ogiulwopaifugh}), which, by Lemma~\ref{togiethgouiho}, 
will return an empty overdetermined block: conflicts are removed. If block $\block_\underapprox$  is nonempty, then system ${K}$ is underdetermined. Therefore, $b=\fff$ is returned (\rref{op:eriughforui}). Otherwise, $b=\ttt$ and $H=\block_\squared$ are returned (\rref{op:erpigftuhpiou}).

\begin{algorithm}[h]
	\caption{$\atomicact{SolveConflict}$}
	\label{alg:elsgouihuip} 
\begin{algorithmic}[1] 
	\Require $(K,\status)$; 
	\Return    $(\status,b,H)$
	  \State
$(\block_\underapprox,\block_\squared,\block_\overapprox)\gets\mathrm{DM}({K})$ with dependent variables set by $\status$ \label{op:roghuphpu}
      \If{$\block_\overapprox\neq\emptyset$} \label{op:erituheoiut}
      \State  $K\;\gets\;K\setminus\block_\overapprox$\,;~ increase $\status$
      \label{op:welriuwou}
      \State
$(\block_\underapprox,\block_\squared,\emptyset)\gets\mathrm{DM}({K})$ with dependent variables set by $\status$ \label{op:ogiulwopaifugh} 
      \EndIf
\State  {$b\gets[\block_\underapprox=\emptyset]$}
       \label{op:eriughforui}
       \If{$b$} $H\gets\block_\squared$
      \EndIf
\label{op:erpigftuhpiou} 
\end{algorithmic}
\end{algorithm}

\paragraph{$\atomicact{Eval}$ algorithm} 
\beq
\status\gets\atomicact{Eval}(H,\status)
\label{toihjiurh}
\eeq
This algorithm requires a status $\status$ and a structurally regular system $H$, enabled at $\status$. It evaluates $H$ for its dependent variables, which sets the values of the leading variables. The corresponding change in status is $\undef\nearrow\ttt$ for both the variables that are evaluated and the equations that are solved.

\subsubsection{Executing a nonstandard instant}
\label{elguehletiughi}
\begin{algorithm}[ht] 
\caption{$\atomicact{ExecRun}$ 
}\label{alg:newmain}
\begin{algorithmic}[1]
   \Require $(\status,\Delta)$; 
   \Return  $(\mathit{Fail},(\wpostset{\status},\wpostset{\Delta}))$ 
      \State $\Phi\gets\{\eqq\in\Delta\mid\status(\guard(e))=\ttt\}$ \label{op:newinitfacts} 
      \State $G \gets [\enbld{\eqq}{\status} \cap \ndef{\status}]$   \label{op:newenabledeqs}
      \State $(\status,b,G_\Sigma,\consistency{G}_\Sigma)\gets\atomicact{IndexReduc}(G,\status)$; increase $\status$
      \label{op:porihujrpoi}
      \If{$\neg b$ \label{op:oesriughoiyu}} \Return\emph{Fail}$(\status)$ 
      \Else \label{op:leriuyliufiu}
      \State $(\status,b,H)\gets\atomicact{SolveConflict}((G_\Sigma\cup\consistency{G}_\Sigma)\setminus\Phi,\status)$; increase $\status$
      \label{op:elkfuiehoeuio}
      \If{$\neg b$} \label{op:lpruightrlu}
      \Return\emph{Fail}$(\status)$ 
      \Else 
      \State $\status\gets\atomicact{Eval}(\status,H)$; increase $\status$
       \label{op:weorpifuhfpeiu} 
      \State $(\wpostset{\status},\wpostset{\Delta})\gets \atomicact{Tick}(\status)$ \label{op:newtick}
      \EndIf
      \EndIf
\end{algorithmic}

\end{algorithm}
$\atomicact{ExecRun}$ requires a finite status $\status$ and a finite context $\Delta$ to serve as initial conditions for the run. It returns, either a documented \emph{Fail}, or a final value for status $\status$, as well as values for $\wpostset{\status}$ and $\wpostset{\Delta}$ to serve as initial conditions for the next instant. A line-by-line description of $\atomicact{ExecRun}$ is given next.
\begin{description}
\item[{\rref{op:newinitfacts}}] Not all the equations belonging to the context are active in the current instant. The active ones are collected in the set $\Phi$ of \emph{facts}.

\item[{\rref{op:newenabledeqs}}]
System $G$ is set to the enabled guarded
equations that are not evaluated yet in status $\sigma$. 

\item[{\rref{op:porihujrpoi}}] We submit $G$ to index reduction ($\atomicact{IndexReduc}$ algorithm, \rref{alg:lkciyughsliu}). 

\item[{\rref{op:oesriughoiyu}}] If $\atomicact{IndexReduc}$ returns $b{=}\fff$, then $\atomicact{ExecRun}$ returns $\mathit{Fail}(\status)$  and stops.

\item[\rref{op:leriuyliufiu}] Otherwise, $\atomicact{ExecRun}$ returns	a pair consisting of $G_\Sigma{=}0$, a structurally nonsingular system determining some leading variables, and the system $\consistency{G}_\Sigma{=}0$ collecting consistency equations. In this case, the algorithm can further progress.

\item[\rref{op:elkfuiehoeuio}] Taking the set $\Phi$ of facts into account, the possible conflict is solved using the $\atomicact{SolveConflict}$ algorithm (\rref{alg:elsgouihuip}). When applying $\atomicact{SolveConflict}$, the dependent variables of $(G_\Sigma\cup\consistency{G}_\Sigma)\setminus\Phi$ include all the variables (both leading and non-leading) of the current system that have status $\status=\undef$, i.e., whose value was not set by  executing the previous nonstandard instant.  

\item[\rref{op:lpruightrlu}] If $\atomicact{SolveConflict}$ fails, $\atomicact{ExecRun}$ returns $\mathit{Fail}(\status)$ and stops. 

\item[{\rref{op:weorpifuhfpeiu}}] Otherwise, $\atomicact{SolveConflict}$ returns a structurally regular system $H$, which is solved using $\atomicact{Eval}$, and the so reached status $\status$ is complete.

\item[{\rref{op:newtick}}] Since $\status$ is {\ssuccessful}, $\atomicact{Tick}$ (\rref{alg:rltoghiuiop}) is performed, so that the values of guards, the status of all {\ScottVars} and the context are known for the next nonstandard instant.
\end{description}

\subsubsection{Important properties}
\label{eliughepiguhp}
\paragraph{Evaluating guards}
Since $\atomicact{ExecRun}$ (\rref{alg:newmain}) performs a \emph{symbolic} interpretation of the model, all the possible modes (valuations for the guards) must be hypothesized and explored at \rref{op:newtick}.\footnote{Assertions may be used to restrict the set of possible modes, by stating some physical knowledge the programmer may have, e.g., `this sequence of modes is not possible'. This is a minor change to our approach, not developed here.} The program is accepted if and only if, for all possible modes, no \emph{Fail} is returned by $\atomicact{ExecRun}$. In case of success, all reachable pairs $(\status,\Delta)$ of statuses and contexts have been explored. 

\paragraph{$\atomicact{ExecRun}$ as a labeled automaton}
We can see $\atomicact{ExecRun}$ as a labeled finite state automaton having nodes labeled with the different encountered status-context pairs, and transitions labeled with the micro-steps. In \rref{fig:clutch-graph}, we depict this graph 
 for the clutch example. The Tick transitions indicate a move to the next nonstandard instant.
\begin{figure}
	\centerline{
		\begin{tikzpicture}
		[
                shorten >=1pt,node distance=3.8cm,auto,
                every node/.style={scale=0.8},
                every edge/.style={thick,draw},
		state/.style={rectangle,scale=.9,shape=rectangle,thick,draw}
               ]
		\node[state,initial] (q0) {\cnfa};
		\node[state] (q1) [above left of=q0] {\cnfb};
		\node[state] (q2) [above right of=q0] {\cnfc};
		\node[state] (q3) [below right of=q0] {\cnfd};
		\node[state] (q4) [below right of=q3] {\cnfe};
		\node[state] (q5) [below left of=q4] {\cnff};
		\node[state] (q6) [above left of=q5] {\cnfg};
		\node[state] (q7) [above left of=q6] {\cnfh};
		\node[state] (q8) [below right of=q5] {\cnfi};
		\node[state] (q9) [below left of=q5] {\cnfj};
		\path[->] (q0) edge[blue] node[blue] {\lbla} (q1)
                                (q0) edge node {\lblb} (q3)
                                (q1) edge[blue] node[blue] {\lblc} (q2)
                                (q2) edge[blue] node[blue] {$\tick$} (q0)
                                (q3) edge[red] node[red] {\lbld} (q4)
                                (q4) edge node {$\tick$} (q5)
                                (q5) edge node {\lble} (q6)
                                (q5) edge[blue] node[blue] {\lblf} (q8)
                                (q6) edge node {\lblg} (q7)
                                (q7) edge node {$\tick$} (q0)
                                (q8) edge[blue] node[blue] {\lblh} (q9)
                                (q9) edge[blue] node[blue] {$\tick$} (q5);
		\end{tikzpicture}
	}
	\caption{ $\atomicact{ExecRun}$ as a labeled automaton for the Simple
          Clutch.
          {Notations: For all statuses (shown in boxes), 
              $\asgnt{v}$ (respectively $\asgnf{v}$) means $v=\ttt$
              (respectively $v=\fff$), and not mentioning $v$ means
              $v=\undef$. $\asgnw{e}$ means that $e_f$ belongs to
              context $\Delta$. The assignment
              $\redundent{e_3}$ refers to \rref{op:newtick} of
              Algorithm~\ref{alg:newmain}. Blue (respectively black) transitions
              belong to a continuous-time (respectively discrete-time)
              dynamics. The red transition is impulsive. A semicolon
              is the sequential composition of computations, and the
              $+$ symbol defines blocks of
              equations.}
              }\label{fig:clutch-graph}
\end{figure}

\paragraph{Rejecting models} Firstly, a system exhibiting logico-numerical fixpoint equations is rejected via synctactic checks, since it does not agree with \rref{defn:mDAEmodified}. Then, a failure is found in the following cases: 
\begin{itemize}
	\item $\atomicact{SigmaMethod}$ (\rref{alg:indexreduction}) itself fails, indicating that no matching is found for the bipartite graph $\cG_G$ of the running dAE system $G{=}0$, a failure criterion for the \sigmamethod. Failure message indicates either over- or under-determination.
	\item The conflict between the predictions from the past and the consistency conditions from the current mode cannot be solved, i.e., $\block_\underapprox\neq\emptyset$ at \rref{op:eriughforui} of $\atomicact{SolveConflict}$, revealing some missing restart equation(s) for the new mode.
\end{itemize}

\paragraph{Detecting continuous modes} Elementary cycles of $\atomicact{ExecRun}$ capture
runs with stationary valuations of the guards and
define the continuous dynamics in each mode. Other
runs capture mode changes and their reset actions.
During runs with constant valuations of the guards, we have $(G_\Sigma\cup\consistency{G}_\Sigma)\setminus\Phi=G_\Sigma$. Hence, \rref{op:roghuphpu} returns $\block_\underapprox=\block_\overapprox=\emptyset$ and, thus, $\block_\squared=G_\Sigma$: we recover the usual structural analysis for DAE systems. 

\paragraph{Solving conflicts at mode changes takes only finitely many nonstandard instants} This follows from the fact that, at a mode change, the maximal shifting degree of context $\Delta$ is finite and the maximal shifting degree of the conflicting subsystem decreases by one at each subsequent time step. See \rref{sec:lerighugholuih} on the Cup-and-Ball example as an illustration of this.

\paragraph{Determinism of the algorithm}
All the algorithms called throughout the execution of \rref{alg:newmain} possess a unique return (at \rref{op:porihujrpoi} and \rref{op:elkfuiehoeuio}). Since \rref{alg:newmain} involves no choice by itself, it is deterministic. Comment~\ref{esrltuiophu} following Lemma~\ref{togiethgouiho} finds its justification here: our design choice in Lemma~\ref{togiethgouiho} was essential in guaranteeing the determinism of our compilation method.

\subsubsection{Generic Blocks within a long mode}
\label{sec:lerigfuiopu}
The results of this section will be used in \rref{sec:standardization}, 
and more specifically in \rref{sec:perughoriugh}.
The successive application of \rref{op:porihujrpoi} and \rref{op:elkfuiehoeuio} of \rref{alg:newmain} yields blocks of a specific form, that is investigated next.
\begin{notation}\rm \label{leirughliu}
	 For $f:\bR^K\rightarrow\bR$ a smooth numerical function and $n\geq 0$ an integer, set
\beq\bea{rcl}
B_\bullet(f,n)&\eqdef&\left[\bea{r} 
B_\bullet^{\head}(f,n) \\ B_\bullet^{\tail}(f,n)
\eea\right] \mbox{ where} 
\\ [4mm]
B_\bullet^{\head}(f,n)&\eqdef&\left[\bea{l} f \\
\postset{f} \\ \,\vdots \\ \ppostset{(n-1)}{f} 
\eea\right] \mbox{ and}
\\ [4mm]
 B_\bullet^{\tail}(f,n)&\eqdef& ~~~\ppostset{n}{f} \enspace .
\eea
\label{pppghuidsliu}
\eeq
$B_\bullet(f,n)$ is called an \emph{$(f,n)$-block}, $B_\bullet^{\head}(f,n)$ is its \emph{head} and $B_\bullet^{\tail}(f,n)$ is its \emph{tail}. By convention, $B_\bullet(f,n)$ is empty for $n<0$ and $B_\bullet^{\head}(f,n)$ is empty for $n=0$.
For $\{f_1,\dots,f_k,\dots,f_K\}$ a tuple of $K$ smooth functions in $K$ dependent variables, define
\beq\bea{l}
F{=}\left[\!\!\!\bea{c}f_1 \\ \vdots \\ f_K\eea\!\!\!\!\right]\!\!, \enspace
 N{=}\left[\!\!\!\bea{c}n_1 \\ \vdots \\ n_K\eea\!\!\!\!\right]\!\!, \enspace
  B_\bullet(F,N) {=} \left[\!\!\!\!\bea{c}
B_\bullet(f_1,n_1) \\ \vdots \\ B_\bullet(f_K,n_K)
\eea\!\!\!\right]
\eea
\label{pweouigtndroui}
\eeq
and the definitions of $B_\bullet^{\head}({F},N)$ and $B_\bullet^{\tail}({F},N)$ follow accordingly.\eproof
\end{notation}
\begin{lemma}
	\label{erlfiuerh} 
Assume that the system is within a long mode. Then, successively applying \emph{\rref{op:porihujrpoi}} and \emph{\rref{op:elkfuiehoeuio}} of \emph{\rref{alg:newmain}} to the system $G{=}0$ returns, either `{Fail}', or a structurally nonsingular system $H{=}B_\bullet^{\tail}({G},N)$ of the form~$(\ref{pweouigtndroui})$.
\end{lemma}
\begin{proof}
Consider \rref{op:porihujrpoi} of \rref{alg:newmain}. Since the system is in a long mode, $\atomicact{SigmaMethod}$ returns, either \emph{Fail}, or $(G_\Sigma,\bar{G}_\Sigma)=(B_\bullet^{\tail}(G,N),B_\bullet^{\head}(G,N))$, where $G$ is the subset of enabled equations at status $\status$.

Move to \rref{op:elkfuiehoeuio}. Since the system is in a long mode, we have $\bar{G}_\Sigma{=}\Phi$, expressing that consistency conditions are guaranteed by previous nonstandard instants. Hence, the Dulmage-Mendelsohn decomposition performed at \rref{op:elkfuiehoeuio} simply returns $\block_\squared{=}G_\Sigma{=}{B}_\bullet^{\tail}({G},N)$. 
\end{proof}

\subsection{{Structural Analysis: general case}}
\label{sec:epwtgouihpiouh}
So far, we restricted ourselves to systems having only long modes. In this section, we consider the general case, where transient modes can also occur, possibly in cascades. The Cup-and-Ball example of Section~\ref{loeruighlrigtuh} can be kept in mind while reading this section.

Our main task is to revisit index reduction. So far, this was performed by invoking the \sigmamethod, which requires the new mode to last for long enough to allow for an unlimited number of time shifts to be applied.
If the current mode is transient, this no longer holds and the continuation of the system dynamics is determined by its successive future modes---see the analysis of the Cup-and-Ball example in \rref{sec:leiruioepui}.

Whereas the \sigmamethod\ applies to time-invariant systems only, the method of \emph{differentiation arrays} (see Section 3.1 of~\cite{CampbellGear1995}) can be adapted to time-varying systems.

\subsubsection{{Index reduction with Difference Arrays}}
The index reduction algorithm proposed in this section applies to time-varying DAE systems. It is, however, less computationally efficient than the \sigmamethod\ (when the latter can be applied). 

\subsubsection*{Difference arrays for time-varying systems}
In the following, $\bT$ denotes a discrete-time index set possessing a minimal element, denoted by $0$. 
\begin{definition} 
	\label{leriuerkl} A \emph{time-varying dAE system} over $\bT$ is a pair $(X,\cF)$, where $X$ is a tuple of numerical variables taking its values in an Euclidian space $\dom$, and
	 $\cF:\bT\ra\cE$ maps every instant $t\in\bT$ to a dAE system $F_t=0$ whose set of dependent variables is contained in $X$.
A \emph{solution} of $(X,\cF)$ is a trajectory $\bfx:\bT\ra\dom$ such that, for every $t\in\bT$, $\bfx(t)$ satisfies $F_t=0$.
\end{definition}
\begin{example}\rm 
	\label{eloriufheofui} Consider $\bT=\bN$, $X=(u,v,w)$, and $\cF$ is defined by
\[
t{=}0: \;
0=\postset{v}-3u
~~; \mbox{ and }
\forall t>0: \left\{\bea{l}
0=\postset{u}-v \\ 0=v-uw \\ 0=\postset{w}+u
\eea\right. \enspace .
\]
The solutions of this system are trajectories $t\mapsto(u(t),v(t),w(t))$ such that
\[
0={v(1)}-3u(0)~~;~~
\forall t>0: \left\{\bea{l}
0={u(t+1)}-v(t) \\ 0=v(t)-u(t)w(t) \\ 0={w(t+1)}+u(t)
\eea\right. \enspace .
\]
Note that solutions are only partially specified at $t=0$ since $w$ is not involved in the initial conditions.\eproof
\end{example}
Given a time-varying dAE system $(X,\cF)$, we consider, for each $k\in\bT$, the following \emph{dAE array}:
	\beq
	\cA_k &\eqdef& \left[\bea{c}	
F_0 \\ F_1 \\ \vdots \\ F_k
\eea	\right] \enspace .
\label{riotgjpip}
	\eeq
	This array allows us to perform the `index reduction' of $F_0$, by adding `latent equations' taken from the future systems $F_1$, $F_2$, etc., generally different from $\postset{F_0},\ppostset{2}{F_0}$, etc. To this end, the set of all variables of array $\cA_k$ decomposes as $Y\cup{X}\cup{W}$, where:
\begin{equation}\bea{l}
\mbox{$Y$ collects the variables in $F_0$ that participate in the consistency;}
\\
	 \mbox{$X$ collects the leading variables of $F_0$;} 
	 \\
	 \mbox{$W$ collects all the remaining variables of the array.}
	 \eea \label{lreiotuhleiuh}
\end{equation}
In its Section 3.1 on `standard indices', Campbell and Gear's landmark paper~\cite{CampbellGear1995} states that, for time-invariant systems ($F_t\equiv{F_0}$), the index is the smallest $k$ such that $\exists W{:}\,\cA_k(X,W,Y){=}0$ uniquely defines $X$ as a function of $Y$, provided that $Y$ is consistent.

Therefore, we consider the following structural translation of the above problem (compare with Problem~\ref{liftuerhpituhepu8} in \rref{sec:structural}):
\begin{problem}
	\label{leriguhelriugl} Find the smallest integer $k$ such that the system $\exists W:\cA_k(X,W,Y){=}0$ is structurally nonsingular in the sense of Definition~$\ref{lerigfueriopg}$.
\end{problem}
For time-invariant systems ($F_t\equiv{F}$ for all $t$), Problem~\ref{leriguhelriugl} is monotonic in that, if this problem has a solution for $k$, then it also has a solution for $k'>k$. However, this no longer holds for time-varying systems, since it can be that system $F_{k+1}=0$ is in conflict with $\cA_k$.

\subsubsection*{Handling finite cascades of transient modes}
Now, we specialize the structural analysis of difference arrays to the handling of finite cascades of transient modes. That is, we assume the following finite chain of modes:
\beq
\underbrace{\mode_0,\mode_1\,\dots,\mode_K}_{\rm finite\;cascade\;of\;transient\;modes~~},\underbrace{\mode_\infty}_{\rm ~~final\;long\;mode}
\label{ioulbguihu}
\eeq
We adopt the following principles in developing our structural analysis:
\begin{itemize}
	\item For the final long mode $\mode_\infty$, it is legitimate to postpone the satisfaction of some of the consistency equations. That is, we reuse the conflict solving approach of the algorithm $\atomicact{SolveConflict}$ (\rref{alg:elsgouihuip}).
	\item For the transient modes $\mode_0,\mode_1\,\dots,\mode_K$, however, this is not legitimate, so that we require the corresponding equations to be satisfied and refuse to erase any of them.
\end{itemize}
We associate to the chain (\ref{ioulbguihu}) the array
	\beq\bea{c}
	\cA \eqdef \left[\bea{c}	
F_0 \\ F_1 \\ \vdots \\ F_K \\ F_{\infty,\Sigma}\cup\consistency{F}_{\infty,\Sigma}
\eea	\right]
\\ [11mm]
\mbox{with a decomposition of its variables as $Y\cup{X}\cup{W}$ following (\ref{lreiotuhleiuh}).}
\eea \label{lguhpiouhp}
	\eeq
	The last row of this array originates from applying index reduction on the long mode $\mode_\infty$. 

Based on the results of \rref{sec:operiuhlo}, in particular Lemma~\ref{lergfuioerhpuip}, the structural analysis for time-varying systems is described in \rref{alg:timevarindexreduction} ($\atomicact{DiffArray}$). 
\begin{algorithm}[t]
	\caption{$\atomicact{DiffArray}$}
	\label{alg:timevarindexreduction}
 \begin{algorithmic}[1]
\Require $(\cA,\status)$ as in (\ref{lguhpiouhp}); 
\Return $(\status,b,F_\Sigma,\consistency{F}_\Sigma)$
\State $k=-1$; \label{op:oeriughopi}
\If{$k<{K}$} \label{op:perioughpiguh}
\State  $k\gets{k{+}1}$ \label{op:orghoiu}
\State  $\cA_k\gets(F_0,\dots,F_k)$
\State  arrange the variables of $\cA_k$ as $(Y_k{\cup}X_k{\cup}W_k)$
\label{op:oerilutriopu}
\State  $(b_\overapprox,b_\underapprox,F_{\Sigma},\consistency{F}_{\Sigma})\gets\atomicact{ExistQuantifEqn}(\cA_k)$ \label{op:oseirugfioeruh}
\If{$\neg b_\overapprox$} return $b\gets\fff$ \label{op:epoguihpouh}
\Else \If{$\neg b_\underapprox$} \textbf{go to } \rref{op:perioughpiguh}
\Else\ {return} $(b\gets\ttt,F_{\Sigma},\consistency{F}_{\Sigma})$; increase $\status$ \label{op:loguhotui}
\EndIf
 \EndIf
 \Else \ consider $\cA$ defined in (\ref{lguhpiouhp})
 \State $\cA\gets\cA$ where $\consistency{F}_{\infty,\Sigma}\gets(\consistency{F}_{\infty,\Sigma}\setminus\cA_K)$
 \label{op:prgruioghioyuku}
 \State  $(\block_\underapprox,\block_\squared,\block_\overapprox)=\DM(\cA)$
 \label{op:hgtrhtighiou}
 \If{$\block_\overapprox\neq\emptyset$}
 \State $\cA\gets\cA$ where  $\consistency{F}_{\infty,\Sigma}\gets\consistency{F}_{\infty,\Sigma}\setminus\block_\overapprox$; \label{op:pseriguhltruip}
 \State \textbf{go to} \rref{op:hgtrhtighiou}
 \Else \label{op:egroghp}
 \If{cond.\,\ref{odigpiufr} and\,\ref{lpiguhlosgioj} of Lemma\,\ref{lergfuioerhpuip} hold} \label{op:owfyugoyu}
  \State partition $\block_\squared=F_\Sigma\cup\consistency{F}_\Sigma$
\State {return} $(b\gets\ttt,F_{\Sigma},\consistency{F}_{\Sigma})$; increase $\status$
\label{op:gsekliuhoauhpu}
  \Else \ {return} $b\gets\fff$
 \EndIf
 \EndIf
 \EndIf
	\end{algorithmic}
\end{algorithm}
This algorithm requires as inputs an array $\cA$ as in~(\ref{lguhpiouhp}) and a status; it returns, either a \emph{Fail} information ($b=\fff$), or a pair $(F_{\Sigma},\consistency{F}_{\Sigma})$, where $F_{\Sigma}{=}0$ is a structurally nonsingular system extracted from $\cA$ and determining the leading variables of $F_0$, and $\consistency{F}_{\Sigma}{=}0$ is a solvable system of consistency equations. The algorithm proceeds according to two phases.

In a first phase (\rref{op:oeriughopi}--\rref{op:loguhotui}), the subchain of transient modes is explored, with increasing array size $k=0,\dots,K$. Using status $\status$, the variables of $\cA_k$ are sorted at \rref{op:oerilutriopu} following (\ref{lreiotuhleiuh}).
At \rref{op:oseirugfioeruh}, array $\cA_k$ is submitted to $\atomicact{ExistQuantifEqn}$ (\rref{alg:erofiueopiu}), which returns the two Boolean flags $b_\overapprox$, indicating whether the overdetermined block is empty, and $b_\underapprox$, indicating whether the underdetermined block is empty.
If $b_\overapprox=\fff$, which indicates an overdetermined system, a failure indication is returned at \rref{op:epoguihpouh}. If $b_\overapprox=\ttt$ and $b_\underapprox=\fff$ is returned, the array is increased (\rref{op:orghoiu}) and reprocessed the same way, until $k=K$ is reached. Otherwise, the first phase returns the pair $(F_\Sigma,\consistency{F}_\Sigma)$. The status is updated as $\irrelevant\nearrow\undef$ for each latent equation added, and the algorithm successfully terminates.

If $k{=}K$ is reached with no success, we move to the second phase, starting with the removal of duplicates between $\cA_K$ and  $\consistency{F}_{\infty,\Sigma}$ (\rref{op:prgruioghioyuku}).
Then, we apply the DM decomposition at \rref{op:hgtrhtighiou} to identify possible conflicts and underdetermined parts. 
We handle conflicts as in the $\atomicact{SolveConflict}$ algorithm (\rref{alg:elsgouihuip}), by deleting conflicting equations from $\consistency{F}_{\infty,\Sigma}$ (\rref{op:pseriguhltruip}) and returning to \rref{op:hgtrhtighiou}; this loop is visited only once thanks to Lemma\,\ref{togiethgouiho}. When reaching \rref{op:egroghp}, we know that condition~\ref{sguiholghouiy} of Lemma\,\ref{lergfuioerhpuip} is satisfied, so that we can conclude by using the successful case of Lemma\,\ref{lergfuioerhpuip} at Lines\,\ref{op:owfyugoyu} and subsequent. The status is updated as $\irrelevant\nearrow\undef$ for each latent equation added (\rref{op:gsekliuhoauhpu}).

\subsubsection{Executing a nonstandard instant}
We first revisit the algorithm $\atomicact{IndexReduc}$.
 If the current mode is long, then $\atomicact{IndexReduc}$ is implemented as \rref{alg:lkciyughsliu}. Otherwise, we will implement it by hypothesizing a future mode trajectory and calling the algorithm $\atomicact{DiffArray}$ (\rref{alg:timevarindexreduction}). $\atomicact{IndexReduc}$, revisited, is specified as \rref{alg:epgioehpguioh}.
\begin{algorithm}[ht]
	\caption{$\atomicact{IndexReduc}$, revisited}
	\label{alg:epgioehpguioh} 
\begin{algorithmic}[1]
	\Require $(G,\status,\mathit{type})$; 
	\Return $(\status,b,G_\Sigma,\consistency{G}_\Sigma)$
	\If{\emph{type}$=$long} 
	  \State $(b,G_\Sigma,\consistency{G}_\Sigma)\gets\atomicact{SigmaMethod}(G)$ with dependent variables set by $\status$; increase $\status$
      \label{op:rtokuygfguihuo}
      \Else
      \State hypothesize $\cF=(G_0,G_1,\dots)$ with $G_0{=}G$; let $\cA$ be the corresponding array \label{op:eofiuerhpoiufeiu}
      \State $(b,G_\Sigma,\consistency{G}_\Sigma)\gets\atomicact{DiffArray}(\cA,\status)$; increase $\status$
      \EndIf
      \If{$b$} \label{op:pgtuioh}
        \Return $(G_\Sigma,\consistency{G}_\Sigma)$
      \EndIf
\end{algorithmic}
\end{algorithm}

In this algorithm, \rref{op:eofiuerhpoiufeiu} requires hypothesizing a continuation for the current mode. This, of course, is problematic if all modes can be visited in this continuation, particularly transient modes. It is therefore essential to be able, at compile time, to restrict the number of possible continuations for the current mode. This question is investigated in \rref{sec:perlgtuiohoiu}. 

\begin{algorithm}[ht] 
\caption{$\atomicact{ExecRun}$ , revisited
}\label{alg:ergfipuhfopiu}
\begin{algorithmic}[1]
   \Require $(\status,\Delta)$; 
   \Return  $(\mathit{Fail},\status,(\wpostset{\status},\wpostset{\Delta}))$ 
   \State $\mathit{type}\gets\atomicact{EvalModeType}(\status)$ \label{op:oelrigfurhoihoi}
      \State $\Phi\gets\{\eqq\in\Delta\mid\status(\guard(e))=\ttt\}$ \label{op:newinitfactsrevis} 
      \State $G \gets [\enbld{\eqq}{\status} \cap \ndef{\status}]$   \label{op:newenabledeqsrevis}
      \State $(\status,b,G_\Sigma,\consistency{G}_\Sigma)\gets\atomicact{IndexReduc}(G,\status,\mathit{type})$; increase $\status$
      \label{op:porihujrpoirevis}
      \If{$\neg b$} \Return\emph{Fail}$(\status)$ \label{op:eogfuiehoriufhoui}
      \Else \label{op:leifuerhfuiohouirevis}
      \State $(\status,b,H)\gets\atomicact{SolveConflict}((G_\Sigma\cup\consistency{G}_\Sigma)\setminus\Phi,\status)$; increase $\status$
      \label{op:elkfuiehoeuiorevis}
      \If{$\neg b$} \label{op:lpruightrlurevis}
      \Return\emph{Fail}$(\status)$ 
      \Else
      \State $\status\gets\atomicact{Eval}(\status,H)$; increase $\status$
       \label{op:weorpifuhfpeiurevis} 
      \State $(\wpostset{\status},\wpostset{\Delta})\gets \atomicact{Tick}(\status)$ \label{op:newtickrevis}
      \EndIf
      \EndIf
\end{algorithmic}
\end{algorithm}
Having revisited index reduction, we can now revisit $\atomicact{ExecRun}$, see \rref{alg:ergfipuhfopiu}. Only two lines are modified with reference to the original algorithm $\atomicact{ExecRun}$:
\begin{description}
\item[\rref{op:oelrigfurhoihoi}] We add the evaluation of the \emph{type} (long or transient) of the current mode.
\item[\rref{op:porihujrpoirevis}] We call instead the $\atomicact{IndexReduc}$ corresponding to \rref{alg:epgioehpguioh}, which selects the appropriate method for index reduction, depending on the \emph{type} of the current mode.
\end{description}

\subsubsection{{Continuations of transient modes}}
\label{sec:perlgtuiohoiu}
We need to investigate this issue for transient modes only. Therefore, we begin with assumptions regarding transient modes. 
Recall that we follow Notations~\ref{lerigftuheroiuof}, hence, by \emph{numerical variable} of an \mDAE\ system $\system$, we mean a variable of the form $\pprime{m}{x}$ ($m$-th derivative of $x$) for some algebraic or state variable $x$ of the considered system. The \emph{degree} of $\pprime{m}{x}$ is $m$. The same holds regarding \mdAE\ system $\system$, with $\ppostset{m}{x}$ ($m$-shifted version of $x$) replacing $\pprime{m}{x}$. We repeatedly use these notations in the sequel.

\paragraph{Assumptions regarding transient modes} 
According to \rref{defn:egfuiher}, a mode $\mode$ is transient if at least one of its $\ttt$-valued guards is transient. Guards are Boolean expressions of predicates over numerical variables. We will  restrict ourselves to the class of nonstandard \mdAE\ models $\system$ with set $X$ of numerical variables, satisfying the following assumption regarding transient predicates over numerical variables:
\begin{assumption}
	\label{lerifuherlifu} Let $\guard_1,\guard_2\in\Guards(\system)$ be any two distinct transient predicates over numerical variables of \mdAE\ model $\system$. Then, for $\vval{X}^1$ and $\vval{X}^2$ any two valuations of the tuple $X$ such that $\vval{X}^1-\vval{X}^2$ is infinitesimal, $\guard_1(\vval{X}^1)=\guard_2(\vval{X}^2)=\ttt$ cannot occur.
\end{assumption}
\myparagraph{Discussion}
Assumption~\ref{lerifuherlifu} formalizes that two different transient predicates never take the value $\ttt$ simultaneously, ``by chance''. For example, for $x,y\in{X}$, the two predicates $\guard_1=\prog{when}\;x{\geq}0$ and $\guard_2=\prog{when}\;y{\leq}5$ specify the events when $x$ becomes ${\geq}{0}$ and when $y$ becomes ${\leq}5$. The assumption says that these two events will never occur simultaneously. For this case, Assumption~\ref{lerifuherlifu} is very reasonable. Of course, if $\guard_1=\prog{when}\;x^3{\geq}8$ and $\guard_2=\prog{when}\;x{\geq}2$, the two guards are considered different (their syntax differ) but generate simultaneous events. It is not the duty of the compiler to discover if two guards happen to be mathematically identical whereas they differ in their syntax. It is rather the responsibility of the programmer to use identical syntax for different occurrences of a same predicate.\eproof

\smallskip

Using Assumption~\ref{lerifuherlifu} we can in particular bound the number of transient modes that can occur at a given instant $t\in\bT$. For $\guard_o\in\Guards(\system)$ a transient predicate, let 
\beq
\Guards(\system,\guard_o) &\subseteq& \Guards(\system) \label{iogelguiohj}
\eeq
be the subset of all guards $\guard$ of $\system$ that are Boolean expressions of predicates such that $\guard=\ttt$ requires $\guard_o=\ttt$. An example would be $\guard=\guard_o\wedge\guard'$. The following obvious lemma holds:
\begin{lemma}
	\label{luthelriuth} Under Assumption~$\ref{lerifuherlifu}$, the number of transient modes that can possibly occur at a given instant $t\in\bT$ is bounded by
\beqq
M(\system) &=& \max_{\guard_o}\bigl|2^{\Guards(\system,\guard_o)}\bigr|
\eeqq
where $\guard_o$ ranges over the set of all transient predicate guards of $\system$, and $|A|$ denotes the cardinal of set $A$.
\end{lemma}
In practice, this bound is generally much smaller than the number of possible modes for $\system$, which is bounded by $|2^{\Guards(\system)}|$.

\paragraph{Assumptions regarding continuations of transient modes} Let $\mode_0=\mode(\nstime)$ be the current mode at time $\nstime\in\bT$, and let $\mode_0,\mode_1,\mode_2,\dots$ be any continuation for $\mode_0$ at instants ${\nstime},\postset{\nstime},\ppostset{2}{\nstime},\dots$. Such continuation can have one of the following two forms:
\beq
\hspace*{-5mm} \mbox{finite cascade}&:&\mode_0,\mode_1,\dots,\mode_{m-1},\mode_{\infty} \label{kufeygoigheoi} \\
\hspace*{-5mm} \mbox{infinite cascade}&:&\mode_0,\mode_1,\dots,\mode_{m-1},\dots \label{leriguhliuh} 
\eeq
where the final mode $\mode_{\infty}$ in (\ref{kufeygoigheoi}) is long (and thus repeated), and all other listed modes, including the ``$\dots$'', are transient. We will state assumptions that statically rule out (\ref{leriguhliuh}) and limit the number of possible continuations in (\ref{kufeygoigheoi}).
For $\mode$ a transient mode, let 
$$X(\mode)\subseteq{X}$$
be the subset of numerical variables $x$ such that $\postset{x}-x$ is possibly non-infinitesimal at this mode. In other words, $X(\mode)$ is the set of variables that are reset by this transient mode. We then define, for $\Modes$ a set of transient modes: 
\[
X(\Modes)=\bigcup_{\mode\in\Modes}X(\mode) ~ .
\]
This set can be identified using structural algorithms.
\begin{example}\rm
	\label{lifuihil} The statement
\[
\dot{x}=f(x,u) \;\prog{reset}\;0\;\prog{when}\;x\geq{1}
\]
expands, according to \rref{defn:mDAEattempt}, as
\[
\left\{\bea{rcl}
\prog{when}\;x\geq{1}&\doo&x^+=0 \\
&\prog{else}&\dot{x}=f(x,u) 
\eea\right.
\]
whose nonstandard semantics is
\beq
\left\{\bea{rclc}
\prog{when}\;x\geq{1}&\doo&\postset{x}=0 &(e_1) \\
&\prog{else}&\frac{\postset{x}-x}{\vsmall}=f(x,u) &(e_2)
\eea\right.
\label{sys:eprigfuhperiuf}
\eeq
Since $\postset{x}-x$ is not guaranteed infinitesimal by $(e_1)$, $\postset{x}-x$ is non-infinitesimal at the instant of zero-crossing (i.e., the first time when $x\geq{1}$ occurs). Thus, $X(\mode_{\rm zc})=\{x\}$, where $\mode_{\rm zc}$ is the mode when the zero-crossing occurs. This can be found at compile time, based on the syntax of \rref{sys:eprigfuhperiuf}.\eproof
\end{example}
\begin{example}\rm 
	\label{erkuyfkuofio} Consider the Cup-and-Ball example with elastic impact, see \rref{sec:leiruioepui}, in its original specification prior to adding the restart condition (\ref{lruifhepriuhiu}). Let $\mode_{\rm imp}$ be the mode when the elastic impact occurs. Then, $X(\mode_{\rm imp})=\emptyset$. Thus, we know that the only continuation is: $\mode_{\rm imp},\mode_{\rm free}$,\, in which $\mode_{\rm free}$ is the mode in which the rope is not straight and, thus, free motion occurs. The same holds after adding  restart condition (\ref{lruifhepriuhiu}).\eproof
\end{example}
In the sequel, we denote by $M$ the set of all possible modes of system $\system$. 
For $Y\subset{X}$ a subset of variables, let 
\[
\Modes(Y)\subseteq M
\]
be the set of all transient modes that can possibly be visited as a consequence of making noninfinitesimal changes to the variables of $Y$---this set can be identified at compile time using structural algorithms.

In bounding the possible continuations of a transient mode $\mode_0$, we consider the following chain of sets of possibly reachable transient modes:
\beq\bea{rl}
\continuation{\mode_0} =& M_0,M_1,M_2,\dots 
\\
\nonumber
\mbox{where}:& M_0=\{\mode_0\} \\
\nonumber
\mbox{and }\forall k>0:& M_k=M(X(M_{k-1})) 
\eea
\label{giohoioigth}
\eeq
Once again thanks to structural algorithms, the following assumption can be checked at compile time on the considered system $\system$:
\begin{assumption}
	\label{elrituhgui} For every transient mode $\mode$, the chain $X_k,k\geq{0}$ in  $\continuation{\mode}$ converges to $\emptyset$ in finitely many steps.
\end{assumption}
\begin{lemma}
	\label{refuilerhuil} 
Assumption~$\ref{ughlrhlruih}$ implies Assumption~$\ref{elrituhgui}$,
\end{lemma}
where Assumption~$\ref{ughlrhlruih}$ states:
\begin{assumption}
	\label{ughlrhlruih} For every subset $Y$ of reset variables, 
	\beqq
	\degree{}{X(\Modes(Y))} &>& \degree{}{Y} 
	\eeqq
	where $\degree{}{Z}$ is the smallest (shifting or differentiation) degree among all variables of $Z$.
\end{assumption}
Assumption~\ref{ughlrhlruih} holds if, when a transient mode is reached due to a predicate on positions $x,y,\dots$, then only velocities $\dot{x},\dot{y},\dots$ or accelerations $\ddot{x},\ddot{y},\dots$ can be reset. This kind of property is often encountered in classes of physical systems, e.g., contact mechanics.
\section{Systems with shifts and differentiations}
\label{sec:mdDAE}
So far, our formalism of multi-mode DAE systems handles only guarded DAEs. For more elaborate examples, we  need to support also left- and right-limits, to be able to explicitly define  restart or reset values for the dynamics of a given mode. With reference to model class (\ref{ow8ugthoerui}), our extension takes the form
\begin{equation}
\!\!\!\bea{rl}
\mbox{if} & \guard_j(\mbox{the } x_i \mbox{ their derivatives and left-limits}) \\
\mbox{do} & f_j(\mbox{the } x_i \mbox{ their derivatives and right-limits}) = 0
\eea
\label{edfsrogieosgfi}
\end{equation}
Right-limits are typically used in defining restart values at mode changes.

Since we will map the dynamics in the nonstandard domain, we will represent the left- and right-limits $x^-$ and $x^+$ through the backward and forward shifts $\preset{x}$ and $\postset{x}$, see \rref{defn:tshifts}.  
We thus consider a syntax offering  derivative $\prime$ and  forward shift $\bullet$, both related by the following identities in the nonstandard semantics: 
\beq
\bea{rcl}
\dot{x} &=& \frac{1}{\vsmall}(\postset{x}-x)
\\ \vspace*{-3mm} \\
\postset{x} &=& x + \vsmall\ttimes\dot{x}
\eea
\label{weorifjqoi}
\eeq
The resulting formalism will be called mdDAE, or \deltaAE, for short. We formalize this next.

\subsection{Preliminaries}
In the nonstandard semantics, shift and derivative are related via identities (\ref{weorifjqoi}). In these preliminaries we investigate additional relations and properties that follow from identities (\ref{weorifjqoi}).
\begin{lemma} \
	\label{erlirgfuih} 
\begin{enumerate}
	\item \label{rtguii} The two operators $\bullet$ and $\prime$ commute: for every variable $x$, we have \mbox{$x^{\bullet\prime}=x^{\prime\bullet}$}.
	\item \label{lirufeli} For $f:\bR^k\mapsto\bR$ a $\mathcal{C}^1$-function and $X$ a vector variable of dimension $k$, define	
	\beq\bea{rcl}
	\postset{f}(X)\eqdef	[f(X)]^\bullet &\eqdef& {f}(\postset{X}) \\ [1mm]
	(f(X))^\prime &\eqdef& \nabla{f}(X).X^\prime
	\eea
	\label{roitughoiu}
	\eeq
	 where $\nabla{f}$ denotes the Jacobian of $f$. 
	Then, differentiation and shift commute, meaning that: $[f(X)]^{\bullet\prime}=[f(X)]^{\prime\bullet}$. \eproof 
\end{enumerate}
\end{lemma}
\begin{proof}
	Statement~\ref{rtguii} follows immediately from identities (\ref{weorifjqoi}). Then, the equalities
	\beqq\bea{rcl}
	[{f}(X)]^{\bullet\prime} = \dot{[\postset{{f}}(X)]}&\!\!\!=\!\!\!& {\nabla{f}}(\postset{X}).X^{\bullet\prime}
	\\ [1mm]
	\mbox{(by Statement~\ref{rtguii})}&\!\!\!=\!\!\!& \postset{[{\nabla{f}}(X).\dot{X}]} 
	\\ [1mm]
	&\!\!\!=\!\!\!& [{f}(X)]^{\prime\bullet}
	\eea
	\eeqq
	prove Statement~\ref{lirufeli}.
\end{proof}
We consider the alphabet $\{\bullet,\prime\}$ related through identities (\ref{weorifjqoi}), representing the shift and the derivative operators, respectively. 
Successive differentiations or shifts of a variable $x$ are denoted by $x^{w}$ for some word $w\in\{\bullet,\prime\}^*$ where $^*$ denotes the usual Kleene closure. We denote by $\bullet{n}$ and $\prime{n}$ the concatenation of $n$ successive $\bullet$ and $\prime$, respectively.
In $x^{w}$, word $w$ is read from left to right, e.g., in $x^{\prime\bullet}$, $x$ is first differentiated, and then $x'$ is shifted. 
We denote by $X^{w}$ the set of all $x^{w}$ for $x$ ranging over the set $X$ of variables. For $\Sigma\subseteq\{\bullet,\prime\}^*$, we set
\beq
X^{(\Sigma)} \eqdef \bigcup_{w\in\Sigma}\,x^{w} &\mbox{ and }&
\{X\} \eqdef  X^{(\{\bullet,\prime\}^*)}
\label{poqiurwehpguo}
\eeq
\begin{definition}
\label{defn:deltaAE}
A \emph{{\deltaAE} system} is a finite set of \emph{guarded equations} of the form:
  \beq
e&::=&~~~\prog{if}\;\guard\;\prog{do}\;f{=}0 \label{operguiophzz}  \\
\postset{\guard}&::=&~~~\mathit{bexp} \label{leriughouzz}
\eeq
where
$f$ is a smooth scalar function over $\{X\}$,  the shifted guard $\postset{\guard}$ is defined by $\mathit{bexp}$, which is a \emph{typed} boolean expression of predicates over $\{X\}$, and the guard $\guard$ is typed $\prog{\llong}/\prog{\transient}$.
\end{definition}
The reader is invited to compare this definition with \rref{defn:mDAEmodified}. The notations of \rref{defn:mDAEmodified} and the explanations that follow it extend to {\deltaAE}.
In a \deltaAE\ model we may find $x,\dot{x},\postset{x},\postset{\ddot{x}},x^{\prime\bullet\prime}$ and so on. 
Using (\ref{weorifjqoi}) and Lemma~\ref{erlirgfuih} repeatedly, $x^{\prime\bullet\prime}$ can be expressed as a function of $x,\postset{x},\ppostset{2}{x}$, $\ppostset{3}{x}$ or as a function of  $x,\dot{x},\ddot{x},x^{\prime{3}}$. 
Consequently, for \mbox{$w\in\{\prime,\bullet\}^*$} a word over the alphabet $\{\prime,\bullet\}$,
it is natural to define the \emph{degree} of $x^{w}$ with respect to $x$ as being $|w|$, the length of word $w$.
For $X$, $\Sigma\subseteq\{\bullet,\prime\}^*$, and $
X^{(\Sigma)}
$ as in (\ref{poqiurwehpguo}), let $E$ be a system of algebraic equations over the set $X^{(\Sigma)}$. 
\beq
\mbox{
\begin{minipage}{14cm}
	We denote by $^{(\bullet)}\!E$  the system obtained by expressing any monomial $x^{w}$ as a function of $x,\postset{x},\ppostset{2}{x},\ppostset{3}{x},\dots$, by repeatedly using the first identity of (\ref{weorifjqoi}) and Lemma~\ref{erlirgfuih}.
\end{minipage}
} 
\label{pw49guhl}
\eeq
System $^{(\bullet)}\!E$ involves only shifts and no differentiations. 
Vice-versa, recalling (\ref{roitughoiu}):
\beq
\mbox{
\begin{minipage}{14cm}
	For $E$ of the shifted form $E = \ppostset{K}{F}$, let $^{(\prime)}\!E$  
be obtained by repeatedly applying, to $E$, until all shifts get eliminated from it, the map
$
G(\postset{X}) \mapsto G(X)+\nabla{G}(X).\dot{X}
$,
where $\nabla G$ denotes the Jacobian of $G$. 
\end{minipage}
}
\label{oweiurghoui}
\label{o4f8ygffgwelio}
\eeq
Note the difference in defining $^{(\bullet)}\!E$ and $^{(\prime)}\!E$.

\subsection{Shifting versus differentiating}
The whole Sections~\ref{sec:loguhiuliugh} and~\ref{sec:csemantics}, where the structural analysis of \mdAE\ systems was developed, extends verbatim to {\deltaAE} systems, up to the following changes:
\begin{itemize}
	\item Replace everywhere the superscript $\bullet{m}$, for $m{\in}\bN$, by a word $w{\in}\{\prime,\bullet\}^*$ of length $|w|{=}m$; define, for $\eqq^{w}::\prog{if}\;\guard\;\prog{do}\;f^{w}{=}0$, the operators
\beq
	\bea{rr}
	\atomicact{ForwardShift}:\hspace*{-3mm}& 	\eqq^{w\bullet} \eqdef \prog{if}\;\guard\;\prog{do}\;f^{w\bullet}=0 
\\
[1mm]
\atomicact{Differentiate}:\hspace*{-3mm}& 
	\eqq^{w\prime} \eqdef \prog{if}\;\guard\;\prog{do}\;f^{w\prime}=0
	\eea
	\label{elpiuegfip}
	\eeq
\item When invoking the \sigmamethod\ (via $\atomicact{IndexReduc}$) in \rref{alg:newmain} (respectively \rref{alg:ergfipuhfopiu}), we can either use differentiation or forward shift at \rref{op:porihujrpoi} (respectively \rref{op:porihujrpoirevis}). The question is: does this choice influence the outcome of the Dulmage-Mendelsohn decomposition at \rref{op:roghuphpu} of $\atomicact{SolveConflict}$ (\rref{alg:elsgouihuip})?
Other steps of \rref{alg:newmain} or \rref{alg:ergfipuhfopiu} are for sure not affected.
\end{itemize}
The alternative ``shift or differentiate'' raises the following question: \emph{what is the effect, on the structural analysis, of chosing a particular alternative each time we reach \rref{op:porihujrpoi} of \rref{alg:newmain} or \rref{op:porihujrpoirevis} of \rref{alg:ergfipuhfopiu}}\,? We address this issue next.

\subsection{Invariance results}
\label{q3p95oui}
We use the notion of weighted bipartite graph associated to a \deltaAE\ system $F{=}0$. We refer the reader to Section~\ref{reoiughfuil} for this notion associated with DAE systems, and we extend it to \deltaAE\ systems by deducing weights from the length $|w|$ of words over the alphabet $\{\bullet,\prime\}^*$.
Thanks to the two-way correspondence (\ref{weorifjqoi}), the following holds:
\begin{lemma}
	\label {ow47ghleuir} 
	For $F{=}0$ any system of algebraic equations over $X^{(\Sigma)}$: $F$, ${^{(\bullet)}\!F}$, and ${^{(\prime)}\!F}$ (when the latter can be considered) possess identical weighted bipartite graphs. As a consequence, applying \emph{\rref{alg:indexreduction}} (the \sigmamethod) to $F$, ${^{(\bullet)}\!F}$, or ${^{(\prime)}\!F}$ yields the same result in terms of offsets and outcome of the DM decomposition.
\end{lemma}
Recall the following definition: An \emph{execution tree} for \rref{alg:newmain} is a labeled tree $T$ representing all its possible runs. Nodes of $T$ are labeled by one of the \ref{op:newtick} action labels of \rref{alg:newmain} and each maximal path 
of $T$ describes a possible run of \rref{alg:newmain} through the sequence of action labels sitting at the nodes of this path. 
In this section, we investigate the effect of
\beq
\mbox{
\begin{minipage}{14cm}  selecting one of the two alternatives
in (\ref{elpiuegfip}), each time index reduction is performed.
\end{minipage}
}
\label{erogirhjugi}
\eeq
Performing (\ref{erogirhjugi}) results in corresponding variations of the context $\Delta$ at the next instant. 
We thus consider the set $\cT$ of all executions trees obtained by performing (\ref{erogirhjugi}) in \rref{alg:newmain}. 
From Lemma~\ref{ow47ghleuir}  we immediately deduce:
\begin{lemma}
	\label {weor8gfyubh} 
$\cT$ is a singleton consisting of an execution tree that we denote by $T_S$. In addition, the contribution, to the bipartite graph, by the context $\Delta$, is independent from these choices. The same holds for \emph{\rref{alg:ergfipuhfopiu}}.
\end{lemma}
Lemma \ref{weor8gfyubh} means that Algorithms~\ref{alg:newmain} and~\ref{alg:ergfipuhfopiu} will visit the same sequence of action labels, irrespectively of all the choices made (either shifting or differentiating). The systems of equations produced by \rref{op:porihujrpoi} of \rref{alg:newmain} and \rref{op:porihujrpoirevis} of \rref{alg:ergfipuhfopiu} may differ, but possess identical weighted bipartite graphs. From now on we focus on \rref{alg:ergfipuhfopiu}, of which \rref{alg:newmain} is a specialization.

\subsection{Executing a nonstandard instant, revisited}
\label{erliuherliu}
In \rref{alg:ergfipuhfopiu}, forward shifting is used when performing index reduction. 
In Section~\ref{sec:mdDAE}, however, we have seen that, in \rref{alg:ergfipuhfopiu} forward shifting and differentiation can be substituted for one another.
This leads us to consider the following strategy, whose motivation is to make standardization easier (see the standardisation of the Cup-and-Ball example in \rref{sec:oguhkiu}): 
\begin{strategy}
	\label{oerigfu79} In \emph{\rref{alg:ergfipuhfopiu}}, we priorize shifts or differentiations depending on the following cases:
\begin{enumerate}
	\item \label{rouiygfoui} The considered nonstandard instant belongs to a cascade of events: we priorize shifts;
	\item \label{w478ghsljkl} The considered nonstandard instant is in a continuous mode: we priorize differentiation;
	\item \label{etyoiryp} We are entering a continuous mode: we priorize differentiation unless it brings leading derivatives.
\end{enumerate}
\end{strategy} 
The justification of point~\ref{w478ghsljkl} is clear: we should reproduce the usual index reduction. In a cascade of mode changes occurring at successive nonstandard instants, we are in a discrete time dynamics and thus point~\ref{rouiygfoui} follows. Finally, point~\ref{etyoiryp} is justified by the fact that, if $\pprime{n}{x}$ is the leading derivative for $x$ in the new mode, we need to know restart values for the lower derivatives $x,\dot{x},\dots,\pprime{(n-1)}{x}$, given by the values for the variables $\postset{x},\postset{{\dot{x}}},\dots,x^{\prime(n-1)\bullet}$.

\subsection{Generic Blocks}
\label{sec:lteiughlituh}
Despite all the different choices (either $\prime$ or $\bullet$) give rise to the same execution tree $T_S$ by Lemma~\ref{weor8gfyubh}, they may result in different standardizations. We thus want to refine our understanding of the effect of choices (\ref{erogirhjugi}) when applying \rref{alg:ergfipuhfopiu}. What we are really interested in is what happens at \rref{op:weorpifuhfpeiurevis} where enabled blocks are solved: it may matter whether shift or differentiation was used at each particular traverse of \rref{op:porihujrpoirevis} of \rref{alg:ergfipuhfopiu}. 
To this end we investigate the effect, on a given $e\in\bE^S$, of the choices (\ref{erogirhjugi}) performed in \rref{alg:ergfipuhfopiu}, where $\bE^S$ is the set of all the guarded equations of a given \deltaAE\ system $S$. 
The successive application of \rref{op:porihujrpoirevis} and \rref{op:elkfuiehoeuiorevis} of \rref{alg:ergfipuhfopiu} yields blocks of a specific form that we investigate next.

\paragraph{Notations:} Let $w\in\{\prime,\bullet\}^*$ be a word of length $n$,  let $w[0],\dots,w[n-1]$ be the successive prefixes of $w$, starting from the empty prefix. For $f:\bR^K\mapsto\bR$ a smooth numerical function and $w$ as above, set
\beqq\bea{rcll}
B(f,w)&\eqdef& \left[\bea{l}
B^{\head}(f,w) \\ B^{\tail}(f,w)
\eea\right] &\mbox{where}
\\ \vspace*{-2mm} \\
B^{\head}(f,w) &\eqdef& \left[\bea{l}
f^{(w[0])} \\ f^{(w[1])} \\ \,\vdots \\ f^{(w[n-1])} 
\eea\right] &\mbox{and}
\\ \vspace*{-2mm} \\
B^{\tail}(f,w) &\eqdef& \quad\, f^{w}
\eea
\eeqq
$B(f,w)$ is called an \emph{$(f,w)$-block}, $B^{\head}(f,w)$ is its \emph{head} and $B^{\tail}(f,w)$ is its \emph{tail.} 
For $\{f_1,\dots,f_k,\dots,f_K\}$ a tuple of $K$ smooth functions in $K$ dependent variables, \mbox{$w_1,\dots,w_k,\dots,w_K\in\{\prime,\bullet\}^*$} a set of $K$ words of respective lengths $n_{1},\dots,n_{k},\dots,n_{K}$, 
define
\beq\bea{l}
F{=}\left[\!\!\!\bea{c}f_1 \\ \vdots \\ f_K\eea\!\!\!\!\right]\!\!,
 W{=}\left[\!\!\!\bea{c}w_1 \\ \vdots \\ w_K\eea\!\!\!\!\right]\!\!,
  B(F,W) {=} \left[\!\!\!\!\bea{c}
B(f_1,w_1) \\ \vdots \\ B(f_K,w_K)
\eea\!\!\!\right]
\eea
\label{pweouigtndroui}
\eeq
and the definitions of $B^{\head}({F},W)$ and $B^{\tail}({F},W)$ follow accordingly. In addition, we will write 
\beqq
B_\bullet^{\head}(f,n) \eqdef B^{\head}(f,\bullet{n})
&\mbox{and}&
B_\prime^{\head}(f,n) \eqdef B^{\head}(f,\prime{n})
\eeqq
\eproof
\paragraph{Exchanging $\prime \leftrightarrow \bullet$:}
In this paragraph we detail such exchanges for the above identified blocks.
Let $f:\bR^m\mapsto\bR$ be a function and let \mbox{$f(X){=}0$} be the associated equation with $m$ dependent variables collected in the vector $X$. 
Consider $w{\in}\{\prime,\bullet\}^*$, and let
\beq
f^{w} &\mapsto&
\bigl(f^{(\prime{n})},f^{(\bullet{n})}\bigr)\,, \mbox{ where } n=|w|\,,
 \label{pw4ttuighui}
\eeq
be the map defined by the following algorithm:
\begin{enumerate}
	\item Initialize $h_0=g_0:=f^{w}$ and $v_0:=w$;
	\item For $j=1,\dots,n$:
\begin{enumerate}
	\item 
	Decompose $v_{j-1}:=v_j.\nu$ where $\nu\in\{\prime,\bullet\}$;
	\item 
\begin{description}
	\item[case $\nu=\prime\;$:] define $h_j:=h_{j-1}$ and  
	\beq
	\bea{rcl} g_j&:=&g_{j-1}\left[\dot{x}\leftarrow\frac{\postset{x}-x}{\vsmall}\right],
	\eea
	\label{werpghuio}
	\eeq
	where $\dot{x}$ ranges over the set of derivatives involved in $g_{j-1}$;
	\item[case $\nu=\bullet$\,:]  $h_{j-1}$ has the form $\postset{\ell}(Y)$ for some function $h$ and $Y$ its set of dependent variables; define $g_j:=g_{j-1}$ and 
	\beq
	h_j(Y)&:=&\nabla{\ell}(Y).\dot{Y},
	\label{wrebhjrbi}
	\eeq
	where $\nabla{\ell}$ is the Jacobian of $\ell$.
\end{description}
\end{enumerate}
\item Return $\bigl(f^{(\prime{n})},f^{(\bullet{n})}\bigr) = (h_n,g_n)$.
\end{enumerate}
\begin{lemma}
	\label{slguihilu} 
Applying map $(\ref{pw4ttuighui})$ to all the entries of $B^\head(f,w),B^\head(f,w)$, and $B^\tail(f,w)$, returns
\beq\bea{rrrrr}
B_\prime^\head(f,{n})
&\mbox{and}&
B_\prime^\tail(f,{n})\,,
\\
B_\bullet^\head(f,{n})
&\mbox{and}&
B_\bullet^\tail(f,{n})\,,
\eea
\label{iytfghsodij}
\eeq
respectively, where $n=|w|$. 
Doing the same for $B(F,W)$ defined in $(\ref{pweouigtndroui})$ returns
\beq\bea{rrrrr}
B_\prime(F,N) &\mbox{and}& 
B_\bullet(F,N) 
\eea
\label{fq38ogakuhyfb}
\eeq
where $N$ is the vector of integers collecting the lengths of the entries of $W$.
\end{lemma}
\begin{ccomment} 
	\label {wogtuhstrkl} \emph{
By regarding all $\dot{x},\ddot{x}$, etc. as dummy variables,\footnote{This means that we regard $\dot{x},\ddot{x}$, etc. as different variable names and forget about the fact that they are related via differentiation. The same term holds whith shifting.}	every structurally nonsingular block $B(F,W)$ $=0$ defines a dAE system over nonstandard time basis $\bT$. This also holds for $B_\prime(F,N)=0$ and $B_\bullet(F,N)=0$.}
\emph{
On the other hand, by not involving shifts, $B_\prime(F,N)=0$ defines a (continuous time, standard) DAE system. This is not true for $B(F,W)=0$ in general.\eproof
	}
\end{ccomment}
This finishes our study of structural analysis for \mDAE\ or \mdAE\ systems. Referring to \rref{fig:approach} in \rref{sec:liguholui}, we are ready to move to the two downgoing arrows, consisting in the standardization of the nonstandard model resulting from the structural analysis, together with the impulse analysis. The rest of the paper is devoted to this.

\section{Background on nonstandard analysis}
\label{sec:NSA}
\label {wo458t7gykuywe}
\label {pqe9rusipg}
The development of the clutch example in \rref{sec:simpleclutch} consisted of three main steps: (1) the mapping of the original clutch model to its nonstandard semantics, (2) the symbolic execution of this nonstandard semantics (\rref{sec:orguhuoi}), and (3) the back-standardization of this symbolic execution (\rref{sec:standardizeclutch}). All of this resulted in the actual code for the simulation results reported in \rref{fig:clutch0}. Formally justifying back-standardization requires a non-trivial use of nonstandard analysis, way beyond the quick and superficial introduction we gave in \rref{sec:nsa}.
In this and the next sections, we provide the mathematical background supporting back-standardization procedures in full generality. 

\emph{Nonstandard analysis} was proposed by Abraham Robinson in the 1960s to allow for the explicit manipulation of ``infinitesimals'' in analysis~\cite{Robinson,Cutland}. Robinson's approach was axiomatic, by adding three new axioms to the  basic Zermelo-Fraenkel (ZFC) framework. 
An alternative presentation was later proposed by Lindstr{\o}m~\cite{Lindstrom}. Its interest is that it does not require any fancy axiomatic material but only makes use of Zorn's lemma, equivalent to the axiom of choice in the ZF set theory. The proposed construction bears some resemblance to the construction of $\bR$ as the set of equivalence classes of Cauchy sequences in $\bQ$ modulo the equivalence relation $\sim$ defined by $(u_n)\sim(v_n)$ iff \mbox{$\lim_{n\rightarrow\infty}(u_n-v_n)=0$.} 

The important point for us is that nonstandard analysis allows the use of  the nonstandard discretization of continuous dynamics ``as if'' it was operational and with infinitesimal discretization error. Iwasaki et al.~\cite{IwasakiFSBG95} first proposed using nonstandard analysis to discuss the nature of time in hybrid systems.  Bliudze and Krob~\cite{BliudzeK09,BliudzePhD} have also used nonstandard analysis as a mathematical support for defining a system theory for hybrid systems. In their study of mathematical foundations of hybrid systems modelers of the ODE class, Benveniste et al.~\cite{JCSS11} used extensively the nonstandard semantics of hybrid systems, with practical consequences for how to design a specification formalism for multimode ODE systems~\cite{lucy:hscc13,benveniste:hal-01549183}.\footnote{Corresponding to the Simulink class of modeling languages, with no DAE.} See also~\cite{DBLP:conf/popl/SuenagaSH13,DBLP:conf/aplas/NakamuraKSI17,DBLP:journals/ngc/Hasuo17} for similar studies using the axiomatic approach to nonstandard analysis.
The following presentation is borrowed from~\cite{JCSS11}. 

\subsection{Intuitive introduction}
We begin with an intuitive introduction to the construction of the 
nonstandard reals.
The goal is to augment $\bR\cup\{\pm\infty\}$ by adding, to each $x$ in
the set, a set of elements that are ``infinitesimally close'' to
it. We will call the resulting set ${\nstr}$.
Another requirement is that all
operations and relations defined on $\bR$ should extend to ${\nstr}$.

A first idea is to represent such additional numbers as convergent
real sequences. For example, elements infinitesimally close to the
real number zero are the sequences $u_n={1}/{n}$, $v_n={1}/{\sqrt{n}}$
and $w_n={1}/{n^2}$. Observe that the above three sequences can be
ordered: $v_n> u_n> w_n> 0$ where $0$ denotes the constant zero
sequence. Of course, infinitely large elements (close to $+\infty$)
can also be considered, e.g., sequences $x_u=n$, $y_n={\sqrt{n}}$, and
$z_n={n^2}$.

Unfortunately, this way of defining ${\nstr}$ does not yield a total order 
since two sequences converging to zero cannot always be
compared: if $u_n$ and $u'_n$ are two such sequences, the three sets
$\{n\mid u_n>u'_n\}$, $\{n\mid u_n=u'_n\}$, and $\{n\mid u_n<u'_n\}$
may even all be infinite. The beautiful idea of Lindstr{\o}m is to
enforce that \emph{exactly one of the above sets is important and the
 other two can be ignored}. 
 
 The key step in Lindstr{\o}m's construction consists in fixing once and for
all a finitely additive positive measure $\mu$ over the set $\bN$ of
integers with the following properties:\footnote{The existence of such
 a measure is non-trivial and is explained later.}
 \[\bea{l}
 \mu: 2^\bN \ra \{0,1\} \\
 \mbox{$\mu(X)=0$ whenever $X$ is finite}  \\
 \mu(\bN)=1 
 \eea
 \]
Once $\mu$ is fixed, one can compare any two sequences $u_n$ and $u'_n$ as follows. Exactly one of the three sets $\{n\mid u_n>u'_n\}$, $\{n\mid u_n=u'_n\}$, or $\{n\mid u_n<u'_n\}$ has $\mu$-measure $1$ whereas the other two must have $\mu$-measure $0$. Thus, we say that 
\[
\bea{rcl}
u>u' &\mbox{if}& \mu(\{n\,{\mid}\,u_n>u'_n\})=1 \\ 
u=u' &\mbox{if}& \mu(\{n\,{\mid}\,u_n=u'_n\})=1 \\ 
u<u' &\mbox{if}& \mu(\{n\,{\mid}\,u_n<u'_n\})=1
\eea
\]
respectively. 
Indeed, the same trick works for many other relations and 
operations on nonstandard real numbers, as we shall
see. We now proceed with a more formal presentation.

\subsection{Nonstandard domains}
For $I$ an arbitrary set, a \emph{filter} $\cF$ over $I$ is a family of subsets of $I$ such that: (1) the empty set does not belong to $\cF$; 
	(2) $P,Q\in\cF$ implies $P\cap Q\in\cF$; and 
	(3) $P\in \cF$ and $P\subset Q\subseteq I$ implies $Q\in\cF$.
Consequently, $\cF$ cannot contain both a set $P$ and its complement
$P^c$. A filter that contains one of the two for any subset
$P\subseteq I$ is called an \emph{ultra-filter}. At this point we
recall Zorn's lemma, known to be equivalent to the axiom of choice:
\begin{lemma}[Zorn's lemma] \label{epr98hepr9} 
Any partially ordered set \mbox{$(X,\leq)$} such that any chain in $X$ 
possesses an upper bound has a maximal element.
\end{lemma}
A filter $\cF$ over $I$ is an ultra-filter if and only if it is maximal with 
respect to set inclusion.
By Zorn's lemma, any filter $\cF$ over $I$ can be extended to an 
ultra-filter over $I$.
Now, if $I$ is infinite, the family of sets $\cF=$
\mbox{$\{P\subseteq I\mid P^c \mbox{ is finite}\}$} is a \emph{free}
filter, meaning it contains no finite set. It can thus be extended to
a {free ultra-filter} over $I$. Hence:
\begin{lemma} \label{efiupeiu}
 Any infinite set has a free ultra-filter.
\end{lemma}
Every free ultra-filter $\cF$ over $I$ uniquely defines, by setting
$\mu(P) = 1$ if $P\in\cF$ and 0 otherwise, a finitely additive
measure\footnote{Observe that, as a consequence, $\mu$ cannot be
  sigma-additive (in contrast to probability measures or Radon
  measures) in that it is \emph{not} true that $\mu(\bigcup_n
  A_n)=\sum_n\mu(A_n)$ holds for an infinite denumerable sequence
  $A_n$ of pairwise disjoint subsets of $\bN$.}  $\mu:
2^I\mapsto\{0,1\}$, which satisfies
\beq
\mbox{
$\mu(I)=1$ and, if $P$ is finite, then $\mu(P)=0$.}
\label{erltuihjeui}
\eeq
Now, fix an infinite set $I$ and a finitely additive measure $\mu$ over $I$ 
as above. Let $\bX$ be a set and consider the Cartesian product 
$\bX^I=(x_i)_{i\in I}$. Define 
\beq
(x_i)\sim(x'_i)
\label{erlouiwleiurhfuil}
\eeq 
if and only if
\mbox{$\mu\{i\in I\mid x_i\neq x'_i\}=0$}. Relation $\sim$ is an equivalence relation whose equivalence classes are denoted by 
$[x_i]$
 and we define 
\beq
\nst{\,\bX}=\bX^I/\sim ~ .
\label{erpfheu}
\eeq 
$\bX$ is naturally embedded into $\nst{\,\bX}$ by
mapping every $x\in \bX$ to the constant tuple such that $x_i=x$ for
every $i\in I$; we denote it by 
$
[x] \mbox{ or, simply, by } x 
$,
if no confusion can result. 
Any algebraic structure over $\bX$ (group, ring,
field) carries over to $\nst{\,\bX}$ by almost pointwise extension. In
particular, if $[x_i]\neq 0$, meaning that $\mu\{i\mid x_i=0\}=0$ we
can define its inverse $[x_i]^{-1}$ by taking $y_i=x^{-1}_i$ if
$x_i\neq 0$ and $y_i=0$ otherwise. This construction yields
$\mu\{i\mid y_i x_i=1\}=1$, whence $[y_i][x_i]=1$ in $\nst{\,\bX}$.
The existence of an inverse for any non-zero element of a
ring is indeed stated by the formula: \mbox{$\forall
 x\,(x=0\mathrel{\lor}\exists y\,(xy=1))$}.\footnote{More generally, the \emph{Transfer Principle} states that every first-order {formula} is true over $\nst{\,\bX}$ if and only if it is true over $\bX$.
Recall that a first-order formula is a formula involving quantifiers over variables but not functions.}

\subsection{Nonstandard reals and integers}
\label{sec:nonstannums}
The above general construction applies to $\bX=\bR$ and 
$I=\bN$.
The result, denoted by $\nstr$, is a field (according to the transfer 
principle).
By the same principle, $\nstr$ is totally ordered by \mbox{$[u_n]\leq[v_n]$} 
iff $\mu\{n\mid u_n{>}v_n\}=0$. 
Call \emph{infinitesimal} any nonstandard real number whose absolute value is smaller than any positive real number: $[\,|x_n|\,]\leq[\varepsilon]$ for any $\varepsilon\in\bR,\varepsilon>0$. For $x$ and $y$ two nonstandard real numbers, write
\[
x \approx y
\]
if $x-y$ is infinitesimal. Call \emph{infinite} any nonstandard real number whose absolute value is larger than any positive real number: $[\,|x_n|\,]\geq[K]$ for any finite $K\in\bR$. Call \emph{finite} a nonstandard real number that is not infinite. Call \emph{bounded} a subset $A\subseteq{\nst{\,\bR}}$ such that there exists a finite standard positive real $K$ such that $|x|\leq{K}$ for every $x\in{A}$.

\begin{lemma}[standard part] 
	\label {owirguhpr} \label{elirsuthoperui} 
For any finite $[x_n]\in\nstr$, there exists a unique standard real number $x\in\bR$, such that $[x]-[x_n]$ is infinitesimal. We call $x$ the \emph{standard part} of $[x_n]$ and denote it by $\st{[x_n]}$. Infinite nonstandard reals have no standard
part in $\bR$. 
\end{lemma}
\begin{proof} 
To prove this, let
$x=\sup\{u\in\bR\mid [u]\leq[x_n]\}$. Since $[x_n]$ is finite, $x$ exists and we only
need to show that $[x_n]{-}x$ is infinitesimal. If not, then there
exists $y\in\bR, y{>}0$ such that either
\mbox{$[x]<[x_n]-[y]$} or \mbox{$[x]>[x_n]+[y]$}, which both contradict
the definition of $x$. The uniqueness of $x$ is clear, thus we can
define $\st{[x_n]}=x$. 
\end{proof}
It is also of interest to apply the general construction
(\ref{erpfheu}) to $\bX=I=\bN$, which results in the set $\nst{\,\bN}$ of
\emph{nonstandard natural numbers}; $\nst{\,\bZ}$ is defined similarly, from $\bX=\bZ$ and $I=\bN$.
The nonstandard set $\nstn$ differs from
$\bN$ by the addition of \emph{infinite natural numbers,} which are
equivalence classes of sequences of integers whose essential limit is
$+\infty$. 

\subsection{Internal functions and sets}
Any sequence $(f_n)$ of functions $f_n:\bR\rightarrow\bR$ pointwise defines a function $[f_n]:\nstr\rightarrow\nstr$ by setting
\beq
[f_n]([x_n]) &\eqdef& [f_n(x_n)]\,.
\label{piufhwpiou}
\eeq
To justify this definition based on a representant of the equivalence class, note that $[x_n]=[x'_n]$ means that the set of integers $n$ such that $x_n\neq{x'_n}$ is neglectible, and we have $f_n(x_n)=f_n(x'_n)$ for every $n$ outside of this set.
A function $\nstr\rightarrow\nstr$ obtained in this way is called \emph{internal.}
Properties of, and operations on, ordinary
functions extend pointwise to internal functions of
$\nstr\rightarrow\nstr$. 
The same notions apply to sets by considering their characteristic functions. 

An internal
set $A=[A_n]$ is called \emph{hyperfinite} if \mbox{$\mu\{n\mid A_n
 \mbox{ finite}\}=1$}; the \emph{cardinal} $|A|$ of $A$ is defined as
$[\,|A_n|\,]$, where $|A_n|$ denotes the cardinal of the finite set $A_n$ in the usual sense.
For $A$ a bounded set of hyperreals, we define its \emph{shadow} $\shadow{A}\subseteq\bR$ as follows: 
\beq
\shadow{A} = \{\st{x}\in{\bR}\mid\exists{y}{\in}{A} \mbox{ such that } y\approx{x}\}\,. \label{oltuheliu}
\eeq
The \emph{nonstandard lifting} of 
\mbox{$f:\bR\ra\bR$} is the internal function
$\nst{f}=[f,f,f,\dots]\,.$
Let \mbox{$f:\bR\ra\bR$} be a standard function that is differentiable at $x\in\bR$, and $\nst{f}=[f,f,f,\dots]$ its nonstandard lifting. Then, we have, for every $x\in\bR$~:
\beq
\dot{f}(x)&=&\st{\frac{\nst{f}(y)-\nst{f}(x)}{y-x}}
\label{50bhjietk}
\eeq
for every $y\in\nstr$ such that $y\approx[x]$---the proof is similar to that of Lemma~\ref{owirguhpr}. 

Now, consider an infinite number $K\in\nstn$ and the set 
\beq
T=\left\{\,
0,\frac{1}{K},\frac{2}{K},\frac{3}{K},\dots\frac{K-1}{K},1
\,\right\} ~ .
\label{erpifuherifp}
\eeq
By definition, if $K=[K_n]$, then $T=[T_n]$ with 
\[
T_n=\left\{\,
0,\frac{1}{K_n},\frac{2}{K_n},\frac{3}{K_n},\dots\frac{K_n-1}{K_n},1
\,\right\}
\]
hence $|T|=[\,|T_n|\,]=[K_n+1]=K+1$. Note that the shadow of $T$ (see (\ref{oltuheliu})) is the standard interval $[0,1]$: for any given real number $a$ in $[0,1]$, $\dfrac{\lfloor aK \rfloor}{K}$ standardizes to $a$.

In our nonstandard semantics, we will be using the following time basis:
\beq
\bT&\eqdef&\{k \vsmall\mid{k}{\in}\nstn\}
\label{eortuheli}
\eeq
where $\vsmall>0$ is some fixed infinitesimal time. This definition is reminiscent of Equation~(\ref{non-standard-time-line}), in Section~\ref{sec:nsa}. The shadow of $\bT$ is $\bR_+$, the set of nonnegative real numbers, and that, informally speaking, $\bT$ is both ``discrete'' (every element has a previous and a next element, except for 0) and ``continuous'' (as the shadow of $\bT$ is $\bR_+$). 
\begin{notation}\rm 
	\label{leruioghpuio} Generic elements of $\bT$ will be denoted by the special symbol $\nstime$ or simply $t$ when no confusion can result.
\end{notation}
\subsection{Integrals}
Now, consider an internal function $f=[f_n]$ and a hyperfinite set 
$A=[A_n]$.
The \emph{sum} of $f$ over $A$ can be defined by
\[
\sum_{a\in A} \; f(a) \eqdef \left[\,
\sum_{a\in A_n} \; f_n(a)
\right] ~ .
\]
If $T$ is as above, and $f: \bR\ra\bR$ is a standard function, we obtain
\beq
\sum_{t\in T}\frac{1}{|T|}\,{{\nst{f}}(t)} = \left[\,
\sum_{t\in T_n}\frac{1}{|T_n|}{f}(t)
\right] ~ .
\label{08t7f087}
\eeq
Now, the continuity of $f$ implies the convergence of the Riemann sums:
$\sum_{t\in T_n}\frac{1}{|T_n|}{f}(t)\rightarrow\int_0^1f(t)\dt$. Hence,
\beq
\int_0^1f(t)\dt &=& \st{\,\sum_{t\in T}\,\frac{1}{|T|}{\nst{f}(t)}\,}.
\label{peruiohpferuio}
\eeq
Under the same assumptions, for any $t\in[0,1]$, 
\beq
\int_0^tf(u)\du &=& \st{\,\sum_{u\in T, u\leq t}\ 
\frac{1}{|T|}{\nst{f}(u)}\,}.
\label{oeiyro87}
\eeq
\subsection{ODE}
Consider the following ODE:
\beq
\dot{x}=f(x,t), \ \ x(0)=x_0\,.
\label{0e9784rfh9pe8}
\eeq 
Assume (\ref{0e9784rfh9pe8}) possesses a solution $x:[0,1]\rightarrow\bR$ such that the function $t \mapsto f(x(t),t)$ is continuous.
Rewriting (\ref{0e9784rfh9pe8}) in its equivalent integral form
\mbox{$x(t)=x_0+\int_0^tf(x(u),u)\du$}\, and using (\ref{oeiyro87})
yields
\beq 
x(t) = \st{\nst{x}(t)} \mbox{ where }\,{\nst{x}(t)\eqdef x_0+\sum_{u\in T, u\leq t}\
  \frac{1}{|T|}\,{\nst{f}(x(u),u)}\,}.
\label{erpifuheipu}
\eeq
One can rewrite the positive infinitesimal quantity $1/|T|$ as $\partial$, so that
$$T=\{\nstime_k=k\partial\mid k=0,\dots,|T|\}\,.$$ 
Then, after substituting in (\ref{erpifuheipu}), one gets that the piecewise-constant right-continuous
function $\nst{x}(\nstime),\nstime{\in}\nst{\bR},0{\leq}\nstime{\leq}1$ satisfies the following difference equation, for $k{=}0,\dots,|T|{-}1$,
\beq\bea{lcl}
\nst{x}(\nstime_{k+1}) &=& \nst{x}(\nstime_{k})+ \partial \,.\,\ttimes \nst{f}(\nst{x}(\nstime_{k}),\nstime_{k})
\\
\nst{x}(\nstime_0) &=& x_0
\eea
\label{pe498hrte489r}
\eeq
By (\ref{erpifuheipu}), the following holds:
\begin{theorem}
	\label{0pw489gthtrgu} 
The solution $x(t),t\in\bR$, of ODE $(\ref{0e9784rfh9pe8})$, and 
$\nst{x}(\nstime),\nstime\in\nst{\bR}_+$, the piecewise constant interpolation of the trajectory of system $(\ref{pe498hrte489r})$, are related by \mbox{$x = \st{{\nst{x}}}$}.
\end{theorem}
The same line of arguments applies for systems of ODEs possessing a unique continuous solution, i.e., $x$ and $f$ take their values in $\bR^m$ in (\ref{0e9784rfh9pe8}). 
Formula (\ref{pe498hrte489r}) can be seen as a \emph{nonstandard
  semantics} for ODE (\ref{0e9784rfh9pe8}). This semantics seems to depend on 
  the choice of infinitesimal step parameter $\partial$.
Property (\ref{erpifuheipu}), though, expresses that 
\beq
\mbox{
\begin{minipage}{11cm}
	 all of these nonstandard semantics are equivalent from the standard viewpoint, regardless of the choice made for $\partial$.
\end{minipage}
}
\eeq
The paper~\cite{Lindstrom} goes beyond our short exposure by including a direct proof of the Peano theorem on the existence of solutions of (\ref{0e9784rfh9pe8}) when $f$ is continuous and bounded.

\subsection{Infinitesimal calculus}
\label{sec:infinitcalc}
We conclude this background material with a summary of the section ``Infinitesimal calculus'' from~\cite{Lindstrom}. We formulate the essential characterizations of properties of the infinitesimal calculus in nonstandard terms. $f$ denotes a standard function and $\nst{f}=[f,f,f,\dots]$ is its nonstandard lifting. 
The reader is referred to that text for the proofs.
\begin{theorem}
	\label{perqoui} \ 
\begin{enumerate}
	\item The function $f:\bR\ra\bR$ is continuous at $a\in\bR$ if and only if $\nst{f}(x)\approx{f(a)}$ for all $x\approx{a}$.
	\item \label{ow4u9ghspsr} The function $f:\bR\ra\bR$ is uniformly continuous on set $A\subseteq\bR$ if and only if $\nst{f}(x)\approx\nst{f}(y)$ for all $x,y\in\nst{A}$ such that $x\approx{y}$.
	\item \label{ouiepuicksv} The function $f:\bR\ra\bR$ is differentiable at $a\in\bR$ if and only if there exists a number $b\in\bR$ such that
	$$
	\frac{\nst{f}(x)-\nst{f}(a)}{x-a} \approx b\, , 
	\mbox{for all $x\in{\nst{\bR}},x\approx{a}, x\neq{a}$.}
	$$
	Moreover, if such a number $b$ exists, then it equals $\dot{f}(a)$.
\end{enumerate}
\end{theorem}

\section{The toolkit supporting standardization}
\label{sec:standardization}
By building on the background of \rref{sec:NSA}, we develop here the specific material we need for our developments. First, we develop the compile-time impulse analysis. Second, we formally justify our use of structural analysis in the nonstandard domain. Third, referring to Comment~\ref{keruif} in support of the clutch example, we formally justify our method for standardizing systems of equations. 
We consider systems of equations of the form
\beq
0=\bfH(\dot{X},\postset{X},V,X) ~ ,
\label{eroifughoiu}
\label{lergtiurhuip}
\eeq
where $Z\eqdef(\postset{X},V)$ collects the dependent variables. Such systems are also written as
\beq
0=H(Z,X,\vsmall) 
\label{ltrghilu}
\eeq
after expanding $\dot{X}$ as $\frac{\postset{X}-X}{\vsmall}$.
In System~(\ref{eroifughoiu}), $V$ collects the algebraic variables, $X$ collects the state variables, and $\frac{\postset{X}-X}{\vsmall}$ is the nonstandard semantics of $\dot{X}$. Writing (\ref{lergtiurhuip}) involves both $\dot{X}$ and $\postset{X}$ but not $\vsmall$. 
As~(\ref{ltrghilu}) is obtained from~(\ref{eroifughoiu}) by performing the expansion, it involves both $\postset{X}$ and $\vsmall$ but not $\dot{X}$.

As an illustration, with reference to the Cup-and-Ball example, System (\ref{uioghrpui}) is of the form (\ref{eroifughoiu})---we repeat it here for convenience:
\beqq\left\{\bea{rll}
& 0= \ddot{x}+{\tension}x & (\eqq_1) \\
& 0= \ddot{y}+{\tension}y+g  & (\eqq_2) \\
& {\postset{\guard}= [s\leq{0}]; \guard(0)=\fff}  & {(\straight_0)} \\
\when \; \guard\; \doo& 0={L^2}{-}(x^2{+}y^2)   & (\straight_1) \\
\prog{and}& \remph{0={L^2}{-}\postset{(x^2{+}y^2)}}   & \remph{(\postset{\straight_1})} \\
\prog{and}& \remph{0={L^2}{-}\ppostset{2}{(x^2{+}y^2)}}   & \remph{(\ppostset{2}{\straight_1})} \\
\prog{and}& 0=\tension+s   & (\straight_2) \\
\when \;\prog{not}\; \guard\; \doo& 0=\tension   & (\straight_3) \\
\prog{and}& 0=({L^2}{-}(x^2{+}y^2))-s   & (\straight_4) \\
\eea\right.
\eeqq
Its counterpart following form (\ref{ltrghilu}) is obtained by expanding the second derivatives by using Euler scheme:
\beqq\left\{\bea{rll}
& 0= \frac{\ppostset{2}{x}-2\postset{x}+x}{\vsmall^2}+{\tension}x & (\eqq_1) \\
& 0= \frac{\ppostset{2}{y}-2\postset{y}+y}{\vsmall^2}+{\tension}y+g  & (\eqq_2) \\
& {\postset{\guard}= [s\leq{0}]; \guard(0)=\fff}  & {(\straight_0)} \\
\when \; \guard\; \doo& 0={L^2}{-}(x^2{+}y^2)   & (\straight_1) \\
\prog{and}& \remph{0={L^2}{-}\postset{(x^2{+}y^2)}}   & \remph{(\postset{\straight_1})} \\
\prog{and}& \remph{0={L^2}{-}\ppostset{2}{(x^2{+}y^2)}}   & \remph{(\ppostset{2}{\straight_1})} \\
\prog{and}& 0=\tension+s   & (\straight_2) \\
\when \;\prog{not}\; \guard\; \doo& 0=\tension   & (\straight_3) \\
\prog{and}& 0=({L^2}{-}(x^2{+}y^2))-s   & (\straight_4) \\
\eea\right.
\eeqq
\begin{ccomment}\rm
	\label{klgfuio} System (\ref{eroifughoiu}) is the generic form of systems we need to solve at \rref{op:weorpifuhfpeiurevis} of \rref{alg:ergfipuhfopiu}, when we eliminate higher order derivatives and shifts by introducing auxiliary variables.\eproof
\end{ccomment} 
We consider the following assumptions, which hold for the nonstandard semantics of any standard \mDAE\ system: 
\begin{assumption} \
	\label{oergfiuehoiu} 
\begin{enumerate}
\item The functions $\bfH(\cdot)$ and $H(\cdot)$ are standard (they result from applying \emph{\rref{alg:ergfipuhfopiu}} to our given standard \mDAE\ model);
\item The state variables $X$ are standard and finite (their values were assigned by the standardized dynamics of the previous mode).
\end{enumerate} 
\end{assumption}

\subsection{{Impulse Analysis}}
\label{sec:impulseanalysis}
We refer the reader to \rref{sec:standardizeclutch} about the standardization of the clutch model at mode changes, and particularly  \rref{sys:wp49guhsopiu} in it. As a prerequisite to standardization, we had to identify which variables could become impulsive at mode changes (the two torques in this example), and then, we had to eliminate them. 

To identify impulsive and non-impulsive variables in a systematic way, we develop the \emph{impulse analysis,} which
 consists in abstracting hyperreals with their ``{magnitude order}'' compared to the infinitesimal $\vsmall$. Since our original \DAE\ models are all standard, only the occurrence of the infinitesimal $\vsmall$ will have a role in this matter. The impulse analysis will be useful as a preparation step for  standardization. 
\begin{definition}[impulse order and analysis]
	\label{kfrygjkasdtf}  \
\begin{enumerate}
	\item 
Given a system of equations such as $(\ref{eroifughoiu})$ or $(\ref{ltrghilu})$, say that a dependent variable $x$ has \emph{impulse order} (or simply \emph{order}) $\imporder\in\bR$, if the solution of the considered system is such that $x\vsmall^\imporder$ is provably a finite non-zero (standard) real number. We denote by $\magorder{x}$ the impulse order of $x$. By convention, the constant $0$ has impulse order $-\infty$. 
\item
Say that $x$ is \emph{impulsive} if $\magorder{x}>0$.
\item
The \emph{impulse analysis} of a system of equations such as $(\ref{eroifughoiu})$ or $(\ref{ltrghilu})$ is the system of constraints satisfied by the impulse orders of the dependent variables of the system.
\end{enumerate}
\end{definition}
\subsubsection{The rules of impulse analysis}
Figures~\ref{fig:lreiufgoeyuy} and~\ref{fig:lrieuiouwi} display the rules defining the translation of a system of equations of the form $(\ref{eroifughoiu})$ or $(\ref{ltrghilu})$ into its impulse analysis, for the restricted class where only rational expressions are involved.
\begin{figure}[ht]
\vspace*{-5mm}
\beqq
e &~::=~& 0 ~\mid~c ~\mid~ \vsmall ~\mid~ x  ~\mid~ e^{c} ~\mid~ e+e ~\mid~ e \times e 
\\
E &::=& e=e ~\mid~ E \;\prog{and}\; E 
\eeqq
\vspace*{-5mm}
\caption{Syntax: $E$ is a system of one or several equations $e=e$. An expression $e$ is $0$, a nonzero (standard) real constant $c$, the infinitesimal $\vsmall$, a variable $x$, the monomial $e^c$, a sum, or a product.}
\label{fig:lreiufgoeyuy}
\[ 
\hspace*{-3mm}\bea{crcl} 
\mbox{(R1)}& \magorder{0}&=&-\infty \\ [1mm] \mbox{(R2)}& \magorder{c}&=&0 \\ [1mm] \mbox{(R3)}&\magorder{\vsmall}&=&-1 \\ [1mm] \mbox{(R4)}& \magorder{e^{c}}&=&c\magorder{e}
\\ [1mm]
\mbox{(R5)}& \magorder{e_1 \times e_2}&=&\magorder{e_1}+\magorder{e_2}
\\ [1mm]
\mbox{(R6)}& \magorder{e_1 + e_2} &\leq& \max\{\magorder{e_1},\magorder{e_2}\}
\eea 
\bea{cc}
\displaystyle\frac{
E \vdash e=e'
}{
\magorder{E} \vdash \magorder{e}=\magorder{e'}
} & \mbox{(R7)}
\\ [8mm]
\displaystyle\frac{\left.\bea{l}
E \vdash x=y+e\, \mbox{ ~or} \\ E \vdash 0=y-x+e 
\eea\right\} \mbox{ and }~
E \nvdash y=x-e
}{
E \vdash E \;\prog{and}\; y=x-e
} & \mbox{(R8)}
\eea
\]
\caption{Rules: The left column displays the impulse order of the primitive expressions. Rule (R7) indicates that $\magorder{e}=\magorder{e'}$ is an equation of the impulse analysis $\magorder{E}$ if $e=e'$ is an equation of $E$; Rule (R8) indicates that, if $E$ involves the equation $x=y+e$ but not the equation $y=x-e$, then we augment $E$ with the latter, i.e., we saturate $E$ with the rule $x=y+e \implies y=x-e$.}
\label{fig:lrieuiouwi}
\end{figure}

\rref{fig:lreiufgoeyuy} describes the syntax of a mini-language specifying such systems of equations. The left column of \rref{fig:lrieuiouwi} gives the rules for mapping expressions to their corresponding impulse orders. All the rules are self-explanatory, with the exception of Rule (R6), where the inequality deserves some explanation.
The sum $e_1{+}e_2$, the dominant terms in the expansion of the $e_i$'s as power series over $\vsmall$ may compensate each other for their impulse orders. For an example of this, see equation $(\eqq^\vsmall_1)$ in \rref{sys:wp49guhsopiu}: rewriting this equation as $
a_1(\omega_1)=\frac{\postset{\omega_1}-\omega_1}{\vsmall}-b_1(\omega_1)\tau_1
$, we see a case of strict inequality for (R6) since $a_1(\omega_1)$ has order zero, whereas it is equal to the difference of two terms of order one, see Section~\ref{eroifueoiug} for details. 

We will use Rule (R6) in the following way, thereby reinforcing it. Consider an equation
\[
\eqq: z=x+y\,.
\]
We can rewrite $\eqq$ in the following equivalent ways: $0=x+y-z\,,\, x=z-y$, or $y=z-x$. To each of them we apply the max rule. This yields, for the impulse analysis of equation $\eqq$, the following system of constraints:
\beq
\mbox{impulse analysis of }\eqq:
\left\{
\bea{lcl}
\magorder{z}\leq\max\{\magorder{x},\magorder{y}\}&;&
\magorder{0}\leq\max\{\magorder{x},\magorder{y},\magorder{z}\}
\\ [1mm]
\magorder{x}\leq\max\{\magorder{z},\magorder{y}\}&;&
\magorder{y}\leq\max\{\magorder{x},\magorder{z}\}
\eea\right.
\label{leriughpiu}
\eeq
Note that the constraint $\magorder{0}\leq\dots$ is vacuously satisfied since $\magorder{0}=-\infty$. Also note that, among the three nontrivial inequalities of (\ref{leriughpiu}), at least two must be saturated. We will use impulse analysis (\ref{leriughpiu}) for handling sums of terms. 
This  reinforcement of the max rule is formalized by Rule (R8) of \rref{fig:lrieuiouwi}, which mechanizes the association, to equation $\eqq$ of (\ref{leriughpiu}), of its different rewritings.

Using the rules of Figures~\ref{fig:lreiufgoeyuy} and~\ref{fig:lrieuiouwi} in the numerical expressions, we map any system of rational equations of the form $(\ref{eroifughoiu})$ or $(\ref{ltrghilu})$ into a system of constraints over impulse orders. 

To cover functions beyond polynomials, we need to extend $\bR\cup\{-\infty\}$ with $+\infty$. In this extension, we take the convention that $-\infty+\infty=-\infty$, justified by the equality $0{\times}x=0$ for any nonstandard $x$. For functions $f(x)=\sum_{k=0}^\infty a_kx^ k$ that can be represented as absolutely converging power series, we  then get
\beq
\magorder{f(x)} = \magorder{\,\sum_{k=0}^\infty a_kx^ k\,} = \magorder{{x}}.\sup(A), \mbox{ where } A=\{k\mid a_k\neq{0}\}
\label{weoriuthlkui}
\eeq
is the support of the series and $\sup(A)$ is the supremum of set $A$.
In particular, if $\magorder{{x}}>0$ and if the support of the series is infinite, we get $\magorder{f(x)}=+\infty$.

\subsubsection{Impulse analysis of systems of equations for restarts}
Here we particularize the impulse analysis to systems of equations of the form (\ref{eroifughoiu}), where the only reason for $\vsmall$ to occur is the expansion of derivatives using the Euler scheme:
\[
0=\bfH\left(\frac{\postset{X}-X}{\vsmall},\postset{X},V,X\right) 
\]
The dependent variables are $\postset{X},V$. It will be convenient to introduce the auxiliary variables
\[
U\eqdef\postset{X}-X\,,
\]
so that the systems we consider take the following form:
\beq
\mbox{dependent variables $\postset{X},V,U$ in}&:&
\left\{\bea{ccl}
0&=& \bfH\left(\frac{U}{\vsmall},\postset{X},V,X\right)
\\ [1mm]
U&=&\postset{X}-X
\eea\right.
\label{sys:lerioefkfuo}
\eeq
The following condition for \rref{sys:lerioefkfuo} can be assumed, based on physical considerations (restart values for an ODE or a DAE cannot be impulsive in a physically meaningful model):
\begin{assumption}
	\label{lwrieuogwh} 
	 Since $X$ is a state, both $X$ (a known value) and $\postset{X}$ must be finite.
\end{assumption}
First, the impulse orders $\magorder{X}$ are all known, from previous nonstandard instants. Next, from Assumption~\ref{lwrieuogwh} we deduce the inequalities
\beq
\magorder{\postset{X}}\leq 0 &\mbox{and}& \magorder{U}\leq 0\,.
\label{leriufokuydf}
\eeq
The impulse orders $\magorder{V}$ are a priori unknown. We have, however, more prior information, thanks to the structural analysis. As part of both the simpler algorithm $\atomicact{ExecRun}$ (\rref{alg:newmain}) and revisited algorithm $\atomicact{ExecRun}$ (\rref{alg:ergfipuhfopiu}), the auxiliary algorithm $\atomicact{SolveConflict}$ (\rref{alg:elsgouihuip}) is called. At the considered mode change, we thus know which consistency equation(s) of the new mode was/were conflicting with the dynamics of the previous mode. Formally, call $G{=}0$ the subsystem collecting all the equations that were erased by algorithm $\atomicact{SolveConflict}$ at \rref{op:elkfuiehoeuio} of \rref{alg:newmain} or \rref{op:elkfuiehoeuiorevis} of \rref{alg:ergfipuhfopiu}. As a result, $G{=}0$ is no longer satisfied at the considered mode change, and thus, $G$ defines a tuple $R$ of finite nonzero variables called \emph{residuals}, by setting
\beq
R=G\,. \label{elifugeyuiu}
\eeq
An example of residual in the clutch is $r=\omega_1-\omega_2$, which is both finite and nonzero at mode change $\guard:\fff\ra\ttt$; this information was found by the structural analysis.
Finally, the system of equations that we need to solve collects all the above items, namely:
\beq
\mbox{dependent variables $\postset{X},V,U,R$ in}&:&
\left\{\bea{ccl}
0&=& \bfH\left(\frac{U}{\vsmall},\postset{X},V,X\right)
\\ [1mm]
U&=&\postset{X}-X
\\ [1mm]
R&=&G
\eea\right.
\label{uoegiuhwoeilu}
\\ \mbox{prior information on impulse orders}&:& \left\{\bea{ccl}\magorder{\frac{1}{\vsmall}}&=&1 \\  [1mm] \magorder{\postset{X}} &\leq& 0 \\ [1mm] \magorder{U} &\leq& 0 \\ [1mm] \magorder{R}&=&0
\eea\right.
\label{ltrgiubriu}
\eeq

\paragraph{Decomposing algebraic dependent variables into impulsive and non-impulsive ones}
In the following, we assume that the vector functions $\bfH$ or $H$ can be represented as absolutely convergent power series in their arguments. We can then apply the rules of Figures~\ref{fig:lreiufgoeyuy} and~\ref{fig:lrieuiouwi}, complemented with (\ref{weoriuthlkui}), to derive the system of equations that impulse orders must satisfy. 
In particular, this allows us to partition the set $V$ of algebraic variables as
\beq
V &=& \impuls{W}\,\cup\,\nonimpuls{W}
\label{psergtiohjpio}
\eeq
where $\impuls{W}$ collects the \emph{impulsive variables}, having impulse order $>0$, and $\nonimpuls{W}$ collects the \emph{non-impulsive variables}. In \rref{sec:restartschemes}, we will propose a numerical scheme for computing restarts that only requires the splitting of algebraic dependent variables into impulsive and non-impulsive ones. For now, let us develop the impulse analysis for the two transitions $\guard:\ttt\ra\fff$ and $\guard:\fff\ra\ttt$ 
of \rref{sys:nscoupledshaftsbroken} assuming linear $f_i$ as in \rref{eq:peowigtu9p}. 

\subsubsection{Example: mode change $\guard:\ttt\ra\fff$ of the clutch}

We develop this example here based on the rules of Figures~\ref{fig:lreiufgoeyuy} and~\ref{fig:lrieuiouwi}. To better illustrate this mechanical reasoning, we forbid ourselves algebraic manipulations that otherwise could help simplifying a manual analysis.
We repeat here the corresponding system of equations with its expansion (\ref{sys:lerioefkfuo})---there is no residual for this mode change:
\beq
\left\{
\bea{lcc}
\frac{u_1}{\vsmall}=a_1(\omega_1)+b_1(\omega_1)\tau_1 &&(\eqq_1^\vsmall) \\
\frac{u_2}{\vsmall}=a_2(\omega_2)+b_2(\omega_2)\tau_2 &&(\eqq_2^\vsmall) \\
\tau_1=0 &&(\eqq_{5}) \\
\tau_2=0 &&(\eqq_{6}) \\
u_1 = \postset{\omega_1}-\omega_1 \\
u_2 = \postset{\omega_2}-\omega_2 
\eea
\right.
\label{sys:reuhgoiweruh}
\eeq
We then saturate \rref{sys:reuhgoiweruh} using Rule (R8) of \rref{fig:lrieuiouwi} (added equations are in $\bemph{\rm blue}$):
\beq
\left\{
\bea{l}
\frac{u_1}{\vsmall}=a_1(\omega_1)+b_1(\omega_1)\tau_1 \;;\;
\bemph{a_1(\omega_1)=\frac{u_1}{\vsmall}-b_1(\omega_1)\tau_1 \;;\;
b_1(\omega_1)\tau_1=\frac{u_1}{\vsmall}-a_1(\omega_1)} \\
\frac{u_2}{\vsmall}=a_2(\omega_2)+b_2(\omega_2)\tau_2 \;;\;
\bemph{a_2(\omega_2)=\frac{u_2}{\vsmall}-b_2(\omega_2)\tau_2 \;;\;
b_2(\omega_2)\tau_2=\frac{u_2}{\vsmall}-a_2(\omega_2)} \\
\tau_1=0  \\
\tau_2=0 \\
u_1 = \postset{\omega_1}-\omega_1 \;;\; \bemph{u_1+\omega_1 = \postset{\omega_1} \;;\; u_1-\postset{\omega_1} = -\omega_1}\\
u_2 = \postset{\omega_2}-\omega_2 \;;\;\bemph{ u_2+\omega_2 = \postset{\omega_2}  \;;\; u_2 - \postset{\omega_2}=-\omega_2 }
\eea
\right.
\label{sys:regfygygfg}
\eeq
and we consider the prior information (\ref{ltrgiubriu}):
\beq
\magorder{\postset{\omega_i}}\leq 0 \;;\; \magorder{u_i}\leq 0 \;;\; \magorder{\frac{1}{\vsmall}}=1
\label{liughpoitugh}
\eeq
We now apply the rules of \rref{fig:lrieuiouwi} to \rref{sys:regfygygfg}, then use (\ref{liughpoitugh}):
\beq
\left\{
\bea{l}
1+\magorder{u_1}\leq 0 \;;\; 0 \leq 1+ \magorder{u_1}\;;\; \mbox{vacuous}  \\
1+\magorder{u_2}\leq 0 \;;\; 0 \leq 1+ \magorder{u_2}\;;\; \mbox{vacuous}  \\
\magorder{\tau_1}=-\infty \\
\magorder{\tau_2}=-\infty \\
\magorder{u_1} \leq 0 \;;\; \mbox{vacuous}  \;;\; \mbox{vacuous} \\
\magorder{u_2} \leq 0 \;;\; \mbox{vacuous}  \;;\; \mbox{vacuous}
\eea
\right.
\label{sys:lrsieuhgoweiurhg}
\eeq
where ``vacuous'' indicates that the equation sitting in the corresponding position in \rref{sys:regfygygfg} generates a vacuously satisfied constraint.
\rref{sys:lrsieuhgoweiurhg} shows that $\magorder{u_i}=-1$, meaning that $u_i$ is continuous. Note that \rref{sys:lrsieuhgoweiurhg} uniquely determines the impulse orders for all dependent variables.

\subsubsection{Example: mode change $\guard:\fff\ra\ttt$ of the clutch}
\label{eroifueoiug}

We repeat here the corresponding system of equations with its expansion (\ref{sys:lerioefkfuo}), and we include the residual $\remph{r\eqdef\omega_1-\omega_2}$ (from structural analysis, we know it is nonzero), highlighted in $\remph{\rm red}$:
\beq
\left\{
\bea{lcc}
\frac{u_1}{\vsmall}=a_1(\omega_1)+b_1(\omega_1)\tau_1 &&(\eqq_1^\vsmall) \\
\frac{u_2}{\vsmall}=a_2(\omega_2)+b_2(\omega_2)\tau_2 &&(\eqq_2^\vsmall) \\
\postset{\omega_1}-\postset{\omega_2}=0 &&(\postset{\eqq_3}) \\
\tau_1+\tau_2=0 &&(\eqq_{4}) \\
u_1 = \postset{\omega_1}-\omega_1 \\
u_2 = \postset{\omega_2}-\omega_2 \\
\remph{r = \omega_1-\omega_2}
\eea
\right.
\label{sys:ilerughoreiugh}
\eeq
We then saturate \rref{sys:ilerughoreiugh} using Rule (R8) of \rref{fig:lrieuiouwi} (added equations are in $\bemph{\rm blue}$):
\beq
\left\{
\bea{l}
\frac{u_1}{\vsmall}=a_1(\omega_1)+b_1(\omega_1)\tau_1 \;;\;
\bemph{a_1(\omega_1)=\frac{u_1}{\vsmall}-b_1(\omega_1)\tau_1} \;;\;
\bemph{b_1(\omega_1)\tau_1=\frac{u_1}{\vsmall}-a_1(\omega_1)} \\
\frac{u_2}{\vsmall}=a_2(\omega_2)+b_2(\omega_2)\tau_2 \;;\;
\bemph{a_2(\omega_2)=\frac{u_2}{\vsmall}-b_2(\omega_2)\tau_2} \;;\;
\bemph{b_2(\omega_2)\tau_2=\frac{u_2}{\vsmall}-a_2(\omega_2)} \\
\postset{\omega_1}-\postset{\omega_2}=0 \;;\; \bemph{\postset{\omega_1}=\postset{\omega_2}} \\
\tau_1+\tau_2=0 ~\,\;;\; \bemph{\tau_1=-\tau_2} \\
u_1 = \postset{\omega_1}-\omega_1 \;;\; \bemph{u_1+\omega_1 = \postset{\omega_1}} \;;\; \bemph{u_1-\postset{\omega_1} = -\omega_1}\\
u_2 = \postset{\omega_2}-\omega_2 \;;\; \bemph{u_2+\omega_2 = \postset{\omega_2} } \;;\; \bemph{u_2 - \postset{\omega_2}=-\omega_2} \\
\remph{r = \omega_1-\omega_2}
\eea
\right.
\label{sys:rkuigewo}
\eeq
We also consider the prior information (\ref{ltrgiubriu}):
\beq
\magorder{\postset{\omega_i}}\leq 0 \;;\; \magorder{u_i}\leq 0 \;;\; \magorder{r}=0 \;;\; \magorder{\frac{1}{\vsmall}}=1 ~ .
\label{ieurfygeoyiu}
\eeq
We now apply the rules of \rref{fig:lrieuiouwi} to \rref{sys:rkuigewo}, then use (\ref{ieurfygeoyiu}):
\beq
\left\{
\bea{l}
1\leq \magorder{\tau_1} \;;\; \mbox{vacuous} \;;\;
\magorder{\tau_1}\leq 1 \\
1\leq \magorder{\tau_2} \;;\; \mbox{vacuous} \;;\;
\magorder{\tau_2}\leq 1 \\
\mbox{vacuous} \;;\; \magorder{\postset{\omega_1}}=\magorder{\postset{\omega_2}} \\
\mbox{vacuous} \;;\; \magorder{\tau_1}=\magorder{\tau_2} \\
\mbox{vacuous} \;;\; \mbox{vacuous} \;;\; 0\leq \max\{\magorder{u_1},\magorder{\postset{\omega_1}}\} \\
\mbox{vacuous} \;;\; \mbox{vacuous} \;;\; 0\leq \max\{\magorder{u_2},\magorder{\postset{\omega_2}}\} \\
0= \magorder{r} = \magorder{\omega_1} = \magorder{\omega_2} \\
\magorder{\postset{\omega_i}}\leq 0 \;,\; \magorder{u_i}\leq 0\;,\; i=1,2
\eea
\right.
\label{sys:elrihgtoui}
\eeq
which proves $\magorder{\tau_1}=\magorder{\tau_2}=1$, expressing that the two torques are impulsive. Note, on the other hand, that (\ref{sys:elrihgtoui}) does not fully determine the impulse orders of ${\postset{\omega_i}}$ and $u_i$, since one may, e.g., have $\magorder{\postset{\omega_i}}<0$ and $\magorder{u_i}=0$.

\subsubsection{{Using impulse analysis in code generation}}
\label{lriuhglreiugh}

Code generation for restarts consists in standardizing System~(\ref{eroifughoiu}) or System~(\ref{ltrghilu}). Recall that standardizing systems of equations requires more care than standardizing numbers, due to impulsive behaviors and singularity issues that result.

We can exploit the impulse analysis through three different approaches, which we illustrate by using the clutch example and particularly the mode change $\guard:\fff\ra\ttt$ when the clutch gets engaged. The system of equations for the corresponding restart is (\ref{sys:wp49guhsopiu}), which we recall for convenience:
\beqq \left\{ \bea{lcc}
  \postset{\omega_1}=\omega_1+{\vsmall}.(a_1(\omega_1)+b_1(\omega_1)\tau_1) &&(\eqq_1^\vsmall) \\
  \postset{\omega_2}=\omega_2+{\vsmall}.(a_2(\omega_2)+b_2(\omega_2)\tau_2) &&(\eqq_2^\vsmall) \\
  \postset{\omega_1}-\postset{\omega_2}=0 &&(\postset{\eqq_3}) \\
  \tau_1+\tau_2=0 &&(\eqq_{4}) \eea \right.
\eeqq
As we have shown, the two torques $\tau_1$ and $\tau_2$ are impulsive of order $1$.

\paragraph{Eliminating impulsive variables}
When this is practical, the simplest method is to eliminate impulsive variables from the restart system, namely the two torques in \rref{sys:wp49guhsopiu}. The details were given in Section~\ref{lwergtuiherlu} and the resulting reduced system of equations is 
\beq
\postset{\omega_1} =\postset{\omega_2} = \displaystyle\frac{b_2(\omega_2) \omega_1 + b_1(\omega_1) \omega_2}{b_1(\omega_1) + b_2(\omega_2)} 
 + \vsmall\,\displaystyle\frac{a_1(\omega_1) b_2(\omega_2) + a_2(\omega_2) b_1(\omega_1)}{b_1(\omega_1) + b_2(\omega_2)}  \;,
\label{eq:odwifgkyyug}
\eeq
whose standardization is simply achieved by substituting $\vsmall\gets{0}$. No numerical information is provided on the torques with this method.

This is a satisfactory solution when elimination of impulsive variables is practical. In our example, they entered linearly in the restart system, so that elimination was straightforward. When this is not the case, elimination becomes costly or even impossible. We thus need to look for alternatives.

\paragraph{Rescaling impulsive variables}

Since $\magorder{\tau_i}=1$, we apply the change of variable $\tau_i=\vsmall\hat{\tau}_i$, which transforms \rref{sys:ilerughoreiugh} into:
\beq
\left\{
\bea{lcc}
{\postset{\omega_1}-\omega_1}=a_1(\omega_1)+b_1(\omega_1)\hat{\tau}_1 &&(\eqq_1^\vsmall) \\
{\postset{\omega_2}-\omega_2}=a_2(\omega_2)+b_2(\omega_2)\hat{\tau}_2 &&(\eqq_2^\vsmall) \\
\postset{\omega_1}-\postset{\omega_2}=0 &&(\postset{\eqq_3}) \\
\hat{\tau}_1+\hat{\tau}_2=0 &&(\eqq_{4}) 
\eea
\right.
\label{sys:lriugeorliu}
\eeq
which yields again the solution (\ref{eq:odwifgkyyug}) for $\postset{\omega_1}=\postset{\omega_2}$. 
However, the rescaled value $\hat{\tau}_1$ can now be computed.
Rescaling impulsive variables is simpler than eliminating them. Unfortunately, this method does not work in full generality since impulse orders can be infinite as  we have shown in (\ref{weoriuthlkui}).
The last method addresses such cases, at the price of a possibly poor numerical conditioning.

\paragraph{Bruteforce solving of the restart system}
Consider once more \rref{sys:wp49guhsopiu},  replace in it the infinitesimal time step $\vsmall$ by a small real positive time step $\delta>0$, and expand the derivatives as before. This yields
\beq \left\{ \bea{lcc}
  \postset{\omega_1}=\omega_1+{\delta}.(a_1(\omega_1)+b_1(\omega_1)\tau_1) &&(\eqq_1^\delta) \\
  \postset{\omega_2}=\omega_2+{\delta}.(a_2(\omega_2)+b_2(\omega_2)\tau_2) &&(\eqq_2^\delta) \\
  \postset{\omega_1}-\postset{\omega_2}=0 &&(\postset{\eqq_3}) \\
  \tau_1+\tau_2=0 &&(\eqq_{4}) \eea \right.
\label{sys:rkeuyfgeuy}
\eeq
Then, it will be proved in Theorem~\ref{eroitutrdh} of \rref{sec:restartschemes} that \emph{solving \rref{sys:rkeuyfgeuy} for its dependent variables and then discarding the values found for the impulsive variables yields a converging approximation for the states $\omega_1^+$ and $\omega_2^+$ at restart}. 
Of course, numerical conditioning for \rref{sys:rkeuyfgeuy} is likely to be less good than  for \rref{sys:lriugeorliu}, so that rescaling is recommended when impulse orders are finite.

\paragraph{Qualitative information obtained from impulse analysis}
The dynamical profile of impulsive variables yields interesting information about the nature of the impulsive behaviors. For the clutch example, we get the following profile for $\magorder{\tau_i}:\dots,0,1,0,\dots$, where the $1$ occurs at the event $\guard:\fff\ra\ttt$. This indicates a Dirac-like impulse. As a second example (not developed here), the impulse analysis of the Cup-and-Ball example with inelastic impact yields the following profile for the impulse order of the tension $\magorder{\tension}:\dots,0,1,0,\dots$, where the $1$ occurs at the event following the mode change $\guard:\fff\ra\ttt$, i.e., when \rref{sys:rkufygekyu} gets executed to evaluate restart values for velocities. This again indicates a Dirac impulse for the tension. Note that having a profile $\dots,0,k,0,\dots$ for an impulsive variable, where $k$ is a positive integer, still indicates a Dirac behavior. The amplitude does not matter: only the profile matters. Developing a systematic analysis of this kind requires further study and is left for future investigations.

\subsection{Nonstandard Structural Analysis}
\label{ptw9ugfdll}
The motivation for this section was developed in \rref{sec:standardizeclutch}.
The following two questions will be addressed:
\begin{enumerate}
	\item \label{elriguheoiu} We apply structural analysis to the nonstandard semantics. Is this justified? In the standard setting, the structural analysis was motivated by the quest for low cost criteria for almost everywhere nonsingularity of systems of algebraic equations. Can we have a similar argument here?
	\item \label{lerigui} We need to formalize and generalize the Comment\,\ref{keruif} of Section\,\ref{lwergtuiherlu}, where we discussed the difference between standardizing functions and standardizing equations.
\end{enumerate}
We begin with question\,\ref{elriguheoiu}. 
\rref{alg:ergfipuhfopiu} returns at its \rref{op:weorpifuhfpeiurevis} systems of algebraic equations of the form $H(Z,X,\vsmall)=0$ introduced in (\ref{ltrghilu}), where tuple $Z$ collects the dependent variables and tuple $X$ collects the state variables. We justify the call of the atomic action $\atomicact{Eval}(\status,H)$ by the fact that ``$H$ is structurally nonsingular''. But the justification of structural analysis was developed in Section\,\ref{epfuiowerhpfoui}, i.e., in a standard setting. 
We thus need a justification for taking the liberty of using it for nonstandard semantics as well. 
Formally, we consider nonstandard systems of equations of the form 
\beq\mbox{
\begin{minipage}{11.3cm}
$\nst{H}\left(Z,\nst{X},\vsmall\right){=}0$, where $\nst{X}{\in}\nstR^m$, $Z$ collects the dependent (nonstandard) variables, $\nst{H} {=} [H,H,H,\dots]$ is the nonstandard lifting of the standard function $H$, and $\vsmall {=} [\delta_n]$ is infinitesimal.
\end{minipage}
}
\label{qweruighp39}
\eeq
\begin{definition}
	\label{opw4u9gthler} 
Say that $\nst{H}\left(Z,\nst{X},\vsmall\right)=0$ is \emph{structurally nonsingular} if and only if so is the standard equation $H\left(Z,{X},\delta\right)=0$ in the sense of Section\,${\ref{epfuiowerhpfoui}}$, where $\delta>0$ is a standard step size. 
\end{definition}
\begin{lemma}
	\label{liughlkuygfyu} 
Lemma~$\ref{oerigtuerio}$ extends to nonstandard systems of the form $(\ref{qweruighp39})$.
\end{lemma}
\begin{proof}
Select $\nst{z}=[z_n]$ and $\nst{x}=[x_n]$ satisfying $\nst{H}\left(\nst{z},\nst{x},\vsmall\right)=0$. By definition, we have $[H\left(z_n,x_n,\delta_n\right)]\sim[0,0,\dots]$, 
where relation $\sim$ was defined in (\ref{erlouiwleiurhfuil}). Equivalently,
\beq
\mu(A)=1, \mbox{ where }A=\{n\mid{H}\left(z_n,x_n,\delta_n\right){=}0\}
\label{eorgtuioeiru}
\eeq
and $\mu$ is the finitely additive measure over $\bN$ introduced in (\ref{erltuihjeui}).
Now, assume that $\nst{H}$ is {structurally nonsingular} according to Definition\,\ref{opw4u9gthler}. Then, using notations (\ref{lerigfui}) of  Section~{\ref{epfuiowerhpfoui}}, for every $n\in{A}$, there exist $\mathbf{U}_n\eqdef\regularU{H_n}{z_n,x_n,\delta_n}$ and $\mathbf{V}_n\eqdef\regularV{H_n}{z_n,x_n,\delta_n}$ satisfying property (\ref{leiuthpieu}) with respect to the system ${H}\left(Z_n,X_n,\delta_n\right){=}0$.

Define the internal sets $\mathbf{U}=[\mathbf{U}_n]$ and \mbox{$\mathbf{V}=[\mathbf{V}_n]$}. We have \mbox{$\nst{\lambda}(\mathbf{V}\setminus\mathbf{U})=0$}, where $\nst{\lambda}$ is the internal set function $\nst{\lambda}=[\lambda,\lambda,\lambda,\dots]$. For every $n\in{A}$, every $\cJ_n\in{\mathbf{V}_n}$ yields an invertible matrix, implying that $\cJ\eqdef[\cJ_n]$ also yields an invertible matrix. As a result, we have exhibited two sets $\mathbf{U}$ and $\mathbf{V}$ satisfying property (\ref{leiuthpieu}) with respect to the system $\nst{H}\left(Z,\nst{X},\vsmall\right)=0$.
\end{proof}

Definition~\ref{opw4u9gthler} thus formally justifies our shameless use of the structural analysis in nonstandard domains. Note that Definition~\ref{opw4u9gthler} tells nothing about how $\mathbf{U}$ is---it could be infinitely small. Neither does it require uniqueness of the solutions of the considered equations.

We now move to question\,\ref{lerigui} and investigate how equations should be standardized.
\begin{theorem}[standardizing equations]
	\label {ow4u89ghsljkwrg} 
Consider a square system of the form $(\ref{qweruighp39})$, and assume that $H$ is of class $\cC^1$. The following two properties hold: 
\begin{enumerate}
	\item \label{leriuhtpwui} If the standard pair $(z,x)$ satisfies $H(z,x,0){=}0$ and the Jacobian $\Jacobian_{\!Z}H$ at the point $(z,x,0)$ is nonsingular, then, for every nonstandard pair $(\nst{x},\vsmall)$ such that $({\nst{x},\vsmall})\approx(x,0)$, there exists $\nst{z}$ such that $\st{\nst{z}}{=}z$ and $\nst{H}\left(\nst{z},\nst{x},\vsmall\right){=}0$.
	
	\item \label{lerituherilut} Vice-versa, let $(\nst{z},\nst{x},\vsmall)$ be a finite nonstandard triple satisfying $\nst{H}\left(\nst{z},\nst{x},\vsmall\right)=0$, and let $(z,x,0)=\st{\nst{z},\nst{x},\vsmall}$ be its standard part. Then $H(z,x,0)=0$ holds.
\end{enumerate}
Under these assumptions, the standardization of the equation $\nst{H}\left(\nst{z},\nst{x},\vsmall\right){=}0$ is independent from the particular value for $\vsmall$ we have chosen (provided that it is infinitesimal).
\end{theorem}
\begin{corollary}
	\label{erktuieghoriugfter} We can standardize $\nst{H}(Z,X,\vsmall){=}0$ as the system $H(Z,X,0){=}0$, provided that the latter is structurally nonsingular.
\end{corollary}
The nonsingularity condition is needed for statement~\ref{leriuhtpwui} to hold. We proceed to the proof of Theorem~\ref{ow4u89ghsljkwrg}.
\begin{proof} We successively prove the two statements.
For {Statement~\ref{leriuhtpwui}}, by the Implicit Function Theorem applied to the standard function $H$, there exists a $\cC^1$-class function $(x',\delta)\mapsto{G}(x',\delta)$ such that $H(G(x',\delta),x',\delta){=}0$ holds for $(x',\delta)$ ranging over a neighborhood $V$ of $(x,0)$. Let $(\nst{x},\vsmall)=([x_n],[\delta_n])$ be such that $\st{\nst{x},\vsmall}{=}(x,0)$. Then, by definition of the standard part, there exists $B\subseteq{\bN}$ with $\mu(B)=1$ such that, for every $n\in{B}$, $({x_n},\delta_n)\in{V}$. The sequence $(G({x_n},\delta_n))_{n\in{B}}$ satisfies 
\beq
\mbox{$H(G({x_n},\delta_n),{x_n},\delta_n){=}0$, for every $n\in{B}$.}
\label{wleiuthupo}
\eeq
We extend the sequence $(G({x_n},\delta_n))_{n\in{B}}$ to $\bN$ by assigning arbitrary values to indices $n \in \bN \setminus B$ and consider the nonstandard $\nst{z}\eqdef[G({x_n},\delta_n)]$. By (\ref{wleiuthupo}), we get $\nst{H}\left(\nst{z},\nst{x},\vsmall\right){=}0$.

{For Statement~\ref{lerituherilut}}, let the finite triple $(\nst{z},\nst{x},\vsmall)$ satisfy $\nst{H}\left(\nst{z},\nst{x},\vsmall\right)=0$ and let $(z,x,0){=}\st{\nst{z},\nst{x},\vsmall}$. Since $(z,x,0){\approx}({\nst{z},\nst{x},\vsmall})$ and $H$ is of class $\cC^1$, we can apply Statement~\ref{ow4u9ghspsr} of Theorem~\ref{perqoui} to derive $H(z,x,0)\approx\nst{H}\left(\nst{z},\nst{x},\vsmall\right)=0$, whence $H(z,x,0)=0$ follows since $H(z,x,0)$ is standard.
\end{proof}

\subsection{DAE of index zero}
Invoking the implicit function theorem, one can also address DAEs of index $0$, i.e., of the form:
\beq
F(\dot{X},X)=0 
\label{whi0gojevfnh}
\eeq
where the Jacobian $\nabla_{\!V}F(V,X)$ is invertible in a neighborhood of the solution of the equation $F(V,X)=0$, where $V$ is the vector of dependent variables. 
\begin{theorem}
	\label {gu9htrkjwr} 
	Let $F(\dot{X},X)=0$ be a structurally nonsingular DAE system of index $0$, and let $\nst{F}=[F,F,F,\dots]$ be the nonstandard lifting of $F$. Consider the following nonstandard transition system having $\postset{X}$ as dependent variables:
\beq
0=\nst{F}\left(\frac{\postset{{X}}-{X}}{\vsmall}\,,\,{X}\right)\,, X(0)=X_0
\label {gfruilouigre}
\eeq
Then, system $(\ref{gfruilouigre})$ possesses a unique solution and we denote by $\nst{X}(\nstime),\nstime{\in}\nst{\bR}_+$ its piecewise constant interpolation.
Furthermore, the solution $X$ of DAE system $F(\dot{X},X){=}0, X(0){=}X_0$ satisfies $X = \st{\nst{X}}$.
\end{theorem}
\begin{proof}
Recall that the time basis of nonstandard semantics is $\bT=\{k{\times}\vsmall\mid{k}{\in}\nst{\bN}\}$. Denoting by $\nstime_k$ the generic element of $\bT$, transition system (\ref{gfruilouigre}) rewrites:
\beq
\nst{F}\left(\frac{{X}(\nstime_{k+1})-{X}(\nstime_{k})}{\vsmall}\,,\,{X}(\nstime_{k})\right)=0\,.
\label {ow34w34w34w34uip}
\eeq
Using the Implicit Function Theorem, let $G$ be the standard function such that $\dot{X}=G(X)$ solves (\ref{whi0gojevfnh}) for $\dot{X}$, and set $\nst{G}=[G,G,G,\dots]$. Then, (\ref{ow34w34w34w34uip}) rewrites
\beq
{X}(\nstime_{k+1}) = {X}(\nstime_{k})+\vsmall.\nst{G}({X}(\nstime_{k})),~X(0)=X_0\,.
\label {wiofeugfybpl}
\eeq
Theorem~\ref{gu9htrkjwr} follows from Theorem~\ref{0pw489gthtrgu} applied to (\ref{wiofeugfybpl}).
\end{proof}
In the next two sections we investigate the standardization of the blocks returned by \rref{alg:ergfipuhfopiu}, as identified in \rref{sec:lerigfuiopu}. 
Since, in this section, we pay attention to the status standard/nonstandard, we will make this status explicit in our development whenever needed.

\subsection{Long modes of \mDAE}
\label{sec:perughoriugh}
We use the notations and results of \rref{sec:lerigfuiopu}.
More precisely, we reuse notations (\ref{pweouigtndroui}) with $\prime$ replacing $\bullet$, and we invoke Lemma~\ref{erlfiuerh}. We consider a system $B_\prime(H,N)=(B_\prime^\tail(H,N),B_\prime^\head(H,N))$ returned by \rref{alg:ergfipuhfopiu} for long modes. By construction of \rref{alg:ergfipuhfopiu}, the algebraic system $G_\Sigma=0$, where $G_\Sigma\eqdef{B_\prime^\tail(H,N)}$, is structurally nonsingular and admits ${\bar{G}}_\Sigma=0$, where ${\bar{G}}_\Sigma\eqdef{B_\prime^\head(H,N)}$, as associated consistency conditions. 

At this point, we need to relate the above blocks to their standard continuous-time counterparts. The reduced index systems $(G_\Sigma,\bar{G}_\Sigma)$ returned by \rref{alg:lkciyughsliu} ($\atomicact{IndexReduc}$) in case of success, coincide with the two systems returned by the \sigmamethod, when replacing differentiation by shifting (see the emphasized sentence when introducing this algorithm). In our study of standardization, we will also need to consider the original outcome of the \sigmamethod\ (using differentiation); we denote this pair by $(G^c_\Sigma,{\bar{G}}^c_\Sigma)$, where the superscript $^c$ refers to continuous time.
Continuous-time systems $(G^c_\Sigma,{\bar{G}}^c_\Sigma)$ define the leading equations and consistency conditions of a (standard) DAE system that is generically regular. 
In the following theorem, we denote by $X=(X(t))_{t{\in}\bR_+}$ the solution of this DAE system. 

Since we are concerned with standardization, we carefully distinguish between $H$, a standard vector function, and $\nst{H}=[H,H,\dots]$, its nonstandard lifting.
\begin{theorem}
	\label{wlpguiohnl} 
The dAE system $B_\bullet^\tail(\nst{H},N){=}0$, with consistency conditions $B_\bullet^\head(\nst{H},N){=}0$, possesses a unique solution. Its piecewise constant interpolation, denoted by $\nst{X}(\nstime),\nstime{\in}\nst{\bR}_+$,  satisfies $X=\st{\nst{X}}$. Note that this standardization is independent from the particular infinitesimal time step $\vsmall$.
\end{theorem}
Theorem~\ref{wlpguiohnl} formalizes the practical informal rule that 
\begin{equation}
\mbox{\begin{minipage}{11.5cm}
	 for long modes, the standardization of 
	 $B_\bullet^\tail(H,N){=}0$ with consistency conditions $B_\bullet^\head(H,N){=}0$ is 
	 $B_\prime^\tail(H,N){=}0$ with consistency conditions $B_\prime^\head(H,N){=}0$,
\end{minipage}}
\label{3hp59ghio}
\end{equation}
where 
\beq\bea{rcl}
B_\prime^{\head}(f,{\prime}n)&\eqdef&\left[\bea{l} f \\
\dot{f} \\ \,\vdots \\ \pprime{(n-1)}{f} 
\eea\right] \mbox{ and}
\\ [4mm]
 B_\prime^{\tail}(f,{\prime}n)&\eqdef& ~~~\pprime{n}{f} \enspace ,
\eea
\label{rliufgeui}
\eeq
whereas $B_\prime^\tail(H,N)$ and $B_\prime^\head(H,N)$ are constructed from (\ref{rliufgeui}) as in Notations~\ref{leirughliu}.
\begin{proof}
The theorem is proved by componentwise induction on $N$, with the help of the Lemma~\ref{wo487ghlseih} to follow.
\end{proof}
\begin{lemma}
	\label {wo487ghlseih} Consider a structurally nonsingular nonstandard dAE system of the form
	\beq\bea{rl}
	\nst{H}(\postset{X},\postset{Y},X,Y) \!\!\! &\eqdef \left[\bea{c}
	\nst{G}(\postset{X},\postset{Y},X,Y) \\ \nst{g}({X}) \\ \nst{g}(\postset{X})
	\eea\right] 
	 \\ &= 0
	\eea
	\label{wp8tughjup}
	\eeq
with leading variables $\postset{X},\postset{Y}$,	and consider 
	\beq\bea{rl}
	\nst{{\widehat{H}}}(\postset{X},\postset{Y},X,Y)  \!\!\! &\eqdef \left[\bea{c}
	\nst{G}(\postset{X},\postset{Y},X,Y) \\ \nst{g}({X}) \\ \nst{(\Jacobian{g})}(X).\dot{X}
	\eea\right]
	 \\ &=0
	 \eea
	\label{p458ghfjkl}
	\eeq
	 where $\Jacobian{g}(X)$ is the Jacobian of the function $g$ at $X$.
	Let $(\nst{X}(\nstime),\nst{Y}(\nstime)),{\nstime\in\nst{\bR}_+}$, and $(\nst{{\widehat{X}}}(\nstime),\nst{{\widehat{Y}}}(\nstime)),{\nstime\in\nst{\bR}_+}$, be the piecewise constant interpolations of the solutions of nonstandard dAE systems $(\ref{wp8tughjup})$ and $(\ref{p458ghfjkl})$, respectively.
	Assume that we know a standard DAE block 
	\[
	\widehat{G}(\dot{X},\dot{Y},X,Y)=0
	\]
	 such that the solution $(X(t),Y(t))_{t\in\bR}$ of 
		\beq\bea{rl}
	\widehat{H}(\dot{X},\dot{Y},X,Y)  \!\!\! &\eqdef \left[\bea{c}
	\widehat{G}(\dot{X},\dot{Y},X,Y) \\ {{g}}(X) \\ {\Jacobian{g}}(X).\dot{X}
	\eea\right] 
	 \\ &= 0
	 \eea
	\label{qw89p7hlui}
	\eeq
	exists, is unique, and satisfies \mbox{$(X,Y){=}\st{{\nst{{\widehat{X}}},\nst{{\widehat{Y}}}}}$}. Then, we also have $(X,Y){=}\st{\nst{{{X}}},\nst{{{Y}}}}\,.$
\end{lemma}
\begin{proof}
	 By statement~\ref{ouiepuicksv} of Theorem~\ref{perqoui}, we have	 
\beqq
\nst{g}(\postset{X}) &=& \nst{g}({X})+\vsmall.\nst{(\Jacobian{g})}(X).\left(
\dot{X}+\varepsilon
\right) 
\eeqq
where $\varepsilon$ is infinitesimal. Hence $\nst{H}$ rewrites
\beqq
\nst{H}(\postset{X},\postset{Y},X,Y) 
&=& \left[\bea{c}
	\nst{G}(\postset{X},\postset{Y},X,Y) \\ \nst{g}({X}) \\ \nst{g}({X})+\vsmall.\nst{(\Jacobian{g})}(X).\left(
\dot{X}+\varepsilon
\right) 
	\eea\right] 
\eeqq
Since $\nst{H}(\postset{X},\postset{Y},X,Y)=0$ implies in particular $\nst{g}({X})=0$, (\ref{wp8tughjup}) rewrites
\beqq
\left[\bea{c}
	\nst{G}(\postset{X},\postset{Y},X,Y) \\ \nst{g}({X}) \\ \nst{(\Jacobian{g})}(X).\left(
\dot{X}+\varepsilon
\right) 
	\eea\right] = 0
\eeqq
and removing the infinitesimal $\varepsilon$ does not change the standardization of this dAE system. We conclude by invoking the assumption of the lemma.
\end{proof}

\subsection{Mode changes of \mDAE}
\label{weroguihergu}
 Assumption\,\ref{oergfiuehoiu} (stated on page~\pageref{oergfiuehoiu}) is in force throughout this section.
We study the standardization of mode change events. Having successfully performed the nonstandard structural analysis, the system at mode changes can be put in the generic form (\ref{eroifughoiu}), and we use the decomposition (\ref{psergtiohjpio}) of the algebraic variables into impulsive and non-impulsive ones. For convenience, we repeat these formulas here:
\beqq\bea{cccc}
(\ref{eroifughoiu}):\dot{X}=\frac{\postset{X}-X}{\vsmall}\mbox{ in }
0=\bfH(\dot{X},\postset{X},V,X)
&
&&
(\ref{psergtiohjpio}): \quad V =\impuls{W}\,\cup\,\nonimpuls{W} 
\eea
\eeqq
Note that, in (\ref{eroifughoiu}), the infinitesimal parameter $\vsmall$ is involved in both time (by being the time step of $\bT$) and space (since derivatives $\dot{x}$ in the original standard model are expanded as $\frac{\postset{x}-x}{\vsmall}$). Concrete instances of this are found in the clutch example, e.g., in \rref{sys:nscoupledshaftsbroken}.

Compared with the previous section that dealt with long modes, the novelty, here, is that we take a different target domain for the standardization, namely the domain of \emph{discrete-time} transition systems---standard transition systems indexed by $\bN$. 
Clearly, we cannot expect mapping to this domain the entire trajectory of System~(\ref{eroifughoiu}). We only aim at standardizing the finite prefix of it needed to specify the restart conditions for the new mode. 
The ``$\bullet$'' shift operator defines the discrete time of this dynamics, hence we do not need to eliminate $\vsmall$ from the time domain. In contrast, we will have to eliminate the occurrences of the infinitesimal parameter $\vsmall$ acting in space. 

Consider again the clutch example and in particular \rref{sys:wp49guhsopiu}, used to determine the restart values $\postset{\omega_1},\postset{\omega_2}$ and the current torques $\tau_1,\tau_2$ from the current values for the states ${\omega_1},{\omega_2}$. The occurrence of $\vsmall$ in the right-hand sides of equations $(e_1^\vsmall),(e_2^\vsmall)$ is the difficulty, since $\vsmall$ is nonstandard---see Comment\,\ref{keruif} regarding this. 

The standardization of (\ref{eroifughoiu}) relies on Theorem~\ref{ow4u89ghsljkwrg}. Unfortunately, it is not true in general that substituting $0$ for $\vsmall$ in (\ref{eroifughoiu}) yields a structurally nonsingular system with respect to the dependent variables $Z\eqdef(\postset{X},V)$. We thus need further work to be able to invoke Theorem~\ref{ow4u89ghsljkwrg}. This will consist in eliminating, from System\,(\ref{eroifughoiu}), the impulsive variables. For the example of the clutch, the corresponding analysis was developed in Section~\ref{lwergtuiherlu}. In the following, we propose a systematic approach.

To this end, we also use the rewriting of (\ref{eroifughoiu}), i.e., we use both:
\beqq
(\ref{lergtiurhuip})&:& 0=\bfH(\dot{X},\postset{X},V,X) \\
(\ref{ltrghilu}) &:&0=H(Z,X,\vsmall) 
\eeqq
where we recall that the form (\ref{ltrghilu}) is obtained by expanding $\dot{X}$ as $\frac{\postset{X}-X}{\vsmall}$ in (\ref{lergtiurhuip})---the same notational convention is used in the assumption below.
Using decomposition (\ref{psergtiohjpio}) of algebraic variables into impulsive and non-impulsive ones: $V =\impuls{W}\,\cup\,\nonimpuls{W}$, and writing $$Y\eqdef(\postset{X},W) ~ ,$$ we assume the following about (\ref{eroifughoiu}):
\begin{assumption}
	\label{pw498gthseroi}  
	There exists a reduced system 
	$\bfK(\dot{X},\postset{X},\nonimpuls{W},X)=0$ with dependent variables $Y$, satisfying the following two conditions:
\begin{enumerate}
	\item System $K(Y,X,0){=}0$ is structurally nonsingular, where $K(Y,X,\vsmall)$ is derived from $\bfK(\dot{X},\postset{X},\nonimpuls{W},X)$ as explained above;
	\item The considered model is a reduction of the original one, in that, equivalently:
	\beqq
	\bfK(\dot{X},\postset{X},\nonimpuls{W},X)=0 &\Longleftrightarrow& \exists \impuls{W}.\bfH(\dot{X},\postset{X},V,X)=0
	\\
	K(Y,X,\vsmall)=0 &\Longleftrightarrow& \exists \impuls{W}.H(Z,X,\vsmall)=0
	\eeqq
\end{enumerate}
\end{assumption}
In words, we can eliminate, from the original model, the impulsive variables (belonging to $\compl{W}$), and the result is structurally nonsingular for $\vsmall=0$. We constructed such a reduced system in Section~\ref{lwergtuiherlu} by eliminating the torques for the released-to-engaged mode change in the clutch example.
\begin{theorem}
	\label{owgu895hlskjrg} Assumptions~$\ref{oergfiuehoiu}$ and~$\ref{pw498gthseroi}$ are in force. Then, standardizing System~$(\ref{eroifughoiu})$ is performed by solving, for $Y$, the structurally nonsingular system $K(Y,X,0)=0$. Note that this standardization is independent from the particular infinitesimal time step $\vsmall$.
\end{theorem}
Note that the latter system is standard.
\begin{proof}
	Theorem~\ref{owgu895hlskjrg} is a consequence of Theorem~\ref{ow4u89ghsljkwrg}.
\end{proof}
Theorem~\ref{owgu895hlskjrg} allows us to compute the restart values for the state variables that are not impulsive, which is sufficient to restart the dynamics in the coming long mode (impulsive variables at mode change cannot be initial conditions for state variables in the new mode).

We have shown in Section~\ref{lwergtuiherlu} that Assumption~\ref{pw498gthseroi} holds for the clutch example. In general, the method would consist in, first, computing the decomposition (\ref{psergtiohjpio}), and, then, eliminating impulsive variables.
\rref{fig:clutch-standard} shows the effective code resulting from the
standardization of the labeled automaton of \rref{fig:clutch-graph}. 
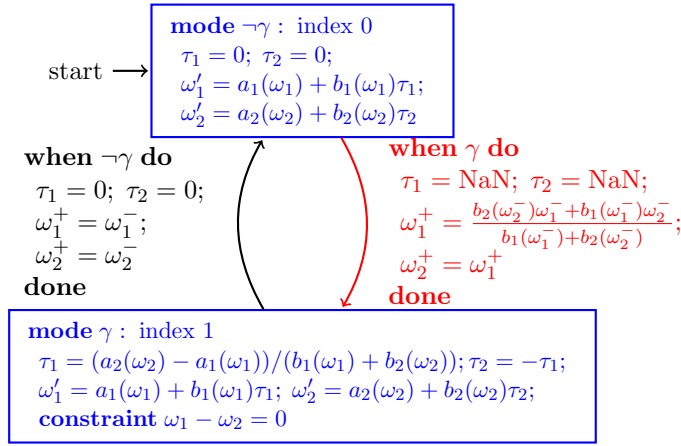
\begin{figure}
	\centerline{
		\begin{tikzpicture}
		[
                shorten >=1pt,node distance=5cm,auto,
                every node/.style={scale=.9},
                every edge/.style={thick,draw},
		state/.style={rectangle,scale=.9,shape=rectangle,thick,draw}
               ]
		\node[state,initial,blue] (q0) {$\begin{array}{l}\mbox{\textbf{mode}} \; \neg \gamma : \; \mbox{index 0} \\ \gencodea\end{array}$};
		\node[state,blue] (q1) [below of=q0] {$\begin{array}{l}\mbox{\textbf{mode}} \; \gamma : \; \mbox{index 1} \\ \gencodec\end{array}$};
		\path[->] (q0) edge[red, bend left] node {$\begin{array}{l}\mbox{\textbf{when}} \; \gamma \; \mbox{\textbf{do}} \\ \gencodeb \\ \mbox{\textbf{done}}\end{array}$} (q1)
                                (q1) edge[bend left] node {$\begin{array}{l}\mbox{\textbf{when}} \; \neg \gamma \; \mbox{\textbf{do}} \\ \gencoded \\ \mbox{\textbf{done}}\end{array}$} (q0);
		\end{tikzpicture}
	}
	\caption{Standardization of $\atomicact{ExecRun}$ for the ideal clutch (\rref{fig:clutch-graph}). 
{Blocks have been standardized and
              then symbolically pivoted. $\mlast{x}$ is
              the previous value of state variable $x$, i.e.,
              the left limit of $x$ when exiting a mode; $x^+$ is the value for restart. Long modes are colored blue; non-impulsive
              (resp. impulsive) state-jumps are colored black
              (resp. red). The dynamics in mode $\neg\gamma$ is
              defined by an ODE system, while in mode $\gamma$, it is
              defined by an over-determined index-1 DAE system
              consisting of an ODE system coupled to an algebraic
              constraint. In the transition from mode $\neg\gamma$ to
              mode $\gamma$, variables $\tau_1$ and $\tau_2$ are
              impulsive, and their standardization is undefined. This
              explains why they are set to $\nan$ (Not a
              Number).}
              }\label{fig:clutch-standard}
\end{figure}

To get practical, Assumption~\ref{pw498gthseroi} relies on computer algebraic techniques for eliminating variables from systems of equations. Algebraic elimination is easy for systems in which impulsive variables enter linearly. If they enter polynomially, Gr\"obner bases can be used, but at a high computational cost. Effective elimination techniques may not exist in general, or may be too costly. Hence, the alternative based on the numerical scheme developed in \rref{sec:restartschemes} is of great practical interest.
\section{Main results}
\label{sec:MainResults}

This section collects results that are essential in grounding our approach. First, we prove that our approach is intrinsic, in that the actual code we generate does not depend on the particular nonstandard scheme used, instead of the explicit first-order Euler scheme, to expand the derivative.
Then, we compare our approach with known schemes, for subclasses for which they are known to exist. 
Finally, we propose a numerical scheme to compute restart values at mode changes, in full generality, for programs accepted by our structural analysis. This numerical scheme does not require eliminating impulsive variables using computer algebra.

\subsection{What if we change the interpretation of the derivative?}
\label{sec:otherschemes}
So far, an essential step in our reasoning was the interpretation (\ref{whgp45g8hseluio}) for the derivative $\dot{\omega}$. However, several nonstandard interpretations of the derivative can be equally considered. What happens if we take an interpretation different from  (\ref{whgp45g8hseluio})? We first discuss this on the clutch example, then provide a general result.

\subsubsection{The clutch example}
\label{frkfuiwgui}
If the system is within a long continuous mode, the following alternative expansion exists for the derivative of $\omega$: for $n\in\bN$ finite but otherwise arbitrary, 
$
\dot{\omega} \approx \frac{\ppostset{n}{(\postset{\omega}-\omega)}}{\vsmall} 
$.
More generally, for 
$\mathbf{a}{=}(\alpha_P,\dots,\alpha_{N})$ such that $P{\leq}N$, $N\geq{0}$, $\alpha_{N}\neq{0}$, and  $\sum_{n=P}^{N}\alpha_n=1$,  one has:
\beq
\dot{\omega} &\approx& \!\!\frac{1}{\vsmall}\sum_{n=P}^{N}\alpha_n{\ppostset{n}{(\postset{\omega}-\omega)}} \enspace .
\label{ergtioeroi}
\eeq
Let us investigate how the analysis of the clutch is modified when using the nonstandard interpretation 
\beq
\dot{\omega_\bfa} &\eqdef& \frac{1}{\vsmall}\sum_{n=P}^{N}\alpha_n{\ppostset{n}{(\postset{\omega}-\omega)}}
\label{erlgtuiohou}
\eeq
for the derivative. The nonstandard semantics of the  clutch model is:
\beq
\left\{
\bea{rllcc}
&&\dot{\omega_{\bfa{1}}}=f_1(\omega_1,\tau_1) &&(\eqq_1^\vsmall) \\
[1mm]
&&\dot{\omega_{\bfa{2}}}=f_2(\omega_2,\tau_2) &&(\eqq_2^\vsmall) \\
[1mm]
\when\;\guard&\doo&{\omega_1}-{\omega_2}=0 &&({\eqq_3}) \\
&\aand&\tau_1+\tau_2=0 &&(\eqq_{4}) \\
\when\;\prog{not}\;\guard&\doo&\tau_1=0 &&(\eqq_{5}) \\
&\aand&\tau_2=0 &&(\eqq_{6}) \\
\eea
\right.
\label{oiweruhowieu}
\eeq
\rref{alg:ergfipuhfopiu} tells us that we must shift forward $({\eqq_3})$ by $N{+}1$ and keep other equations unchanged; other latent equations $(\ppostset{n}{\eqq_3})$ for $n\leq{N}$ are erased when solving the conflict with the previous mode:
\beq
\left\{
\bea{rllcl}
&&\dot{\omega_{\bfa{1}}}=f_1(\omega_1,\tau_1)  &&(\eqq_1^\vsmall) \\
[1mm]
&&\dot{\omega_{\bfa{2}}}=f_2(\omega_2,\tau_2)  &&(\eqq_2^\vsmall) \\
[1mm]
\when\;\guard&\doo&\ppostset{(N+1)}{\omega_1}-\ppostset{(N+1)}{\omega_2}=0  &&(\ppostset{(N{+}1)}{\eqq_3}) \\
&\aand&\tau_1+\tau_2=0   &&(\eqq_{4})  \\
\when\;\prog{not}\;\guard&\doo&\tau_1=0  &&(\eqq_{5}) \\
&\aand&\tau_2=0  &&(\eqq_{6}) 
\eea
\right.
\label{eroiweuhyoi}
\eeq
Using $\tau_1=-\tau_2\eqdef\tau$, the resetting equations when entering the mode $\guard = \ttt$ are
\beq
\left\{
\bea{l}
\dot{\omega_{\bfa{1}}}=f_1(\omega_1,\tau) \\
[1mm]
\dot{\omega_{\bfa{2}}}=f_2(\omega_2,-\tau) \\
[1mm]
\ppostset{(N+1)}{\omega_1}-\ppostset{(N+1)}{\omega_2}=0  \\
\eea
\right.
\label{pgtoiuaoiaug}
\eeq
Assuming again that the $f_i$'s are linear in $\tau$, i.e., of the form $f_i(\omega_i,\tau)=a_i(\omega_i)+b_i(\omega_i)\tau$, (\ref{pgtoiuaoiaug}) rewrites:
\[
\left\{
\bea{lc}
\sum_{n=P}^{N}\alpha_n\ppostset{n}{(\postset{\omega_1}-\omega_1)}
=\vsmall.(a_1(\omega_1)+b_1(\omega_1)\tau) \\
[1mm]
\sum_{n=P}^{N}\alpha_n\ppostset{n}{(\postset{\omega_2}-\omega_2)}
=\vsmall.(a_2(\omega_2)-b_2(\omega_2)\tau) \\
[1mm]
\ppostset{(N+1)}{\omega_1}-\ppostset{(N+1)}{\omega_2}=0 \\
\eea
\right. \enspace .
\]
We can eliminate $\tau$:
\begin{equation}
\bea{rcl}
0&=& 
\sum_{n=P}^{N}\alpha_n\bigl(b_2(\omega_2)\ppostset{n}{(\postset{\omega_1}-\omega_1)} 
+b_1(\omega_1)\ppostset{n}{(\postset{\omega_2}-\omega_2)}\bigr) 
\\ [1mm] &&
-\vsmall.(a_1(\omega_1)b_2(\omega_2)+a_2(\omega_2)b_1(\omega_1))  \\ 
[2mm]
0&=&\ppostset{(N+1)}{\omega_1}-\ppostset{(N+1)}{\omega_2} \\
\eea
\label{eiopgfuehpiu}
\end{equation}
We can zero the term in $\vsmall$ in the first equation: doing so keeps the system structurally nonsingular, hence it is safe. It then remains to standardize
\begin{equation}
\bea{rcl}
0&=& 
\sum_{n=P}^{N}\alpha_n\bigl(b_2(\omega_2)\ppostset{n}{(\postset{\omega_1}-\omega_1)} +b_1(\omega_1)\ppostset{n}{(\postset{\omega_2}-\omega_2)}\bigr) \\
[2mm]
0&=&\ppostset{(N+1)}{\omega_1}-\ppostset{(N+1)}{\omega_2} \\
\eea
\label{epgtuiohpeui}
\end{equation}
and, finally, solve it for the dependent variable $\ppostset{(N+1)}{\omega}\eqdef\ppostset{(N+1)}{\omega_1}=\ppostset{(N+1)}{\omega_2}$. For $n=P,\dots,N$, $\ppostset{n}{\omega_i}$ standardizes as the left-limit $\omega^-_i$ for $i=1,2$, whereas $\ppostset{(N+1)}{\omega}$ standardizes as the right-limit $\omega^+$. Hence, (\ref{epgtuiohpeui}) standardizes as
\[
0=\alpha_N\bigl(b_2(\omega^-_2){({\omega^+}-\omega^-_1)}+b_1(\omega^-_1){({\omega^+}-\omega^-_2)}\bigr)
\]
which does not depend on $\alpha_N$ since $\alpha_N\neq{0}$. Finally, we get 
$$
{\omega^+}{=}\frac{
b_2(\omega^-_2){\omega^-_1}+b_1(\omega^-_1){\omega^-_2}
}{
b_1(\omega^-_1)+b_2(\omega^-_2)
}\;,
$$
which  coincides with (\ref{eq:owgtuiho}).

\subsubsection{A general result}
In this section, we consider again the generic nonstandard interpretation~(\ref{erlgtuiohou}) for the derivative. We investigate under which assumptions the result of our whole compilation suite is independent from the particular way derivatives are expanded when mapping the considered \mDAE\ system to its nonstandard semantics.
We consider the following assumption on $\system$, thus characterizing a restricted class of \deltaAE:
\begin{assumption}
	\label{lirtylweiru} 
\mDAE-system $\system$ involves only long modes, no transient modes.
\end{assumption}
Let us compare, under Assumption~\ref{lirtylweiru} and Comment~\ref{klgfuio}, the results of \rref{alg:newmain} on $\system$ for the following two nonstandard semantics for the derivative:
\beq\bea{rcl}
{\vsmall}.\dot{x} &\!\!=\!\!& \postset{x}-x 
\\
\postset{x} &\!\!=\!\!& \vsmall.\dot{x}+x
\eea
\label{wotgihio} 
\eeq
and

\beq\bea{rcl}
\vsmall.\dot{x_\bfa} &\!\!\!\!=\!\!& \!\!\sum_{n=P}^{N}\alpha_n{\ppostset{n}{(\postset{x}-x)}}
\\ [2mm]
{\alpha_{N}}.\ppostset{(N+1)}{x} &\!\!\!\!=\!\!& 
\vsmall.\dot{x_\bfa} +\!\!\sum_{n=P}^{N}(\alpha_{n}-\alpha_{n-1})\ppostset{n}{x}
\eea
\label{lerigheil}
\eeq
where $\mathbf{a}=(\alpha_P,\dots,\alpha_{N})$ satisfies $P{\leq}N$, $N{\geq}0$, $\alpha_{N}{\neq}{0}$, and  $\sum_{n=P}^{N}\alpha_n=1$, 
whereas $\alpha_{P-1}{=}0$ by convention. Note that nothing prevents $P$ from being negative.
For the two formulas, we have given both the direct and inverse forms. Equations (\ref{wotgihio}) are obtained by taking $P{=}N{=}0$ and $\alpha_N{=}1$ in (\ref{lerigheil}).

\subsubsection*{Comparing the structural analyses}

Let $\system$ be an \mDAE\ and let $H_\system$ be the system $H$ returned by \rref{alg:newmain} at \rref{op:elkfuiehoeuio}, i.e., the outcome of index reduction, for both a long mode and a mode change event. 

We will consider the following problem: how should $H_\system$ be modified to become the outcome of the index reduction of $\system$, when derivatives are expanded using~(\ref{lerigheil}) instead of the first-order Euler expansion~(\ref{wotgihio})? We denote by $\nst{\system}_\bfa$ and $\nst{\system}$ the corresponding two nonstandard semantics for $\system$.
Regarding this question, two lines of \rref{alg:newmain} are important: \rref{op:porihujrpoi} and \rref{op:elkfuiehoeuio}. We investigate them successively.

\paragraph{\rref{op:porihujrpoi} of \rref{alg:newmain}} 
This consists in applying the \sigmamethod\ to the original \mDAE\ system $\system$, within a given mode; let $G{=}0$ be the dynamics of $\system$ in this mode. The outcome of this algorithm only depends on the weighted bipartite graph $\cG$ associated to  $G$. In the original \sigmamethod, the weight of an edge $(f,x)$ of graph $\cG$ is equal to the highest differentiation degree of $x$ in $f$. Similarly, if $G$ is a dAE system (in discrete time), the weight of an edge $(f,x)$ of $\cG$ is equal to the highest shifting degree of $x$ in $f$. Having recalled this, we can now compare the two expansions:
\begin{itemize}
	\item \emph{Using expansion $(\ref{wotgihio})$:} 
\begin{enumerate}
	\item We begin with DAE system $G$ and its bipartite graph $\cG$;
	\item Replacing DAE system $G$ by dAE system $\nst{\,G}$ using expansion $(\ref{wotgihio})$ yields ${\nst{\,\cG}}$, that is identical to $\cG$.
\end{enumerate}
	\item \emph{Using expansion $(\ref{lerigheil})$:} 
\begin{enumerate}
	\item We begin with DAE system $G$ and its bipartite graph $\cG$;
	\item Replacing DAE system $G$ by dAE system $\nst{\,G}$ using expansion $(\ref{lerigheil})$ yields ${\nst{\,\cG}_\bfa}$, in which the vertices and edges are exactly those in $\cG$ and all weights have been multiplied by $$M \eqdef N{+}1 \enspace .$$ 
\end{enumerate}
\end{itemize}
The case $N=0$ is the base case of the forward first-order Euler scheme, so that we do not need to consider it. Therefore, in the sequel, we assume $N>0$, i.e., $M>1$.

Using the notations of~(\ref{erilughroghwgouyg}), let $(c_f)_{f{\in}{F}}$ and $(c_f^\bfa)_{f{\in}{F}}$ be the offsets associated to ${\nst{\,\cG}}$ and ${\nst{\,\cG}_\bfa}$. Then, by Lemma~\ref{rilugtlo}, one gets:
\begin{lemma} \label{trghiopio}
	 The \sigmamethod\ succeeds on ${\nst{\,\cG}}$ if and only if it succeeds on ${\nst{\,\cG}}_\bfa$, and, in this case, $c_f^\bfa=M\times{c_f}$ for every $f{\in}{F}$ and $d_x^\bfa=M\times{d_x}$ for every $x{\in}X({F})$.
\end{lemma}
In particular, each guarded equation in $\nst{\system}_\bfa$ corresponds to an equation in $\nst{\system}$ shifted $M$ times.  
We denote by 
\beq
\mbox{$(G_\Sigma,\consistency{G}_\Sigma)$ and $(G^\bfa_{\Sigma},\parconsistency{G}^\bfa_{\Sigma})$} \label{rtyiohjeproi}
\eeq 
the corresponding two returns of the \sigmamethod.

\paragraph{\rref{op:elkfuiehoeuio} of \rref{alg:newmain}} If the current instant sits within a long mode, then this line just erases the consistency conditions $\consistency{G}_\Sigma$ and $\parconsistency{G}^\bfa_{\Sigma}$. So, we are left with $G_\Sigma$ if we use expansion (\ref{wotgihio}), and $G^\bfa_{\Sigma}$ if we use expansion (\ref{lerigheil}). Now, dividing the infinitesimal step $\vsmall$ by $M$ in $G_\Sigma$ yields $G^\bfa_{\Sigma}$. Theorem~\ref{wlpguiohnl} tells us that the standardizations of the two systems $G_\Sigma$ and $G^\bfa_{\Sigma}$ yield the same DAE system. 

As such, the difficulty is the handling of mode changes by \rref{op:elkfuiehoeuio}, when the two different expansions are used. Is the handling of conflicts preserved, when the two different expansions are used? More precisely, for $\nstime_o\in\bT$ an instant of mode change:
\begin{itemize}
	\item[--] Let $G{=}0$ and $\preset{G}{=}0$ be the DAE system at the current ($\nstime\geq\nstime_o$) and previous ($\nstime<\nstime_o$) modes, respectively; we denote by $(c_f,d_x)$ and $(\previous{c_g},\previous{d_y})$ the offsets returned by the \sigmamethod\ for the current and previous mode, respectively;
	\item[--] Let $(G_\Sigma,\consistency{G}_\Sigma)$ and $(G^\bfa_{\Sigma},\parconsistency{G}^\bfa_{\Sigma})$ be the returns of the \sigmamethod\ when applied to $G{=}0$ and ${G^\bfa}{=}0$;
		\item[--] Let $(\preset{}G_\Sigma,\preset{}\consistency{G}_\Sigma)$ and $(\preset{}G^\bfa_{\Sigma},\preset{}\parconsistency{G}^\bfa_{\Sigma})$ be the returns of the \sigmamethod\ when applied to $\preset{G}{=}0$ and $\preset{G}^\bfa{=}0$;
		\item[--] We denote by $\Delta$ and $\Delta^\bfa$, and  by $\Phi$ and $\Phi^\bfa$, the contexts and sets of facts constructed by $\atomicact{ExecRun}$ if expansions (\ref{wotgihio}) and (\ref{lerigheil}) are used. 
\end{itemize}
\rref{op:elkfuiehoeuio} takes as inputs the systems $(G_\Sigma\cup\consistency{G}_\Sigma)\setminus\Phi$ and $(G^\bfa_{\Sigma}\cup\parconsistency{G}^\bfa_{\Sigma})\setminus\Phi^\bfa$. When applying $\atomicact{SolveConflict}$, the dependent variables for $(G_\Sigma\cup\consistency{G}_\Sigma)\setminus\Phi$ are all the variables of this system, both leading and non-leading, whose values were not already set by executing the previous nonstandard instant. The same holds for the dependent variables of $(G^\bfa_{\Sigma}\cup\parconsistency{G}^\bfa_{\Sigma})\setminus\Phi^\bfa$. The following result holds:
\begin{theorem}
	\label{erlituherltuih} Set $F\eqdef(G_\Sigma\cup\consistency{G}_\Sigma)\setminus\Phi$ with dependent variables defined by context $\Delta$, and $F^\bfa\eqdef(G^\bfa_\Sigma\cup\parconsistency{G}^\bfa_\Sigma)\setminus\Phi^\bfa$ with dependent variables defined by context $\Delta^\bfa$. Then, the mappings	
	\[\bea{lcl}
	\psi_f:F^\bfa\mapsto{F}&\emph{defined by}& \psi_f\left(\ppostset{M(c_f-m)}{f}\right)=\ppostset{(c_f-m)}{f}
	\\ [1mm]
	\psi_x:X(F^\bfa)\mapsto{X(F)} &\emph{defined by}& 
	\psi_x\left(\ppostset{M(d_x-m)}{x}\right)=\ppostset{(d_x-m)}{x}
	\eea
	\]
	preserve the Dulmage-Mendelsohn decompositions of $\cG_F$ and $\cG_{F^\bfa}$, i.e.,
	$g\in\block^\bfa_{\overapprox/\squared/\underapprox}$ if and only if $\psi_f(g)\in\block_{\overapprox/\squared/\underapprox}$, and	$y\in\block^\bfa_{\overapprox/\squared/\underapprox}$ if and only if $\psi_x(y)\in\block_{\overapprox/\squared/\underapprox}$.
\end{theorem}
The proof of this theorem decomposes into a series of lemmas.
Let $f{=}0$ range over the set of equations of system $G{=}0$, and let $n=c_f$ be the equation offset computed for $f$ when applying the \sigmamethod\ to $G{=}0$. 
\begin{lemma}
	\label{elruioolguy} \ 
\begin{enumerate}
	\item The following formulas hold:
\begin{equation}\hspace*{-6mm}\bea{rclcrclc}
G_\Sigma\cup\consistency{G}_\Sigma &\!\!\!=\!\!\!& \displaystyle\bigcup_{f\in{G}}B(f,c_f)
&\!\!\mbox{and}\!\!&
G^\bfa_\Sigma\cup\consistency{G}^\bfa_\Sigma &\!\!\!=\!\!\!& \displaystyle\bigcup_{f\in{G}}B(f,M{c_f}) &(a)
\\ [5mm]
\Delta&\!\!\!=\!\!\!&\displaystyle\bigcup_{f\in\preset{G}} B(f,\previous{c_f}-1)
&\!\!\mbox{and}\!\!&
\Delta^\bfa&\!\!\!=\!\!\!&\displaystyle\bigcup_{f\in\preset{G}} B(f,M\,\previous{c_f}-1)  &(b)
\\ [5mm]
\Phi&\!\!\!=\!\!\!&\displaystyle\bigcup_{f\in\preset{G}\cap{G}}B(f,k_f)
&\!\!\mbox{and}\!\!&
\Phi^\bfa&\!\!\!=\!\!\!&\displaystyle\bigcup_{f\in\preset{G}\cap{G}}B(f,K_f)
 &(c) \\ [5mm]
\mbox{where } k_f&\!\!\!=\!\!\!&\min(\previous{c_f}-1,c_f)&\!\!\mbox{and}\!\!&K_f&\!\!\!=\!\!\!&\min(M\,\previous{c_f}-1,Mc_f)  &(d)
\eea
\label{toguihukuh}
\end{equation}
\item If $\previous{c_f}\leq c_f$, then $k_f=\previous{c_f}-1 \mbox{ and } K_f=M\,\previous{c_f}-1$; 

If $\previous{c_f}> c_f$, then $k_f=c_f \mbox{ and } K_f=Mc_f$.
\end{enumerate}
\end{lemma}
\begin{proof} For the first statement, 
	 (\ref{toguihukuh}-$a$) is a direct consequence of Lemma~\ref{trghiopio} and Lemma~\ref{erlfiuerh} from \rref{sec:lerigfuiopu}. (\ref{toguihukuh}-$b$) and (\ref{toguihukuh}-$c$) follow from Definition~\ref{eporgfuieropiufl} and Lemma~\ref{trghiopio}.
	 For the second statement, we note that, since $M>1$, $\previous{c_f}\leq c_f$ is equivalent to $\previous{c_f}-\frac{1}{M}\leq c_f$, or  $M\,\previous{c_f}-1\leq Mc_f$. The case $\previous{c_f}\leq c_f$ follows. The other case is proved similarly.
\end{proof}
In the following lemma, we use Notations~$\ref{leirughliu}$ and recall that $X(G)$ denotes the set of variables involved in the set $G$ of functions. We also introduce the following notations.
\begin{notation}
	\label{erokfuiioh}\rm 
	For $G{=}0$ a DAE system, we extend by convention its offsets $(c_f)_{f\in{G}}$ and $(d_x)_{x\in{X(G)}}$ to arbitrary $f$ and $x$ by setting $c_f=d_x=-\infty$ if $f\not\in{G}$ and $x\not\in{X}(G)$. 
	For $-\infty\leq{m}$ and $0\leq{n}$, and $\zeta$ a variable or a function, we set
	\[
	\bfB(\zeta,m,n)\eqdef\left[\bea{l} \ppostset{\widehat{m}}{\zeta} \\ \ppostset{(\widehat{m}+1)}{\zeta}
\\ \,\vdots \\ \ppostset{n}{\zeta} 
\eea\right] , \mbox{ where } \widehat{m}=\max(m,0)
	\] 
	with the convention that $\bfB(\zeta,m,n)$ is empty if $m>n$.
\end{notation}

\begin{lemma}
	\label{guhoiuy} Set $F\eqdef(G_\Sigma\cup\consistency{G}_\Sigma)\setminus\Phi$ and $F^\bfa\eqdef(G^\bfa_\Sigma\cup\parconsistency{G}^\bfa_\Sigma)\setminus\Phi^\bfa$. Then:
\begin{enumerate}
\item $F$ and $F^\bfa$ rewrite as follows:
\beq
F=\bigcup_{f\in{G}}\bfB(f,k_f{+}1,c_f) \mbox{ and } F^\bfa=\bigcup_{f\in{G}}\bfB(f,K_f{+}1,Mc_f) \enspace .
\label{erliueipsruiu}
\eeq
	\item Let $Z$ and $Z^\bfa$ be the sets of dependent variables of $F{=}0$ and $F^\bfa{=}0$. Then:
	\beq
	Z = \bigcup_{x\in{X}(G)}\bfB(x,\previous{d_x},d_x)
	\mbox{ and }
	Z^\bfa = \bigcup_{x\in{X}(G)}\bfB(x,M\previous{d_x},Md_x) \enspace .
	\label{lrieghjhgiohlio}
	\eeq
\end{enumerate}
\end{lemma}
\begin{proof}
Obvious.
\end{proof}
If $\previous{c_f}\leq{c_f}$, then $k_f{+}1=\previous{c_f}$ and $K_f{+}1=M\previous{c_f}$; thus, $\bfB(f,k_f{+}1,c_f)$ and $\bfB(f,K_f{+}1,c_f)$ are both nonempty. If $\previous{c_f}>{c_f}$, then $k_f{+}1={c_f}{+}1$ and $K_f{+}1=M{c_f}{+}1$; thus, $\bfB(f,k_f{+}1,c_f)$ and $\bfB(f,K_f{+}1,c_f)$ are both empty. Hence, (\ref{erliueipsruiu}) rewrites:
\beq
F=\bigcup_{f\in{G}:\previous{c_f}\leq{c_f}}\bfB(f,\previous{c_f},c_f) \mbox{ and } F^\bfa=\bigcup_{f\in{G}:\previous{c_f}\leq{c_f}}\bfB(f,M\previous{c_f},Mc_f) \enspace .
\label{liuhrthiojtoij}
\eeq
Similarly, (\ref{lrieghjhgiohlio}) rewrites:
\beq
\hspace*{-4mm}	Z = \bigcup_{x\in{X}(G):\previous{d_x}\leq{d_x}}\bfB(x,\previous{d_x},d_x)
	\mbox{ and }
	Z^\bfa = \bigcup_{x\in{X}(G):\previous{d_x}\leq{d_x}}\bfB(x,M\previous{d_x},Md_x) \enspace .
	\label{leruighlrithliu}
\eeq
In words, $F^\bfa$ and its set of dependent variables are obtained, from $F$ and its set of dependent variables, via an $M$-stretching. From formulas (\ref{liuhrthiojtoij}) and (\ref{leruighlrithliu}), we deduce the following result:
\begin{lemma}
	\label{erliufheiuh} Let $f\in{G}$ and $x\in{X}(G)$. Then, $(\ppostset{k}{f},\ppostset{l}{x})$ is an edge of $\cG_F$ if and only if $(\ppostset{Mk}{f},\ppostset{Ml}{x})$ is an edge of $\cG_{F^\bfa}$.
\end{lemma}
The situation is illustrated in Figures~\ref{werlfiuliu}, \ref{legiotoluiohyi}, and~\ref{erligfushepui}, for the clutch and two variants of it, for $N=1$ and $N=2$.

\begin{figure}[ht]
\centerline{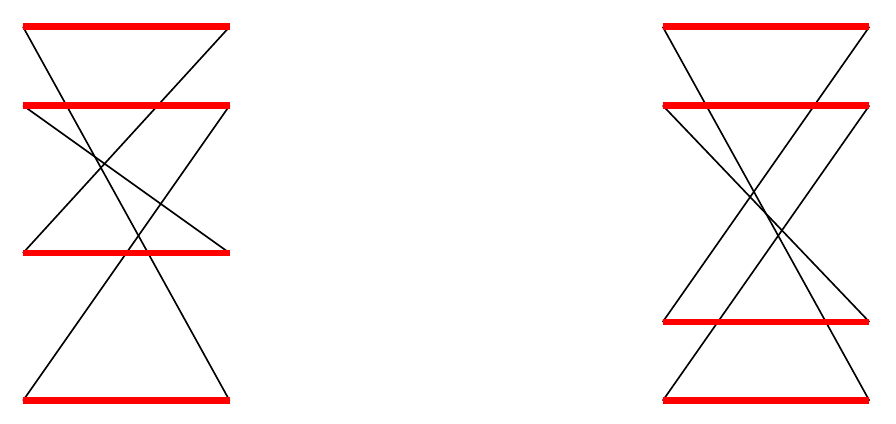}
	\caption{The clutch is an example in which $\previous{c_f}=c_f$ holds. Maximal cardinality matchings for $N=1$ (left) and $N=2$ (right) are given. Subsequently applying the Dulmage-Mendelsohn decomposition returns $\block_\overapprox=\{e_3\}$ on the left, and $\block_\overapprox=\{e_3,\postset{e_3}\}$ on the right.} 
	\label{werlfiuliu}
\end{figure}

\begin{figure}[ht]
\centerline{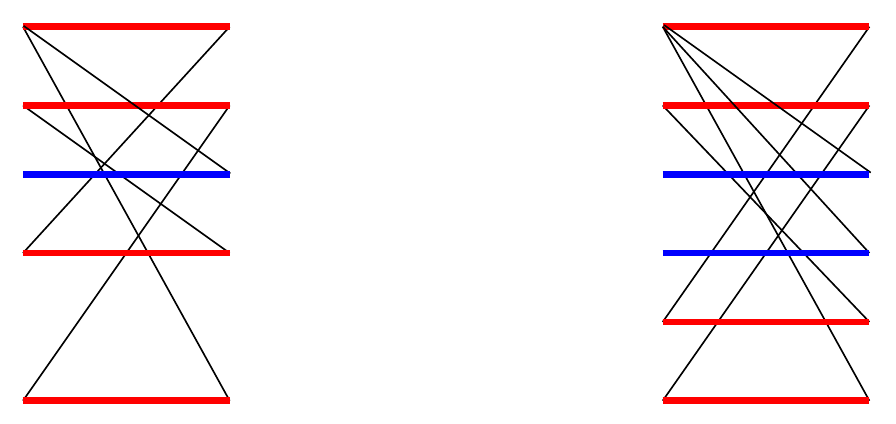}
	\caption{An example with $\previous{c_f}<c_f$: a variation on the clutch. The previous mode does not involve $\postset{\omega_1}$, hence $\omega_1$ remains a dependent variable when entering the new mode. Maximal cardinality matchings for $N=1$ (left) and $N=2$ (right) are shown. With reference with the example of Figure~\ref{werlfiuliu}, the added edges correspond to consistency equations and are colored in {\color{blue}blue}. Subsequently applying the Dulmage-Mendelsohn decomposition returns $\block_\overapprox=\emptyset$ on both. Actually, the two matchings are complete. As a result, the mode change is fully determined, despite the new mode for the clutch having underdetermined consistency conditions (because the common rotation velocity is free).} 
	\label{legiotoluiohyi}
\end{figure}

\begin{figure}[ht]
\centerline{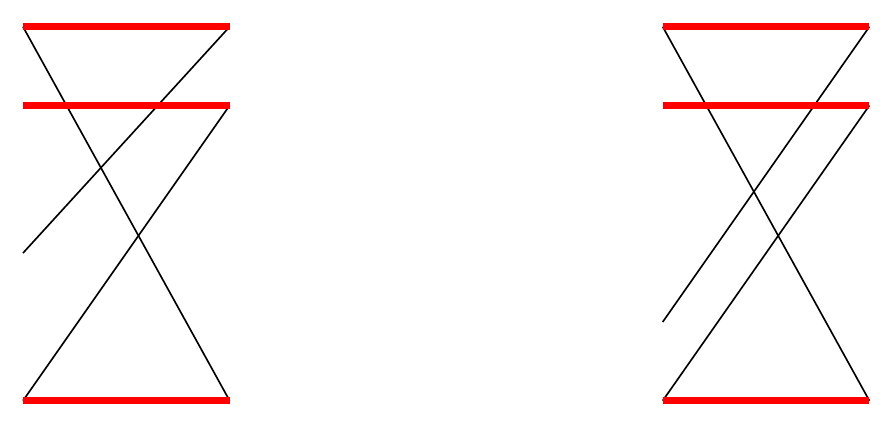}
	\caption{An example with $\previous{c_f}>c_f$: another variation on the clutch. The previous mode determines $\ppostset{2}{\omega_2}$, hence $\postset{\omega_2}$ is no longer a dependent variable when entering the new mode. Maximal cardinality matchings for $N=1$ (left) and $N=2$ (right) are shown. Subsequently applying the Dulmage-Mendelsohn decomposition on the left-hand graph returns $\block_\overapprox=\cG_F$, hence \rref{op:elkfuiehoeuio} of \rref{alg:newmain} returns $H=\emptyset$. No equation of the new mode is satisfied at the first instant and we move to the next instant. The same holds for the right-hand graph.} 
	\label{erligfushepui}
\end{figure}

%

\subsubsection*{Procedure for constructing matchings of maximal cardinality for $F$ and $F^\bfa$}
We want to design a procedure for determining such matchings so that the two constructions parallel each other. This will allow us to compare the results. 

\paragraph{Step 1: starting from $G_\Sigma$ and $G_\Sigma^\bfa$}
By construction, $G_\Sigma$ and $G_\Sigma^\bfa$ are both structurally nonsingular and determine the leading variables of the new mode, given a consistent valuation of other variables. They both admit a complete matching. 
\beq
\mbox{
\begin{minipage}{11cm}
 Let $\cN_1$ be a complete matching for $G_\Sigma$. Its edges have the  form $\bigl(\ppostset{c_f}{f},\ppostset{d_x}{x}\bigr)$ for appropriate pairs $(f,x)\in{G}\times{X}(G)$. Then, 
$
\cN^\bfa_1 \eqdef\left\{\left.
\bigl(\ppostset{Mc_f}{f},\ppostset{Md_x}{x}\bigr)\;
\right| \bigl(\ppostset{c_f}{f},\ppostset{d_x}{x}\bigr)\in\cN_1\;
\right\}
$
is a complete matching for $G_\Sigma^\bfa$.
\end{minipage}
} \label{lerituhlugkhliu}
\eeq
This parallel construction is illustrated in Figure~\ref{werlfiuliu}.

\paragraph{Step 2: Removing, from the set of dependent variables of $G_\Sigma$ and $G_\Sigma^\bfa$, the variables involved in the contexts $\Delta$ and $\Delta^\bfa$} Let 
\[
Y \eqdef \left\{
x\in{X}(G)
\;\left| \ppostset{d_x}{x}\in{X}(\Delta)
\right.\right\}
\mbox{ and }
Y^\bfa \eqdef \left\{
x\in{X}(G)
\;\left| \ppostset{Md_x}{x}\in{X}(\Delta^\bfa)
\right.\right\}\;.
\]
Then, $Y=Y^\bfa$ holds. Denote by $\widehat{G}_2$ (resp. $\widehat{G}_2^\bfa$) the system $G_\Sigma$ (resp. $G_\Sigma^\bfa$) where the variables involved in the context $\Delta$ (resp. $\Delta^\bfa$) have been removed. The bipartite graphs of $\widehat{G}_2$ and $\widehat{G}_2^\bfa$, denoted by $\widehat{\cG}_2$ and $\widehat{\cG}_2^\bfa$, are isomorphic. Hence:
\beq
\mbox{
\begin{minipage}{11cm}
	 If $\widehat{\cN}_2$ denotes a matching of maximal cardinality of $\widehat{\cG}_2$, then, 
$
\widehat{\cN}^\bfa_2 \eqdef\bigl\{\left.
\bigl(\ppostset{Mc_f}{f},\ppostset{Md_x}{x}\bigr)\;
\right| \bigl(\ppostset{c_f}{f},\ppostset{d_x}{x}\bigr)\in\widehat{\cN}_2\;
\bigr\}
$
is a matching of maximal cardinality for $\widehat{\cG}_2^\bfa$.
\end{minipage}
} \label{tgihgekrguihjgu}
\eeq
 In particular, their Dulmage-Mendelsohn decompositions are isomorphic too. An example of this is illustrated in Figure~\ref{erligfushepui}.

We start from $\widehat{\cN}_2$ and $\widehat{\cN}^\bfa_2$, perform their respective Dulmage-Mendelsohn decompositions, and remove, from both bipartite graphs, their respective overdetermined blocks $\block_\overapprox$ and $\block_\overapprox^\bfa$. The resulting bipartite graphs are denoted by $\cG_2$ and $\cG_2^\bfa$: they are also isomorphic. Hence:
\beq\mbox{
\begin{minipage}{11cm}
	 Let $\cN_2$ be a matching of maximal cardinality for $\cG_2$. The formula
$
\cN^\bfa_2 \eqdef\left\{\left.
\bigl(\ppostset{Mc_f}{f},\ppostset{Md_x}{x}\bigr)\;
\right| \bigl(\ppostset{c_f}{f},\ppostset{d_x}{x}\bigr)\in\cN_2\;
\right\}
$
defines a matching of maximal cardinality for $\cG_2^\bfa$.
\end{minipage}
}
\label{erlgtioio}
\eeq
$\cN_2$ involves all equations of $\cG_2$ since, by Lemma~\ref{togiethgouiho}, both $\cG_2$ and $\cG_2^\bfa$ have empty overdetermined blocks in their respective Dulmage-Mendelsohn decompositions.

\paragraph{Step 3: Adding to $G_\Sigma$ the consistency subsystem $\consistency{G}_\Sigma$ and subtracting $\Phi$; performing the corresponding changes to $G_\Sigma^\bfa$} 
We start from $\cG_2,\cN_2$ and $\cG^\bfa_2,\cN^\bfa_2$. We do not need to remove equations belonging to $\Phi$ and $\Phi^\bfa$, since such equations, if any, have already been removed at Step 2, when moving from $\widehat{\cN}_2$ to ${\cN}_2$. Therefore, the only final modification we need to perform is adding the consistency equations. The following result holds:
\begin{lemma}
	\label{elriguhjtoiudfh} 
	
\begin{enumerate}
	\item \label{leriuthuyfroui}
	Let $f\in{G}$ such that $\previous{c_f}\leq{c_f}$ and let $(\ppostset{c_f}{f},\ppostset{d_x}{x})\in\cN_2$. For every $m$ such that $0<{m}\leq{c_f}-\previous{c_f}$, the edge $(\ppostset{(c_f-m)}{f},\ppostset{(d_x-m)}{x})$ is an augmenting path of $\cN_2$ with respect to system $F$.
	\item \label{idgljhoirjrg}
	Let $f\in{G}$ such that $\previous{c_f}\leq{c_f}$ and let $(\ppostset{Mc_f}{f},\ppostset{Md_x}{x})\in\cN_2^\bfa$. For every $m$ such that $0<{m}\leq{M}({c_f}-\previous{c_f})$, the edge $(\ppostset{(Mc_f-m)}{f},\ppostset{(Md_x-m)}{x})$ is an augmenting path of $\cN_2^\bfa$ with respect to system $F^\bfa$.
\end{enumerate}
 \end{lemma}
\begin{proof}
Consider Statement~\ref{leriuthuyfroui}. If ${c_f}\leq\previous{c_f}$ for every $f\in{G}$, then matching $\cN_2$ is of maximal cardinality for $F$, so that no augmenting path exists. From now on, we assume ${c_f}>\previous{c_f}$ for some $f\in{G}$. Select any such $f$ and let $(\ppostset{c_f}{f},\ppostset{d_x}{x})$ be the edge of $\cN_2$ involving $f$. Then, $\ppostset{(c_f-1)}{f}$ and $\ppostset{(d_x-1)}{x}$ are involved in $F$, but  $\ppostset{(d_x-1)}{x}$ is not a vertex of $\cN_2$. Thus, for every $f\in{G}$ such that ${c_f}>\previous{c_f}$, then $(\ppostset{(c_f-1)}{f},\ppostset{(d_x-1)}{x})$ is an augmenting path of $\cN_2$ in $F$. The same reasoning applies for $(\ppostset{(c_f-m)}{f},\ppostset{(d_x-m)}{x})$, provided that $0<{m}\leq{c_f}-\previous{c_f}$. This proves Statement~\ref{leriuthuyfroui}. The same reasoning applies for proving Satement~\ref{idgljhoirjrg}.
\end{proof}
\begin{lemma}
	\label{elritueilrutp} 
\begin{enumerate}
	\item \label{leroyijhlruish} The following formulas define matchings of maximal cardinality for $F$ and $F^\bfa$, respectively:
	\beqq
	\cM &=& \left\{\left.
	\bigl(\ppostset{(c_f-m)}{f},\ppostset{(d_x-m)}{x}\bigr)~~~~\;
	\right| 
	\bigl(\ppostset{c_f}{f},\ppostset{d_x}{x}\bigr)\in\cN_2 , 0<{m}\leq{c_f}-\previous{c_f}
	\right\}
	\\
	\cM^\bfa &=& \left\{\left.
	\bigl(\ppostset{M(c_f-m)}{f},\ppostset{M(d_x-m)}{x}\bigr)
	\right| 
	\bigl(\ppostset{c_f}{f},\ppostset{d_x}{x}\bigr)\in\cN_2 , 0<{m}\leq {c_f}-\previous{c_f}
	\right\}
	\eeqq
	\item 
\label{ephojliuhlk}	Consider the surjective mapping $\psi:\cM^\bfa\mapsto\cM$, defined by  $$\psi\bigl(\ppostset{M(c_f-m)}{f},\ppostset{M(d_x-m)}{x}\bigr)=\bigl(\ppostset{(c_f-m)}{f},\ppostset{(d_x-m)}{x}\bigr)$$
	and define 
	\beqq
	\psi_f:F^\bfa\mapsto{F}&\emph{by}& \psi_f(\ppostset{M(c_f-m)}{f})=\ppostset{(c_f-m)}{f} \enspace ,
	\\ \psi_x:X(F^\bfa)\mapsto{X(F)} &\emph{by}& 
	\psi_x(\ppostset{M(d_x-m)}{x})=\ppostset{(d_x-m)}{x} \enspace .
	\eeqq
	 Then, $\psi$
	preserves the Dulmage-Mendelsohn decompositions of $\cG_F$ and $\cG_{F^\bfa}$, meaning that
	$g\in\block^\bfa_{\overapprox/\squared/\underapprox}$ if and only if $\psi_f(g)\in\block_{\overapprox/\squared/\underapprox}$, and	$y\in\block^\bfa_{\overapprox/\squared/\underapprox}$ if and only if $\psi_x(y)\in\block_{\overapprox/\squared/\underapprox}$.
\end{enumerate}
\end{lemma}
\begin{proof}
Statement~\ref{leroyijhlruish} is a direct consequence of Lemma~\ref{elriguhjtoiudfh} and (\ref{erlgtioio}). We thus focus on Statement~\ref{ephojliuhlk}. By Statement~\ref{leroyijhlruish}, we can see that $(\ppostset{(c_f-m)}{f},\ppostset{(d_x-m)}{x})\in\cM$ if and only if 
$(\ppostset{M(c_f-m)}{f},\ppostset{M(d_x-m)}{x})\in\cM^\bfa$. Let $g=\ppostset{(c_f-m)}{f}$ be a vertex of $\cG_{F}$ associated with a function. Then, $g$ is unmatched in $\cM$ if and only if $\psi_f(g)$ is unmatched in $\cM^\bfa$. We conclude by invoking Lemma~\ref{erliufheiuh}: the mappings $\psi_x$ and $\psi_f$ preserve the overdetermined blocks. A similar reasoning holds regarding the preservation of underdetermined blocks, by starting from unmatched vertices of variable type.
\end{proof}
This finishes the proof of Theorem~\ref{erlituherltuih}.

\subsubsection*{Comparing the standardizations}
The independence of the standardization with respect to the particular semantics of the derivative is immediate within long (continuous) modes. We thus focus on mode changes, where we must target a standard discrete-time dynamical system indexed by $\bN$. Throughout this section, Assumptions~\ref{oergfiuehoiu} (Page~\pageref{oergfiuehoiu}),~\ref{pw498gthseroi} (Page~\pageref{pw498gthseroi}) and~\ref{lirtylweiru} (Page~\pageref{lirtylweiru}) are in force. 
We want to compare the respective standardizations of the returns of \rref{alg:newmain} ($\atomicact{ExecRun}$), when using the two expansions (\ref{wotgihio}) and (\ref{lerigheil}) for the derivatives.

We first adapt the generic system~(\ref{eroifughoiu}), computing the restart at a mode change:
\beq
\mbox{expand $\dot{X}$ as $\frac{1}{\vsmall}\sum_M^N\alpha_n\ppostset{n}{(\postset{X}-X)}$ in}\hspace*{-5mm}
&&0=\bfH(\dot{X},\postset{X},V,X) \enspace ,
\label{lseriughpui}  \label{ruihpeiru} 
 \\ 
\mbox{also written}\hspace*{-5mm}&&0=H(Z,X,\vsmall) \enspace .
\label{eprfui9hepiruoh} 
\eeq
In  (\ref{lseriughpui}), $X$ collects the state variables, while $V$ collects the algebraic dependent variables, decomposed into impulsive and non-impulsive variables: $V=\impuls{W}\,\cup\,\nonimpuls{W}$.
Since $N\geq{0}$, the dependent variables of system (\ref{lseriughpui}) are $Z\eqdef(\ppostset{(N+1)}{X},V)$, and we set $Y\eqdef(\ppostset{(N+1)}{X},\nonimpuls{W})$; in words, $Y$ collects the non-impulsive dependent variables.

In~(\ref{lseriughpui}), the functions $F$ and $G$ are standard; so is the forward shift $x\mapsto\postset{x}$, since it will be mapped to a corresponding forward shift in the discrete time index $\bN$. Hence, the infinitesimal time step $\vsmall$, arising in the expansion of the derivative in $F$, is the only reason for System~(\ref{lseriughpui}) to be nonstandard. 

Assumption~$\ref{pw498gthseroi}$ states that impulsive variables can be eliminated from system~(\ref{lseriughpui}), resulting in the following reduced system:
\beq
\bfK(\dot{X},\postset{X},\nonimpuls{W},X)=0 ~\Longleftrightarrow~ \exists \impuls{W}.\bfH(\dot{X},\postset{X},V,X)=0 \enspace .
\label{epsiuvhopeiu}
\eeq 
Recall that $\vsmall$ enters System~(\ref{epsiuvhopeiu})-left, i.e., the left-hand part of the equivalence given by Equation~(\ref{epsiuvhopeiu}), via the expansion used for $\dot{X}$ in~(\ref{lseriughpui}).
As part of Assumption~$\ref{pw498gthseroi}$:
\beq
\mbox{System (\ref{epsiuvhopeiu})-left remains structurally nonsingular when $\vsmall=0$.}
\label{rtioghrjoi}
\eeq
In~(\ref{epsiuvhopeiu}), $(\ppostset{(N+1)}{X},\nonimpuls{W})$ standardizes as the right-limit $(X^+,\nonimpuls{W}^+)$ whereas all the $\ppostset{n}{X}, n=P,\dots,0$ standardize as the left-limit $X^-$. Therefore, in the expansion used for $\dot{X}$ in (\ref{lseriughpui}), all terms cancel out but the last one.
 Hence, System~(\ref{epsiuvhopeiu}) standardizes as:
\begin{itemize}
	\item Case $N>0$: substitute $\vsmall\gets{0}\mbox{ in }
\bfK\left(\frac{\alpha_N}{\vsmall}(X^+-X^-),\remph{X^-},\nonimpuls{W},X\right)=0$;
	\item Case $N=0$: substitute $\vsmall\gets{0}\mbox{ in }
\bfK\left(\frac{\alpha_N}{\vsmall}(X^+-X^-),\remph{X^+},\nonimpuls{W},X\right)=0$.
\end{itemize}
The difference lies in the $X^-$ vs. $X^+$, highlighted in {\color{red}red}.
These substitutions are both legitimate, thanks to~(\ref{rtioghrjoi}). In both cases, the result does not depend on $\alpha_{N}$, since $\alpha_{N}\neq{0}$.
Hence, we get the following result:
\begin{theorem}
	\label{erofuiehrpuioh}
	Under Assumptions~$\ref{oergfiuehoiu}$,~$\ref{pw498gthseroi}$, and~$\ref{lirtylweiru}$, the standardization of the restart conditions does not depend on the expansion $(\ref{lerigheil})$ for the derivative.
\end{theorem}
\begin{ccomment}\rm
	\label{elrgfiorohjyfgj} So far, Theorem~\ref{erofuiehrpuioh} applies to mDAE systems possessing only long modes. We could extend this theorem in several directions. First, we could include in our modeling language left- or right-limits: $x^-$ or $x^+$. Then, our reasoning would apply with no change provided that $x^-$ and $x^+$ are encoded as $\preset{x}$ or $\postset{x}$ if expansion (\ref{wotgihio}) is used, and $\ppreset{(M+1)}{x}$ or $\ppostset{(M+1)}{x}$ if expansion (\ref{lerigheil}) is used. Removing Assumption~$\ref{lirtylweiru}$ could be handled similarly, by redefining transient modes as modes of nonstandard duration $M\vsmall$ if expansion (\ref{lerigheil}) is used, instead of $\vsmall$ if expansion (\ref{wotgihio}) is used.
\end{ccomment}

\subsection{{A correctness result}}
\label{sec:mainthm}
Our approach, consisting of the structural analysis in the nonstandard domain (\rref{alg:newmain} and its refined version \rref{alg:ergfipuhfopiu}) followed by standardization, either fails, in which case it returns proper diagnosis, or succeeds and produces executable simulation code. The natural question is: 
\begin{question} \rm
	\label{lergtuioh} Does this approach compute the solution of the given \mDAE\ system?
\end{question}
Once again, answering Question\,\ref{lergtuioh} in its full generality is out of reach as, to the best of our knowledge, there is no general mathematical definition of the solution of an \mDAE\ system. Pathologies may exist in the complement of the long `continuous' modes: possibly unbounded cascades of events, chattering or sliding modes, Zeno behaviors for events of mode changes or transient modes, and more. Even when the set of events of mode changes is `gentle' (a finite or diverging sequence of instants), the characterization of the impulses may not be available. Still, as mentioned in the Introduction, notions of solutions of \mDAE\ systems exist for some restricted classes.

In this section, we provide an answer to Question\,\ref{lergtuioh} for a restricted subclass of \mDAE\ systems, by reproducing the results of~\cite{DBLP:series/lncs/BenvenisteCEGOP19} and giving their detailed proofs. For the description of the class of systems considered, we use the notion of function of bounded variation: a function $f:\bR\rightarrow\bR$ has \emph{bounded variation} if it is the primitive of a Lebesgue integrable function~\cite{DunfordSchwartz}. As a consequence, if $f$ has bounded variation, then
\beq
\lim_{h\searrow{0}}\;\int_t^{t+h}\dot{f}(s)ds &=&{0} \enspace .
\label{liguhlsigtuh}
\eeq
\begin{definition}[semilinear systems]
	\label{leriuthoui} 
	Call \emph{semilinear} an \mDAE\ system $\system$ such that, for 
	each mode $\mu$, the active DAE system takes the restricted form
	\beq\left\{\bea{l}
	0=A(X_s)\dot{X}+B_\mu(X) \\ 0=C_\mu(X)
	\eea\right.
	\label{leriuljytdfutr}
	\eeq
	where:
\begin{enumerate}
\item \label{lhglruigth} $X_s$ collects the smooth elements of $X$, i.e., those being continuous and of bounded variation around each instant of mode change. Other elements of $X$ might be discontinuous. 
\item \label{epioguhip} Both matrix $A(.)$ and the mode-dependent vectors $B_\mu(.)$ and $C_\mu(.)$ are continuous functions of their arguments, and remain bounded in a neighborhood of each instant of mode change.
	\item \label{eoriuheopiu} The matrix
	\[
	\Jacobian=\left[\bea{c}
	A(X_s) \\ \Jacobian_{\!X}{C}_\mu(X)
	\eea\right]
	\]
	is regular, where $\Jacobian_{\!X}{C}_\mu(X)$ denotes the Jacobian matrix of $C_\mu$ with respect to $X$, evaluated at $X$.
\end{enumerate}
\end{definition}
\begin{ccomment}\rm
	\label{elrtuiogehluio} 
The special form (\ref{leriuljytdfutr}) and Condition~\ref{eoriuheopiu} for the Jacobian matrix express that mode dynamics have been processed using the \sigmamethod: differentiating the algebraic constraint once yields a regular linear system determining $\dot{X}$. Then, Conditions~\ref{lhglruigth} and~\ref{epioguhip} together make the impulse analysis easier.
\end{ccomment}
\begin{ccomment}\rm
	\label{epruhypiu} 
In~\cite{DBLP:series/lncs/BenvenisteCEGOP19}, Section 4.3, it is shown that multi-body systems with contacts yield semilinear \mDAE\ systems.
\end{ccomment}
For semilinear \mDAE\ systems, the solution at an instant $t_*$ of mode change can be determined as follows. After proper reordering of the entries of $X$, we decompose vector $X$ into its smooth and nonsmooth parts and decompose matrix $A$ accordingly:
\[
X=\left[\bea{l}
X_s \\ X_n
\eea\right] ~\mbox{ and }~ A=\bigl[\,A_s ~~ A_n\,\bigr] \enspace .
\]
Since $X_s$ is smooth at any instant, by Condition~\ref{lhglruigth}, matrix $A\bigl(X_s(t_*)\bigr)$ is known prior to computing the mode change. {In contrast, $\dot{X_n}(u)du$ is a Dirac measure at mode change.}

To evaluate the restart values for all elements of $X$, we integrate the system dynamics over the interval $[t_*-\varepsilon,t_*+\varepsilon]$ for $\varepsilon>0$ sufficiently small, taking into account the dynamics before and after the mode change. In the following calculations, $\mu^-$ and $\mu^+$ denote the mode prior and after the mode change, and $B_+$ and $C_+$ stand for $B_{\mu^+}$ and $C_{\mu^+}$, respectively.
\small
\beqq
&&\int_{[t_*-\varepsilon,t_*)}\bigl(A(X_s(u))\dot{X}(u){+}B_-(X(u))\bigr)du
+\int_{[t_*,t_*+\varepsilon]}\bigl(A(X_s(u))\dot{X}(u){+}B_{+}(X(u))\bigr)du
\\
=&&\int_{[t_*-\varepsilon,t_*)}\bigl(A_s(X_s(u))\dot{X_s}(u){+}B_-(X(u))\bigr)du
+\int_{[t_*,t_*+\varepsilon]}\bigl(A_s(X_s(u))\dot{X_s}(u){+}B_+(X(u))\bigr)du
\\
&\!\!\!\!\!\!+\!\!\!\!\!\!\!\!\!&
\int_{[t_*-\varepsilon,t_*+\varepsilon]}A_n(X_s(u))\dot{X_n}(u)du
\\ [1mm]
\approx&& 0 ~+~ 0~+~ A_n(X_s(t_*))\bigl(X_n(t_*{+}\varepsilon)-X_n(t_*{-}\varepsilon)\bigr)
\\
\approx&&A(X_s(t_*))\bigl(X(t_*{+}\varepsilon)-X(t_*{-}\varepsilon)\bigr)
\eeqq
\normalsize
The evaluation to zero of the first two integrals follows from (\ref{liguhlsigtuh}). For the third integral, we use the fact that $\dot{X_n}(u)du$ is approximately a Dirac measure at $t_*$ in the indicated interval. 
This reasoning leads to the following result determining the restart conditions at $t_*$:
\begin{theorem}[\cite{DBLP:series/lncs/BenvenisteCEGOP19}]
	\label{serliguserhligu} The restart conditions $X^+$ are determined from the left-limits $X^-$ by using the following system of equations:
\beq
\left\{\bea{l}
0=A(X^-_s)(X^+-X^-) \\ 0=C_+(X^+)
\eea\right.
\label{elrifuhliu}
\eeq
which is locally regular thanks to Condition~$\ref{eoriuheopiu}$ of Definition~$\ref{leriuthoui}$.
\end{theorem}
Whether we can solve system~(\ref{elrifuhliu}) depends on the particular system at hand.
It remains to compare scheme~(\ref{elrifuhliu}) with our approach. The nonstandard semantics of~(\ref{leriuljytdfutr}) is:
\beqq\left\{\bea{l}
	0=A(X_s)({\postset{X}-X})+{\vsmall}\times{B}_\mu(X) \\ 0=C_\mu(X)
	\eea\right.
	\eeqq
	and applying the \sigmamethod\ to it yields the augmented system
	\beq\left\{\bea{l}
	0=A(X_s)({\postset{X}-X})+{\vsmall}\times{B}_\mu(X) \\ 0=C_\mu(X)\\ 0=C_\mu(\postset{X})
	\eea\right. \enspace .
	\label{aeruioghpeiug}
	\eeq
The reasoning performed at Section~\ref{seoiguhweroui} applies to semilinear \mDAE\ systems. Following (\ref{ekerughlruigth}), we unfold (\ref{aeruioghpeiug}) at the two successive instants $\preset{t_*}$ and $t_*$:
\beq\bea{rl}
\preset{t_*}\!\!\!\!&:\left\{\bea{l}
	0=A(\preset{X_s})({{X}-{\preset{X}}})+{\vsmall}\times{B}_-(\preset{X}) \\ 0=C_-(\preset{X})\\ 0=C_-({X})
	\eea\right.
	\\ [6mm]
	{t_*}\!\!\!\!&:\left\{\bea{l}
	0=A(X_s)({\postset{X}-X})+{\vsmall}\times{B}_+(X) \\ 0=C_+(X)\\ 0=C_+(\postset{X})
	\eea\right.
	\eea
	\label{elguiohwepogui}
	\eeq
	Eliminating $\preset{X}$ in (\ref{elguiohwepogui}) yields
	\beqq\left\{\bea{lc}
	0=C_-({X}) & (\eqq^-) \\ [1mm]
	0=A(X_s)({\postset{X}-X})+{\vsmall}\times{B}_+(X)  & (\eqq^+_1) \\  [1mm] 0=C_+(X)   & (\eqq^+_2) \\ [1mm] 0=C_+(\postset{X})   & (\eqq^+_3) 
	\eea\right.
	\eeqq
	Equations $(\eqq^-)$ and $(\eqq^+_2)$ are in conflict: we discard the latter. This finally yields
	\beq\left\{\bea{lc}
	0=C_-({X}) & (\eqq^-) \\ [1mm]
	0=A(X_s)({\postset{X}-X})+{\vsmall}\times{B}_+(X)  & (\eqq^+_1) \\ [1mm] 0=C_+(\postset{X})   & (\eqq^+_3) 
	\eea\right.
	\label{erpoguifehpiouh}
	\eeq
	In System (\ref{erpoguifehpiouh}), $(\eqq^-)$ is a context, whereas $(\eqq^+_1)$ and $(\eqq^+_3)$ determine $\postset{X}$ as a function of $X$. Setting $\vsmall\gets{0}$ in $(\eqq^+_1,\eqq^+_3)$ yields a structurally nonsingular system by Condition~\ref{eoriuheopiu} of Definition~\ref{leriuthoui}. Hence, by Theorem~\ref{ow4u89ghsljkwrg}, setting $\vsmall\gets{0}$ in $(\eqq^+_1,\eqq^+_3)$ performs the standardization: the scheme of Theorem~\ref{serliguserhligu} is recovered.
\begin{theorem}[correctness result]
	\label{religftuehligu} For semilinear \mDAE\ systems, our approach computes the correct solution.
\end{theorem}
\begin{ccomment}\rm 
	 Our approach computes solutions for any \mDAE\ system whose structural analysis (\rref{alg:newmain} or \rref{alg:ergfipuhfopiu}) succeeds. This is way beyond the class of semilinear \mDAE\ systems. It would be desirable to compare the so obtained scheme with schemes known for more dedicated physics.\eproof
\end{ccomment}
\begin{ccomment}\rm
	\label{epriguhpui}  
Hilding Elmqvist and Martin Otter have developed the  \href{https://modiasim.github.io/ModiaMath.jl/stable/man/Overview.html}{ModiaMath} tool for semilinear \mDAE\ systems by effectively implementing scheme~(\ref{elrifuhliu}).\eproof
\end{ccomment}
	
\subsection{{Numerical scheme for restarts}}
\label{sec:restartschemes}
In this section, we focus on the system of equations defining restarts, in its form (\ref{ltrghilu}), which we recall  for convenience: 
\beq
H(Z,X,\vsmall)=0 \enspace . \label{elrigtueho}
\eeq
Our aim is to justify our claim that, if we remain `blind' with respect to impulsive variables, the solution of this system is well approximated by the standard system 
\beq
H(Z,X,\delta)=0 \enspace ,
\label{loijhlighiuygf}
\eeq
 where the infinitesimal $\vsmall$ has been replaced by a small standard parameter $\delta>0$. This is essentially the proposed numerical scheme.
Throughout this section, Assumptions~$\ref{oergfiuehoiu}$ and~$\ref{pw498gthseroi}$ of Section~\ref{weroguihergu} are in force. 

We briefly recall the corresponding notations for convenience: in  (\ref{elrigtueho}), $X$ collects the state variables and $Z\eqdef(\postset{X},V)$ collects the dependent variables. The algebraic variables $V$ can be split into impulsive and non-impulsive variables: $V=\impuls{W}\,\cup\,\nonimpuls{W}$, and we set $Y\eqdef(\postset{X},W)$. Assumption~$\ref{pw498gthseroi}$ states that impulsive variables can be eliminated from System~(\ref{loijhlighiuygf}), resulting in the reduced system 
\beq
K(Y,X,\delta){=}0 \enspace ,
\label{lguihuio2}
\eeq 
that remains structurally nonsingular when $\delta=0$.
\begin{theorem}
	\label{eroitutrdh} 
The following statements hold:
\begin{enumerate}
\item \label{eprituehrp} Let $Y_*$ be a solution of $(\ref{lguihuio2})$. There exists a selection map $\delta\mapsto(Y(\delta),\compl{W}(\delta))$, from $\bR_{>0}$ to the set of solutions of the system $(\ref{loijhlighiuygf})$, such that $\lim_{\delta\searrow{0}}Y(\delta){=}Y_*$.
\item \label{lreuiwgfo} Conversely, let $(Y(\delta),\compl{W}(\delta))$ be a  solution of  $(\ref{loijhlighiuygf})$ such that $Y_*\eqdef\lim_{\delta\searrow{0}}Y(\delta)$ exists; then, $Y_*$ is a solution of $(\ref{lguihuio2})$. 
\end{enumerate}
\end{theorem}
\begin{proof} We successively prove the two statements.

\myparagraph{Statement~\ref{eprituehrp}}
	By the Implicit Function Theorem applied to the standard function $K$, there exists a function $(\hat{X},\delta)\mapsto{Y}(\hat{X},\delta)$ of class $\cC^1$ such that $Y(X,0){=}Y_*$ and $K(Y(\hat{X},\delta),\hat{X},\delta){=}0$ holds for $(\hat{X},\delta)$ ranging over a neighborhood of $(X,0)$. Then, by Assumption~\ref{pw498gthseroi}, 
	$
	\exists \impuls{W}.H(Z,X,\delta){=}0
	~\Longleftrightarrow~
	K(Y,X,\delta){=}0
	$
	holds, which concludes the proof of this statement.

\myparagraph{Statement~\ref{lreuiwgfo}}
This is trivial and does not use Assumption~$\ref{pw498gthseroi}$. Since  $K$ is smooth, the two conditions (1) $K(Y(\delta),X,\delta){=}0$ for every $\delta>0$, and (2) $Y(\delta)$ converges to $Y_*$, imply $K(Y_*,X,0){=}0$.
\end{proof}
\begin{corollary}
	\label{erligtuherliut}
	 Let Assumption~$\ref{pw498gthseroi}$ be in force and let the systems 
	 $(\ref{loijhlighiuygf})$ and $(\ref{lguihuio2})$ possess a unique solution for $Z$ and $Y$, respectively. Then,  $H(Z(\delta),X,\delta){=}0$ implies $K(\lim_{\delta\searrow{0}}Y(\delta),X,0){=}0$.
\end{corollary}

For its practical application, this scheme requires the prior knowledge of the decomposition $V=\impuls{W}\,\cup\,\nonimpuls{W}$ of algebraic variables into impulsive and non-impulsive ones: we solve System (\ref{loijhlighiuygf}) for $Z$ and simply discard the values of $\impuls{W}$ from this solution. The drawback of this scheme is that system (\ref{loijhlighiuygf}) will be likely ill-conditioned.

Thus, in practice, we would recommend a rescaling of the impulsive variables based on the impulsion order: for every impulsive variable $w$ or order $\mu\eqdef\magorder{w}$, we perform the substitution $w\gets\nu\times\delta^{-\mu}$, where $\nu$ is the new variable replacing $w$---such a rescaling works only for variables of finite impulsion order, see \rref{sec:impulseanalysis}.

\subsection{Experimental results: mode changes of a nonsemilinear system}
\label{sec:nonsemilin}
To exercise a numerical method computing state-jumps at mode changes,
when some variables are impulsive, let us consider the following
multimode DAE system:
\begin{equation}
\left\{
\bea{rllcc}
&&\omega'_1=a_1\omega_1 + b_1\tau_1^3 &&(\eqq_1) \\
&&\omega'_2=a_2\omega_2 + b_2\tau_2 &&(\eqq_2) \\
\when\;\guard&\doo&\omega_1-\omega_2=0 &&(\eqq_{3}) \\
&\aand&\tau_1+\tau_2=0 &&(\eqq_{4}) \\
\when\;\prog{not}\;\guard&\doo&\tau_1=0 &&(\eqq_{5}) \\
&\aand&\tau_2=0 &&(\eqq_{6}) \\
\eea
\right.
\label{sys:mdae:nonsemilin}
\end{equation}
This system is identical to the clutch
model~(\ref{sys:coupledshafts}), with the only difference that torque
$\tau_1$ appears non-linearly in equation $(\eqq_1)$. Therefore, this
system is not semilinear, and the direct approach based on
linear-algebraic techniques used in \rref{sec:mainthm} does not
apply. Nevertheless, we shall demonstrate on this example that our
general approach applies. We focus on the mode change
$\guard:\fff\ra\ttt$. First, we perform impulse analysis following
\rref{sec:impulseanalysis}. Based on this impulse analysis, we apply a
change of variable that bypasses the need for eliminating
impulsive variables. We then standardize the resulting algebraic
system of equations and we prove that it possesses a unique solution.

Using the structural analysis method detailed in
Section~\ref{sec:orguhuoi}, the state-jump is solution of the
following system of algebraic equations, where $\postset{\omega_1}$,
$\postset{\omega_2}$, $\tau_1$ and $\tau_2$ are the unknowns:
\beq \left\{ \bea{lcc}
  0 = \frac{\postset{\omega_1}-\omega_1}{\vsmall} - a_1\omega_1 - b_1\tau_1^3 &&(\eqq_1^\vsmall) \\
  0 = \frac{\postset{\omega_2}-\omega_2}{\vsmall} - a_2\omega_2 - b_2\tau_2 &&(\eqq_2^\vsmall) \\
  0 = \postset{\omega_1}-\postset{\omega_2} &&(\postset{\eqq_3}) \\
  0 = \tau_1+\tau_2 &&(\eqq_{4}) \eea \right.
\label{sys:sj:nonsemilin}
\eeq 
Applying the method of \rref{sec:impulseanalysis}, the impulse
orders of the dependent variables
$\postset{\omega_1}-\omega_1,\postset{\omega_2}-\omega_2,\tau_1,\tau_2$
are solution of the following system of equalities and inequalities:
\begin{equation}
\left\{
\bea{lcc}
1+\magorder{\postset{\omega_1}-\omega_1}\leq\max\{\magorder{a_1\omega_1},\magorder{b_1}+3\magorder{\tau_1}\} \\
\magorder{a_1\omega_1}\leq\max\{1+\magorder{\postset{\omega_1}-\omega_1},\magorder{b_1}+3\magorder{\tau_1}\} \\
\magorder{b_1}+3\magorder{\tau_1}\leq\max\{1+\magorder{\postset{\omega_1}-\omega_1},\magorder{a_1\omega_1}\} \\
1+\magorder{\postset{\omega_2}-\omega_2}\leq\max\{\magorder{a_2\omega_2},\magorder{b_2}+\magorder{\tau_2}\}  \\
\magorder{a_2\omega_2}\leq\max\{1+\magorder{\postset{\omega_2}-\omega_2},\magorder{b_2}+\magorder{\tau_2}\} \\
\magorder{b_2}+\magorder{\tau_2}\leq\max\{1+\magorder{\postset{\omega_2}-\omega_2},\magorder{a_2\omega_2}\} \\
\magorder{\postset{\omega_1} - \omega_1}=\magorder{\postset{\omega_2} - \omega_2}=0 \\
\magorder{\tau_1}=\magorder{\tau_2} 
\eea
\right.
\label{sys:ia:nonsemilin}
\end{equation}
System (\ref{sys:ia:nonsemilin}) admits a unique solution:
$\magorder{\postset{\omega_1}-\omega_1} =
\magorder{\postset{\omega_2}-\omega_2} = 0$ and
$\magorder{\tau_1} = \magorder{\tau_2} = \frac{1}{3}$.
This impulse analysis tells us that if a solution to
System~(\ref{sys:sj:nonsemilin}) exists, then unknowns $\tau_1$ and
$\tau_2$ are of order
$O(\vsmall^{-\frac{1}{3}})$. The
following substitution helps solving System (\ref{sys:sj:nonsemilin}):
$\tau_i = \vsmall^{-\frac{1}{3}}\hat{\tau}_i$. With this substitution, System (\ref{sys:sj:nonsemilin}) becomes:
\beq \left\{ \bea{lcc}
  0=\postset{\omega_1}-\omega_1 - \bigl(\vsmall a_1\omega_1+b_1\hat{\tau}_1^3\bigr) && (\eqq_1^\vsmall) \\
  0=\postset{\omega_2}-\omega_2 - \bigl( \vsmall a_2\omega_2+\vsmall^{\frac{2}{3}} b_2\hat{\tau}_2\bigr) && (\eqq_2^\vsmall) \\
  0=\postset{\omega_1}-\postset{\omega_2}  && (\postset{\eqq_3}) \\
  0=\hat{\tau}_1+\hat{\tau}_2  && (\eqq_{4}) \eea \right.
\label{sys:sj:nonsemilin:subs}
\eeq 
Let us replace each unknown by its unique decomposition into a
standard part and an infinitesimal part. Regarding velocities, we
have: $\postset{\omega_i} = \postset{\bar{\omega}_i} + \phi_i$, where
$\phi_i \approx 0$, for $i=1,2$. Next, since the rescaled variables
$\hat{\tau}_i$ are non-impulsive, the same substitution is legitimate
by Lemma~\ref{owirguhpr}: $\hat{\tau}_i = \bar{\tau}_i + \psi_i$, with
$\psi_i \approx 0$, for $i=1,2$. System (\ref{sys:sj:nonsemilin:subs})
then becomes:
\beq \left\{ \bea{lcc}
  \postset{\bar{\omega}_1} + \phi_1 - \omega_1 = \vsmall a_1\omega_1 + b_1\left(\bar{\tau}^3_1+3\bar{\tau}^2_1\psi_1+3\bar{\tau}_1\psi^2_1+\psi^3_1\right) && (\eqq_1^\vsmall) \\
  \postset{\bar{\omega}_2} + \phi_2 - \omega_2 = \vsmall a_2\omega_2 + \vsmall^{\frac{2}{3}} b_2 \left(\bar{\tau}_2+\psi_2\right) && (\eqq_2^\vsmall) \\
  \postset{\bar{\omega}_1} + \phi_1 - \postset{\bar{\omega}_2} - \phi_2 = 0 && (\postset{\eqq_3}) \\
  \bar{\tau}_1 + \psi_1 + \bar{\tau}_2 + \psi_2 = 0 && (\eqq_{4}) \eea \right.
\label{sys:sj:nonsemilin:subs1}
\eeq 
Zeroing the infinitesimals in these equations yields the following
nonsingular system of equations, which, by
Theorem~\ref{ow4u89ghsljkwrg}, is then the standardization of System
(\ref{sys:sj:nonsemilin:subs1}):
\beq \left\{ \bea{lcc}
  \postset{\bar{\omega}_1} - \omega_1 =  b_1\bar{\tau}^3_1 && (\bar{\eqq}_1^\vsmall) \\
  \postset{\bar{\omega}_2} - \omega_2 = 0 && (\bar{\eqq}_2^\vsmall) \\
  \postset{\bar{\omega}_1} - \postset{\bar{\omega}_2}  = 0 && (\postset{\bar{\eqq}_3}) \\
  \bar{\tau}_1 + \bar{\tau}_2 = 0 && (\bar{\eqq}_{4}) \eea \right.
\label{sys:sj:nonsemilin:subs2}
\eeq 
This standard system of equations admits a unique solution: 
\beq
\postset{\bar{\omega}_1} = \postset{\bar{\omega}_2} = \omega_2 
~~\mbox{ and }~~
\bar{\tau}_1 = -\bar{\tau}_2 = \sqrt[3]{\frac{\omega_2 - \omega_1}{b_1}}\,.
\label{sys:sj:nonsemilin:sol}
\eeq

We now show that solving System (\ref{sys:sj:nonsemilin}) with a small
standard positive parameter instead of the nonstandard $\vsmall$
yields a numerical scheme converging to the solution
(\ref{sys:sj:nonsemilin:sol}), if we disregard impulsive variables.
Algorithm\,\ref{alg:nonsemilin:algo} describes our numerical method,
whereas Figure\,\ref{fig:nonsemilin:neta} displays the results.

\begin{algorithm}[ht]
\caption{Numerical method for approximating the state-jump of non-semilinear system (\ref{sys:mdae:nonsemilin}) for mode change $\guard:\fff\ra\ttt$; $F_h(u)$ denotes the value of the right-hand side of System (\ref{sys:sj:nonsemilin}) for a $4$-tuple $u$ of values for the unknowns and $\vsmall$ replaced by $h$.}
\label{alg:nonsemilin:algo}
\begin{algorithmic}[1]
   \Require Difference equations $F_h(u) : \bR {\times} \bR^N {\rightarrow} \bR^N$, $u_0{\in} \bR^N$ initial guess, \mbox{$\pi : \bR^N {\rightarrow} \bR^M$} projection on the set of non-impulsive variables, $h_0{>}0$ initial time step, $\epsilon {>}0$ tolerance bound, $0{<}\theta{<}1$ scaling factor ; 
   \Return $u$ estimated state-jump
   \State $u \gets u_0$; $x \gets \pi(u_0)$; $h \gets h_0$
   \Repeat
     \State $y \gets x$
     \State $u \gets \mbox{NSolve}(\mbox{system}=F_h, \mbox{guess}=u, \mbox{tolerance}=h)$
     \State $x \gets \pi(u)$; $h \gets \theta h$
     \Until{$\left|x-y\right|_{\infty} \leq \epsilon$}
   \end{algorithmic}
\end{algorithm}

\begin{figure}[ht]
  \centering
    \includegraphics[width=0.45\textwidth]{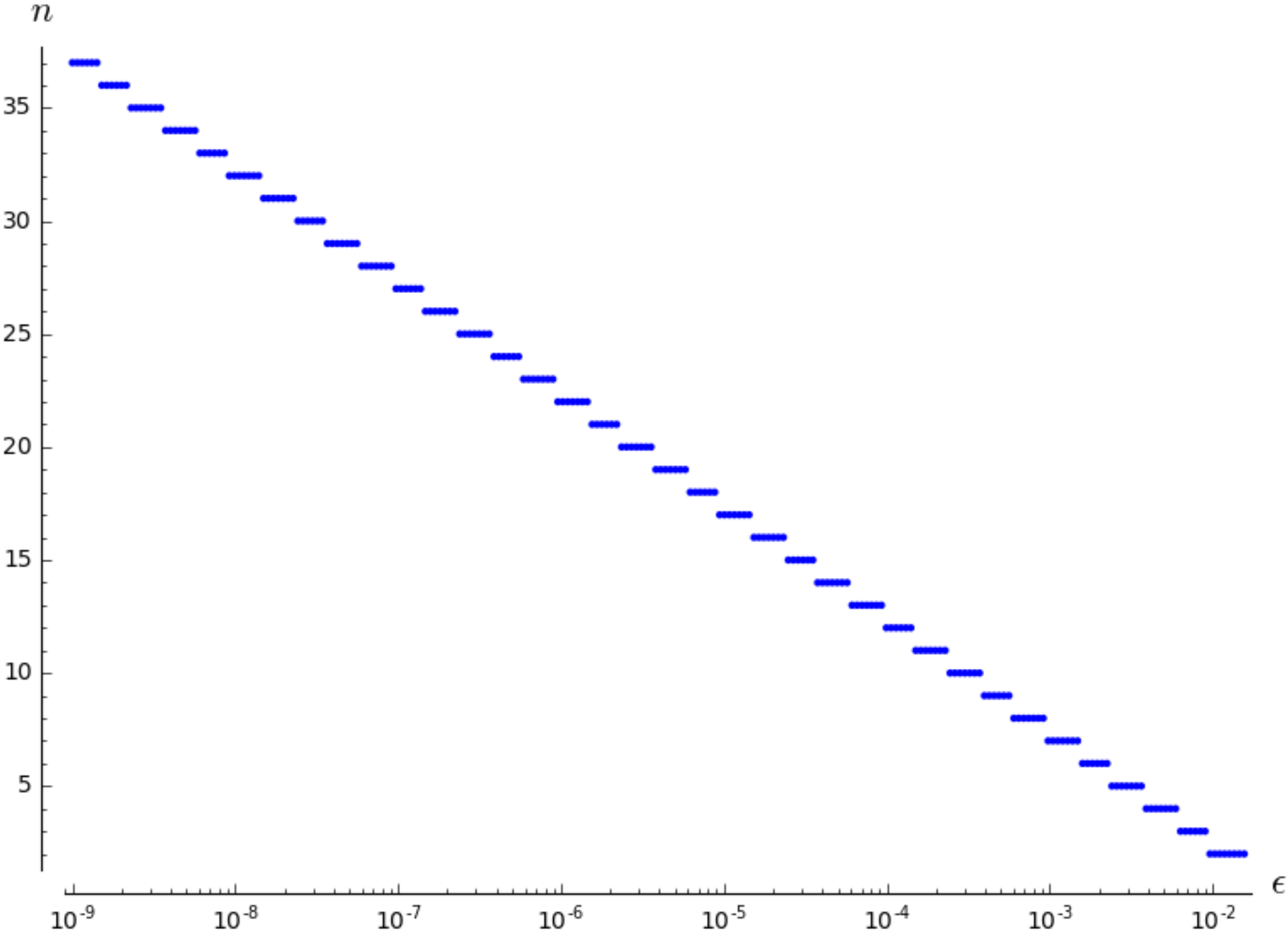}~\hfill~
    \includegraphics[width=0.45\textwidth]{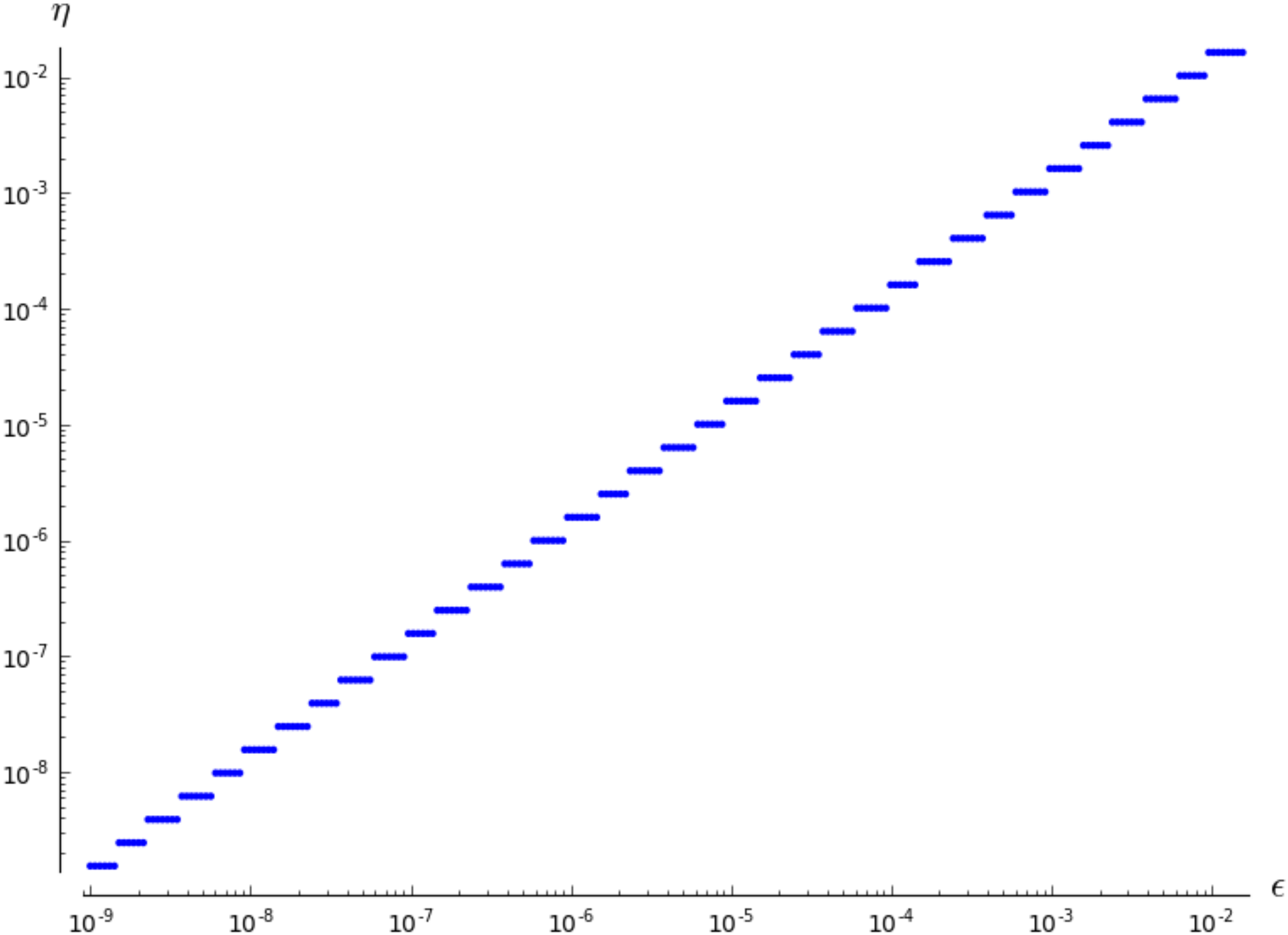} \\
    \hfill~(a)~\hfill~(b)~\hfill
  \caption{Numerical method for computing the state-jump of non-semilinear system (\ref{sys:mdae:nonsemilin}): (a) number of iterations as a function of $\epsilon$; (b) error $\eta$ as a function of $\epsilon$.}\label{fig:nonsemilin:neta}
  \end{figure}

Experimental results on this method have been obtained using Sage. The
nonlinear solver used is the least squares solver of the Scipy's
\texttt{optimize} library. Figure~\ref{fig:nonsemilin:neta} gives, on
the left, the number of iterations required, and on the right, the
absolute error $\eta = |x-\bar{x}|_{\infty}$ between the computed solution
and the exact solution $\bar{x} = (\omega_2,\omega_2)$. The scaling
factor is $\theta=0.5$, the initial time step $h_0 = 10^{-2}$ and
$u_0 = (\omega_1, \omega_2, 0, 0)$. It can be observed that the method
converges in not more than $37$ iterations, down to
$\epsilon = 10^{-9}$. It should also be remarked that, experimentally, the
absolute error is bounded : $\eta < 2\epsilon$.


\section{{Toward a tool supporting our approach}}
\label{sec:RLDC2}
The general approach we presented sets formidable challenges for a tool development. We identify three important requirements for such a tool. 

\begin{itemize}
\item First, it should implement the theory: structural analysis must be adjusted to every different mode and mode change. Both aspects are included in what we call a `multimode structural analysis'.
\item Second, the described approach aims at the compilation of \mDAE\ systems. As such, `on-the-fly' techniques for the structural analysis of such systems, as advocated in~\cite{Broman:EECS-2012-173} and~\cite{10.1007/978-3-319-47169-3_15}, cannot be taken into consideration. As soon as a model is structurally unsound in any given mode or at any given mode change, it should be rejected at compile time. For accepted models, the output of the tool should enable the generation of efficient simulation code.
\item Third, enumerating all modes and mode changes to perform multimode structural analysis is bound to collapse. An \mDAE\ compiler should be able to handle models in a clever way to avoid such enumeration.
\end{itemize}
We believe that we achieved an important first step towards alleviating these challenges with our tool IsamDAE under development.
IsamDAE implements the structural analysis of all modes of an \mDAE\ without any enumeration of the modes, instead relying on Binary Decision Diagrams (BDD, see for instance~\cite{Bryant1986}) for implicit representations of the mode-dependent structure of an \mDAE.
This is described more thoroughly in~\cite{Caillaud2020a,caillaud:hal-02476541}.

In order to illustrate how a multimode model is currently handled by IsamDAE, this Section focuses on the so-called RLDC2 model,
a simple
electronic circuit shown in Figure~\ref{fig:rldc2-sch}. This example, provided to us by Sven-Erik Mattsson, is actually a difficult one to handle for the existing Modelica tools because of its mode-dependent index and structure, even though it exhibits no impulsive behavior.
Indeed, we show once more how this model is mishandled by two
leading Modelica tools, namely, Dymola and OpenModelica, and hint at
reasons for this fact, thereby justifying our work once again.
We then detail its handling by the IsamDAE tool, in order to provide an insight about why this tool could be expanded into a complete multimode structural analysis suite.
\begin{figure}[!h]
\begin{center}
\includegraphics[width=0.80\linewidth]{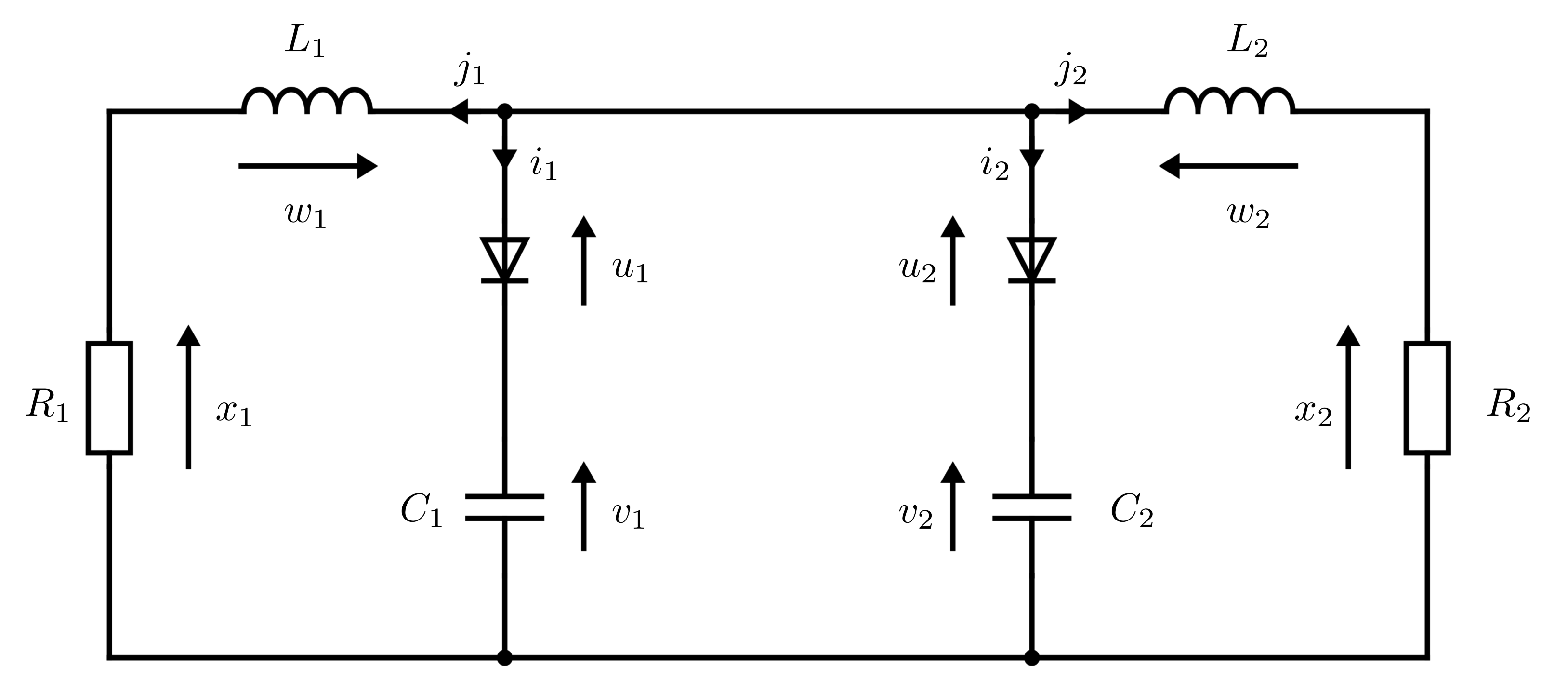}
\end{center}
\caption{Schematics of the RLDC2 circuit.}
\label{fig:rldc2-sch}
\end{figure}

\subsection{The RLDC2 model}
\label{sub:rldc2-model}

The circuit consists in two RLC circuits interconnected in parallel
and in which two diodes have been inserted. The diodes are considered
to be ideal, meaning that they are not ruled by the customary Shockley
law $i = I \left(e^{u/U} - 1\right)$, but rather by the complementarity
condition $0 \leq i ~ \bot ~ -u \geq 0$, meaning that, at any moment,
both $i$ and $-u$ are nonnegative and $i u = 0$ (Section~\ref{leriugfehliu} already mentions complementarity conditions). Figure~\ref{fig:idealdiode} illustrates these two diode laws.

\begin{figure}[!h]
  \begin{center}
    \includegraphics[width=0.5\linewidth]{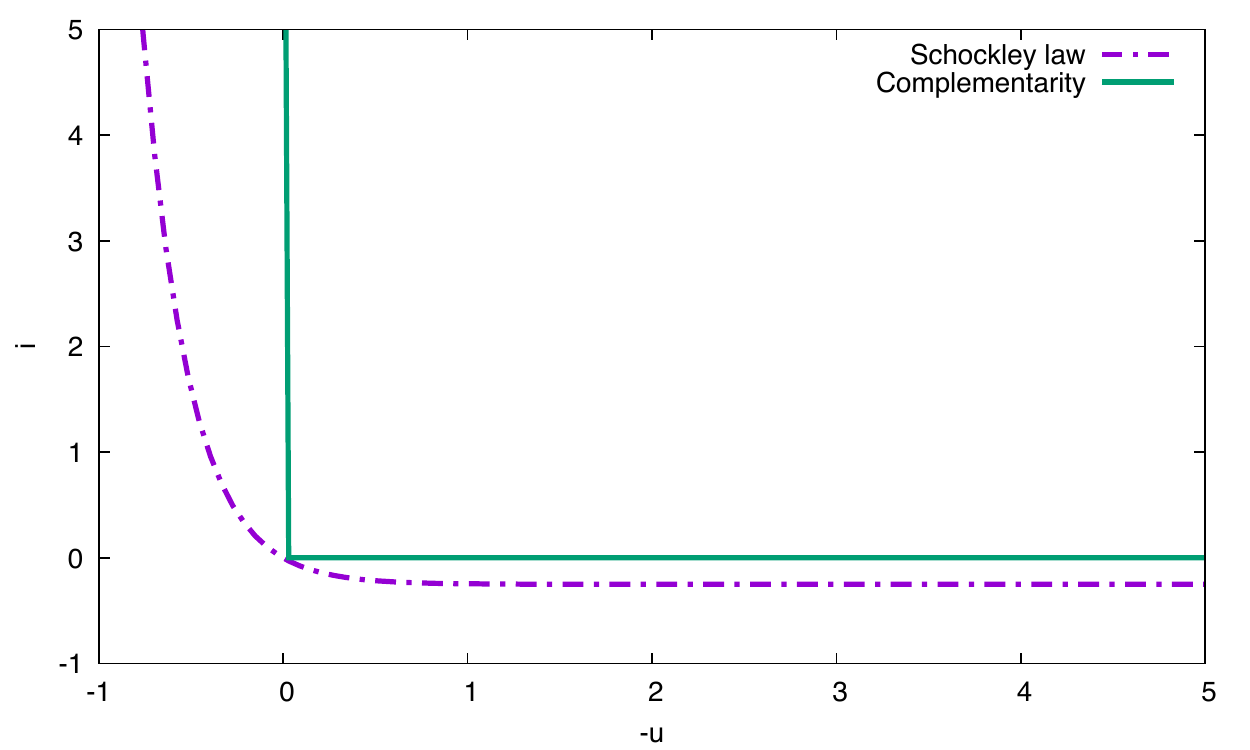}
  \end{center}
  \caption{Shockley law (dashed, {\color{magenta}magenta}) vs. ideal complementarity condition (solid, {\color{green}green}).}
  \label{fig:idealdiode}
\end{figure}

As in Section~\ref{leriugfehliu}, we redefine the graph
of this complementarity condition as a parametric curve,
represented by the following three equations:
\[
  \begin{array}{rclc}
    s &=& \text{if } \guard \text{ then } i \text{ else } -u & (S) \\
    0 &=& \text{if } \guard \text{ then } u \text{ else } i & (Z) \\
    \guard &=& (s \geq 0) & (G)
    \end{array}
  \]
  To avoid a logico-numerical fixpoint, we 
	consider instead the system
  obtained by replacing $s$ with its left-limit $s^-$ in equation $(G)$:
  \[
  \begin{array}{rcll}
    s &=& \text{if } \guard \text{ then } i \text{ else } -u & (S) \\
    0 &=& \text{if } \guard \text{ then } u \text{ else } i & (Z) \\
    \guard &=& (s^- \geq 0) & (G^-)
    \end{array}
  \]
  Under the assumption that $u$ and $i$ are continuous functions of
  time (which turns out to be a valid assumption for the RLDC2
  circuit), $s(t) = s^-(t)$ holds at every instant $t$. This implies
  that systems $(S), (Z), (G)$ and $(S), (Z), (G^-)$ are
  equivalent. 
  Using this encoding of the complementarity condition, the RLDC2
  circuit is modeled as the mDAE shown in Figure~\ref{fig:rldc2}, with
  two Boolean variables $\guard_1$ and $\guard_2$ and four multimode equations
  $(S_1)$, $(Z_1)$, $(S_2)$ and $(Z_2)$,
	incident to a set of
  variables that depends on the mode of the system 
	(i.e., the values of
  Boolean variables $\guard_1$ and $\guard_2$). Note that this model belongs to the class introduced in \rref{defn:mDAEmodified}.
\begin{figure}[!ht]
  \begin{equation}
    \begin{array}{rclc}
      0 &=& i_1 + i_2 + j_1 + j_2 & (K_1) \\
      x_1 + w_1 &=& u_1 + v_1 & (K_2) \\
      u_1 + v_1 &=& u_2 + v_2 & (K_3) \\
      u_2 + v_2 &=& x_2 + w_2 & (K_4) \\
      w_1 &=& L_1 \cdot j'_1 & (L_1) \\
      w_2 &=& L_2 \cdot j'_2 & (L_2) \\
      i_1 &=& C_1 \cdot v'_1 & (C_1) \\
      i_2 &=& C_2 \cdot v'_2 & (C_2) \\
      x_1 &=& R_1 \cdot j_1 & (R_1) \\
      x_2 &=& R_2 \cdot j_2 & (R_2) \\
      s_1 &=& \text{if } \guard_1 \text{ then } i_1 \text{ else } -u_1 & (S_1) \\
      s_2 &=& \text{if } \guard_2 \text{ then } i_2 \text{ else } -u_2 & (S_2) \\
      0 &=& \text{if } \guard_1 \text{ then } u_1 \text{ else } i_1 & (Z_1) \\
      0 &=& \text{if } \guard_2 \text{ then } u_2 \text{ else } i_2 & (Z_2) \\
      \guard_1 &=& (s_1^- \geq 0) & (\guard_1^-) \\
      \guard_2 &=& (s_2^- \geq 0) & (\guard_2^-)
    \end{array}
    \label{eq:rldc2-system}
  \end{equation}
  \caption{The RLDC2 model (\emph{n.b.}: variable $y'$ denotes the time derivative of $y$).}\label{fig:rldc2}
\end{figure}

\subsection{Handling the RLDC2 model with Modelica tools}
\label{sub:rldc2-modelica}

The Modelica model of the RLDC2 circuit is given in
Figure~\ref{fig:rldc2-modelica}; it is a direct translation of the
model above. The authors had the opportunity to test it on two
implementations of the Modelica language: OpenModelica
v1.12\footnote{\url{https://openmodelica.org}} and Dymola
2019.\footnote{\url{https://www.3ds.com/products-services/catia/products/dymola/}}
\begin{figure}[!p]
  \begin{center}
    \includegraphics[width=0.7\linewidth]{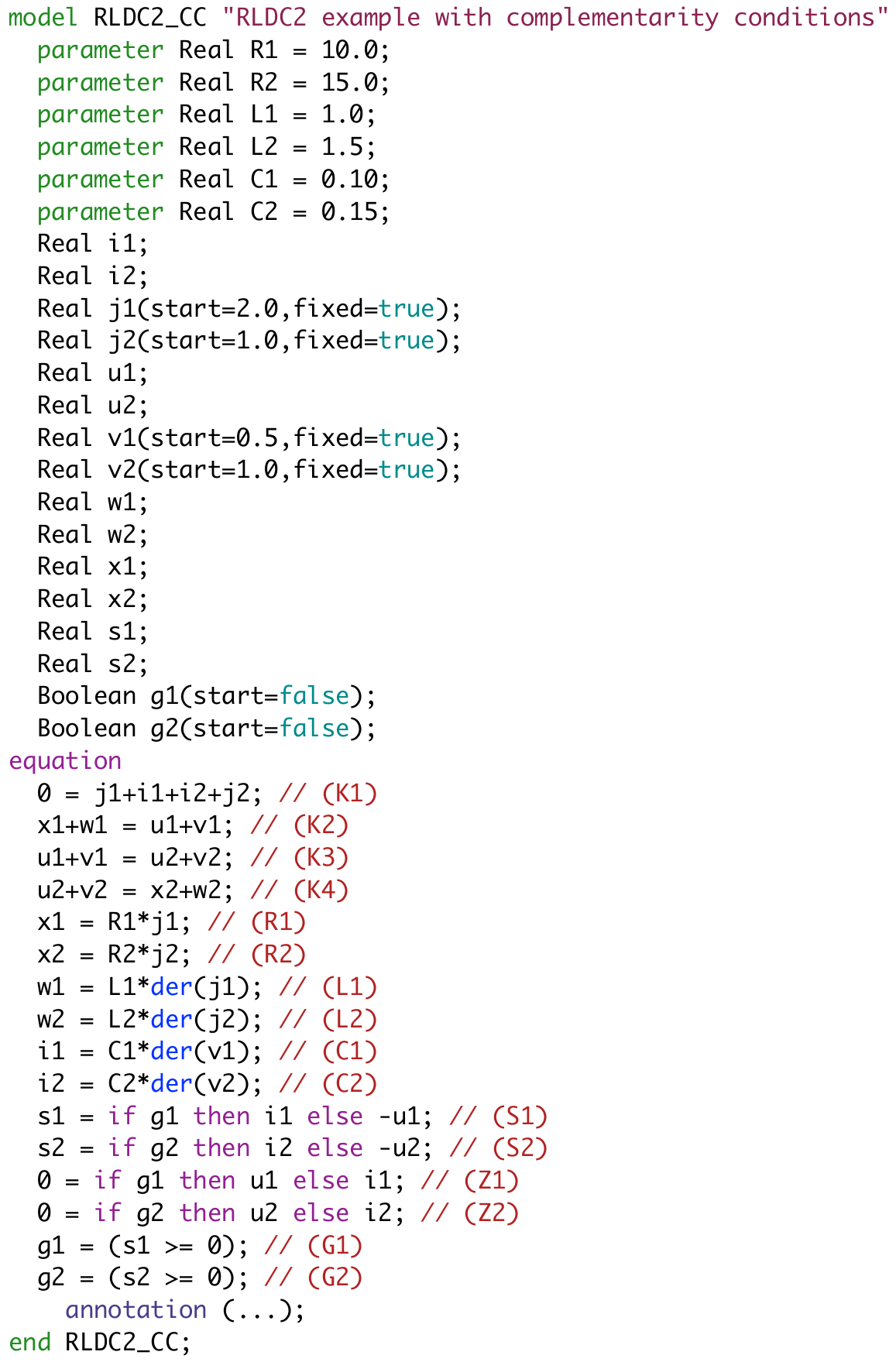}
    \end{center}
  \caption{Modelica model of the RLDC2 circuit.}\label{fig:rldc2-modelica}
  \end{figure}
    
  Identical results are obtained with both tools: the model is deemed
  nonsingular at compile time, yet a runtime error immediately occurs
  (at time $t=0$) during simulation, as shown in
  Figure~\ref{fig:dymola-omc}.
    
    \begin{figure}[ht]
      \begin{tabular}{cc}
    (a) &
          \includegraphics[width=0.6\linewidth]{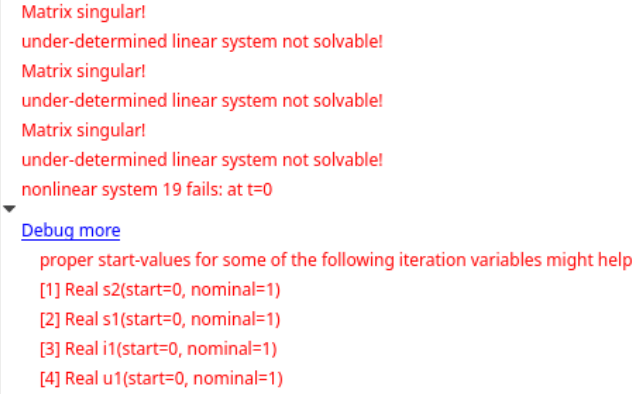} 
					\\
        \vspace{1em} & \\
    (b) &
        \includegraphics[width=0.8\linewidth]{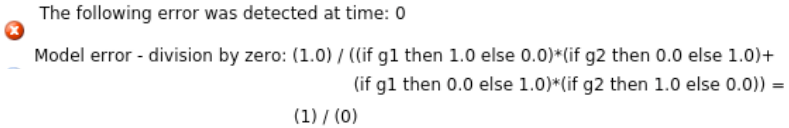}
        \end{tabular}
    \caption{Error messages produced during simulation of the RLDC2 model with (a) OpenModelica and (b) Dymola, both at the initial time of the simulation.}\label{fig:dymola-omc}
    \end{figure}
    
    The issue is that the structural analysis implemented in these
    tools treats the model as a single-mode DAE, disregarding the
    mode-dependent variability of the incidence relations. The
    consequence is that, despite the model being deemed structurally
    nonsingular by these tools, blocks of equations that are
    structurally singular in some modes are produced. This immediately
    leads to runtime errors: both tools attempt a pivoting of the
    Jacobian matrix by an element that is equal to zero in some modes.

\subsection{Structural analysis of the RLDC2 model}
\label{ssec:cbdg}
\begin{figure}[ht]
\centering
\includegraphics[width=0.45\linewidth]{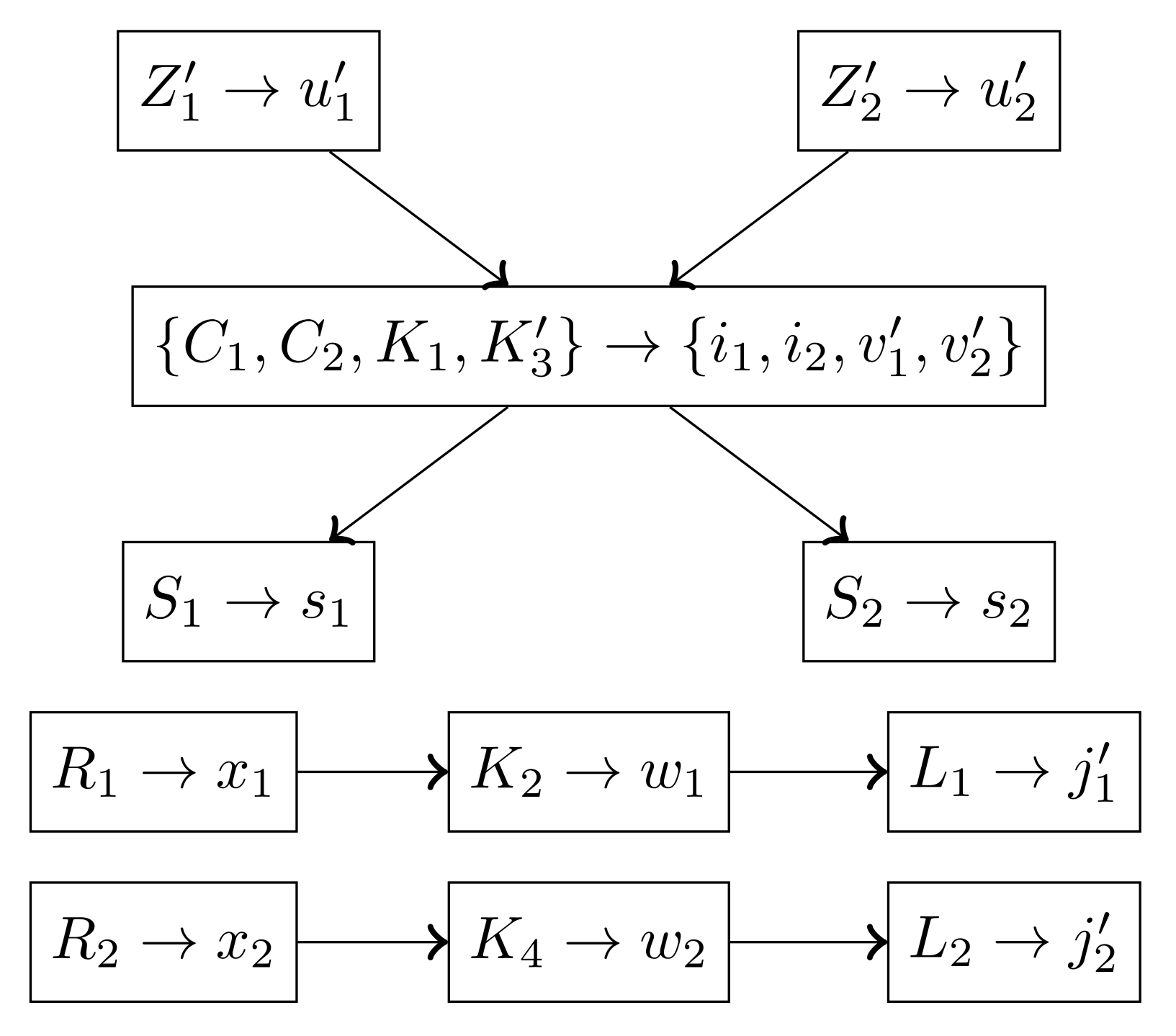}
\caption{Block dependency graph for the RLDC2 example with both diodes
  passing (i.e., $\guard_1 = \guard_2 = \ttrue$).}
\label{fig:circuit-dag}
\centering
\includegraphics[width=0.55\linewidth]{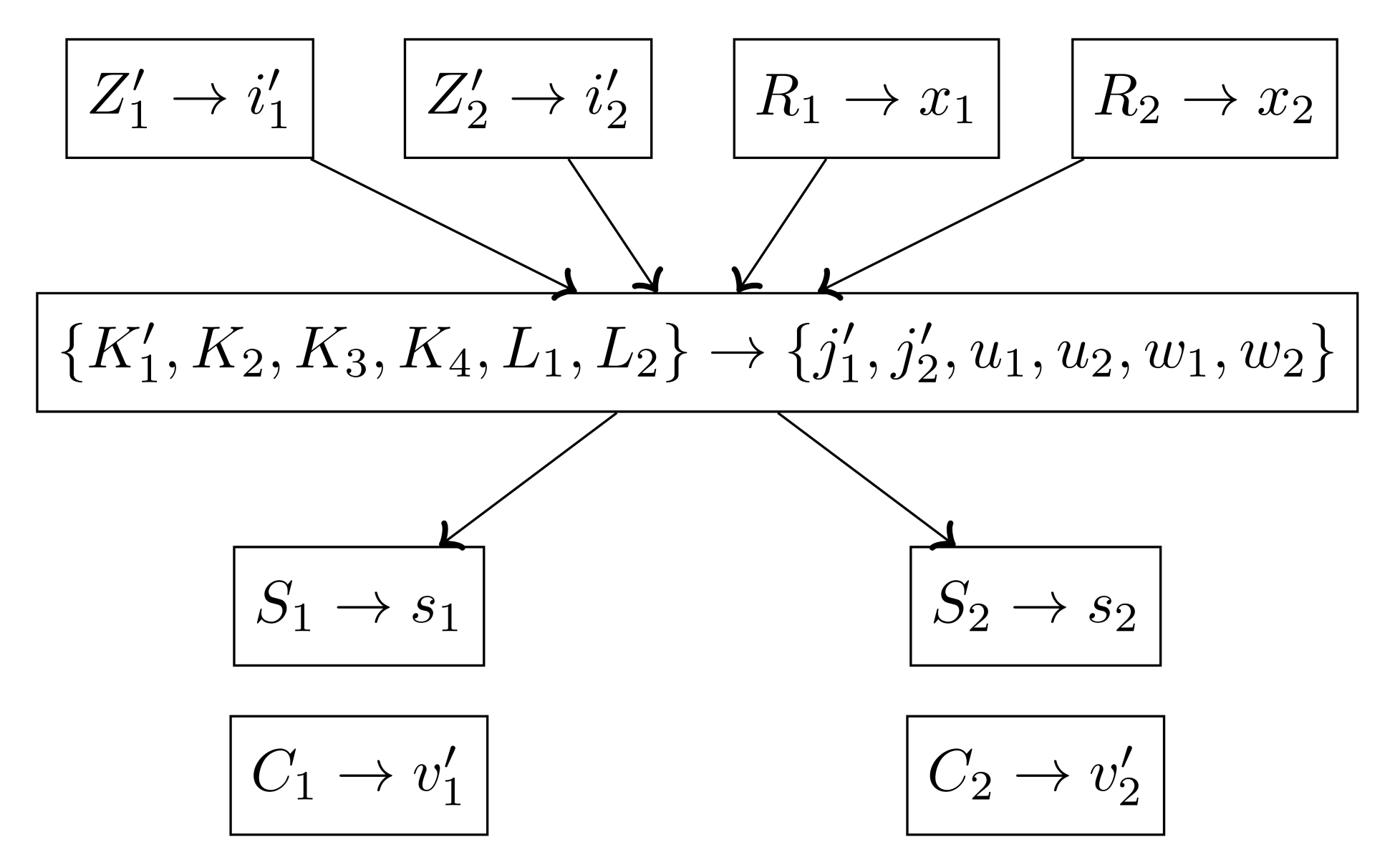}
        \caption{Block dependency graph for the RLDC2 example with both diodes blocking (i.e., $\guard_1 = \guard_2 = \ffalse$).}
        \label{fig:circuit-dag-open}
\end{figure}
The \mDAE\ system of \rref{fig:rldc2} has two modes of nonzero index, when diodes are either both passing or both open. Having performed index reduction on each of them, we show the corresponding block dependency graphs in Figures~\ref{fig:circuit-dag} and~\ref{fig:circuit-dag-open}, respectively---see Lemma~\ref{jytdfeoguip} and \rref{fig:BTF}. These two graphs significantly differ, highlighting the fact that a mode-dependent structural analysis is necessary.

  \begin{figure*}[!htp]
    \begin{center}
    \includegraphics[width=0.9\linewidth]{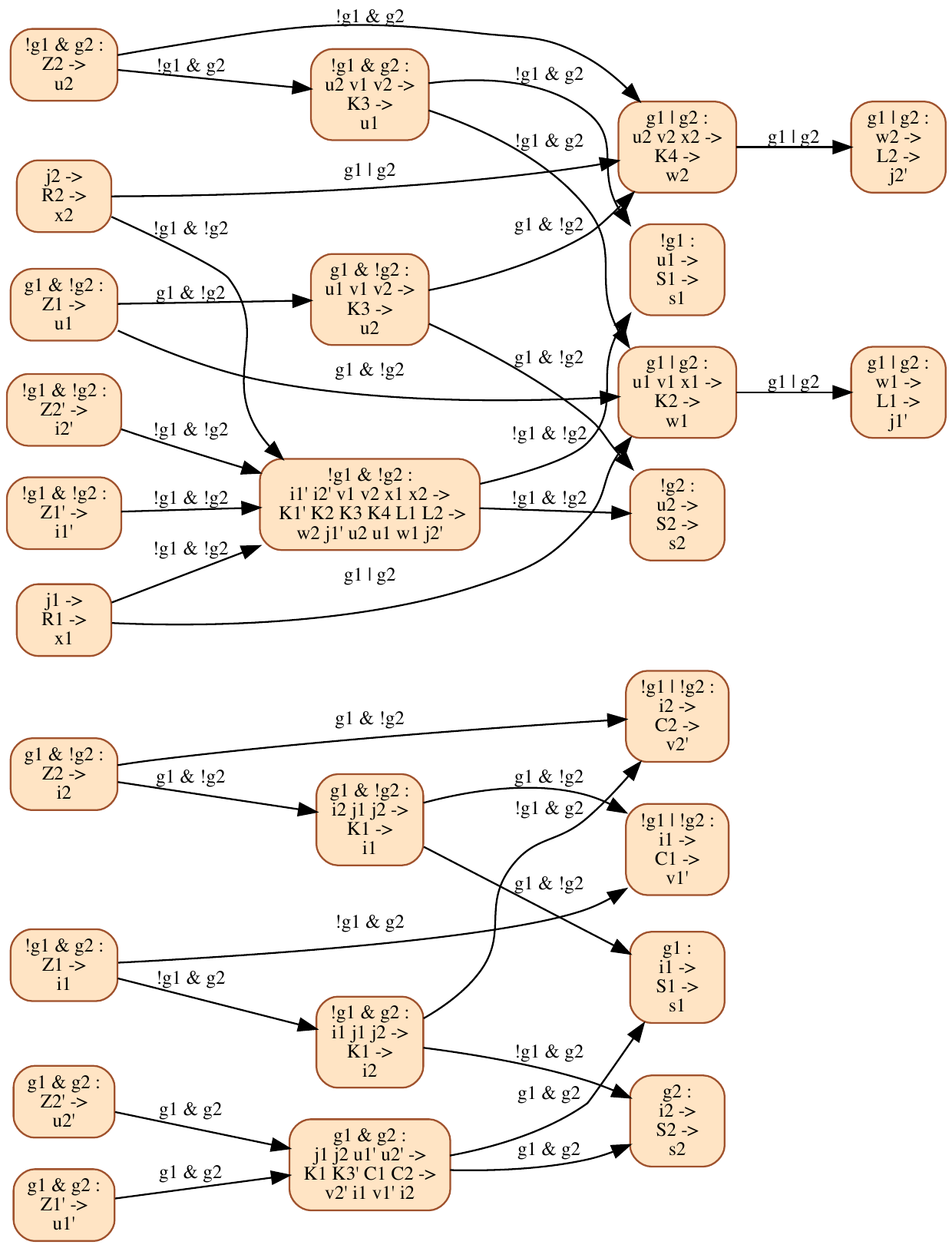}
    \end{center}
    \caption{Block dependency graph of the RLDC2 model, generated
    by IsamDAE. Vertices are labeled
    $p : R \rightarrow B \rightarrow W$, where: $p$
    is a propositional formula defining in which modes the block is
    evaluated; $R$ is the set of variables to read; $B$ is the set of
    equations of the block; $W$ is the set of variables to write. Edges
    are labeled by a propositional formula, defining in which modes the
    dependency applies---notation ``!g'' means ``not g''.}\label{fig:rldc2:cdg}
  \end{figure*}
  
Figure\,\ref{fig:rldc2:cdg} shows a graph generated
by the structural analysis performed by our software IsamDAE~\cite{Caillaud2020a,caillaud:hal-02476541}, for the \mDAE\ system of \rref{fig:rldc2}. This graph does not consist of one graph per mode, but, rather, of a set of edges labeled by a propositional formula characterizing the set of modes in which the considered branch is involved. This data structure prevents the combinatorial explosion arising in \mDAE\ systems composed of a large number of \mDAE\ subsystems, when modes are enumerated.

  The block triangular form of the system can then be found at runtime for any given
  mode, for instance by evaluating the block dependency graph in this
  mode and performing a topological sort on the resulting graph.
  For the RLDC2 example, when evaluated in the mode in which both diodes are passing (resp. blocking), the graph of Figure~\ref{fig:rldc2:cdg} yields the dependency graph of Figure~\ref{fig:circuit-dag} (resp. Figure~\ref{fig:circuit-dag-open}); the same holds in the remaining modes.
	All equation
  blocks can be turned into efficient simulation code at compile time,
  so that the computational overhead due to a mode switching at
  runtime is minimized.

\subsection{Perspectives}

As shown on the RLDC2 example above, our tool IsamDAE addresses the issue of performing the structural analysis of all modes of an \mDAE\ while avoiding any explicit enumeration of these modes.
Although the example detailed above is, for clarity purposes, a small one, IsamDAE is provably able to handle models with a very large number of modes, as detailed in~\cite{Caillaud2020a,caillaud:hal-02476541}.

As such, the tool already provides the data structures and key algorithms required for addressing the second and third challenges invoked at the beginning of this section.
This encourages us to implement, in the same tool, the structural analysis of mode changes, in the same `all-at-once' fashion as for the modes themselves.
This is indeed the main perspective of our work, as we believe it will be a crucial step towards the design of a mathematically sound compiler for multimode DAE systems.
\section{Conclusion}
We have proposed sound foundations for the compilation and code
generation for physical systems modeling languages relying on DAE for
their continuous dynamics; major instances are Modelica and
Simscape. As its main contribution, our theory explicitly considers
events of mode change and properly handles any multimode system
having at most finite cascades of mode changes---sliding modes are not
supported by our method.

A key step was the extension of {structural analysis} from
single-mode to multimode systems. In this extension, the handling of
mode changes was the main difficulty. We illustrated our approach on
small examples that are not properly supported by
existing tools (with the exception of
\href{https://modiasim.github.io/ModiaMath.jl/stable/man/Overview.html}{ModiaMath}). These
examples exhibit mode-dependent index, which is not well supported by
existing techniques.

For the clutch example, the most natural model consists in
specifying the dynamics for the two modes ``released'' and
``engaged''; then, the compiler automatically generates the restart conditions at mode
changes.
The ``Westinghouse air brake'' example motivated the consideration of assertions---not studied here but included in our tool under development. 
The Cup-and-Ball example is interesting in that, when the rope
gets straight, inelastic or elastic impact can be hypothesized. This
example pointed the issue of long vs. transient modes (modes having
zero duration). Depending on whether the straight rope mode was
assumed long (inelastic impact) or transient (elastic impact), the
original model was accepted or rejected. For the latter,
underspecification was detected and this information was returned to
the user, who could refine his/her model by specifying an impact law. The
so enriched model was accepted and code properly generated at mode
change.
 
Generally, our method allows us to reject models on the basis of
under/over-specification. Logico-numerical fixpoint equations
(cases in which numerical variables determine the value of a Boolean condition that, in turn, guards the equation allowing to evaluate the same
numerical variables) are prohibited by setting restrictions on the
language syntax.
 
The structure of the code we generate is reminiscent of the so-called
dummy derivatives method~\cite{MattssonSoderlin1993}: for both the
continuous dynamics and the handling of transient modes or mode change
events, we call an algebraic equation system solver for computing
derivatives or next state values, as a function of the current state,
while meeting algebraic constraints.

A first requirement on our theory was that it should support the
analysis of any model related to any kind of physics---we aim at being
physics agnostic. This support should hold regardless of the
particular properties of the model, such as conditions
ensuring existence and uniqueness of solutions. The motivation for
this was that the user is not expected to check such conditions before
submitting her/his model.

A second requirement was the ability to handle
the conflicts that may result, at mode changes, between the dynamics
of the previous mode and the consistency conditions implied by the new
mode.

We do not see how these two requirements could be addressed without
relying on nonstandard analysis. Mapping real time to a discrete
time with infinitesimal step size allowed us to cast both (DAE-based) continuous
dynamics, and the restart conditions at mode changes, to a
uniform framework. This allowed us to properly solve the conflicts mentioned above, between the dynamics of the previous mode and the
consistency conditions implied by the new mode.
To make all of this doable, we developed a toolbox of new results on
nonstandard analysis, mainly related to the needs of structural
analysis. We proved that the result of our compilation remains
independent from the particular scheme we used to map derivatives to
the nonstandard domain.

From our experience in using nonstandard analysis, we can state that
this is an extremely convenient tool to establish the mathematical
foundations for the compilation of continuous-time systems modeling
tools, for both the Simulink and Modelica classes.

A desirable objective of a work like ours is to prove that every
execution scheme produced by the approach actually generates
solutions of the given source \mDAE\ system. As a reference, such
soundness analyses were performed for all synchronous
languages~\cite{Berry96,BenvenisteCG00,BenvenisteCEHGS03}. However, a particular difficulty holds: no mathematical definition of
solutions of an mDAE system exists in full generality, that would serve
as a reference. So far, we were only able to prove such a correctness
result for the subclass of ``semi-linear multimode systems'', that
may involve impulsive behaviors at mode changes.

In this work, we paid no specific attention to the efficiency of our
structural analysis algorithm, and did not discuss how to generate scalable simulation code.  We also only considered flat guards, whereas nested
guards (naturally occurring, for instance, in the case of nested ``if-then-else'' constructs) 
 are very useful in a practical
formalism: flattening a model with nested guards is not desirable for
getting efficient code. Work is in progress to address these issues by
relying on the technique of Binary Decision Diagrams (BDD), and first
results are reported in~\cite{Caillaud2020a}, where our IsamDAE tool under development is described.

\subsubsection*{Acknowledgements}
The authors wish to thank a number of people for this work. First, Hilding Elmqvist and Martin Otter introduced them to the subject in the early 2010's; our collaboration then established for a couple of years, leading to both a first version of our approach and the handling of semi-linear mDAE systems; since then, Hilding and Martin have focused on their grand child of Modelica called Modia. We had extensive fruitful discussions with John Pryce on both structural analysis of DAEs, including the {\sigmamethod\}, and physical modeling in general; this helped us bring out the need for an explicit distinction between long and transient modes. Khalil Ghorbal participated to the first version of this approach and further interacted with the authors for this revised approach. Finally, Vincent Acary introduced us to nonsmooth systems techniques and solvers and we had frequent exchanges regarding the notions of solution for mDAE systems, revealing to us how strange these can be.

\clearpage
\appendix
\section{Appendix: constructive execution of an instant in systems having only long modes}
\label{eriughepuih}
In this appendix, we briefly develop the execution of a run, for \mDAE\ systems according to \rref{defn:mDAEattempt}, i.e., without delaying the effect of guards by one nonstandard instant. Since the effect of guards may not be postponed, the current mode is not known from evaluating the previous instant. We still refuse to solve mixed logico-numerical fixpoint equations. Hence we must instead identify the current mode progressively and constructively, according to an iterative scheme of the following form:
\begin{enumerate}
	\item Enter the new nonstandard instant; the status of a subset of the \ScottVars\ is known from previous instant;
	\item \label{urhgilohsdoug} Using the subset of evaluated guards, select the corresponding set of enabled equations and submit them to index reduction, conflict solving, and then evaluate them;
	\item \label{eouirfghpui9} This yields more enabled guards; evaluate them and return to Step~\ref{urhgilohsdoug}.
\end{enumerate}
This iteration is performed until, either all \ScottVars\ are evaluated to $\ttt,\fff$, or $\irrelevant$, or no additional guard gets enabled at Step~\ref{eouirfghpui9}---a failure. Since knowing if the current mode is long or transient modifies the way index reduction is performed, and since knowing this requires having evaluated all the guards, we are unable to process general \mDAE\ systems, and we must restrict ourselves to systems having only long modes.

To execute a nonstandard instant through a sequence of microsteps, the algorithm $\atomicact{ExecRunProgress}$ (\rref{alg:prognewmain}) is iterated until a failure is reported or a $\atomicact{Tick}$ atomic action occurs. In case of success, the pair $(\status,\Delta)$ produced by $\atomicact{Tick}$ produces the initial status and the context for starting the next instant.
\begin{algorithm}[ht] 
\caption{$\atomicact{ExecRunProgress}$; compare with \rref{alg:newmain}
}\label{alg:prognewmain}
\begin{algorithmic}[1]
   \Require $(\status,\Delta)$; 
   \Return  $(\mathit{Fail},\mbox{updated}\;(\status,\Delta))$
   \If{$\mathit{Success}(\status)$} \label{op:prognewtestsuccess}
      \State $(\status,\Delta) \gets \atomicact{Tick}(\status)$ \label{op:prognewtick} 
      \State $\Phi\gets\{\eqq\in\Delta\mid\status(\guard(e))=\ttt\}$ \label{op:prognewinitfacts} 
   \Else
      \State $G \gets [\enbld{\eqq}{\status} \cap \ndef{\status}]$   \label{op:prognewenabledeqs}
      \State $(\status,b,G_\Sigma,\consistency{G}_\Sigma)\gets\atomicact{IndexReduc}(G,\status)$; increase $\status$
      \label{op:progporihujrpoi}
      \If{$b=\fff$} \Return\emph{Fail}$(\status)$ \label{op:progeogfuiehoriufhoui}
      \Else \label{op:progleifuerhfuiohoui}
      \State $(\status,b,H)\gets\atomicact{SolveConflict}((G_\Sigma\cup\consistency{G}_\Sigma)\setminus\Phi,\status)$; increase $\status$
      \label{op:progelkfuiehoeuio}
      \If{b=\fff} \label{op:proglpruightrlu}
      \Return\emph{Fail}$(\status)$ 
      \Else\ $(\status,\Phi)\gets\atomicact{Eval}(\Delta,\Phi,\status,H)$; increase $\status$
       \label{op:progweorpifuhfpeiu} 
      \EndIf
      \EndIf
      \EndIf
\end{algorithmic}
\end{algorithm}
$\atomicact{ExecRunProgress}$ requires a finite coherent status $\status$ and a finite context $\Delta$; it returns updated $\status$ and $\Delta$, or documented failure information. Its details are commented next.
\begin{description}
\item[{\rref{op:prognewtestsuccess}}] Function
$\mathit{Success}(\status)$ checks if status $\status$ is
{\ssuccessful} according to \rref{defn:success}.

\item[{\rref{op:prognewtick}}] If $\status$ is {\ssuccessful}, then $\atomicact{Tick}$ (\rref{alg:rltoghiuiop}) is performed and the instant is completed.

\item[{\rref{op:prognewinitfacts}}] Not all the equations belonging to the context are active in the current instant. A system $\Phi$ collecting all the \emph{facts} is introduced and is initialized to the subset of the context $\Delta$ consisting of all its enabled equations, taking into account the guards that are already known to be true at initialization of the run.

\item[{\rref{op:prognewenabledeqs}}]
The system $G$ is set to the enabled guarded
equations not evaluated yet, in status $\sigma$. 

\item[{\rref{op:progporihujrpoi}}] We submit $G$ to the index reduction ($\atomicact{IndexReduc}$ atomic action, \rref{alg:lkciyughsliu}). This returns,
\begin{itemize}
	\item  either $b=\fff$ (\rref{op:progeogfuiehoriufhoui}), in which case  $\atomicact{ExecRun}$ returns $\mathit{Fail}(\status)$ and stops, or 
	\item (\rref{op:progleifuerhfuiohoui}) a pair consisting of $G_\Sigma$, a structurally nonsingular system determining some leading variables, and the system $\consistency{G}_\Sigma$ collecting consistency equations. In this case, the algorithm can further progress.
\end{itemize}

\item[\rref{op:progelkfuiehoeuio}] With the additional help of the set $\Phi$ of facts, the possible conflict is solved using the $\atomicact{SolveConflict}$ atomic action (\rref{alg:elsgouihuip}).

\item[\rref{op:proglpruightrlu}] $\atomicact{SolveConflict}$ fails if no structurally regular system survives as a result of applying \rref{op:progelkfuiehoeuio}. In this case, $\atomicact{ExecRun}$ returns $\mathit{Fail}(\status)$ and stops.

\item[{\rref{op:progweorpifuhfpeiu}}] Otherwise, $\atomicact{SolveConflict}$ returns a structurally regular system $H$, which is solved using $\atomicact{Eval}$ (\ref{toihjiurh}), which updates the status $\status$ by giving a value to more \ScottVars. As a result, the set $\Phi$ of facts is updated.
\end{description}
%
\clearpage
\bibliographystyle{abbrv}
\bibliography{hybrid,modelica,synchrone,lucy}

\end{document}